\documentclass[11pt]{article}

\usepackage[colorlinks, linkcolor=blue, citecolor=blue]{hyperref}            
\usepackage{color}
\usepackage{graphicx,subfigure,amsmath,amssymb,amsfonts,bm,epsfig,epsf,url,dsfont}
\usepackage{amsthm}
\usepackage{tikz}
\usepackage{bbm}      
\usepackage{booktabs}
\usepackage{cases}
\usepackage{fullpage}
\usepackage[small,bf]{caption}
\usepackage{natbib}
\usepackage[top=1in,bottom=1in,left=1in,right=1in]{geometry}
\usepackage{fancybox}

\newcommand{\bR}{\mathbb{R}}
\newcommand{\mL}{\mathcal{L}}
\newcommand{\mO}{\mathcal{O}}
\newcommand{\mG}{\mathcal{G}}

\newcommand{\mP}{\mathcal{P}}

\newcommand{\TV}{d_{\text{TV}}}
\newcommand{\KL}{d_{\text{KL}}}

\newcommand{\bP}{\mathbb{P}}
\newcommand{\bE}{\mathbb{E}}

\newtheorem{theorem}{Theorem}[section]

\newtheorem{conjecture}{Conjecture}[section]
\newtheorem{lemma}[theorem]{Lemma}

\newenvironment{fminipage}%
  {\begin{Sbox}\begin{minipage}}%
  {\end{minipage}\end{Sbox}\fbox{\TheSbox}}

\newenvironment{algbox}[0]{\vskip 0.2in
\noindent 
\begin{fminipage}{6.3in}
}{
\end{fminipage}
\vskip 0.2in
}

\begin{document}

\title{Reducibility and Computational Lower Bounds for Problems with Planted Sparse Structure}

\author{Matthew Brennan\thanks{Massachusetts Institute of Technology. Department of EECS. Email: \texttt{brennanm@mit.edu}.}
\and 
Guy Bresler\thanks{Massachusetts Institute of Technology. Department of EECS. Email: \texttt{guy@mit.edu}.}
\and
Wasim Huleihel\thanks{Massachusetts Institute of Technology. Research Laboratory of Electronics. Email: \texttt{wasimh@mit.edu}.}}
\date{\today}

\maketitle

\begin{abstract}
Recently, research in unsupervised learning has gravitated towards exploring statistical-computational gaps induced by sparsity. A line of work initiated in \cite{berthet2013complexity} has aimed to explain these gaps through reductions from conjecturally hard problems in complexity theory. However, the delicate nature of average-case reductions has limited the development of techniques and often led to weaker hardness results that only apply to algorithms robust to different noise distributions or that do not need to know the parameters of the problem. We introduce several new techniques to give a web of average-case reductions showing strong computational lower bounds based on the planted clique conjecture. Our new lower bounds include:
\begin{itemize}
\item \textbf{Planted Independent Set:} We show tight lower bounds for detecting a planted independent set of size $k$ in a sparse Erd\H{o}s-R\'{e}nyi graph of size $n$ with edge density $\tilde{\Theta}(n^{-\alpha})$.
\item \textbf{Planted Dense Subgraph:} If $p > q$ are the edge densities inside and outside of the community, we show the first lower bounds for the general regime $q = \tilde{\Theta}(n^{-\alpha})$ and $p - q = \tilde{\Theta}(n^{-\gamma})$ where $\gamma \ge \alpha$, matching the lower bounds predicted in \cite{chen2016statistical}. Our lower bounds apply to a deterministic community size $k$, resolving a question raised in \cite{hajek2015computational}.
\item \textbf{Biclustering:} We show lower bounds for the canonical simple hypothesis testing formulation of Gaussian biclustering, slightly strengthening the result in \cite{ma2015computational}.
\item \textbf{Sparse Rank-1 Submatrix:} We show that detection in the sparse spiked Wigner model is often harder than biclustering, and are able to obtain tight lower bounds for these two problems with different reductions from planted clique.
\item \textbf{Sparse PCA:} We give a reduction between sparse rank-1 submatrix and sparse PCA to obtain tight lower bounds in the less sparse regime $k \gg \sqrt{n}$, when the spectral algorithm is optimal over the natural SDP. We give an alternate reduction recovering the lower bounds of \cite{berthet2013complexity, gao2017sparse} in the simple hypothesis testing variant of sparse PCA. We also observe a subtlety in the complexity of sparse PCA that arises when the planted vector is biased.
\item \textbf{Subgraph Stochastic Block Model:} We introduce a model where two small communities are planted in an Erd\H{o}s-R\'{e}nyi graph of the same average edge density and give tight lower bounds yielding different hard regimes than planted dense subgraph.
\end{itemize}

Our results demonstrate that, despite the delicate nature of average-case reductions, using natural problems as intermediates can often be beneficial, as is the case in worst-case complexity. Our main technical contribution is to introduce a set of techniques for average-case reductions that: (1) maintain the level of signal in an instance of a problem; (2) alter its planted structure; and (3) map two initial high-dimensional distributions simultaneously to two target distributions approximately under total variation. We also give algorithms matching our lower bounds and identify the information-theoretic limits of the models we consider.
\end{abstract}

\pagebreak

\tableofcontents

\pagebreak

\section{Introduction}

The field of statistics is undergoing a dramatic conceptual shift, with computation moving from the periphery to center stage. Prompted by the demands of modern data analysis, researchers realized two decades ago that a new approach to estimation was needed for high-dimensional problems in which the dimensionality of the data is at least as large as the sample size. High-dimensional problems are inherently underdetermined, often precluding nontrivial rates of estimation. However, this issue typically disappears if the underlying signal is known to have an appropriate structure, such as low rank or sparsity. As a result, high-dimensional structured estimation problems have received significant attention in both the statistics and computer science communities. Prominent examples include estimating a sparse vector from linear observations, sparse phase retrieval, low-rank matrix estimation, community detection, subgraph and matrix recovery problems, random constraint satisfiability, sparse principal component analysis and covariance matrix estimation. Although structural assumptions can yield nontrivial estimation rates, the statistically optimal estimators for these problems typically entail an exhaustive search over the set of possible structures and are thus not efficiently computable. Conversely, all known efficient algorithms for these problems are statistically suboptimal, requiring more data than strictly necessary. This phenomenon has led to a number of conjectured statistical-computational gaps for high-dimensional problems with structure. This raises an intriguing question: how are these gaps related to one another and are they emerging for a common reason?

In the last few years, several lines of work have emerged to make rigorous the notion of what is and what is not achievable statistically by efficient algorithms. In the seminal work of \cite{berthet2013complexity}, a conjectured statistical-computational gap for sparse principal component analysis (PCA) was shown to follow from the planted clique conjecture. This marked the first result basing the hardness of a natural statistics problem on an average-case hardness assumption and produced a framework for showing statistical-computational gaps by approximately mapping in total variation. This subsequently led to several more reductions from the planted clique conjecture to show statistical-computational gaps for problems including submatrix detection/biclustering \cite{ma2015computational}, submatrix localization \cite{cai2015computational}, planted dense subgraph \cite{hajek2015computational}, RIP certification \cite{wang2016average}, sparse PCA and sparse canonical correlation analysis \cite{wang2016statistical, gao2017sparse}. We draw heavily from the framework for average-case reductions laid out in these papers. More recently, focus has shifted to showing unconditional hardness results for restricted models of computation and classes of algorithms. An exciting line of work has emerged surrounding applications of the Sum of Squares (SOS) semidefinite programming hierarchy to problems with statistical-computational gaps. SOS Lower bounds have been shown for planted clique \cite{barak2016nearly} and for sparse PCA \cite{krauthgamer2015semidefinite, ma2015sum, hopkins2017power}. Tight computational lower bounds have also been shown in the statistical query model for planted clique and planted random $k$-SAT \cite{feldman2012statistical, feldman2015complexity}.

One reason behind this focus on showing hardness in restricted models of computation is that average-case reductions are inherently delicate, creating obstacles to obtaining satisfying hardness results. As described in \cite{Barak2017}, these technical obstacles have left us with an unsatisfying theory of average-case hardness. Reductions in worst-case complexity typically take a \textit{general} instance of a problem $A$ to a \textit{structured} instance of a problem $B$. For example, a classic reduction from $\textsc{3SAT}$ to $\textsc{Independent-Set}$ produces a very specific type of graph with a cluster of seven vertices per clause corresponding to each satisfying assignment such that two vertices are connected if together they yield an inconsistent assignment. If such a reduction were applied to a random $\textsc{3SAT}$ instance, the resulting graph instance would be far from any natural graph distribution. Unlike reductions in worst-case complexity, average-case reductions between natural decision problems need to precisely map the distributions on instances to one another without destroying the underlying signal in polynomial-time. The delicate nature of this task has severely limited the development of techniques and left open reductions between decision problems that seem to be obviously equivalent from the standpoint of algorithm design. For example, it remains unknown whether refuting random constraint satisfaction problems with $10m$ clauses is equivalent to refuting those with $11m$ clauses or whether the planted clique conjecture at edge density $1/2$ implies the conjecture at edge density $0.49$. There are also a variety of negative results further demonstrating new obstacles associated with average-case complexity, including that it likely cannot be based on worst-case complexity \cite{bogdanov2006worst}. For more on average-case complexity, we refer to the survey of \cite{bogdanov2006average}.

In order to overcome these average-case difficulties, prior reductions have often made assumptions on the robustness of the underlying algorithm such as that it succeeds for any noise distributions from a fixed class as in \cite{berthet2013complexity, wang2016statistical, cai2015computational}. This corresponds to composite vs. composite hypothesis testing formulations of detection problems, where the composite null hypothesis $H_0$ consists of the class of noise distributions. Other reductions have shown hardness for precise noise distributions but for algorithms that do not need to exactly know the parameters of the given instance \cite{ma2015computational, gao2017sparse}. This typically corresponds to simple vs. composite hypothesis testing where the composite alternative $H_1$ consists of models defined by varying parameters such as the sparsity $k$ or signal strength. The strongest prior reduction from planted clique is that to the sparsest regime of planted dense subgraph in \cite{hajek2015computational}. A lower bound is shown for a simple vs. simple hypothesis testing variant of the problem, with each consisting of a single distribution. However, the community in their formulation of planted dense subgraph was binomially distributed and therefore still assumed to be unknown exactly to the algorithm. Prior reductions have also shown hardness at particular points in the parameter space, deducing that an algorithm cannot always perform better than a conjectured computational barrier rather than showing that no algorithm can ever perform better. For example, prior reductions for sparse PCA have only shown tight hardness around the single parameter point where the signal is $\theta = \tilde{\Theta}(1)$ and the sparsity is $k = \tilde{\Theta}(\sqrt{n})$. Simplifying parameters in their reductions, both \cite{berthet2013complexity} and \cite{gao2017sparse} approximately map a planted clique instance on $n$ vertices with clique size $k$ to an instance of sparse PCA with $\theta \approx \tilde{\Theta}(k^2/n)$ which is only tight to the conjectured barrier of $\theta^* = \Theta(\sqrt{k^2/n})$ when $k = \tilde{\Theta}(\sqrt{n})$.

These assumptions leave a subtle disparity between the existing average-case lower bounds for many problems and algorithmic upper bounds. Many algorithmic results assume a canonical generative model or implicitly assume knowledge of parameters. For example, even in the recent literature on robust algorithms for problems with sparsity in \cite{balakrishnan2017computationally, li2017robust}, the setup is in the context of specific canonical generating models, such as the spiked covariance model for sparse PCA. Even when corrupted by adversarial noise, the spiked covariance model is far in distribution from many sub-Gaussian formulations of sparse PCA. Despite existing average-case lower bounds, hardness for the canonical generative models for many problems has remained open. This includes biclustering with a flat planted $k \times k$ submatrix selected uniformly at random in Gaussian noise, sparse PCA with a $k$-sparse principal component chosen uniformly at random to have entries equal to $\pm 1/\sqrt{k}$ and planted dense subgraph with deterministic community size.

\subsection{Overview}

\begin{figure*}[t!]
\centering
\begin{tikzpicture}[scale=0.45]
\node at (0, 0) (PC) {Planted Clique};
\node at (-12, -3) (PIS) {Planted Independent Set};
\node at (-2, -3) (SPDS) {Low-Density PDS};
\node at (10, -3) (BC) {Biclustering};
\node at (-2, -6) (PDS) {General PDS};
\node at (5.5, -6) (ROS) {Rank-1 Submatrix};
\node at (14.5, -6) (BSPCA) {Biased Sparse PCA};
\node at (0.5, -9) (SPCA) {Sparse PCA};
\node at (10.5, -9) (SSBM) {Subgraph Stochastic Block Model};

\draw[->] (PC) -- (PIS);
\draw[->] (PC) -- (SPDS);
\draw[->] (PC) -- (BC);
\draw[->] (SPDS) -- (PDS);
\draw[->] (BC) -- (PDS);
\draw[->] (BC) -- (BSPCA);
\draw[->] (BC) -- (ROS);
\draw[->] (ROS) -- (SPCA);
\draw[->] (ROS) -- (SSBM);
\end{tikzpicture}
\caption{Graph of average-case reductions for detection problems showing tight statistical-computational gaps given the planted clique conjecture.}
\end{figure*}
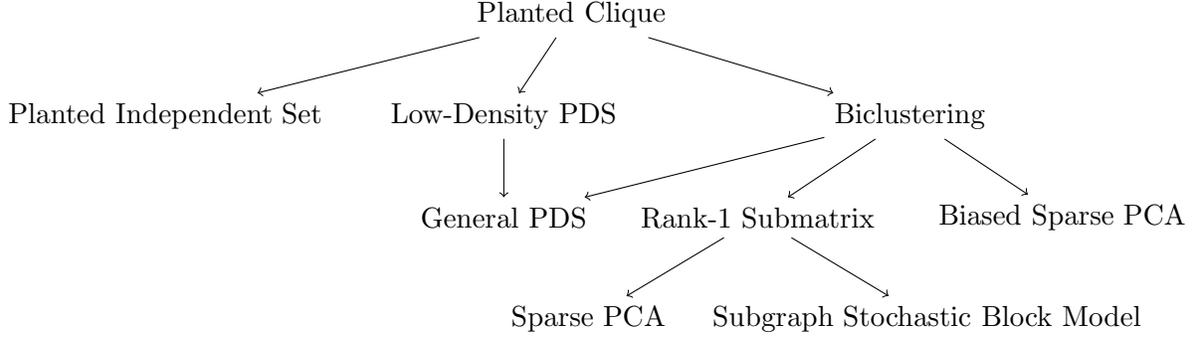

The aim of this paper is threefold: (1) to demonstrate that a web of average-case reductions among problems with statistical-computational gaps is feasible even for showing strong computational lower bounds; (2) to introduce a number of new techniques for average-case reductions between problems; and (3) to fully characterize the computationally hard regime of several models. The graph of our reductions is shown in Figure 1.1. Our new lower bounds are as follows.
\begin{itemize}
\item \textbf{Planted Independent Set:} We show tight lower bounds for detecting a planted independent set of size $k$ in a sparse Erd\H{o}s-R\'{e}nyi graph of size $n$ with edge density $\tilde{\Theta}(n^{-\alpha})$.
\item \textbf{Planted Dense Subgraph:} If $p > q$ are the edge densities inside and outside of the community, we show the first lower bounds for the general regime $q = \tilde{\Theta}(n^{-\alpha})$ and $p - q = \tilde{\Theta}(n^{-\gamma})$ where $\gamma \ge \alpha$, matching the lower bounds predicted in \cite{chen2016statistical}. Our lower bounds apply to a deterministic community size $k$, resolving a question raised in \cite{hajek2015computational}.
\item \textbf{Biclustering:} We show lower bounds for Gaussian biclustering as a simple hypothesis testing problem to detect a uniformly at random planted flat $k \times k$ submatrix. Our alternative reduction matches the barriers in \cite{ma2015computational}, where a computational lower bound was shown for a composite hypothesis testing variant of biclustering. We show hardness for the natural simple hypothesis testing problem where the $k \times k$ submatrix is chosen uniformly at random and has equal entries.
\item \textbf{Sparse Rank-1 Submatrix:} We show that detection in the sparse spiked Wigner model has a different computational threshold from biclustering when $k \gg \sqrt{n}$. Surprisingly, we are able to obtain tight lower bounds matching these different detection thresholds with different reductions from planted clique.
\item \textbf{Sparse PCA:} We give a reduction between rank-1 submatrix and sparse PCA to obtain tight lower bounds in the less sparse regime $k \gg \sqrt{n}$, when the spectral algorithm is optimal over the SDP. This yields the first tight characterization of a computational barrier for sparse PCA over an entire parameter regime. We also give an alternate reduction recovering the lower bounds of \cite{berthet2013complexity} and \cite{gao2017sparse} in the canonical simple hypothesis testing variant of sparse PCA.
\item \textbf{Biased Sparse PCA:} We show that any assumption on the sparse principal component having a constant fraction more or fewer positive entries than negative entries yields a detection-recovery gap that is not present in sparse PCA. 
\item \textbf{Subgraph Stochastic Block Model:} We introduce a model where two small communities are planted in an Erd\H{o}s-R\'{e}nyi graph of the same average edge density. Parallel to the difference between biclustering and sparse rank-1 submatrix when $k \gg \sqrt{n}$, we show that detection in this model is much harder than in planted dense subgraph when $k \gg \sqrt{n}$.
\end{itemize}
Our lower bounds for planted independent set, the general regime of planted dense subgraph, sparse rank-1 submatrix, sparse PCA when $k \gg \sqrt{n}$, biased sparse PCA and the subgraph stochastic block model are novel. As previously mentioned, lower bounds for sparse PCA when $k \ll \sqrt{n}$, for biclustering and for planted dense subgraph in the sparsest regime were previously known. In each of these cases, we strengthen the existing lower bounds to the apply to the canonical generative model. We show computational lower bounds for simple vs. simple hypothesis testing in all cases other than for sparse PCA, rank-1 submatrix and the subgraph stochastic block model all in the regime $k \gg \sqrt{n}$. This is a consequence of our underlying reduction technique, reflection cloning, and appears unavoidable given our methods. However, we do show that the distribution we reduce to is in some sense close to the canonical generative model.

Our results demonstrate that, despite the delicate nature of average-case reductions, using natural problems as intermediates can often be beneficial as in reductions in worst-case complexity. Our main technical contribution is to introduce several techniques for mapping problems approximately in total variation without degrading the underlying planted sparse structure. These techniques are:
\begin{itemize}
\item \textbf{Distributional Lifting:} A variant of graph lifts that iteratively maps to intermediate matrices with entries from chosen distributions. Varying the underlying distribution produces different relationships between the resulting edge density and size of a planted subgraph. This is in introduced in Sections 4, 5 and 6.
\item \textbf{Rejection Kernels:} A general framework for a change in measure of the underlying noise distribution while preserving the planted sparse structure. This unifies ideas introduced in \cite{hajek2015computational}, \cite{ma2015computational} and \cite{gao2017sparse}. This is introduced in Section 5.
\item \textbf{Reflection Cloning:} A method of increasing the sparsity of a planted rank-1 structure in noise while preserving the level of signal in the planted structure. This is introduced in Section 7.
\item \textbf{Random Rotations for Sparse PCA:} An average-case connection between the sparse spiked Wigner model and the spiked covariance model of sparse PCA. This is introduced in Section 8.
\end{itemize}
We consider two main variants of distributional lifting using matrices with Poisson and Gaussian entries as intermediates. Poisson and Gaussian lifting lead to two very different parameter scalings, which when combined fully characterize general planted dense subgraph. Reflection cloning can be viewed as a more randomness-efficient variant of Gaussian lifting that we use to show sharper lower bounds for the sparse spiked Wigner model and sparse PCA. We also give algorithms matching our lower bounds and identify the information-theoretic limits of the problems that we consider in Section 9.


\subsection{Hardness Results from an Algorithmic Perspective}

\begin{figure*}[t!]
\centering
\begin{tikzpicture}[scale=0.45]
\tikzstyle{every node}=[font=\footnotesize]
\def\xmin{0}
\def\xmax{11}
\def\ymin{-1}
\def\ymax{11}

\draw[->] (\xmin,\ymin) -- (\xmax,\ymin) node[right] {$\beta$};
\draw[->] (\xmin,\ymin) -- (\xmin,\ymax) node[above] {$\alpha$};

\node at (5, -1) [below] {$\frac{1}{2}$};
\node at (6.66, -1) [below] {$\frac{2}{3}$};
\node at (10, -1) [below] {$1$};
\node at (0, 0) [left] {$0$};
\node at (0, 10) [left] {$1$};
\node at (0, 5) [left] {$\frac{1}{2}$};
\node at (0, 1.67) [left] {$\frac{1}{6}$};

\filldraw[fill=cyan, draw=blue] (0, 5) -- (6.66, 1.66) -- (5, 0) -- (0, 5);
\filldraw[fill=gray!25, draw=gray] (0, 5) -- (5, 0) -- (10, 0) -- (10, -1) -- (0, -1) -- (0, 5);
\filldraw[fill=green!25, draw=green] (0, 5) -- (6.66, 1.66) -- (10, 5) -- (10, 10) -- (0, 10) -- (0, 5);
\filldraw[fill=orange, draw=red] (5, 0) -- (6.66, 1.66) -- (10, 0) -- (5, 0);
\filldraw[fill=red, draw=magenta] (6.66, 1.66) -- (10, 5) -- (10, 0) -- (6.66, 1.66);
\draw (5, 0) -- (0, 10);

\node at (2, 1.9) {SDP};
\node at (1, 5.3) {$k$-eig};
\node at (8, -0.5) {spectral};
\node at (7.8, 5.3) {sum if biased};
\node at (4.5, 8) {previous reductions};
\end{tikzpicture}
\caption{Algorithms for sparse PCA with $d = \Theta(n)$, $k = \tilde{\Theta}(n^\beta)$ and $\theta = \tilde{\Theta}(n^{-\alpha})$. Lines represent the strongest guarantees of each algorithm. The line marked as previous reductions shows the strongest previously known planted clique lower bounds for sparse PCA when $k \lesssim \sqrt{n}$. No planted clique lower bounds were known for $k \gg \sqrt{n}$.}
\label{fig:algsspca}
\end{figure*}
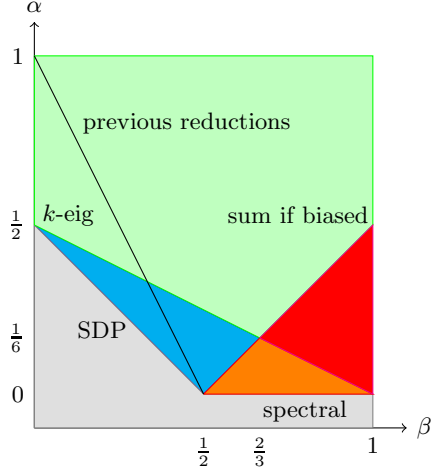

In this section, we motivate our computational lower bounds and techniques using algorithms for sparse PCA as an example. Consider the detection problem for sparse PCA where either $X_1, X_2, \dots, X_n$ are sampled i.i.d. from $N(0, I_d)$ or are sampled i.i.d. from $N(0, I_d + \theta vv^\top)$ for some latent $k$-sparse unit vector $v$ with nonzero entries equal to $\pm 1/\sqrt{k}$. The task is to detect which of the two distributions the samples originated from. For now assume that $d = \Theta(n)$. Now consider the following four algorithms:
\begin{enumerate}
\item \textbf{Semidefinite Programming:} Form the empirical covariance matrix $\hat{\Sigma} = \frac{1}{n} \sum_{i = 1}^n X_i X_i^\top$ and solve the convex program
\begin{align*}
\max_Z \quad &\text{Tr}\left(\hat{\Sigma} Z\right) \\
\text{s.t.} \quad &\text{Tr}(Z) = 1, |Z|_1 \le k, Z \succeq 0
\end{align*}
As shown in \cite{berthet2013complexity}, thresholding the resulting maximum solves the detection problem as long as $\theta = \tilde{\Omega}(\sqrt{k^2/n})$.
\item \textbf{Spectral Algorithm:} Threshold the maximum eigenvalue of $\hat{\Sigma}$. If the data are sampled from $N(0, I_d)$, then the largest eigenvalue is with high probability at most
$$\lambda_{\text{max}}(\hat{\Sigma}) \le \frac{d}{n} + \sqrt{\frac{d}{n}} + 1 + o(1)$$
by standard bounds on the singular values of random Gaussian matrices. Since $d = \Theta(n)$, this algorithm succeeds as long as $\theta = \Omega(1)$. This algorithm was considered in \cite{krauthgamer2015semidefinite}.
\item \textbf{Sum Test:} Sum the entries of $\hat{\Sigma}$ and threshold the absolute value of the sum. If $v$ has sum exactly zero, then this test will not succeed. However, if we assume that $\hat{\Sigma}$ has even 51\% of its nonzero entries of one sign, then this test succeeds if $\theta = \tilde{\Omega}(\sqrt{n}/k)$.
\item \textbf{$k$-Sparse Eigenvalue:} Compute and threshold the $k$-sparse unit vector $u$ that maximizes $u^\top \hat{\Sigma} u$. This can be found by finding the largest eigenvector of each $k \times k$ principal submatrix of $\hat{\Sigma}$. Note that this takes exponential time. It was shown in \cite{berthet2013complexity} that this succeeds as long as $\theta = \tilde{\Omega}(\sqrt{k/n})$.
\end{enumerate}
The boundaries at which these algorithms begin to succeed are shown in Figure \ref{fig:algsspca} for the regime $k = \tilde{\Theta}(n^\beta)$ and $\theta = \tilde{\Theta}(n^{-\alpha})$. The computational lower bound mapping to $\theta \approx k^2/n$ in \cite{berthet2013complexity} and \cite{gao2017sparse} is also drawn. As shown, the only point in the parameter diagram for which it matches an algorithmic upper bound is $\alpha = 0$ and $\beta = 1/2$, corresponding to when $\theta = \tilde{\Theta}(1)$ and $k = \tilde{\Theta}(\sqrt{n})$.

For sparse PCA with $d = \Theta(n)$, the SDP is the best known algorithm up to $k = \Theta(\sqrt{n})$, at which point the spectral algorithm has stronger guarantees. This algorithmic transition at $k = \Theta(\sqrt{n})$ is characteristic of all of the problems we consider. For the biased variant of sparse PCA where the sum test succeeds, the sum test always does strictly better than the spectral algorithm. Furthermore, the biased variant ceases to have a statistical computational gap around $k = \Theta(n^{2/3})$. While the sum test yields an improved algorithm for detection, unlike the other three algorithms considered above, it does not translate into an algorithm for recovering the support of the sparse component. Given a conjecture about recovery in planted dense subgraph, we show that the best recovery algorithm for biased sparse PCA can only match the guarantees of the spectral algorithm. Thus the biased variant induces a detection-recovery gap when $k \gg \sqrt{n}$. We show that the disappearance of a statistical computation gap at $k = \Theta(n^{2/3})$ and a detection-recovery gap when $k \gg \sqrt{n}$ are features of the problems we consider that admit a sum test. These are biased sparse PCA, planted independent set, planted dense subgraph and biclustering. Distributional lifting gives tight planted clique lower bounds for these problems.

In contrast, rank-1 submatrix, the subgraph stochastic block model and sparse PCA do not admit a sum test. Given the planted clique conjecture, rank-1 submatrix and sparse PCA have no detection-recovery gap and retain their statistical-computational gap for all sparsities $k$. Reflection cloning shows tight lower bounds for these problems in the regime $k \gg \sqrt{n}$, where spectral algorithms become optimal. It is surprising that the planted clique conjecture can tightly capture completely different sets of computational barriers for different problems, illustrating its power as an average-case hardness assumption. Although analogues of the sum test, spectral algorithms and semidefinite programs all have equivalent guarantees up to logarithmic factors for planted clique, reductions from planted clique show tight hardness in problems for which this is not true.

\subsection{Prior Work}

This work is part of a growing body of literature giving rigorous evidence for computational-statistical gaps in high-dimensional inference problems. We focus on average-case reductions to directly relate computational-statistical gaps in different problems, as opposed to giving worst-case evidence for hardness in statistical problems \cite{zhang2014lower, hardt2014computational, chan2016approximability}. A survey of prior results on computational-statistical gaps with a focus on predictions from statistical physics can be found in \cite{bandeira2018notes} and a general analysis of gaps for algorithms from several convex relaxation hierarchies can be found in \cite{chandrasekaran2013computational}.

\paragraph{Planted Clique and Independent Set.} Our computational lower bounds are based on average-case reductions from the problem of finding a planted clique of size $k$ in an Erd\H{o}s-R\'{e}nyi graph with $n$ vertices. The planted clique problem was introduced in \cite{alon1998finding}, where a spectral algorithm was shown to recover the planted clique if $k = \Omega(\sqrt{n})$. A number of algorithms for planted clique have since been studied, including approximate message passing, semidefinite programming, nuclear norm minimization and several combinatorial approaches \cite{feige2000finding, mcsherry2001spectral, feige2010finding, ames2011nuclear, dekel2014finding, deshpande2015finding, chen2016statistical}. All of these algorithms require that $k = \Omega(\sqrt{n})$, which has led to the planted clique conjecture that no polynomial time algorithm can recover the planted clique if $k = o(\sqrt{n})$. It was also shown in \cite{alon2007testing} that recovering and detecting the planted clique are equivalent up to $\log n$ factors in $k$. There have been a number of previous average-case reductions from the planted clique conjecture, which we discuss in more detail in the prior work section on average-case reductions.

Several works have considered finding independent sets in sparse Erd\H{o}s-R\'{e}nyi graphs, similar to the regime with edge density $q = \tilde{\Theta}(n^{-\alpha})$ where $\alpha \in (0, 2)$ we consider here. In \cite{coja2015independent, gamarnik2014limits, rahman2017local}, the authors examine greedy and local algorithms to find independent sets in the regime $q = \tilde{\Theta}(n^{-1})$ in random regular, Erd\H{o}s-R\'{e}nyi and other random graphs. In \cite{feige2005finding}, a spectral algorithm is given to find a planted independent set in the regime $q = \tilde{\Theta}(n^{-1})$ and in \cite{coja2003finding}, the planted independent set recovery problem is shown to be possible in polynomial time in the regime $\alpha \in (0, 1)$ when $q \gg \frac{n}{k^2}$ even in a semirandom model. The algorithms of \cite{chen2016statistical} also apply to recovering planted independent sets after taking the complement of the input graph.

\paragraph{Planted Dense Subgraph and Community Detection.} The planted dense subgraph detection problem was considered in \cite{arias2014community, butucea2013detection, verzelen2015community, hajek2015computational} and generalizations of the recovery problem were considered in \cite{chen2016statistical, hajek2016information, montanari2015finding, candogan2018finding}. In \cite{hajek2015computational}, a reduction from planted clique was given for the regime $p = cq$ for some constant $c > 1$ and $q = \tilde{\Theta}(n^{-\alpha})$ and $k$ is binomially distributed, where $k$, $n$, $p$ and $q$ are the size of the community, size of the graph, community edge density and graph edge density, respectively. Our results strengthen this lower bound to apply for deterministic $k$ and for all $p > q$ with $p - q = O(q)$ where $q = \tilde{\Theta}(n^{-\alpha})$. When $p = \omega(q)$, the resulting regime is the planted dense subgraph problem considered in \cite{bhaskara2010detecting}. The computational barrier for this problem is conjectured to be the log-density threshold $k = \tilde{\Theta}(n^{\log_q p})$ when $k \ll \sqrt{n}$ and is achieved by very different algorithms than those that are optimal when $p = O(q)$ \cite{chlamtac2012everywhere, chlamtavc2017minimizing}. Recently, it was shown in \cite{chlamtac2018sherali} that $\tilde{\Omega}(\log n)$ rounds of the Sherali-Adams hierarchy cannot solve the planted dense subgraph detection problem below the log-density threshold in the regime $p = \omega(q)$. Hardness below the log-density threshold has been used as an average-case assumption in several reductions, as outlined in the prior work section on average-case reductions.

Community detection in the stochastic block model has been the focus of an extensive body of literature surveyed in \cite{abbe2017community}. It recently has been shown that the two-community stochastic block model does not exhibit statistical-computational gaps for partial and exact recovery, which are possible when the edge density scales like $\Theta(n^{-1})$ \cite{mossel2012stochastic, mossel2013proof, massoulie2014community} and $\Theta(n^{-1}\log n)$ \cite{mossel2014consistency, hajek2016achieving, abbe2016exact}, respectively. In contrast, the subgraph variant of the two-community stochastic block model that we introduce has computational-statistical gaps for partial recovery, exact recovery and detection, given the planted clique conjecture. The $k$-block stochastic block model is also conjectured to have statistical-computational gaps starting at $k \ge 4$ \cite{abbe2015detection}.

\paragraph{Biclustering and the Spiked Wigner Model.} Gaussian biclustering was considered as a detection problem in \cite{butucea2013detection, ma2015computational, montanari2015limitation} and as a recovery problem in \cite{shabalin2009finding, kolar2011minimax, balakrishnan2011statistical, cai2015computational, chen2016statistical, hajek2016information}. In \cite{ma2015computational}, a reduction from planted clique was given for a simple vs. composite hypothesis testing variant of the biclustering detection problem, where the size and mean entries of the planted submatrix were allowed to vary. In \cite{cai2015computational}, submatrix localization with sub-Gaussian noise was shown to be hard assuming a variant of the planted clique conjecture for regular graphs.

A large body of literature has studied the spectrum of the spiked Wigner model \cite{peche2006largest, feral2007largest, capitaine2009largest, benaych2011eigenvalues}. Spectral algorithms and information-theoretic lower bounds for the spiked Wigner model detection and recovery problems were considered in \cite{montanari2015limitation, perry2016statistical, perry2016optimality}. The sparse spiked Wigner model where the sparsity $k$ of the planted vector satisfies $k = \Theta(n)$ was studied in \cite{perry2016statistical, perry2016optimality, banks2018information}. The sparse spiked Wigner model with $k = \tilde{\Theta}(n^{\beta})$ for some $\beta \in (0, 1)$ was considered in \cite{hopkins2017power}, where the authors showed sum of squares lower bounds matching our planted clique reductions.

\paragraph{Sparse PCA.} Since its introduction in \cite{johnstoneSparse04}, sparse principal component analysis has been studied broadly in the statistics and computer science communities. A number of algorithms solving sparse PCA under the spiked covariance model have been proposed \cite{amini2009high, ma2013sparse, cai2013sparse, berthet2013optimal, berthet2013complexity, shen2013consistency, krauthgamer2015semidefinite, deshpande2014sparse, wang2016statistical}. The information-theoretic limits for detection and recovery in the spiked covariance model have also been examined extensively \cite{amini2009high, vu2012minimax, berthet2013optimal, birnbaum2013minimax, cai2013sparse, wang2016statistical, cai2015optimal}. The computational limits of sparse PCA problems have also been considered in the literature. Degree four SOS lower bounds for the spiked covariance model were shown in \cite{ma2015sum}. In \cite{berthet2013complexity}, the authors give a reduction from planted clique to a sub-Gaussian composite vs. composite hypothesis testing formulation of sparse PCA as a detection problem, and \cite{berthet2013optimal} gives a reduction from planted clique showing hardness for semidefinite programs. In \cite{gao2017sparse}, the authors give a reduction from planted clique to a simple vs. composite hypothesis testing formulation of detection in the spiked covariance model matching our reduction when $k \ll \sqrt{n}$. In \cite{wang2016statistical}, the authors give a reduction from planted clique to a sub-Gaussian variant of the sparse PCA recovery problem. As mentioned in the introduction, these planted clique lower bounds do not match the conjectured algorithmic upper bounds when $k$ differs in a polynomial factor from $\sqrt{n}$.

\paragraph{Average-Case Reductions.} While the theory of worst-case complexity has flourished to the point that many natural problems are now known to be NP-hard or even NP-complete, the theory of average-case complexity is far less developed. In the seminal work of \cite{levin1986average}, it was shown that an average-case complete problem exists. However, no natural problem with a natural distribution on inputs has yet been shown to be average-case complete. As mentioned in this section, there are obfuscations to basing average-case complexity on worst-case complexity \cite{bogdanov2006worst}. For more on the theory of average-case complexity, see Section 18 of \cite{arora2009computational} and \cite{bogdanov2006average}.

As previously mentioned, there have been a number of average-case reductions from planted clique to average-case problems in both the computer science and statistics literature. These include reductions to testing $k$-wise independence \cite{alon2007testing}, biclustering detection and recovery \cite{ma2015computational, cai2015computational, caiwu2018}, planted dense subgraph \cite{hajek2015computational}, RIP certification \cite{wang2016average, koiran2014hidden}, matrix completion \cite{chen2015incoherence}, minimum circuit size and minimum Kolmogorov time-bounded complexity \cite{hirahara2017average} and sparse PCA \cite{berthet2013optimal, berthet2013complexity, wang2016statistical, gao2017sparse}. The planted clique conjecture has also been used as a hardness assumption for average-case reductions in cryptography \cite{juels2000hiding, applebaum2010public}, as described in Sections 2.1 and 6 of \cite{Barak2017}. There have also been a number of average-case reductions from planted clique to show worst-case lower bounds such as hardness of approximation. Planted clique has been used to show worst-case hardness of approximating densest $k$-subgraph \cite{alon2011inapproximability}, finding approximate Nash equilibria \cite{minder2009small, hazan2011hard, austrin2013inapproximability}, signalling \cite{dughmi2014hardness, bhaskar2016hardness}, approximating the minmax value of 3-player games \cite{eickmeyer2012approximating}, aggregating pairwise comparison data \cite{shah2016feeling} and finding endogenously formed communities \cite{balcan2013finding}.

A number of average-case reductions in the literature have started with different average-case assumptions than the planted clique conjecture. Variants of planted dense subgraph have been used to show hardness in a model of financial derivatives under asymmetric information \cite{arora2011computational}, link prediction \cite{baldin2018optimal}, finding dense common subgraphs \cite{charikar2018finding} and online local learning of the size of a label set \cite{awasthi2015label}. Hardness conjectures for random constraint satisfaction problems have been used to show hardness in improper learning complexity \cite{daniely2014average}, learning DNFs \cite{daniely2016complexity} and hardness of approximation \cite{feige2002relations}. There has also been a recent reduction from a hypergraph variant of the planted clique conjecture to tensor PCA \cite{zhang2017tensor}.

\paragraph{Lower Bounds for Classes of Algorithms.} As described in the introduction, recently there has been a focus on showing unconditional hardness results for restricted models of computation and classes of algorithms. In \cite{jerrum1992large}, it was shown that the Metropolis process cannot find large cliques in samples from planted clique. The fundamental limits of spectral algorithms for biclustering and low-rank planted matrix problems were examined in \cite{montanari2015limitation}. Integrality gaps for SDPs solving sparse PCA, planted dense subgraph and submatrix localization were shown in \cite{krauthgamer2015semidefinite} and \cite{chen2016statistical}. SOS lower bounds have been shown for a variety of average-case problems, including planted clique \cite{deshpande2015improved, raghavendra2015tight, hopkins2016integrality, barak2016nearly}, sparse PCA \cite{ma2015sum}, sparse spiked Wigner and tensor PCA \cite{hopkins2017power}, maximizing random tensors on the sphere \cite{bhattiprolu2017sum} and random CSPs \cite{kothari2017sum}. Lower bounds for relaxations of planted clique and maximum independent set in the Lov\'{a}sz-Schrijver hierarchy are shown in \cite{feige2003probable} and lower bounds for Sherali-Adams relaxations of planted dense subgraph in the log-density regime are shown in \cite{chlamtac2018sherali}. Tight lower bounds have been shown in the statistical query model for planted clique \cite{feldman2013statistical}, random CSPs \cite{feldman2015complexity} and robust sparse mean estimation \cite{diakonikolas2016statistical}. It also has been recently shown that planted clique with $k \ll \sqrt{n}$ is hard for regular resolution \cite{atserias2018clique}. In \cite{hopkins2017efficient}, a meta-algorithm for Bayesian estimation based on low-degree polynomials, SDPs and tensor decompositions is introduced and shown to achieve the best known upper bound for the $k$-block stochastic block model, with a matching lower bound for the meta-algorithm.


\subsection{Notation}

In this paper, we adopt the following notational conventions. Let $\mL(X)$ denote the distribution law of a random variable $X$. Given a distribution $\mathbb{P}$, let $\mathbb{P}^{\otimes n}$ denote the distribution of $(X_1, X_2, \dots, X_n)$ where the $X_i$ are i.i.d. according to $\mathbb{P}$. Similarly, let $\mathbb{P}^{\otimes m \times n}$ denote the distribution on $\mathbb{R}^{m \times n}$ with i.i.d. entries distributed as $\mathbb{P}$. Given a finite or measurable set $\mathcal{X}$, let $\text{Unif}[\mathcal{X}]$ denote the uniform distribution on $\mathcal{X}$. Let $\TV$, $\KL$ and $\chi^2$ denote total variation distance, Kullback-Leibler divergence and $\chi^2$ divergence, respectively. Given a measurable set $\mathcal{X}$, let $\Delta(\mathcal{X})$ denote the set of all distributions $\pi$ on $\mathcal{X}$. If $\mathcal{X}$ is itself a set of distributions, we refer to $\Delta(\mathcal{X})$ as the set of priors on $\mathcal{X}$. Throughout the paper, $C$ refers to any constant independent of the parameters of the problem at hand and will be reused for different constants.

Let $N(\mu, \sigma^2)$ denote a normal random variable with mean $\mu$ and variance $\sigma^2$ when $\mu \in \mathbb{R}$ and $\sigma \in \mathbb{R}_{\ge 0}$. Let $N(\mu, \Sigma)$ denote a multivariate normal random vector with mean $\mu \in \mathbb{R}^d$ and covariance matrix $\Sigma$, where $\Sigma$ is a $d \times d$ positive semidefinite matrix. Let $\beta(x, y)$ denote a beta distribution with parameters $x, y > 0$ and let $\chi^2(k)$ denote a $\chi^2$-distribution with $k$ degrees of freedom. Let $\mathcal{B}_0(k)$ denote the set of all unit vectors $v \in \mathbb{R}^d$ with $\| v \|_0 \le k$. Let $[n] = \{1, 2, \dots, n\}$ and $\binom{[n]}{k}$ denote the set of all size $k$ subsets of $[n]$. Let $\mG_n$ denote the set of all simple graphs on vertex set $[n]$. Let the Orthogonal group on $\bR^{d \times d}$ be $\mO_d$. Let $\mathbf{1}_S$ denote the vector $v \in \mathbb{R}^n$ with $v_i = 1$ if $i \in S$ and $v_i = 0$ if $i \not \in S$ where $S \subseteq [n]$. For subsets $S \subseteq \mathbb{R}$, let $\mathbf{1}_S$ denote the indicator function of the set $S$. Let $\Phi$ denote the cumulative distribution of a standard normal random variable with $\Phi(x) = \int_{-\infty}^x e^{-t^2/2} dt$. Given a simple undirected graph $G$, let $V(G)$ and $E(G)$ denote its vertex and edge sets, respectively. The notation $a(n) \gg b(n)$ will denote $a$ growing polynomially faster in $n$ than $b$. In other words, $a \gg b$ if $\liminf_{n \to \infty} \log_n a(n) > \limsup_{n \to \infty} \log_n b(n)$. The notation $a = \tilde{\Theta}(b)$ denotes the equality $\lim_{n \to \infty} \log_n a(n) = \lim_{n \to \infty} \log_n b(n)$. Here, $a \lesssim b$ denotes $a(n) = O(b(n) \cdot \text{poly}(\log n))$ or in other words $a \le b$ up to polylogarithmic factors in $n$.

\section{Summary of Results}

\subsection{Detection and Recovery Problems}

We consider problems $\mP$ with planted sparse structure as both detection and recovery tasks, which we denote by $\mP_D$ and $\mP_R$, respectively.

\paragraph{Detection.} In detection problems $\mP_D$, the algorithm is given a set of observations and tasked with distinguishing between two hypotheses:
\begin{itemize}
\item a \emph{uniform} hypothesis $H_0$, under which observations are generated from the natural noise distribution for the problem; and
\item a \emph{planted} hypothesis $H_1$, under which observations are generated from the same noise distribution but with a latent planted sparse structure.
\end{itemize}
In all of the detection problems we consider, $H_0$ is a simple hypothesis consisting of a single distribution and $H_1$ is either also simple or a composite hypothesis consisting of several distributions. When $H_1$ is a composite hypothesis, it consists of a set of distributions of the form $P_\theta$ where $\theta$ is the latent sparse structure of interest. Often $H_1$ is a simple hypothesis consisting of a single distribution which is a mixture of $P_\theta$ with $\theta$ in some sense chosen uniformly at random. In both cases, we will abuse notation and refer to $H_1$ as a set of distributions. Given an observation $X$, an algorithm $A(X) \in \{0, 1\}$ \emph{solves} the detection problem \emph{with nontrivial probability} if there is an $\epsilon > 0$ such that its Type I$+$II error satisfies that
$$\limsup_{n \to \infty} \left( \bP_{H_0}[A(X) = 1] + \sup_{P \in H_1} \bP_{X \sim P}[A(X) = 0] \right) \le 1 - \epsilon$$
where $n$ is the parameter indicating the size of $X$. We refer to this quantity as the asymptotic Type I$+$II error of $A$ for the problem $\mP_D$. If the asymptotic Type I$+$II error of $A$ is zero, then we say $A$ \emph{solves} the detection problem $\mP_D$. Our reductions under total variation all yield exact correspondences between asymptotic Type I$+$II errors. Specifically, they show that if a polynomial time algorithm has asymptotic Type I$+$II error of $\epsilon$ on the problem of interest then there is a polynomial time algorithm with asymptotic Type I$+$II error $\epsilon$ on the problem being reduced from.

\paragraph{Recovery.} In recovery problems $\mP_R$, the algorithm is given an observation from $P_\theta$ for some latent $\theta$ from a space $\Theta$ and the task is to recover the support $S(\theta)$ of the sparse structure $\theta$. There are several variants of the recovery task. Given a randomized algorithm with output $A(X) \in \{ S(\theta) : \theta \in \Theta \}$ and a distribution $\pi$ on the latent space $\Theta$, the variants of the recovery task are as follows.
\begin{itemize}
\item \textit{Partial Recovery}: $A$ solves partial recovery if
$$\bE_{X \sim \bE_\pi P_\theta}[|A(X) \cap S(\theta)|] = \Omega(|S(\theta)|) \quad \text{as } n \to \infty$$
\item \textit{Weak Recovery}: $A$ solves weak recovery if
$$\bE_{X \sim \bE_\pi P_\theta}[|A(X) \Delta S(\theta)|] = o(|S(\theta)|) \quad \text{as } n \to \infty$$
\item \textit{Exact Recovery}: $A$ solves exact recovery with nontrivial probability $\epsilon > 0$ if for all $\theta \in \Theta$
$$\liminf_{n \to \infty} \bP_{X \sim \bE_\pi P_\theta} \left[ A(X) = S(\theta) \right] \ge \epsilon$$
\end{itemize}
Here, $\bE_\pi P_\theta$ denotes the mixture of $P_\theta$ induced by $\pi$ and $\Delta$ denotes the symmetric difference between two sets. Whenever the corresponding detection problem $\mP_D$ has a simple hypothesis $H_1$, $\pi$ will be the prior on $\Theta$ as in $H_1$, which typically is a uniform prior. When $\mP_D$ has a composite hypothesis $H_1$, then an algorithm $A$ solves each of the three variants of the recovery task if the above conditions are met for \emph{all} distributions $\pi$. We remark that this is equivalent to the above conditions being met only for distributions $\pi$ with all of their mass on a single $\theta \in \Theta$. Given a problem $\mP$, the notation $\mP_R$ will denote the exact recovery problem, and $\mP_{PR}$ and $\mP_{WR}$ will denote partial and weak recovery, respectively. All of our recovery reductions will apply to all recovery variants simultaneously. In other words, a partial, weak or exact recovery algorithm for the problem of interest implies the same for the problem being reduced from. The computational and statistical barriers for partial, weak and exact recovery generally differ by sub-polynomial factors for the problems we consider. Therefore the polynomial order of the barriers for $\mP_R$ will in general also be those for $\mP_{PR}$ and $\mP_{WR}$. This is discussed in more detail in Section 9.

An instance of a detection problem $\mP_D$ hereby refers to an observation $X$. If $\mP_D$ is a simple vs. simple hypothesis testing problem, then the instance $X$ takes one of two distributions -- its distribution under $H_0$ and $H_1$, which we respectively denote by $\mL_{H_0}(X)$ and $\mL_{H_1}(X)$. If $H_1$ is composite, then the distribution of $X$ under $P$ is denoted as $\mL_{P}(X)$ for each $P \in H_1$. An instance of a recovery problem $\mP_R$ refers to an observation from some $P_\theta$ for some latent $\theta \in \Theta$ if the corresponding detection problem has a composite $H_1$ or from $\mL_{H_1}(X) = \bE_\pi P_\theta$ if the corresponding detection problem has a simple $H_1$.

\paragraph{Computational Model.} The algorithms we consider here are either unconstrained or run in randomized polynomial time. An unconstrained algorithm refers to any randomized function or Markov transition kernel from one space to another. These algorithms considered in order to show that information-theoretic lower bounds are asymptotically tight. An algorithm that runs in randomized polynomial time has access to $\text{poly}(n)$ independent random bits and must run in $\text{poly}(n)$ time where $n$ is the size of the input. For clarity of exposition, we assume that explicit expressions can be exactly computed and assume that $N(0, 1)$ and Poisson random variables can be sampled in $O(1)$ operations. 

\subsection{Problem Formulations}

In this section, we define the problems that we show computational lower bounds for and the conjectures on which these lower bounds are based. Each problem we consider has a natural parameter $n$, which typically denotes the number of samples or dimension of the data, and sparsity parameter $k$. Every parameter for each problem is implicitly a function of $n$, that grows or decays polynomially in $n$. For example, $k = k(n) = \tilde{\Theta}(n^{\beta})$ for some constant $\beta \in (0, 1)$ throughout the paper. For simplicity of notation, we do not write this dependence on $n$. We mostly will be concerned with the polynomial order of growth of each of the parameters and not with subpolynomial factors. We now formally define the problems we consider.

\paragraph{Planted Clique and Independent Set.} The hypotheses in the planted clique detection problem $\textsc{PC}_D(n, k, p)$ are
$$H_0: G \sim G(n, p) \quad \text{and} \quad H_1 : G \sim G(n, k, p)$$
where $G(n, p)$ is an Erd\H{o}s-R\'{e}nyi random graph with edge probability $p$ and $G(n, k, p)$ is a sample of $G(n, p)$ with a clique of size $k$ planted uniformly at random. As mentioned in the Prior Work section, all known polynomial time algorithms for planted clique fail if $k \ll \sqrt{n}$. This has led to the following hardness conjecture.

\begin{conjecture}[PC Conjecture]
Fix some constant $p \in (0, 1)$. Suppose that $\{ A_n \}$ is a sequence of randomized polynomial time algorithms $A_n : \mG_n \to \{0, 1\}$ and $k_n$ is a sequence of positive integers satisfying that $\limsup_{n \to \infty} \log_n k_n < \frac{1}{2}$. Then if $G$ is an instance of $\textsc{PC}_D(n, k, p)$, it holds that
$$\liminf_{n \to \infty} \left( \bP_{H_0}\left[A_n(G) = 1\right] + \bP_{H_1}\left[A_n(G) = 0\right] \right) \ge 1.$$ 
\end{conjecture}

The hardness assumption we use throughout our results is the planted clique conjecture. Other than to show hardness for planted dense subgraph in the sparsest regime, we will only need the planted clique conjecture with edge density $p = 1/2$. An interesting open problem posed in \cite{hajek2015computational} is to show that the PC conjecture at $p = 1/2$ implies it for any fixed constant $p < 1/2$.

The hypotheses in the planted independent set detection problem $\textsc{PIS}_D(n, k, p)$ are
$$H_0: G \sim G(n, p) \quad \text{and} \quad H_1 : G \sim G_I(n, k, p)$$
where $G_I(n, k, p)$ is a sample of $G(n, p)$ where all of the edges of vertex set of size $k$ removed uniformly at random. The recovery $\textsc{PC}_R(n, k, p)$ and $\textsc{PIS}_R(n, k, p)$ problems are to estimate the latent clique and independent set supports given samples from $G(n, k, p)$ and $G_I(n, k, p)$, respectively.

\paragraph{Planted Dense Subgraph.} The hypotheses in the planted dense subgraph detection problem $\textsc{PDS}_D(n, k, p, q)$ are
$$H_0: G \sim G(n, q) \quad \text{and} \quad H_1 : G \sim G(n, k, p, q)$$
where $G(n, k, p, q)$ is the distribution on $\mG_n$ formed by selecting a size $k$ subset $S$ of $[n]$ uniformly at random and joining every two nodes in $S$ with probability $p$ and every other two nodes with probability $q$. The recovery problem $\textsc{PDS}_R(n, k, p, q)$ is to estimate the latent planted dense subgraph support $S$ given samples from $G(n, k, p, q)$.

A phenomenon observed in \cite{hajek2015computational} and in \cite{chen2016statistical} is that planted dense subgraph appears to have a detection-recovery gap in the regime where $k \gg \sqrt{n}$. The following is a formulation of the conjectured sharper recovery lower bound.

\begin{conjecture}[PDS Recovery Conjecture] Suppose that $G \sim G(n, k, p, q)$ and
$$\liminf_{n \to \infty} \log_n k > \frac{1}{2} \quad \text{and} \quad \limsup_{n \to \infty} \log_n \left(\frac{k^2(p-q)^2}{q(1-q)} \right) < 1$$
then there is no sequence of randomized polynomial-time algorithms $A_n : \mG_n \to \binom{[n]}{k}$ such that $A_n(G)$ achieve exact recovery of the vertices in the latent planted dense subgraph as $n \to \infty$.
\end{conjecture}

This conjecture asserts that the threshold of $n/k^2$ on the signal $\frac{(p-q)^2}{q(1 - q)}$ of an instance of PDS is tight for the recovery problem. In contrast, our results show that the tight detection threshold for PDS given the PC conjecture is lower, at $n^2/k^4$. We will use this conjecture to establish similar detection-recovery gaps for biased sparse PCA and biclustering. We note that a related detection-recovery gap for BC was shown in \cite{cai2015computational}. The same lower bound was established for strong recovery algorithms that solve biclustering for all sub-Gaussian noise distributions assuming hardness of planted clique for a different distribution on random graphs than Erd\H{o}s-R\'{e}nyi.

\paragraph{Subgraph Stochastic Block Model.} We introduce a planted subgraph variant of the two community stochastic block model. Detection in this model cannot be solved with the edge-thresholding test that produced the conjectural detection-recovery gap in planted dense subgraph. Let $G_B(n, k, q, \rho)$ denote the set of distributions on $\mG_n$ generated a graph $G$ as follows. Fix any two positive integers $k_1$ and $k_2$ satisfying that
$$\frac{k}{2} - k^{1-\delta} \le k_1, k_2 \le \frac{k}{2} + k^{1 - \delta}$$
where $\delta = \delta_{\text{SSBM}} > 0$ is a small constant that will remained fixed throughout the paper. Let $S = [k_1]$ and $T = [k_1+k_2]\backslash [k_1]$. Then generate the edges of $G$ independently as follows:
\begin{enumerate}
\item include edges within $S$ or within $T$ with probability at least $q + \rho$;
\item include edges between $S$ and $T$ with probability at most $q - \rho$; and
\item include all other edges with probability $q$.
\end{enumerate}
Then permute the vertices of $G$ according to a permutation selected uniformly at random. The communities of the graph are defined to be the images of $S$ and $T$ under this permutation. Note that $G_B(n, k, q, \rho)$ defines a set of distributions since $k_1$ and $k_2$ are permitted to vary and the edges between $S$ and $T$ are included independently with a probability at least $q + \rho$ for each edge. Thus given the random permutation, $G$ is distributed as an inhomogeneous random graph with independent edges. The subgraph stochastic block model detection problem $\textsc{SSBM}_D(n, k, q, \rho)$ has hypotheses given by
$$H_0: G \sim G(n, q) \quad \text{and} \quad H_1 : G \sim \bP \quad \text{ for some } \quad \bP \in G_B(n, k, q, \rho)$$

\paragraph{Biclustering.} Let $\mathcal{M}_{n, k} \subseteq \mathbb{R}^{n \times n}$ be the set of sparse matrices supported on a $k \times k$ submatrix with each nonzero entry equal to $1$. The biclustering detection problem $\textsc{BC}_D(n, k, \mu)$ has hypotheses
$$H_0: M \sim N(0, 1)^{\otimes n \times n} \quad \text{and} \quad H_1 : M \sim \mu \cdot A + N(0, 1)^{\otimes n \times n} \text{ where } A \sim \text{Unif}\left[\mathcal{M}_{n,k}\right]$$
The recovery problem $\textsc{BC}_R$ is to estimate the latent support matrix $A$ given a sample from $H_1$.

\paragraph{Rank-1 Submatrix and Sparse Spiked Wigner.} In rank-1 submatrix, sparse spiked Wigner and sparse PCA, the planted sparse vectors will have sufficiently large entries for support recovery to be possible. We consider the following set of near-uniform magnitude unit vectors
$$\mathcal{V}_{d,k} = \left\{ v \in \mathbb{S}^{d-1}  : k - \frac{k}{\log k} \le \| v \|_0 \le k \text{ and } |v_i| \ge \frac{1}{\sqrt{k}} \text{ for } i \in \text{supp}(v) \right\}$$
The function $\log k$ can be replaced by any sub-polynomially growing function but is given explicitly for simplicity. The detection problem $\textsc{ROS}_D(n, k, \mu)$ has hypotheses
$$H_0: M \sim N(0, 1)^{\otimes n \times n} \quad \text{and} \quad H_1 : M \sim \mu \cdot rc^\top + N(0, 1)^{\otimes n \times n} \text{ where } r, c \in \mathcal{V}_{n, k}$$
The recovery problem $\textsc{ROS}_R$ is to estimate the latent supports $\text{supp}(r)$ and $\text{supp}(c)$. An $n \times n$ GOE matrix $\text{GOE}(n)$ is a symmetric matrix with i.i.d. $N(0, 1)$ entries below its main diagonal and i.i.d. $N(0, 2)$ entries on its main diagonal. The sparse spiked Wigner detection problem $\textsc{SSW}_D$ has hypotheses
$$H_0: M \sim \text{GOE}(n) \quad \text{and} \quad H_1 : M \sim \mu \cdot rr^\top + \text{GOE}(n) \text{ where } r \in \mathcal{V}_{n, k}$$
and the recovery problem $\textsc{SSW}_R$ is to estimate the latent supports $\text{supp}(r)$. A simple intermediate variant that will be useful in our reductions is $\textsc{SROS}$, which is $\textsc{ROS}$ constrained to have a symmetric spike $r = c$. Note that if $M$ is an instance of $\textsc{SROS}_D(n, k, \mu)$ then it follows that $\frac{1}{\sqrt{2}}(M + M^\top)$ is an instance of $\textsc{SSW}_D(n, k, \mu/\sqrt{2})$.

\paragraph{Sparse PCA.} Detection in the \emph{spiked covariance model} of $\textsc{SPCA}_D(n, k, d, \theta)$ has hypotheses
\begin{align*}
&H_0: X_1, X_2, \dots, X_n \sim N(0, I_d)^{\otimes n} \quad \text{and} \\
&H_1 : X_1, X_2, \dots, X_n \sim N\left(0, I_d + \theta vv^\top\right)^{\otimes n} \text{ where } v \in \mathcal{V}_{d, k}
\end{align*}
The recovery task $\textsc{SPCA}_R$ is to estimate $\text{supp}(v)$ given observations $X_1, X_2, \dots, X_n$ sampled from $N\left(0, I_d + \theta vv^\top\right)^{\otimes n}$ where $v \in \mathcal{V}_{d, k}$. We also consider a simple hypothesis testing variant $\textsc{USPCA}_D$ of sparse PCA with hypotheses
\begin{align*}
&H_0: X_1, X_2, \dots, X_n \sim N(0, I_d)^{\otimes n} \quad \text{and} \\
&H_1 : X_1, X_2, \dots, X_n \sim N\left(0, I_d + \theta vv^\top\right)^{\otimes n} \text{ where } v \sim \text{Unif}[S_k]
\end{align*}
where $S_k \subseteq \mathbb{R}^d$ consists of all $k$-sparse unit vectors with nonzero coordinates equal to $\pm 1/\sqrt{k}$.

\paragraph{Biased Sparse PCA.} We introduce a variant of the spiked covariance model with an additional promise. In particular, $v$ is restricted to the set $\mathcal{BV}_{d, k}$ of vectors in $\mathcal{V}_{d,k}$ with some overall positive or negative bias. Formally, if $\| v \|_0^+$ denotes the number of positive entries of $v$ then
$$\mathcal{BV}_{d, k} = \left\{ v \in \mathcal{V}_{d, k} : \| v \|_0^+ \ge \left( \frac{1}{2} + \delta \right) k \text{ or } \| v \|_0^+ \le \left( \frac{1}{2} - \delta \right) k \right\}$$
where $\delta = \delta_{\textsc{BSPCA}} > 0$ is an arbitrary constant that will remain fixed throughout the paper. The detection problem $\textsc{BSPCA}_D(n, k, d, \theta)$ has hypotheses
\begin{align*}
&H_0: X_1, X_2, \dots, X_n \sim N(0, I_d)^{\otimes n} \quad \text{and} \\
&H_1 : X_1, X_2, \dots, X_n \sim N\left(0, I_d + \theta vv^\top\right)^{\otimes n} \text{ where } v \in \mathcal{BV}_{d, k}
\end{align*}
The recovery problem $\textsc{BSPCA}_R$ is do estimate $\text{supp}(v)$ given observations $X_1, X_2, \dots, X_n$ sampled from $N\left(0, I_d + \theta vv^\top\right)^{\otimes n}$ where $v \in \mathcal{BV}_{d, k}$. We also consider a simple hypothesis testing variant $\textsc{UBSPCA}_D$ defined similarly to $\textsc{USPCA}_D$ with $v \sim \text{Unif}[BS_k]$ where $BS_k$ is the set of all $k$-sparse unit vectors with nonzero coordinates equal to $1/\sqrt{k}$.

\subsection{Our Results}

\begin{figure*}[t!]
\centering

\subfigure[$\textsc{PIS}(n, k, q)$ and $\textsc{PDS}(n, k, p, q)$ with $p = cq = \tilde{\Theta}(n^{-\alpha})$]{
\begin{tikzpicture}[scale=0.35]
\tikzstyle{every node}=[font=\footnotesize]
\def\xmin{0}
\def\xmax{11}
\def\ymin{0}
\def\ymax{11}

\draw[->] (\xmin,\ymin) -- (\xmax,\ymin) node[right] {$\beta$};
\draw[->] (\xmin,\ymin) -- (\xmin,\ymax) node[above] {$\alpha$};

\node at (5, 0) [below] {$\frac{1}{2}$};
\node at (6.66, 0) [below] {$\frac{2}{3}$};
\node at (10, 0) [below] {$1$};
\node at (0, 0) [left] {$0$};
\node at (0, 10) [left] {$2$};
\node at (0, 3.33) [left] {$\frac{2}{3}$};
\node at (0, 5) [left] {$1$};

\filldraw[fill=cyan, draw=blue] (0,0) -- (5, 0) -- (6.66, 3.33) -- (0, 0);
\filldraw[fill=orange, draw=red] (5, 0) -- (6.66, 3.33) -- (10, 5) -- (5, 0);
\filldraw[fill=gray!25, draw=gray] (5, 0) -- (10, 5) -- (10, 0) -- (0, 0) -- (0, 0) -- (5, 0);
\filldraw[fill=red, draw=magenta] (6.66, 3.33) -- (10, 5) -- (10, 10) -- (6.66, 3.33);
\filldraw[fill=green!25, draw=green] (0, 0) -- (6.66, 3.33) -- (10, 10) -- (0, 10) -- (0, 0);

\node at (4, 1) {HH};
\node at (8.5, 1) {EE};
\node at (7.1, 2.9) {EH};
\node at (8.5, 5.25) {EI};
\node at (4, 5.25) {II};
\end{tikzpicture}}
\quad
\subfigure[$\textsc{PDS}(n, k, p, q)$ with $q = \tilde{\Theta}(n^{-\alpha})$ and $p - q = \tilde{\Theta}(n^{-5\alpha/4})$]{
\begin{tikzpicture}[scale=0.35]
\tikzstyle{every node}=[font=\footnotesize]
\def\xmin{0}
\def\xmax{11}
\def\ymin{0}
\def\ymax{11}

\draw[->] (\xmin,\ymin) -- (\xmax,\ymin) node[right] {$\beta$};
\draw[->] (\xmin,\ymin) -- (\xmin,\ymax) node[above] {$\alpha$};

\node at (5, 0) [below] {$\frac{1}{2}$};
\node at (6.66, 0) [below] {$\frac{2}{3}$};
\node at (10, 0) [below] {$1$};
\node at (0, 0) [left] {$0$};
\node at (0, 10) [left] {$2$};
\node at (0, 2.22) [left] {$\frac{4}{9}$};
\node at (0, 6.66) [left] {$\frac{4}{3}$};

\filldraw[fill=cyan, draw=blue] (0,0) -- (5, 0) -- (6.66, 2.22) -- (0, 0);
\filldraw[fill=orange, draw=red] (5, 0) -- (6.66, 2.22) -- (10, 3.33) -- (5, 0);
\filldraw[fill=gray!25, draw=gray] (5, 0) -- (10, 3.33) -- (10, 0) -- (0, 0) -- (0, 0) -- (5, 0);
\filldraw[fill=red, draw=magenta] (6.66, 2.22) -- (10, 3.33) -- (10, 6.66) -- (6.66, 2.22);
\filldraw[fill=green!25, draw=green] (0, 0) -- (6.66, 2.22) -- (10, 6.66) -- (10, 10) -- (0, 10) -- (0, 0);

\node at (4, 0.66) {HH};
\node at (8.5, 0.66) {EE};
\node at (7.1, 1.92) {EH};
\node at (8.5, 3.4) {EI};
\node at (4, 3.4) {II};
\end{tikzpicture}}
\quad
\subfigure[$\textsc{SSBM}(n, k, q, \rho)$ with $q = \tilde{\Theta}(1)$ and $\rho = \tilde{\Theta}(n^{-\alpha})$]{
\centering
\begin{tikzpicture}[scale=0.35]
\tikzstyle{every node}=[font=\footnotesize]
\def\xmin{0}
\def\xmax{11}
\def\ymin{0}
\def\ymax{11}

\draw[->] (\xmin,\ymin) -- (\xmax,\ymin) node[right] {$\beta$};
\draw[->] (\xmin,\ymin) -- (\xmin,\ymax) node[above] {$\alpha$};

\node at (5, 0) [below] {$\frac{1}{2}$};
\node at (10, 0) [below] {$1$};
\node at (0, 0) [left] {$0$};
\node at (0, 10) [left] {$2$};
\node at (0, 2.5) [left] {$\frac{1}{2}$};

\filldraw[fill=cyan, draw=blue] (0,0) -- (5, 0) -- (10, 2.5) -- (0, 0);
\filldraw[fill=gray!25, draw=gray] (5, 0) -- (10, 2.5) -- (10, 0) -- (0, 0) -- (0, 0) -- (5, 0);
\filldraw[fill=green!25, draw=green] (0, 0) -- (10, 2.5) -- (10, 10) -- (0, 10) -- (0, 0);

\node at (4, 0.5) {H};
\node at (8.5, 0.5) {E};
\node at (4, 3.4) {I};
\end{tikzpicture}}
\quad
\subfigure[$\textsc{BC}(n, k, \mu)$ with $\mu = \tilde{\Theta}(n^{-\alpha})$]{
\begin{tikzpicture}[scale=0.35]
\tikzstyle{every node}=[font=\footnotesize]
\def\xmin{0}
\def\xmax{11}
\def\ymin{-1}
\def\ymax{11}

\draw[->] (\xmin,\ymin) -- (\xmax,\ymin) node[right] {$\beta$};
\draw[->] (\xmin,\ymin) -- (\xmin,\ymax) node[above] {$\alpha$};

\node at (5, -1) [below] {$\frac{1}{2}$};
\node at (6.66, -1) [below] {$\frac{2}{3}$};
\node at (10, -1) [below] {$1$};
\node at (0, 0) [left] {$0$};
\node at (0, 10) [left] {$1$};
\node at (0, 3.33) [left] {$\frac{1}{3}$};
\node at (0, 5) [left] {$\frac{1}{2}$};

\filldraw[fill=cyan, draw=blue] (0,0) -- (5, 0) -- (6.66, 3.33) -- (0, 0);
\filldraw[fill=orange, draw=red] (5, 0) -- (6.66, 3.33) -- (10, 5) -- (5, 0);
\filldraw[fill=gray!25, draw=gray] (5, 0) -- (10, 5) -- (10, -1) -- (0, -1) -- (0, 0) -- (5, 0);
\filldraw[fill=red, draw=magenta] (6.66, 3.33) -- (10, 5) -- (10, 10) -- (6.66, 3.33);
\filldraw[fill=green!25, draw=green] (0, 0) -- (6.66, 3.33) -- (10, 10) -- (0, 10) -- (0, 0);

\node at (4, 1) {HH};
\node at (8.5, 1) {EE};
\node at (7.1, 2.9) {EH};
\node at (8.5, 5.25) {EI};
\node at (4, 5.25) {II};
\end{tikzpicture}}
\quad
\subfigure[$\textsc{ROS}(n, k, \mu)$ and $\textsc{SSW}(n, k, \mu)$ with $\frac{\mu}{k} = \tilde{\Theta}(n^{-\alpha})$]{
\centering
\begin{tikzpicture}[scale=0.35]
\tikzstyle{every node}=[font=\footnotesize]
\def\xmin{0}
\def\xmax{11}
\def\ymin{-1}
\def\ymax{11}

\draw[->] (\xmin,\ymin) -- (\xmax,\ymin) node[right] {$\beta$};
\draw[->] (\xmin,\ymin) -- (\xmin,\ymax) node[above] {$\alpha$};

\node at (5, -1) [below] {$\frac{1}{2}$};
\node at (10, -1) [below] {$1$};
\node at (0, 0) [left] {$0$};
\node at (0, 10) [left] {$1$};
\node at (0, 5) [left] {$\frac{1}{2}$};

\filldraw[fill=cyan, draw=blue] (0,0) -- (5, 0) -- (10, 5) -- (0, 0);
\filldraw[fill=gray!25, draw=gray] (5, 0) -- (10, 5) -- (10, -1) -- (0, -1) -- (0, 0) -- (5, 0);
\filldraw[fill=green!25, draw=green] (0, 0) -- (10, 5) -- (10, 10) -- (0, 10) -- (0, 0);

\node at (4, 1) {HH};
\node at (8.5, 1) {EE};
\node at (4, 5.25) {II};
\end{tikzpicture}}

\subfigure[$\textsc{SPCA}(n, k, d, \theta)$ with $d = \Theta(n)$ and $\theta = \tilde{\Theta}(n^{-\alpha})$]{
\begin{tikzpicture}[scale=0.35]
\tikzstyle{every node}=[font=\footnotesize]
\def\xmin{0}
\def\xmax{11}
\def\ymin{-1}
\def\ymax{11}

\draw[->] (\xmin,\ymin) -- (\xmax,\ymin) node[right] {$\beta$};
\draw[->] (\xmin,\ymin) -- (\xmin,\ymax) node[above] {$\alpha$};

\node at (5, -1) [below] {$\frac{1}{2}$};
\node at (10, -1) [below] {$1$};
\node at (0, 0) [left] {$0$};
\node at (0, 10) [left] {$1$};
\node at (0, 5) [left] {$\frac{1}{2}$};

\filldraw[fill=cyan, draw=blue] (0, 5) -- (10, 0) -- (5, 0) -- (0, 5);
\filldraw[fill=gray!25, draw=gray] (0, 5) -- (5, 0) -- (10, 0) -- (10, -1) -- (0, -1) -- (0, 5);
\filldraw[fill=green!25, draw=green] (0, 5) -- (10, 0) -- (10, 10) -- (0, 10) -- (0, 5);
\filldraw[fill=black, draw=gray] (0, 5) -- (3.33, 3.33) -- (5, 0) -- (0, 5);
\node at (6, 1) {HH};
\node at (2, 1) {EE};
\node at (5, 5) {II};
\end{tikzpicture}}
\quad
\subfigure[$\textsc{BSPCA}(n, k, d, \theta)$ with $d = \Theta(n)$ and $\theta = \tilde{\Theta}(n^{-\alpha})$]{
\begin{tikzpicture}[scale=0.35]
\tikzstyle{every node}=[font=\footnotesize]
\def\xmin{0}
\def\xmax{11}
\def\ymin{-1}
\def\ymax{11}

\draw[->] (\xmin,\ymin) -- (\xmax,\ymin) node[right] {$\beta$};
\draw[->] (\xmin,\ymin) -- (\xmin,\ymax) node[above] {$\alpha$};

\node at (5, -1) [below] {$\frac{1}{2}$};
\node at (6.66, -1) [below] {$\frac{2}{3}$};
\node at (10, -1) [below] {$1$};
\node at (0, 0) [left] {$0$};
\node at (0, 10) [left] {$1$};
\node at (0, 5) [left] {$\frac{1}{2}$};

\filldraw[fill=cyan, draw=blue] (0, 5) -- (6.66, 1.66) -- (5, 0) -- (0, 5);
\filldraw[fill=gray!25, draw=gray] (0, 5) -- (5, 0) -- (10, 0) -- (10, -1) -- (0, -1) -- (0, 5);
\filldraw[fill=green!25, draw=green] (0, 5) -- (6.66, 1.66) -- (10, 5) -- (10, 10) -- (0, 10) -- (0, 5);
\filldraw[fill=orange, draw=red] (5, 0) -- (6.66, 1.66) -- (10, 0) -- (5, 0);
\filldraw[fill=red, draw=magenta] (6.66, 1.66) -- (10, 5) -- (10, 0) -- (6.66, 1.66);
\filldraw[fill=black, draw=gray] (0, 5) -- (3.33, 3.33) -- (5, 0) -- (0, 5);
\filldraw[fill=black, draw=gray] (5, 0) -- (6, 2) -- (6.67, 1.67) -- (5, 0);

\node at (2, 0.8) {EE};
\node at (6, 6) {II};
\node at (4.95, 1.8) {HH};
\node at (6.8, 0.8) {EH};
\node at (9, 2.2) {EI};
\end{tikzpicture}}
\caption{Parameter regimes by problem plotted as signal vs. sparsity. Sparsity is $k = \tilde{\Theta}(n^{\beta})$. First labels characterize detection and second labels characterize exact recovery. Recovery is not considered for $\textsc{SSBM}$ and weak recovery is considered for $\textsc{SPCA}$ and $\textsc{BSPCA}$. In Easy (E) regimes, there is a polynomial-time algorithm. In Hard (H) regimes, the PC or PDS conjecture implies there is no polynomial-time algorithm. In Impossible (I) regimes, the task is information-theoretically impossible. Hardness in black regions is open.}
\label{fig:diagrams}
\end{figure*}
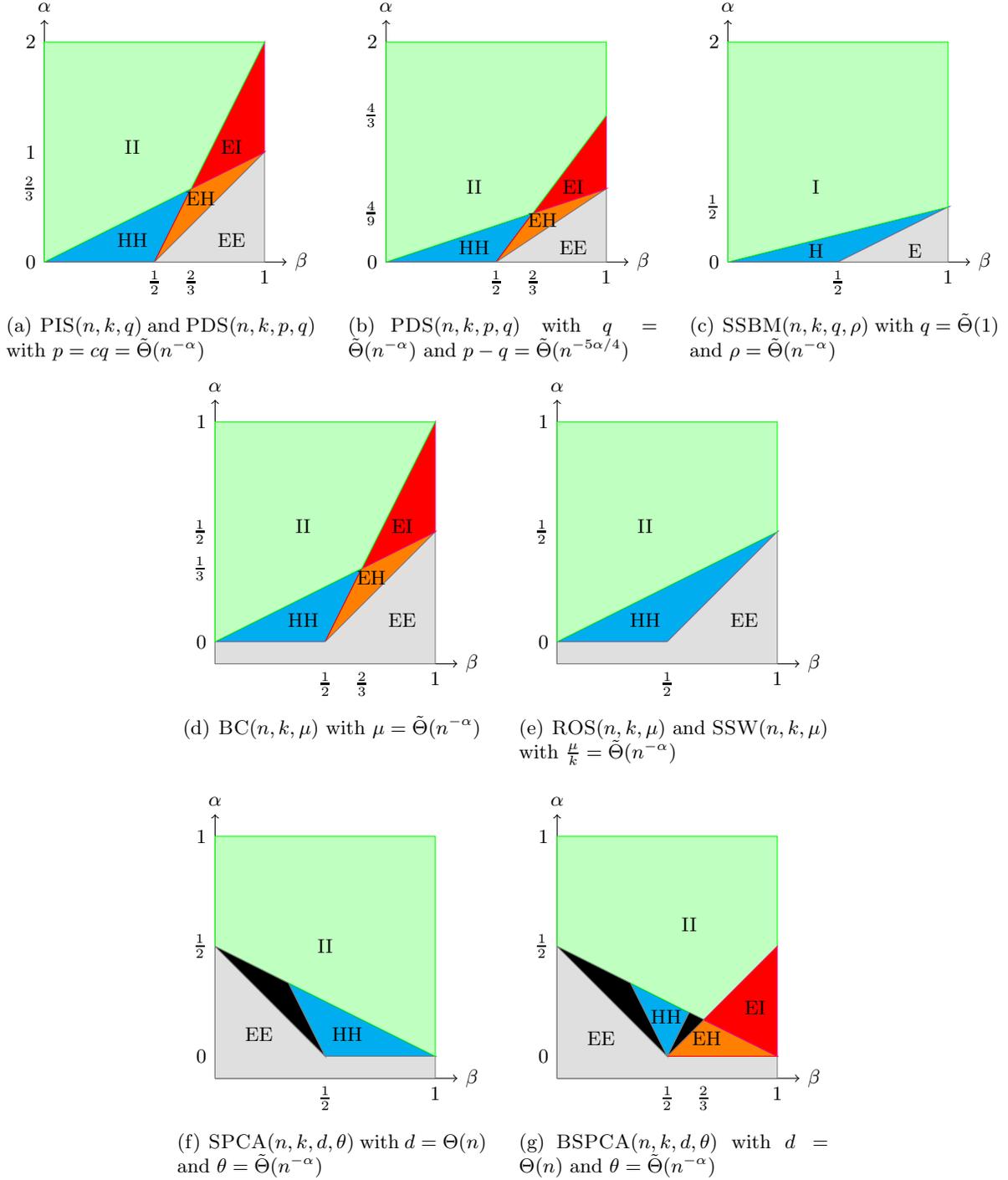

In this paper, our aim is to establish tight characterizations of the statistical-computational gaps for the problems formulated in the previous section. Each problem has three regimes for each of its detection and recovery variants. In the easy regime, there is a polynomial-time algorithm for the task. In the hard regime, the PC or PDS recovery conjecture implies that there is no polynomial-time algorithm but there is an inefficient algorithm solving the task. In the impossible regime, the task is information-theoretically impossible. Our results complete the hardness picture for each of these problems other than sparse PCA when $k \ll \sqrt{n}$ and biased sparse PCA. Our results are informally stated below and depicted visually in Figure \ref{fig:diagrams}.

\begin{theorem}[Informal Main Theorem] \label{lem:2a}
Given the PC and PDS recovery conjectures, the easy, hard and impossible regimes of the problems in Section 2.2 are classified in Figure \ref{fig:classification} as following one of the configurations of regimes in Figure \ref{fig:types}.
\end{theorem}

\begin{figure*}
\begin{center}
\begin{tabular}{|c | c c c|}
\hline
\textsc{Type I} & $k \ll n^{1/2}$ & $n^{1/2} \lesssim k \ll n^{2/3}$ & $n^{2/3} \ll k$ \\
\hline
\textnormal{Impossible} & $\textnormal{SNR} \ll \frac{1}{k}$ & $\textnormal{SNR} \ll \frac{1}{k}$ & $\textnormal{SNR} \ll \frac{n^2}{k^4}$ \\
\textnormal{Hard} & $\textnormal{SNR} \gtrsim \frac{1}{k}$ & $\frac{1}{k} \lesssim \textnormal{SNR} \ll \frac{n^2}{k^4}$ & \textnormal{None} \\
\textnormal{Easy} & (A) \textnormal{None} & $\textnormal{SNR} \gtrsim \frac{n^2}{k^4}$ & $\textnormal{SNR} \gtrsim \frac{n^2}{k^4}$ \\
& (B) $\textnormal{SNR} \gtrsim 1$ & & \\
\hline
\end{tabular}
\vspace{5mm}

\begin{tabular}{|c | c c c|}
\hline
\textsc{Type II} & $k \ll n^{1/2}$ & $n^{1/2} \lesssim k \ll n^{2/3}$ & $n^{2/3} \ll k$ \\
\hline
\textnormal{Impossible} & $\textnormal{SNR} \ll \sqrt{\frac{k}{n}}$ & $\textnormal{SNR} \ll \sqrt{\frac{k}{n}}$ & $\textnormal{SNR} \ll \sqrt{\frac{n}{k^2}}$ \\
\textnormal{Hard} & $\sqrt{\frac{k}{n}} \lesssim \textnormal{SNR} \ll \frac{k^2}{n}$ & $\sqrt{\frac{k}{n}} \lesssim \textnormal{SNR} \ll \frac{n}{k^2}$ & \textnormal{None} \\
\textnormal{Easy} & $\textnormal{SNR} \gtrsim \sqrt{\frac{k^2}{n}}$ & $\textnormal{SNR} \gtrsim \sqrt{\frac{n}{k^2}}$ & $\textnormal{SNR} \gtrsim \sqrt{\frac{n}{k^2}}$ \\
\hline
\end{tabular}
\vspace{5mm}

\begin{tabular}{|c | c c | c | c c |}
\hline
\textsc{Type III} & $k \ll n^{1/2}$ & $n^{1/2} \lesssim k$ & \textsc{Type IV} & $k \ll n^{1/2}$ & $n^{1/2} \lesssim k$ \\
\hline
\textnormal{Impossible} & $\textnormal{SNR} \ll \frac{1}{k}$ & $\textnormal{SNR} \ll \frac{1}{k}$ & \textnormal{Impossible} & $\textnormal{SNR} \ll \sqrt{\frac{k}{n}}$ & $\textnormal{SNR} \ll \sqrt{\frac{k}{n}}$ \\
\textnormal{Hard} & $\textnormal{SNR} \gtrsim \frac{1}{k}$ & $\frac{1}{k} \lesssim \textnormal{SNR} \ll \frac{n}{k^2}$ & \textnormal{Hard} & $\sqrt{\frac{k}{n}} \lesssim \textnormal{SNR} \ll \frac{k^2}{n}$ & $\sqrt{\frac{k}{n}} \lesssim \textnormal{SNR} \ll 1$ \\
\textnormal{Easy} & (A) \textnormal{None} & $\textnormal{SNR} \gtrsim \frac{n}{k^2}$ & \textnormal{Easy} & $\textnormal{SNR} \gtrsim \sqrt{\frac{k^2}{n}}$ & $\textnormal{SNR} \gtrsim 1$ \\
& (B) $\textnormal{SNR} \gtrsim 1$ & & & & \\
\hline
\end{tabular}
\end{center}
\caption{Types of hardness regimes by $\text{SNR}$ in Theorem \ref{lem:2a}. \textsc{Type I} and \textsc{Type III} each have two variants A and B depending on the Easy regime when $k \ll n^{1/2}$.}
\label{fig:types}
\end{figure*}
\begin{figure*}
\begin{center}
\begin{tabular}{| c | c | c | c |}
\hline
\textsc{Problems} & \textsc{Parameter Regime} & \textsc{SNR} & \textsc{Type} \\
\hline
$\textsc{PIS}_D(n, k, q), \textsc{PDS}_D(n, k, cq, q)$ & $q = \tilde{\Theta}(n^{-\alpha})$ for fixed $\alpha \in [0, 1)$ and $c > 1$ & $q$ & \textsc{Type IA} \\
\hline
$\textsc{PIS}_R(n, k, q), \textsc{PDS}_R(n, k, cq, q)$ & $q = \tilde{\Theta}(n^{-\alpha})$ for fixed $\alpha \in [0, 1)$ and $c > 1$ & $q$ & \textsc{Type IIIA} \\
\hline
$\textsc{PDS}_D(n, k, p, q)$ & $q = \tilde{\Theta}(n^{-\alpha})$ and $p - q = \tilde{\Theta}(n^{-\gamma})$ with & $\frac{(p - q)^2}{q(1 - q)}$ & \textsc{Type IA} \\
& $p > q$ for fixed $\alpha, \gamma \in [0, 1)$ & & \\
\hline
$\textsc{PDS}_R(n, k, p, q)$ & $q = \tilde{\Theta}(n^{-\alpha})$ and $p - q = \tilde{\Theta}(n^{-\gamma})$ with & $\frac{(p - q)^2}{q(1 - q)}$ & \textsc{Type IIIA} \\
& $p > q$ for fixed $\alpha, \gamma \in [0, 1)$ & & \\
\hline
$\textsc{SSBM}_D(n, k, q, \rho)$ & $q = \Theta(1)$ and $\rho = \tilde{\Theta}(n^{-\alpha})$ & $\rho^2$ & \textsc{Type IIIA} \\
 & for fixed $\alpha \in [0, 1)$ & & \\
\hline
$\textsc{BC}_D(n, k, \mu)$ & $\mu = \tilde{\Theta}(n^{-\alpha})$ for fixed $\alpha \in [0, 1)$ & $\mu^2$ & \textsc{Type IB} \\
\hline
$\textsc{BC}_R(n, k, \mu)$ & $\mu = \tilde{\Theta}(n^{-\alpha})$ for fixed $\alpha \in [0, 1)$ & $\mu^2$ & \textsc{Type IIIB} \\
\hline
$\textsc{ROS}_D(n, k, \mu)$, $\textsc{ROS}_R(n, k, \mu)$, & $\mu = \tilde{\Theta}(n^{-\alpha})$ for fixed $\alpha \in [0, 1)$ & $\frac{\mu^2}{k^2}$ & \textsc{Type IIIB} \\
$\textsc{SSW}_D(n, k, \mu)$, $\textsc{SSW}_R(n, k, \mu)$ & & & \\
\hline
$\textsc{SPCA}_D(n, k, d, \theta)$, & $d = \Theta(n)$ and $\theta = \tilde{\Theta}(n^{-\alpha})$ & $\theta$ & \textsc{Type II} \\
$\textsc{SPCA}_{WR}(n, k, d, \theta)$, & for fixed $\alpha \in [0, 1)$ & & \\
$\textsc{BSPCA}_{WR}(n, k, d, \theta)$ & & & \\
\hline
$\textsc{BSPCA}_D(n, k, d, \theta)$ & $d = \Theta(n)$ and $\theta = \tilde{\Theta}(n^{-\alpha})$  & $\theta$ & \textsc{Type IV} \\
& for fixed $\alpha \in [0, 1)$ & & \\
\hline
\end{tabular}
\end{center}
\caption{Classification of regimes for each problem as in Theorem \ref{lem:2a}. For each problem, $k$ is in the regime $k = \tilde{\Theta}(n^{\beta})$ where $\beta \in (0, 1)$ is a constant.}
\label{fig:classification}
\end{figure*}

We remark that all of the computational lower bounds in Theorem \ref{lem:2a} follow from the PC conjecture other than those for $\textsc{PIS}_R, \textsc{PDS}_R, \textsc{BC}_R$ and $\textsc{BSPCA}_{WR}$ which follow from the PDS conjecture. The computational lower bounds implicit in the hard regimes in Figures \ref{fig:types} and \ref{fig:classification} are the focus of the present work. Section 4 introduces $\textsc{PC-Lifting}$ to reduce from $\textsc{PC}_D$ to $\textsc{PIS}_D$. Section 5 introduces rejection kernels and general $\textsc{Distributional-Lifting}$ which are then applied in Section 6 to reduce from $\textsc{PC}_D$ to all regimes of $\textsc{PDS}_D$. Section 7 introduces reflection cloning to reduce from $\textsc{BC}_D$ to $\textsc{ROS}_D$, $\textsc{SSW}_D$ and $\textsc{SSBM}_D$. Section 8 introduces random rotations to reduce from $\textsc{SSW}_D$ to $\textsc{SPCA}_D$ and from $\textsc{BC}_D$ to $\textsc{BSPCA}_D$. In Section 6, we also give a reduction from $\textsc{PDS}_R$ to $\textsc{BC}_R$ and in Section 9, we reduce from $\textsc{PDS}_R$ to $\textsc{BSPCA}_{WR}$. In Section 9, we establish the algorithmic upper bounds and information-theoretic lower bounds needed to complete the proof of Theorem \ref{lem:2a}. In Section 10, we show that our detection lower bounds imply recovery lower bounds. 

Note that this gives a complete characterization of the easy, hard and impossible regions for all of the problems we consider other than sparse PCA and biased sparse PCA. The SDP relaxation of the MLE for sparse PCA was shown in \cite{berthet2013optimal} to succeed if $d = \Theta(n)$ and $k \ll n^{1/2 - \delta}$ down to the signal level of $\theta \approx k/\sqrt{n}$, which is generally conjectured to be optimal for efficient algorithms. There is a gap between the lower bounds we prove here, which match those of \cite{berthet2013optimal} and \cite{gao2017sparse}, and this conjecturally optimal threshold when $k \ll \sqrt{n}$. There is also a gap for biased sparse PCA detection when $k \gg \sqrt{n}$. These gaps are depicted as the black region in Figures \ref{fig:diagrams}(f) and \ref{fig:diagrams}(g). Showing tight hardness for sparse PCA for any parameters $k \ll \sqrt{n}$ remains an interesting open problem. Notably, we do not consider the recovery problem for the subgraph stochastic block model and only consider weak, rather than exact, recovery for sparse PCA and biased sparse PCA. While we obtain computational lower bounds that apply to these recovery variants, obtaining matching information-theoretic lower bounds and algorithms remains an open problem. These open problems are discussed in more detail in Section 11.

In Figure \ref{fig:diagrams} and Theorem \ref{lem:2a}, we depict the implications of our hardness results for $\textsc{SPCA}$ and $\textsc{BSPCA}$ in the case where $d = \Theta(n)$ for simplicity. Our hardness results extend mildly beyond this regime, but not tightly. However, we remark that the assumptions $d = \Theta(n)$ or $d = \Omega(n)$ have become commonplace in the literature on lower bounds for sparse PCA. The reductions from PC to sparse PCA in \cite{berthet2013complexity}, \cite{wang2016statistical} and \cite{gao2017sparse} construct hard instances that are only tight when $d = \Omega(n)$. The sum of squares lower bounds for sparse PCA in \cite{ma2015sum} and \cite{hopkins2017power} also are in the regime $d = \Theta(n)$. The SOS lower bounds in \cite{hopkins2017power} are actually for the spiked Wigner model rather than the spiked covariance model of sparse PCA that we consider.

\subsection{Our Techniques}

\paragraph{Distributional Lifting.} We introduce several new techniques resembling graph lifts to increase the size $k$ of a sparse structure, while appropriately maintaining the level of signal and independence in the noise distribution. Given a graph $G$, the main idea behind our techniques is to replace the $\{0, 1\}$-valued edge indicators with Gaussian and Poisson random variables. We then increase the size of an instance by a factor of two iteratively, while maintaining the signal and independence, through distributional tricks such as Poisson splitting and the rotational invariance of independent Gaussians. In \cite{hajek2015computational}, the reduction from planted clique to planted dense subgraph also expands an input graph. Rather than proceed iteratively, their method expands the graph in one step. The main technical issue arising from this is that the diagonal entries of the graph's adjacency matrix are mapped to low-density subgraphs in the hidden community. Showing that these are not detectable requires a subtle argument and that the hidden community is randomly sized according to a Binomial distribution. By proceeding incrementally as in our approach, the diagonal entries become much easier to handle. However for many regimes of interest, incremental approaches that preserve the fact that the instance is a graph seem to unavoidably introduce dependence between edges. Our insight is to map edges to other random variables that can be preserved incrementally while maintaining independence and the desired parameter scaling to produce tight lower bounds. Using Poisson and Gaussian variants of this lifting procedure, we are able to reduce from planted clique in a wide range of parameter regimes. Gaussian lifting also recovers a simple vs. simple hypothesis testing variant of the lower bounds for biclustering shown in \cite{ma2015computational} as an intermediate step towards reducing to planted dense subgraph.

\paragraph{Rejection Kernels.} We give a simple scheme based on rejection sampling that approximately maps a sample from $\text{Bern}(p)$ to a sample from $P$ and from $\text{Bern}(q)$ to $Q$ where $p, q \in [0, 1]$ and $P$ and $Q$ are two distributions on $\mathbb{R}$. By thresholding samples from a pair of distributions $P'$ and $Q'$, and then mapping the resulting Bernoulli random variables to a pair of target distributions $P$ and $Q$, this method yields an efficient procedure to simultaneously perform two changes of measure. This method is used in distributional lifting to map from edge indicators to a pair of chosen distributions. This framework extends and uses similar ideas to the approximate sampling methods introduced in \cite{hajek2015computational} and \cite{ma2015computational}. 

\paragraph{Reflection Cloning.} While distributional lifting appears to accurately characterize the hard regimes in many detection problems with an optimal test involving summing the entire input matrix, it is fundamentally lossy. In each iteration, these lifting techniques generate additional randomness in order to maintain independence in the noise distribution of the instance. In a problem like rank-1 submatrix detection that does not admit a sum test, these cloning techniques do not come close to showing hardness at the computational barrier. We introduce a more sophisticated cloning procedure for cases of Gaussian noise that introduces significantly less randomness in each iteration. Let $\mathcal{R}$ denote the linear operator on $n \times n$ matrices that reflects the matrix about its vertical axis of symmetry and let $\mathcal{F}$ denote the linear operator that multiplies each entry on the right half of the matrix by $-1$. Then one step of reflection cloning replaces a matrix $W$ with
$$W \gets \frac{1}{\sqrt{2}} \left( \mathcal{R} W^\sigma + \mathcal{F} W^\sigma \right)$$
where $\sigma$ is a random permutation. Reflection cloning then repeats this for rows instead of columns. If $W = uv^\top + G$ has even dimensions and $G$ has i.i.d. $N(0, 1)$ entries, then reflection cloning effectively doubles the sparsity of $u$ and $v$ while mildly decreasing the signal. Importantly, it can be checked that the Gaussian noise matrix retains the fact that it has independent entries. It is interesting to note that this cloning procedure precludes the success of a sum test by negating parts of the signal.

\paragraph{Random Rotations and Sparse PCA.} We introduce a simple connection between sparse PCA, biclustering and rank-1 submatrix through random rotations. This yields lower bounds matching those of \cite{gao2017sparse} and \cite{berthet2013complexity}. Although often suboptimal, the random rotations map we introduce tightly gives lower bounds in the regime $k \gg \sqrt{n}$, using reflection cloning as an intermediate. This marks the first tight computational lower bound for sparse PCA over an entire parameter regime. It also illustrates the utility of natural average-case problems as reduction intermediates, suggesting that webs of reductions among problems can be useful beyond worst-case complexity.

\section{Average-Case Reductions under Total Variation}

The typical approach to show computational lower bounds for detection problems is to reduce an instance of one problem to a random object close in total variation distance to an instance of another problem in randomized polynomial time. More precisely, let $\mP$ and $\mP'$ be detection problems and $X$ and $Y$ be instances of $\mP$ and $\mP'$, respectively. Suppose we are given a polynomial-time computable map $\phi$ taking an object $X$ to $\phi(X)$ with total variation distance to $Y$ decaying to zero simultaneously under each of $H_0$ and $H_1$. Then any algorithm that can distinguish $H_0$ and $H_1$ for $\mP'$ in polynomial time when applied to $\phi(X)$ also distinguishes $H_0$ and $H_1$ for $\mP$. Taking $\mP$ to be PC and $\mP'$ to be the problem of interest then yields a computational hardness result for $\mP'$ conditional on the PC conjecture. The general idea in this approach is formalized in the following simple lemma.

\begin{lemma} \label{lem:3a}
Let $\mP$ and $\mP'$ be detection problems with hypotheses $H_0, H_1, H_0', H_1'$ and let $X$ and $Y$ be instances of $\mP$ and $\mP'$, respectively. Suppose there is a polynomial time computable map $\phi$ satisfying
$$\TV\left(\mL_{H_0}(\phi(X)), \mL_{H_0'}(Y)\right) + \sup_{\bP \in H_1} \inf_{\pi \in \Delta(H_1')} \TV\left( \mL_{\bP}(\phi(X)), \int_{H_1'} \mL_{\bP'}(Y) d\pi(\bP') \right) \le \delta$$
If there is a polynomial time algorithm solving $\mP'$ with Type I$+$II error at most $\epsilon$, then there is a polynomial time algorithm solving $\mP$ with Type I$+$II error at most $\epsilon + \delta$.
\end{lemma}

\begin{proof}
Let $\psi$ be a polynomial time computable test function solving $\mP'$ with Type I$+$II error at most $\epsilon$. Note that for any observation $X$ of $\mP$, the value $\psi \circ \phi(X) \in \{0, 1\}$ can be computed in polynomial time. This is because the fact that $\phi(X)$ can be computed in polynomial time implies that $\phi(X)$ has size polynomial in the size of $X$. We claim that $\psi \circ \phi$ solves the detection problem $\mP'$. Now fix some distribution $\bP \in H_1$ and prior $\pi \in \Delta(H_1')$. By the definition of total variation,
\begin{align*}
\left| \bP_{H_0} \left[ \psi \circ \phi(X) = 1 \right] - \bP_{H_0'} \left[ \psi(Y) = 1 \right] \right| &\le \TV\left(\mL_{H_0}(\phi(X)), \mL_{H_0'}(Y)\right) \\
\left| \bP_{X \sim \bP} \left[ \psi \circ \phi(X) = 0 \right] - \int_{H_1'} \bP_{Y \sim \bP'} \left[ \psi(Y) = 0 \right] d\pi(\bP') \right| &\le \TV\left( \mL_{\bP}(\phi(X)), \int_{H_1'} \mL_{\bP'}(Y) d\pi(\bP') \right)
\end{align*}
Also note that since $\pi$ is a probability distribution,
$$\int_{H_1'} \bP_{Y \sim \bP'} \left[ \psi(Y) = 0 \right] d\pi(\bP') \le \sup_{\bP' \in H_1'} \bP_{Y \sim \bP'} \left[ \psi(Y) = 0 \right]$$
Combining these inequalities with the triangle inequality yields that
\begin{align*}
\bP_{H_0} \left[ \psi \circ \phi(X) = 1 \right] + \bP_{X \sim \bP} \left[ \psi \circ \phi(X) = 0 \right] \le \epsilon &+ \TV\left(\mL_{H_0}(\phi(X)), \mL_{H_0'}(Y)\right) \\
&+ \TV\left( \mL_{\bP}(\phi(X)), \int_{H_1'} \mL_{\bP'}(Y) d\pi(\bP') \right)
\end{align*}
Fixing $\bP$ and choosing the prior $\pi$ so that the second total variation above approaches its infimum yields that the right hand side above is upper bounded by $\epsilon + \delta$. The fact that this bound holds for all $\bP \in H_1$ proves the lemma.
\end{proof}

We remark that the second term in the total variation condition of Lemma 1 can be interpreted as ensuring that each distribution $\bP \in H_1$ is close to a distribution formed by taking a prior $\pi$ over the distributions in hypothesis $H_1'$. In light of this lemma, to reduce one problem to another it suffices to find such a map $\phi$. In the case that $\mP$ and $\mP'$ are both simple hypothesis testing problems, the second term is simply $\TV\left(\mL_{H_1}(\phi(X)), \mL_{H_1'}(Y)\right)$.

Throughout the analysis of our average-case reductions, total variation will be the key object of interest. We will make use of several standard results concerning total variation, including the triangle inequality, data processing inequality and tensorization of total variation. The latter two results are stated below.

\begin{lemma}[Data Processing]
Let $P$ and $Q$ be distributions on a measurable space $(\mathcal{X}, \mathcal{B})$ and let $f : \mathcal{X} \to \mathcal{Y}$ be a Markov transition kernel. If $A \sim P$ and $B \sim Q$ then
$$\TV\left(\mL(f(A)), \mL(f(B))\right) \le \TV(P, Q)$$
\end{lemma}

\begin{lemma}[Tensorization]
Let $P_1, P_2, \dots, P_n$ and $Q_1, Q_2, \dots, Q_n$ be distributions on a measurable space $(\mathcal{X}, \mathcal{B})$. Then
$$\TV\left( \prod_{i = 1}^n P_i, \prod_{i = 1}^n Q_i \right) \le \sum_{i = 1}^n \TV\left( P_i, Q_i \right)$$
\end{lemma}

A typical analysis of a multi-step algorithm will proceed as follows. Suppose that $A = A_2 \circ A_1$ is an algorithm with two steps $A_1$ and $A_2$. Let $P_0$ be the input distribution and $P_2$ be the target distribution that we would like to show is close in variation to $A(P_0)$. Let $P_1$ be an intermediate distribution which is close in total variation to $A_1(P_0)$. By the triangle inequality,
\begin{align*}
\TV\left( A(P_0), P_2 \right) &\le \TV\left( A(P_0), A_2(P_1) \right) + \TV\left( A_2(P_1), P_2 \right) \\
&= \TV\left( A_2\circ A_1(P_0), A_2(P_1) \right) + \TV\left( A_2(P_1), P_2 \right) \\
&\le \TV\left( A_1(P_0), P_1 \right) + \TV\left( A_2(P_1), P_2 \right)
\end{align*}
by the data-processing inequality. Thus total variation accumulates over the steps of a multi-step algorithm. This style of analysis will appear frequently in our reductions. Another lemma about total variation that will be useful throughout this work is as follows. The proof is given in Appendix \ref{app3}.

\begin{lemma} \label{lem:5tv}
For any random variable $Y$ and event $A$ in the $\sigma$-algebra $\sigma\{Y \}$, it holds that
$$\TV\left( \mL(Y | A), \mL(Y) \right) = \bP[Y \in A^c]$$
\end{lemma}

\section{Densifying Planted Clique and Planted Independent Set}

In this section, we give a reduction increasing the ambient edge density in planted clique while increasing the relative size of the clique, which shows tight hardness for the planted independent set problem. This reduction, which we term $\textsc{PC-Lifting}$, serves as an introduction to the general distributional lifting procedure in the next section. Distributional lifting will subsequently be specialized to produce Poisson and Gaussian variants of the procedure, which will be used to prove hardness for biclustering and different regimes of planted dense subgraph.

\subsection{Detecting Planted Generalized Diagonals}

We first prove a technical lemma that will be used in all of our cloning procedures. Given a matrix $M$, let $M^{\sigma_1, \sigma_2}$ denote the matrix formed by permuting rows according to $\sigma_1$ and columns according to $\sigma_2$. Let $\text{id}$ denote the identity permutation.

\begin{lemma} \label{lem:4a}
Let $P$ and $Q$ be two distributions such that $Q$ dominates $P$ and $\chi^2(P, Q) \le 1$. Suppose that $M$ is an $n \times n$ matrix with all of its non-diagonal entries i.i.d. sampled from $Q$ and all of its diagonal entries i.i.d. sampled from $P$. Suppose that $\sigma$ is a permutation on $[n]$ chosen uniformly at random. Then
$$\TV\left( \mL(M^{\text{id},\sigma}), Q^{\otimes n \times n} \right) \le \sqrt{\frac{\chi^2(P, Q)}{2}}$$
\end{lemma}

\begin{proof}
Let $\sigma'$ be a permutation of $[n]$ chosen uniformly at random and independent of $\sigma$. By Fubini's theorem we have that
$$\chi^2\left( \mL(M^{\text{id}, \sigma}), Q^{\otimes n \times n} \right) + 1 = \int \frac{\bE_\sigma \left[ \bP_{M^{\text{id}, \sigma}} ( X | \sigma) \right]^2}{\bP_{Q^{\otimes n \times n}}(X)} dX = \bE_{\sigma, \sigma'} \int \frac{\bP_{M^{\text{id}, \sigma}} ( X | \sigma) \bP_{M^{\text{id}, \sigma'}} ( X | \sigma')}{\bP_{Q^{\otimes n \times n}}(X)} dX$$
Now note that conditioned on $\sigma$, the entries of $M^{\text{id}, \sigma}$ are independent with distribution
$$\bP_{M^{\text{id}, \sigma}} ( X | \sigma) = \prod_{i = 1}^n P\left(X_{i \sigma(i)} \right) \prod_{j \neq \sigma(i)} Q\left(X_{ij} \right)$$
Therefore we have that
\begin{align*}
\int \frac{\bP_{M^{\text{id}, \sigma}} ( X | \sigma) \bP_{M^{\text{id}, \sigma'}} ( X | \sigma')}{\bP_{Q^{\otimes n \times n}}(X)} dX &= \int \left( \prod_{i : \sigma(i) = \sigma'(i)} \frac{P\left(X_{i\sigma(i)}\right)^2}{Q\left(X_{i\sigma(i)}\right)} \right) \left( \prod_{i : \sigma(i) \neq \sigma'(i)} P\left(X_{i\sigma(i)}\right) \right) \\
&\quad \quad \times \left( \prod_{i : \sigma(i) \neq \sigma'(i)} P\left(X_{i\sigma'(i)}\right) \right) \left( \prod_{(i, j) : j \neq \sigma(i), j \neq \sigma'(i)} Q\left(X_{ij}\right) \right)dX \\
&= \prod_{i : \sigma(i) = \sigma'(i)} \left( \int \frac{P\left(X_{i\sigma(i)}\right)^2}{Q\left(X_{i\sigma(i)}\right)} dX_{i\sigma(i)} \right) \\
&= \left( 1 + \chi^2(P, Q) \right)^{|\{ i : \sigma(i) = \sigma'(i) \}|}
\end{align*}
If $\tau = \sigma' \circ \sigma^{-1}$, then $\tau$ is a uniformly at random chosen permutation and $Y = |\{ i : \sigma(i) = \sigma'(i) \}|$ is the number of fixed points of $\tau$. As in \cite{pitman1997some}, the $i$th moment of $Y$ is the $i$th Bell number for $i \le n$ and for $i > n$, the $i$th moment of $Y$ is at most the $i$th Bell number. Since a Poisson distribution with rate $1$ has its $i$th moment given by the $i$th Bell number for all $i$, it follows that for each $t \ge 0$ the MGF $\bE[e^{tY}]$ is at most that of a Poisson with rate $1$, which is $\exp(e^t - 1)$. Setting $t = \log(1 + \chi^2(P, Q)) > 0$ yields that
$$\chi^2\left( \mL(M^{\text{id}, \sigma}), Q^{\otimes n \times n} \right) = \bE\left[ (1 + \chi^2(P, Q))^Y \right] - 1 \le \exp\left(\chi^2(P, Q)\right) - 1 \le 2 \cdot \chi^2(P, Q)$$
since $e^x \le 1 + 2x$ for $x \in [0, 1]$. Now by Cauchy-Schwarz we have that
$$\TV\left( \mL(M^{\text{id}, \sigma}), Q^{\otimes n \times n} \right) \le \frac{1}{2} \sqrt{\chi^2\left( \mL(M^{\text{id}, \sigma}), Q^{\otimes n \times n} \right)} \le \sqrt{\frac{\chi^2(P, Q)}{2}}$$
which completes the proof of the lemma.
\end{proof}

\subsection{Planted Clique Lifting}

In this section, we analyze the reduction $\textsc{PC-Lifting}$, which is given in Figure \ref{fig:pclifting}. This reduction will be shown to approximately take an instance of $\textsc{PC}(n, n^{1/2-\epsilon}, 1/2)$ to $\textsc{PC}(N, K, 1 - q)$ where $N = \tilde{\Theta}(n^{1 + \alpha/2})$, $K = \tilde{\Theta}(n^{1/2 + \alpha/2 - \epsilon})$ and $q = \tilde{\Theta}(n^{-\alpha})$. By taking the complement of the resulting graph, this shows planted clique lower bounds for $\textsc{PIS}_D(N, K, q)$ up to the boundary $\frac{N^2}{K^4} \gg q$, exactly matching the computational boundary stated in Theorem \ref{lem:2a}. The reduction $\textsc{PC-Lifting}$ proceeds iteratively, with $\textsc{PC}(n, k, p)$ approximately mapped at each step to $\textsc{PC}(2n, 2k, p^{1/4})$.

Given a labelled graph $G$ on $n$ vertices and a permutation $\sigma$ on $[n]$, let $G^{\sigma}$ denote the labelled graph formed by permuting the vertex labels of $G$ according to $\sigma$. Given disjoint subsets $S, T \subseteq [n]$, let $G[S]$ denote the induced subgraph on the set $S$ and $G[S \times T]$ denote the induced bipartite subgraph between $S$ and $T$. Also let $B(m, n, p)$ denote the random bipartite graph with parts of sizes $m$ and $n$, respectively, where each edge is included independently with probability $p$. Let $G(n, p, S)$, where $S$ is a $k$-subset of $[n]$, denote an instance of $G(n, k, p)$ in which the planted clique is conditioned to be on $S$.

\begin{figure}[t!]
\begin{algbox}
\textbf{Algorithm} \textsc{PC-Lifting}

\vspace{2mm}

\textit{Inputs}: Graph $G \in \mG_n$, number of iterations $\ell$, function $w$ with $w(n) \to \infty$
\begin{enumerate}
\item For each pair of vertices $\{i, j\} \not \in E(G)$, add the edge $\{i, j\}$ to $E(G)$ independently with probability $1 - 2 \cdot w(n)^{-1}$
\item Initialize $H \gets G$, $m \gets n$ and $p \gets 1 - w(n)^{-1}$
\item Repeat for $\ell$ iterations:
\begin{enumerate}
\item[a.] For each pair $\{ i, j \}$ of distinct vertices in $[m]$, sample $x^{ij} \in \{0, 1\}^4$ such that
\begin{itemize}
\item If $\{i, j\} \in E(H)$, then $x^{ij} = (1, 1, 1, 1)$
\item If $\{i, j\} \not \in E(H)$, then $x^{ij} = v$ with probability
$$\bP\left[x^{ij} = v\right] = \frac{p^{|v|_1/4} \left(1 - p^{1/4}\right)^{4 - |v|_1}}{1 - p}$$
for each $v \in \{0, 1\}^4$ with $v \neq (1, 1, 1, 1)$
\end{itemize}
\item[b.] Construct the graph $H'$ on the vertex set $[2m]$ such that for distinct $i, j \in [m]$
\begin{itemize}
\item $\{i, j \} \in E(H')$ if $x^{ij}_1 = 1$
\item $\{2m + 1 - i, j \} \in E(H')$ if $x^{ij}_2 = 1$
\item $\{i, 2m + 1 - j \} \in E(H')$ if $x^{ij}_3 = 1$
\item $\{2m + 1 - i, 2m + 1 - j \} \in E(H')$ if $x^{ij}_4 = 1$
\end{itemize}
and for each $i \in [m]$, add the edge $\{i, 2m + 1 - i \} \in E(H')$
\item[c.] Generate a permutation $\sigma$ on $[2m]$ uniformly at random
\item[d.] Update $H \gets (H')^\sigma, p \gets p^{1/4}$ and $m \gets 2m$
\end{enumerate}
\item Output $H$
\end{enumerate}
\vspace{1mm}
\end{algbox}
\caption{Planted clique lifting procedure in Lemma \ref{lem:4b}.}
\label{fig:pclifting}
\end{figure}

\begin{lemma}[Planted Clique Lifting] \label{lem:4b}
Suppose that $n$ and $\ell$ are such that $\ell = O(\log n)$ and are sufficiently large. Let $w(n) > 2$ be an increasing function with $w(n) \to \infty$ as $n \to \infty$. Then $\phi = \textsc{PC-Lifting}$ is a randomized polynomial time computable map $\phi : \mG_n \to \mG_{2^\ell n}$ such that under both $H_0$ and $H_1$, it holds that
$$\TV\left( \phi\left(\textsc{PC}(n, k, 1/2)\right), \textsc{PC}\left(2^\ell n, 2^\ell k, \left(1 -w(n)^{-1}\right)^{\frac{1}{4^\ell}}\right) \right) \le \frac{2}{\sqrt{w(n)}}$$
\end{lemma}

\begin{proof}
If $\ell = O(\log n)$, this algorithm runs in randomized polynomial time. Let $\varphi$ denote the function that applies a single iteration of Step 3 to an input graph $H$. Let $\phi_\ell$ be the algorithm that outputs the value of $H$ in $\phi$ after $\ell$ iterations of Step 3. Note that $\phi_0$ outputs $H$ after just applying Steps 1 and 2, and $\phi_{\ell+1} = \varphi \circ \phi_\ell$ for all $\ell \ge 0$.

We first consider a single iteration of Step 3 applied to $G \sim G(n, p, S)$, where $G(n, p, S)$ is the distribution of Erd\H{o}s-R\'{e}nyi graphs with a planted clique on a fixed vertex set $S \subseteq [n]$ of size $|S| = k$ and $p \ge 1/2$. For each pair of distinct $\{i, j\} \not \in \binom{S}{2}$, it holds that $\mathbf{1}_{\{i, j\} \in E(G)} \sim \text{Bern}(p)$ and hence by the probability in Step 3a, that $x^{ij} \sim \text{Bern}(p^{1/4})^{\otimes 4}$. Therefore the graph $H'$ constructed in Step 3b satisfies that:
\begin{itemize}
\item $S' = S \cup \{ 2n + 1 - i : i \in S\}$ forms a clique of size $2k$;
\item $\{2n + 1 - i, i\} \in E(H')$ for each $i \in [n]$; and
\item each other edge is in $E(H')$ independently with probability $p^{1/4}$.
\end{itemize}
Now consider the graph $\varphi(G) = H = (H')^\sigma$ conditioned on the set $\sigma(S')$. We will show that this graph is close in total variation to $G(2n, p^{1/4}, \sigma(S'))$. Let $T_1 = [n] \backslash S$ and $T_2 = [2n]\backslash \{ 2n + 1 - i : i \in S\}$. Note that every pair of vertices of the form $\{2n + 1 - i, i\}$ in $H'$ are either both in $S'$ or between $T_1$ and $T_2$. This implies that every pair of distinct vertices not in $\sigma(S')^2$ or $\sigma(T_1) \times \sigma(T_2)$ is in $E(H)$ independent with probability $p^{1/4}$, exactly matching the corresponding edges in $G(2n, p^{1/4}, \sigma(S'))$. Coupling these corresponding edges yields only the edges between $\sigma(T_1)$ and $\sigma(T_2)$ uncoupled. Therefore we have that
$$\TV\left( \mL(H | \sigma(S')), G\left(2n, p^{1/4}, \sigma(S')\right) \right) = \TV\left( \mL\left(H[\sigma(T_1) \times \sigma(T_2)]\right), B\left(n - k, n - k, p^{1/4}\right) \right)$$
Now let the $(n - k) \times (n - k)$ matrix $M$ have $1$'s on its main diagonal and each other entry sampled i.i.d. from $\text{Bern}(p^{1/4})$. If $\tau$ is a random permutation on $[n-k]$, then the adjacency matrix of $H[\sigma(T_1) \times \sigma(T_2)]$ conditioned on $\sigma(S')$ is distributed as $\mL\left( M^{\text{id}, \tau} \right)$, since $T_1$ and $T_2$ are disjoint. Therefore it follows that
\begin{align*}
\TV\left( \mL\left(H[\sigma(T_1) \times \sigma(T_2)]\right), B\left(n - k, n - k, p^{1/4}\right) \right) &= \TV\left( \mL\left( M^{\text{id}, \tau} \right), \text{Bern}(p^{1/4})^{\otimes(n-k) \times (n-k)} \right) \\
&\le \sqrt{\frac{\chi^2(\text{Bern}(1), \text{Bern}(p^{1/4}))}{2}} \\
&\le \sqrt{1-p^{1/4}}
\end{align*}
by Lemma \ref{lem:4a} and since $p \ge 1/2$. It follows by the triangle inequality that
$$\TV\left( \varphi(G(n, p, S)), G(2n, 2k, p^{1/4}) \right) \le \bE_{\sigma(S')} \left[ \TV\left(\mL(H | \sigma(S')), G\left(2n, p^{1/4}, \sigma(S')\right) \right) \right]$$
Letting $S$ be chosen uniformly at random over all subsets of $[n]$ of size $k$, applying the triangle inequality again and combining the inequalities above yields that
$$\TV\left( \varphi(G(n, k, p)), G(2n, 2k, p^{1/4}) \right) \le \bE_S\left[ \TV\left( \varphi(G(n, p, S)), G(2n, 2k, p^{1/4}) \right) \right] \le\sqrt{1-p^{1/4}}$$
A nearly identical but slightly simpler argument shows that
$$\TV\left( \varphi(G(n, p)), G(2n, p^{1/4}) \right) \le \sqrt{1-p^{1/4}}$$
For each $\ell \ge 0$, let $p_\ell = \left(1 -w(n)^{-1}\right)^{\frac{1}{4^\ell}}$ be the value of $p$ after $\ell$ iterations of Step 2. Now note that for each $\ell \ge 0$, we have by triangle inequality and data processing inequality that
\begin{align*}
\TV\left( \phi_{\ell + 1}\left(G(n, k, 1/2)\right), G\left(2^{\ell+1} n, 2^{\ell+1} k, p_{\ell+1} \right) \right) &\le \TV\left( \varphi\left(\phi_\ell\left(G(n, k, 1/2)\right) \right), \varphi\left(G\left(2^\ell n, 2^\ell k, p_\ell \right) \right) \right) \\
&\quad + \TV\left( \varphi\left(G\left(2^\ell n, 2^\ell k, p_\ell \right) \right), G\left(2^{\ell+1} n, 2^{\ell+1} k, p_{\ell+1} \right) \right) \\
&\le \TV\left( \phi_\ell\left(G(n, k, 1/2)\right), G\left(2^\ell n, 2^\ell k, p_\ell \right) \right) \\
&\quad + \sqrt{1-p_{\ell+1}}
\end{align*}
and an identical inequality for $\phi_\ell(G(n, 1/2))$. Noting that this total variation is zero when $\ell = 0$ and applying these inequalities inductively yields that
$$\TV\left( \phi_\ell\left(G(n, k, 1/2)\right), G\left(2^\ell n, 2^\ell k, p_\ell \right) \right) \le \sum_{i = 1}^\ell \sqrt{1-p_{i}}$$
and an identical inequality for $\phi_\ell(G(n, 1/2))$. Now note that if $x \le 1/2$ then $(1 - x)^{1/4} \ge 1 - x/3$. Iterating this inequality yields that $1 - p_{i} \le 3^{-i} w(n)^{-1}$. Therefore
$$\sum_{i = 1}^\ell \sqrt{1-p_{i}} \le \frac{1}{\sqrt{w(n)}} \sum_{i = 1}^\ell 3^{-i/2} < \frac{2}{\sqrt{w(n)}}$$
This completes the proof of the lemma.
\end{proof}

The next theorem formally gives the hardness result guaranteed by the reduction analyzed above together with the PC conjecture. There will be many theorems of this form throughout the paper, which will typically resolve to applying a total variation bound guaranteed in a previous lemma with Lemma \ref{lem:3a}, and analyzing the asymptotic regime of several parameters.

\begin{theorem} \label{lem:4c}
Let $\alpha \in [0, 2)$ and $\beta \in (0, 1)$ be such that $\beta < \frac{1}{2} + \frac{\alpha}{4}$. There is a sequence $\{ (N_n, K_n, q_n) \}_{n \in \mathbb{N}}$ of parameters such that:
\begin{enumerate}
\item The parameters are in the regime $q = \tilde{\Theta}(N^{-\alpha})$ and $K = \tilde{\Theta}(N^\beta)$ or equivalently,
$$\lim_{n \to \infty} \frac{\log q_n^{-1}}{\log N_n} = \alpha \quad \text{and} \quad \lim_{n \to \infty} \frac{\log K_n}{\log N_n} = \beta$$
\item For any sequence of randomized polynomial-time tests $\phi_n : \mG_{N_n} \to \{0, 1\}$, the asymptotic Type I$+$II error of $\phi_n$ on the problems $\textsc{PIS}_D(N_n, K_n, q_n)$ is at least $1$ assuming the PC conjecture holds with density $p = 1/2$.
\end{enumerate}
Therefore the computational boundary for $\textsc{PIS}_D(n, k, q)$ in the parameter regime $q = \tilde{\Theta}(n^{-\alpha})$ and $k = \tilde{\Theta}(n^\beta)$ is $\beta^* = \frac{1}{2} + \frac{\alpha}{4}$.
\end{theorem}

\begin{proof}
If $\beta < \alpha$ then PIS is information-theoretically impossible. Thus we may assume that $\beta \ge \alpha$. Let $\gamma = \frac{2\beta - \alpha}{2 - \alpha}$ and note that $\gamma \in (0, 1/2)$. Now set
$$\ell_n = \left\lceil \frac{\alpha \log_2 n}{2 - \alpha} \right\rceil, \quad \quad k_n = \lceil n^{\gamma} \rceil, \quad \quad N_n = 2^{\ell_n} n \quad \quad K_n = 2^{\ell_n} k_n, \quad \quad q_n = 1 - (1 - w(n)^{-1})^{1/4^{\ell_n}}$$
where $w(n)$ is any sub-polynomial increasing function tending to infinity. By Lemma \ref{lem:4b}, there is a randomized polynomial time algorithm mapping $\text{PC}_D(n, k_n, 1/2)$ to $\text{PC}_D(N_n, K_n, 1 - q_n)$ with total variation converging to zero as $n \to \infty$. Now note that flipping every edge to a non-edge and non-edge to an edge maps $\text{PC}_D(N_n, K_n, 1 - q_n)$ to $\text{PIS}_D(N_n, K_n, q_n)$. This map with Lemma 1 now implies that property 2 above holds. We now verify property 1. Note that
$$\lim_{n \to \infty} \frac{\log K_n}{\log N_n} = \lim_{n \to \infty} \frac{\left\lceil \frac{\alpha \log_2 n}{2 - \alpha} \right\rceil \cdot \log 2 + \left( \frac{2\beta - \alpha}{2 - \alpha} \right) \log n}{\left\lceil \frac{\alpha \log_2 n}{2 - \alpha} \right\rceil\cdot \log 2 + \log n} = \frac{\frac{\alpha}{2 - \alpha} + \frac{2\beta - \alpha}{2 - \alpha}}{\frac{\alpha}{2 - \alpha} + 1} = \beta$$
Note that as $n \to \infty$, it follows that since $4^{-\ell_n} \log(1 - w(n)^{-1}) \to 0$,
$$q_n = 1 - (1 - w(n)^{-1})^{1/4^{\ell_n}} = 1 - e^{4^{-\ell_n} \log(1 - w(n)^{-1})} \sim 4^{-\ell_n} \log(1 - w(n)^{-1})$$
Now it follows that
$$\lim_{n \to \infty} \frac{\log q_n^{-1}}{\log N_n} = \lim_{n \to \infty} \frac{2\left\lceil \frac{\alpha \log_2 n}{2 - \alpha} \right\rceil \log 2 - \log(1 - w(n)^{-1})}{\left\lceil \frac{\alpha \log_2 n}{2 - \alpha} \right\rceil\cdot \log 2 + \log n} = \frac{\frac{2\alpha}{2 - \alpha}}{\frac{\alpha}{2 - \alpha} + 1} = \alpha$$
which completes the proof.
\end{proof}

\section{Rejection Kernels and Distributional Lifting}

In this section, we generalize the idea in $\textsc{PC-Lifting}$ to apply to any distribution with a natural cloning operation, analogous to Step 3a in $\textsc{PC-Lifting}$. Before describing this general distributional lifting procedure, we first will establish several results on applying rejection kernels, a general method for changes of measure such as from Bernoulli edge indicators to Gaussians, that we will need throughout our reductions.

\subsection{Rejection Kernels}

All of our remaining reductions will involve approximately mapping from a pair of Bernoulli random variables, typically edge indicators in random graphs, to a given pair of random variables. Similar entry-wise transformations of measure were used in \cite{ma2015computational} and \cite{gao2017sparse} for mapping from Bernoulli random variables to Gaussian random variables. We generalize these maps to arbitrary distributions and give a simple algorithm using rejection sampling to implement them. The general objective is to construct a single randomized function $\textsc{rk} : \{0, 1\} \to \mathbb{R}$ that simultaneously maps $\text{Bern}(p)$ to the distribution $f_X$ and $\text{Bern}(q)$ to $g_X$, approximately in total variation distance. For maps from instances $G$ of planted clique, such a map with $p = 1$ and $q = 1/2$ approximately sends the edge indicators $\mathbf{1}_{\{ i, j \} \in E(G)}$ to $f_X$ if $i$ and $j$ are in the planted clique and to $g_X$ otherwise.

We first describe the general structure of the maps $\textsc{rk}$ and their precise total variation guarantees in the following lemma. Then we give particular rejection kernels that we will use in our reductions.

\begin{figure}[t!]
\begin{algbox}
\textbf{Algorithm} \textsc{rk}$(B)$

\vspace{2mm}

\textit{Parameters}: Input $B \in \{0, 1\}$, a pair of PMFs or PDFs $f_X$ and $g_X$ that can be efficiently computed and sampled, Bernoulli probabilities $p, q \in [0, 1]$, number of iterations $N$
\begin{enumerate}
\item Initialize $Y \gets 0$
\item For $N$ iterations do:
\begin{enumerate}
\item[a.] If $B = 0$, sample $Z \sim g_X$ and if
$$p \cdot g_X(Z) \ge q \cdot f_X(Z)$$
then with probability $1 - \frac{q \cdot f_X(Z)}{p \cdot g_X(Z)}$, update $Y \gets Z$ and break
\item[b.] If $B = 1$, sample $Z \sim f_X$ and if
$$(1 - q) \cdot f_X(Z) \ge (1 - p) \cdot g_X(Z)$$
then with probability $1 - \frac{(1 - p) \cdot g_X(Z)}{(1 - q) \cdot f_X(Z)}$, update $Y \gets Z$ and break
\end{enumerate}
\item Output $Y$
\end{enumerate}
\vspace{1mm}
\end{algbox}
\caption{Rejection kernel in Lemma \ref{lem:5zz}}
\label{fig:rej-kernel}
\end{figure}

\begin{lemma} \label{lem:5zz}
Let $f_X$ and $g_X$ be probability mass or density functions supported on subsets of $\mathbb{R}$ such that $g_X$ dominates $f_X$. Let $p, q \in [0, 1]$ be such that $p > q$ and let
$$S = \left\{ x \in \mathbb{R} : \frac{1 - p}{1 - q} \le \frac{f_X(x)}{g_X(x)} \le \frac{p}{q} \right\}$$
Suppose that $f_X(x)$ and $g_X(x)$ can be computed in $O(T_1)$ time and samples from $f_X$ and $g_X$ can be generated in randomized $O(T_2)$ time. Then there is a randomized $O(N(T_1 + T_2))$ time computable map $\textsc{rk} : \{0, 1\} \to \mathbb{R}$ such that $\TV\left( \textsc{rk}(\text{Bern}(p)), f_X \right) \le \Delta$ and $\TV\left( \textsc{rk}(\text{Bern}(q)), g_X \right) \le \Delta$ where
\begin{align*}
\Delta = \max &\left\{ \frac{\bP_{X \sim f_X} [X \not \in S]}{p - q} + \left( \bP_{X \sim g_X}[X \not \in S] + \frac{q}{p} \right)^{N}, \right. \\
&\, \, \, \, \left. \frac{\bP_{X \sim g_X} [X \not \in S]}{p - q} + \left( \bP_{X \sim f_X}[X \not \in S] + \frac{1 - p}{1 - q} \right)^{N} \right\}
\end{align*}
\end{lemma}

\begin{proof}
Let $\textsc{rk}$ be implemented as shown in Figure \ref{lem:5zz} and note that $\textsc{rk}$ runs in randomized $O(N(T_1 + T_2))$ time. Define $S_0$ and $S_1$ by
$$S_0 = \left\{ x \in \mathbb{R} : \frac{f_X(x)}{g_X(x)} \le \frac{p}{q} \right\} \quad \text{and} \quad S_1 = \left\{ x \in \mathbb{R} : \frac{1 - p}{1 - q} \le \frac{f_X(x)}{g_X(x)} \right\}$$
Now define the distributions by the densities or mass functions
\begin{align*}
\varphi_0(x) &= \frac{p \cdot g_X(x) - q \cdot f_X(x)}{p \cdot \bP_{X \sim g_X} [X \in S_0] - q \cdot \bP_{X \sim f_X} [X \in S_0]} \quad \text{for } x \in S_0 \\
\varphi_1(x) &= \frac{(1 - q) \cdot f_X(x) - (1 - p) \cdot g_X(x)}{(1 - q) \cdot \bP_{X \sim f_X} [X \in S_1] - (1 - p) \cdot \bP_{X \sim g_X} [X \in S_1]} \quad \text{for } x \in S_1
\end{align*}
that are both zero elsewhere. Note that these are both well-defined PDFs or PMFs since they are nonnegative and normalized by the definitions of $S_0$ and $S_1$. First consider the case when $B = 0$. For the sake of this analysis, consider the iterations of Step 2 beyond the first update $Y \gets Z$. Now let $A_i$ be the event that the update $Y \gets Z$ occurs in the $i$th iteration of Step 2a and let $A_i'$ be the event that the first update occurs in the $i$th iteration. Note that $A_i' = A_1^C \cap A_2^C \cap \cdots \cap A_{i-1}^C \cap A_i$. The probability of $A_i$ is
$$\bP[A_i] = \int_{S_0} g_X(x) \left( 1 - \frac{q \cdot f_X(x)}{p \cdot g_X(x)} \right) dx = \bP_{X \sim g_X}[X \in S_0] - \frac{q}{p} \cdot \bP_{X \sim f_X}[X \in S_0]$$
Now note that since the sample $Z$ in the $i$th iteration and $A_i$ are independent of $A_1, A_2, \dots, A_{i-1}$, it holds that $f_{Y|A_i'} = f_{Y|A_i}$. The density of $Y$ given the event $A_i'$ is therefore given by
$$f_{Y|A_i'}(x) = f_{Y|A_i}(x) = \bP[A_i]^{-1} \cdot g_{X}(x) \cdot \left( 1 - \frac{q \cdot f_X(x)}{p \cdot g_X(x)} \right) = \varphi_0(x)$$
for each $x \in S_0$. If $A = A_1' \cup A_2' \cup \cdots \cup A_N' = A_1 \cup A_2 \cup \cdots \cup A_N$ is the event that the update $Y \gets Z$ occurs in an iteration of Step 2a, then it follows by independence that
\begin{align*}
\bP\left[A^C\right] &= \prod_{i = 1}^N \left(1 - \bP[A_i]\right) = \left( 1 - \bP_{X \sim g_X}[X \in S_0] + \frac{q}{p} \cdot \bP_{X \sim f_X}[X \in S_0] \right)^N \\
&\le \left( \bP_{X \sim g_X}[X \not \in S] + \frac{q}{p} \right)^{N}
\end{align*}
since $S \subseteq S_0$. Note that $f_{Y|A}(x) = \varphi_0(x)$ and $\textsc{rk}(0)$ is $Y$ if $B = 0$. Therefore it follows by Lemma \ref{lem:5tv} that
$$\TV\left( \textsc{rk}(0), \varphi_0 \right) = \bP\left[A^C \right] \le \left( \bP_{X \sim g_X}[X \not \in S] + \frac{q}{p} \right)^{N}$$
A symmetric argument shows that when $B = 1$,
$$\TV\left( \textsc{rk}(1), \varphi_1 \right) \le \left( \bP_{X \sim f_X}[X \not \in S] + \frac{1 - p}{1 - q} \right)^{N}$$
Now note that
\begin{align*}
\left\| \varphi_0 - \frac{p \cdot g_X - q \cdot f_X}{p - q} \right\|_1 &= \int_{S_0} \left| \frac{p \cdot g_X(x) - q \cdot f_X(x)}{p \cdot \bP_{X \sim g_X} [X \in S_0] - q \cdot \bP_{X \sim f_X} [X \in S_0]} - \frac{p \cdot g_X(x) - q \cdot f_X(x)}{p - q} \right| dx \\
&\quad \quad + \int_{S_0^C} \frac{q \cdot f_X(x) - p \cdot g_X(x)}{p - q} dx \\
&= \left| 1 - \frac{p \cdot \bP_{X \sim g_X} [X \in S_0] - q \cdot \bP_{X \sim f_X} [X \in S_0]}{p - q} \right| \\
&\quad \quad + \frac{q \cdot \bP_{X \sim f_X} [X \not \in S_0] - p \cdot \bP_{X \sim g_X} [X \not \in S_0]}{p - q} \\
&= \frac{2(q \cdot \bP_{X \sim f_X} [X \not \in S_0] - p \cdot \bP_{X \sim g_X} [X \not \in S_0])}{p - q} \\
&\le \frac{2 \cdot \bP_{X \sim f_X}[ X \not \in S]}{p - q}
\end{align*}
since $S_0 \subseteq S$. A similar computation shows that
\begin{align*}
\left\| \varphi_1 - \frac{(1 - q) \cdot f_X - (1 - p) \cdot g_X}{p - q} \right\|_1 &= \frac{2((1 - p) \cdot \bP_{X \sim g_X} [X \not \in S_1] - (1 - q) \cdot \bP_{X \sim f_X} [X \not \in S_1])}{p - q} \\
&\le \frac{2 \cdot \bP_{X \sim g_X}[X \not \in S]}{p - q} \\
\end{align*}
Now note that
$$f_X = p \cdot \frac{(1 - q) \cdot f_X - (1 - p) \cdot g_X}{p - q} + (1 - p) \cdot \frac{p \cdot g_X - q \cdot f_X}{p - q}$$
Therefore by the triangle inequality, we have that
\begin{align*}
\TV\left( \textsc{rk}(\text{Bern}(p)), f_X \right) &\le \TV\left( \textsc{rk}(\text{Bern}(p)), p \cdot \varphi_1 + (1 - p) \cdot \varphi_0 \right) + \TV\left( p \cdot \varphi_1 + (1 - p) \cdot \varphi_0, f_X \right) \\
&\le p \cdot \TV\left( \textsc{rk}(1), \varphi_1\right) + \frac{p}{2} \cdot \left\| \varphi_1 - \frac{(1 - q) \cdot f_X - (1 - p) \cdot g_X}{p - q} \right\|_1 \\
&\quad \quad + (1 - p) \cdot \TV\left( \textsc{rk}(0), \varphi_0 \right) + \frac{1 - p}{2} \cdot \left\| \varphi_0 - \frac{p \cdot g_X - q \cdot f_X}{p - q} \right\|_1 \\
&\le p \cdot \left( \frac{\bP_{X \sim g_X} [X \not \in S]}{p - q} + \left( \bP_{X \sim f_X}[X \not \in S] + \frac{1 - p}{1 - q} \right)^{N} \right) \\
&\quad \quad + (1 - p) \cdot \left( \frac{\bP_{X \sim f_X} [X \not \in S]}{p - q} + \left( \bP_{X \sim g_X}[X \not \in S] + \frac{q}{p} \right)^{N} \right) \\
&\le \Delta
\end{align*}
Similarly, note that
$$g_X = q \cdot \frac{(1 - q) \cdot f_X - (1 - p) \cdot g_X}{p - q} + (1 - q) \cdot \frac{p \cdot g_X - q \cdot f_X}{p - q}$$
The same triangle inequality applications as above show that
\begin{align*}
\TV\left( \textsc{rk}(\text{Bern}(q)), g_X \right) &\le q \cdot \left( \frac{\bP_{X \sim f_X} [X \not \in S]}{p - q} + \left( \bP_{X \sim g_X}[X \not \in S] + \frac{q}{p} \right)^{N} \right) \\
&\quad \quad + (1 - q) \cdot \left( \frac{\bP_{X \sim g_X} [X \not \in S]}{p - q} + \left( \bP_{X \sim f_X}[X \not \in S] + \frac{1 - p}{1 - q} \right)^{N} \right) \\
&\le \Delta
\end{align*}
completing the proof of the lemma.
\end{proof}

We will denote $\textsc{rk}$ as defined with the parameters in the lemma above as $\textsc{rk}(p \to f_X, q \to g_X, N)$ from this point forward. We now give the particular rejection kernels we will need in our reductions and their total variation guarantees. The proofs of these guarantees are deferred to Appendix \ref{app5}. The first rejection kernel maps from the edge indicators in planted clique to Poisson random variables and will be essential in Poisson lifting.

\begin{lemma} \label{lem:5a}
Let $n$ be a parameter and let $\epsilon > 0, c > 1$ and $q \in (0, 1)$ be fixed constants satisfying that $3\epsilon^{-1} \le \log_c q^{-1}$. If $\lambda = \lambda(n)$ satisfies that $0 < \lambda \le n^{-\epsilon}$, then the map
$$\textsc{rk}_{\text{P1}} = \textsc{rk}(1 \to \text{Pois}(c\lambda), q \to \text{Pois}(\lambda), N)$$
where $N = \lceil 6 \log_{q^{-1}} n \rceil$ can be computed in $O(\log n)$ time and satisfies that
$$\TV\left(\textsc{rk}_{\text{P1}}(1), \text{Pois}(c\lambda) \right) = O_n(n^{-3}) \quad \text{and} \quad \TV\left(\textsc{rk}_{\text{P1}}(\text{Bern}(q)), \text{Pois}(\lambda) \right) = O_n(n^{-3})$$
\end{lemma}

The next lemma gives another approximate map to Poisson random variables from Bernoulli random variables corresponding to the edge indicators in the edge-dense regime of planted dense subgraph. We use the following lemma to apply Poisson lifting after Gaussian lifting in order to deduce hardness in the general regime of PDS. The proof is very similar to that of Lemma \ref{lem:5a}. In typical applications of the lemma, we take $\lambda = n^{-\epsilon}$ and $c - 1 = \Theta(\rho)$ where $\rho \to 0$. 

\begin{lemma} \label{lem:5b}
Let $\epsilon \in (0, 1)$ be a fixed constant and let $n$ be a parameter. Suppose that:
\begin{itemize}
\item $\lambda = \lambda(n)$ satisfies that $0 < \lambda \le n^{-\epsilon}$;
\item $c = c(n) > 1$ satisfies that $c = O_n(1)$; and
\item $\rho = \rho(n) \in (0, 1/2)$ satisfies that $\rho \ge n^{-K}$ for sufficiently large $n$ where $K = \Theta_n(1)$ is positive and
$$(K + 3)\epsilon^{-1} \le \log_c (1 + 2\rho) = O_n(1)$$
\end{itemize}
Then the map
$$\textsc{rk}_{\text{P2}} = \textsc{rk}\left(\frac{1}{2} + \rho \to \text{Pois}(c\lambda), \frac{1}{2} \to \text{Pois}(\lambda), N \right)$$
where $N = \left\lceil 6\rho^{-1} \log n \right\rceil$ can be computed in $\text{poly}(n)$ time and satisfies
\begin{align*}
\TV\left(\textsc{rk}_{\text{P2}}(\text{Bern}(1/2 + \rho)), \text{Pois}(c\lambda) \right) &= O_n(n^{-3}), \quad \text{and} \\
\TV\left(\textsc{rk}_{\text{P2}}(\text{Bern}(1/2)), \text{Pois}(\lambda) \right) &= O_n(n^{-3})
\end{align*}
\end{lemma}

The next lemma of this section approximately maps from Bernoulli to Gaussian random variables, yielding an alternative to the Gaussianization step in \cite{ma2015computational}. Gaussian random variables appear in two different contexts in our reductions: (1) the problems $\textsc{ROS}, \textsc{SPCA}$ and $\textsc{BC}$ have observations sampled from multivariate Gaussians; and (2) random matrices with Gaussian entries are used as intermediate in our reductions to $\textsc{PDS}$ in the general regime and $\textsc{SSBM}$. In both cases, we will need the map in the following lemma. As in the proofs of the previous two lemmas, this next lemma also verifies the conditions of Lemma \ref{lem:5zz} for Gaussians and derives an upper bound on $\Delta$.

\begin{lemma} \label{lem:5c}
Let $n$ be a parameter and suppose that $p = p(n)$ and $q = q(n)$ satisfy that $p > q$, $p, q \in [0, 1]$, $\min(q, 1 - q) = \Omega_n(1)$ and $p - q \ge n^{-O_n(1)}$. Let $\delta = \min \left\{ \log \left( \frac{p}{q} \right), \log \left( \frac{1 - q}{1 - p} \right) \right\}$. Suppose that $\mu = \mu(n) \in (0, 1)$ is such that
$$\mu \le \frac{\delta}{2 \sqrt{6\log n + 2\log (p-q)^{-1}}}$$
Then the map
$$\textsc{rk}_{\text{G}} = \textsc{rk}\left(p \to N(\mu, 1), q \to N(0, 1), N \right)$$
where $N = \left\lceil 6\delta^{-1} \log n \right\rceil$ can be computed in $\text{poly}(n)$ time and satisfies
$$\TV\left(\textsc{rk}_{\text{G}}(\text{Bern}(p)), N(\mu, 1) \right) = O_n(n^{-3}) \quad \text{and} \quad \TV\left(\textsc{rk}_{\text{G}}(\text{Bern}(q)), N(0, 1) \right) = O_n(n^{-3})$$
\end{lemma}

\subsection{Distributional Lifting}

In this section, we introduce a general distributional lifting procedure to reduce from an instance of planted clique to subgraph problems with larger planted subgraphs. The key inputs to the procedure are two parameterized families of distributions $P_\lambda$ and $Q_\lambda$ that have a natural cloning map, as described below.

The general distributional lifting procedure begins with an instance $G \in \mG_n$ of a planted dense subgraph problem such as planted clique and applies a rejection kernel element-wise to its adjacency matrix. This yields a symmetric matrix $M$ with zeros on its main diagonal, i.i.d. entries sampled from $P_{\lambda_0}$ on entries corresponding to clique edges and i.i.d. entries sampled from $Q_{\lambda_0}$ elsewhere. As an input to the procedure, we assume a random cloning map $f_{\text{cl}}$ that exactly satisfies
$$f_{\text{cl}}(P_\lambda) \sim P_{g_{\text{cl}(\lambda)}}^{\otimes 4} \quad \text{and} \quad f_{\text{cl}}(Q_\lambda) \sim Q_{g_{\text{cl}(\lambda)}}^{\otimes 4}$$
for some parameter update function $g_{\text{cl}}$. There is a natural cloning map $f_{\text{cl}}$ for Gaussian and Poisson distributions, the two families we apply distributional lifting with. Applying this cloning map entry-wise to $M$ and arranging the resulting entries correctly yields a matrix $M$ of size $2n \times 2n$ with zeros on its diagonal and a planted submatrix of size $2k \times 2k$. The only distributional issue that arises are the anti-diagonal entries, which are now all from $Q_{g_{\text{cl}(\lambda)}}$ although some should be from $P_{g_{\text{cl}(\lambda)}}$. We handle these approximately in total variation by randomly permuting the rows and columns and applying Lemma \ref{lem:4a}. Iterating this procedure $\ell$ times yields a matrix $M'$ of size $2^\ell n \times 2^\ell n$ with a planted submatrix of size $2^\ell k \times 2^\ell k$. If $\lambda_{i+1} = g_{\text{cl}}(\lambda_i)$, then $M'$ has all i.i.d. entries from $Q_{\lambda_\ell}$ under $H_0$ and a planted submatrix with i.i.d. entries from $P_{\lambda_\ell}$ under $H_1$. We then threshold the entries of $M'$ to produce the adjacency matrix of a graph with i.i.d. edge indicators, conditioned on the vertices in the planted subgraph.

A natural question is: what is the purpose of the families $P_{\lambda}$ and $Q_{\lambda}$? If the initial and final distributions are both graph distributions with Bernoulli edge indicators, it a priori seems unnecessary to use matrix distributions without Bernoulli entries as intermediates. Our main reason for introducing these intermediate distributions is that they achieve the right parameter tradeoffs to match the best known algorithms for PDS, while cloning procedures that stay within the set of graph distributions do not. Consider the target planted dense subgraph instance of $\text{PDS}(n, k, p, q)$ where $p = 2 q$ and $q = \tilde{\Theta}(n^{-\alpha})$. To produce lower bounds tight with the computational barrier of $q \approx n^2/k^4$ in Theorem \ref{lem:2a}, a lifting procedure mapping $n \to 2n$ and $k \to 2k$ at each step would need its cloning map to satisfy
$$f_{\text{cl}}(\text{Bern}(q)) \sim Q = \text{Bern}(q/4)^{\otimes 4} \quad \text{and} \quad f_{\text{cl}}(\text{Bern}(p)) \sim P = \text{Bern}(p/4)^{\otimes 4}$$
where $p = 2q$. It is not difficult to verify that for any random map $f_{\text{cl}} : \{0, 1\} \to \{0, 1\}^4$, it would need to hold that
$$\frac{1 - p}{1 - q} \le \frac{P(x)}{Q(x)} \le \frac{p}{q} \quad \text{for all } x \in \{0, 1\}^4$$
However, $P(1, 1, 1, 1)/Q(1, 1, 1, 1) = 16 > 2 = p/q$, so no such map can exist. Another approach is to relax $f_{\text{cl}}$ to be an approximate map like a rejection kernel. However, this seems to induce a large loss in total variation from the target distribution on $M'$. Our solution is to use a rejection kernel to map to distributions with natural cloning maps $f_{\text{cl}}$, such as a Poisson or Gaussian distribution, to front-load the total variation loss to this approximate mapping step and induce no entry-wise total variation loss later in the cloning procedure. We choose the precise distributions $P_{\lambda}$ and $Q_{\lambda}$ to match the parameter tradeoff along the computational barrier. Note that this general distributional lifting procedure can also be used to map to problems other than variants of subgraph detection, such as biclustering, by not truncating in Step 4. We remark that the $\textsc{PC-Lifting}$ reduction presented in the previous section is almost an instance of distributional lifting with $P_{\lambda} = \text{Bern}(1)$ for all $\lambda$ and $Q_{\lambda} = \text{Bern}(\lambda)$, with the parameter update $g_{\text{cl}}(\lambda) = \lambda^{1/4}$. However in $\textsc{PC-Lifting}$, the planted anti-diagonal entries between the vertices $i$ and $2m + 1 - i$ in Step 3b are from the planted distribution $P_\lambda$, rather than $Q_\lambda$ as in distributional lifting. This requires a slightly different analysis of the planted anti-diagonal entries with Lemma \ref{lem:4a}.

We now proceed to describe distributional lifting shown in Figure \ref{fig:distlift} and prove its guarantees. Given two distributions $P$ and $Q$, let $M_n(Q)$ denote the distribution on $n \times n$ symmetric matrices with zero diagonal entries and every entry below the diagonal sampled independently from $Q$. Similarly, let $M_n(S, P, Q)$ denote the distribution on random $n \times n$ symmetric matrices formed by:
\begin{enumerate}
\item sampling the entries of the principal submatrix with indices in $S$ below its main diagonal independently from $P$;
\item sampling all other entries below the main diagonal independently from $Q$; and
\item placing zeros on the diagonal.
\end{enumerate}
Let $M_n(k, P, Q)$ denote the distribution of matrices $M_n(S, P, Q)$ where $S$ is a size $k$ subset of $[n]$ selected uniformly at random. Given a matrix $M \in \mathbb{R}^{n \times n}$ and index sets $S, T \subseteq [n]$, let $M[S \times T]$ denote the $|S| \times |T|$ submatrix of $M$ with row indices in $S$ and column indices in $T$. Also let $G(n, p, q, S)$ where $S$ is a $k$-subset of $[n]$ denote an instance of $G(n, k, p, q)$ where the planted dense subgraph is conditioned to be on $S$. The guarantees of distributional lifting are as follows.

\begin{figure}[t!]
\begin{algbox}
\textbf{Algorithm} \textsc{Distributional-Lifting}

\vspace{2mm}

\textit{Inputs}: Graph $G \in \mG_n$, number of iterations $\ell$, parameterized families of target planted and noise distributions $P_\lambda$ and $Q_{\lambda}$, a TV-approximation $Q'_{\lambda}$ to $Q_{\lambda}$ that can be efficiently sampled, rejection kernel $\textsc{rk}$ approximately mapping $\text{Bern}(p') \to P_{\lambda_0}$ and $\text{Bern}(q') \to Q_{\lambda_0}$, threshold $t$, cloning map $f_{\text{cl}}$ and corresponding parameter map $g_{\text{cl}}$
\begin{enumerate}
\item Form the symmetric matrix $M \in \mathbb{R}^{n \times n}$ with $M_{ii} = 0$ and off-diagonal terms
$$M_{ij} = \textsc{rk}\left(\mathbf{1}_{\{i, j \} \in E(G)}\right)$$
\item Initialize $W \gets M$, $m \gets n$ and $\lambda \gets \lambda_0$
\item Repeat for $\ell$ iterations:
\begin{enumerate}
\item[a.] For each pair of distinct $i, j \in [m]$, let $(x_{ij}^1, x_{ij}^2, x_{ij}^3, x_{ij}^4) = f_{\text{cl}}(W_{ij})$
\item[b.] Let $W' \in\mathbb{R}^{2m \times 2m}$ be the symmetric matrix with $W_{ii}' = 0$ and
\begin{align*}
W'_{ij} &= x^1_{ij} \\
W'_{(2m+1 - i)j} &=  x^2_{ij} \\
W'_{i(2m + 1 - j)} &=  x^3_{ij} \\
W'_{(2m+1 - i)(2m + 1 - j)} &=  x^4_{ij}
\end{align*}
for all distinct $i, j \in [m]$ and
$$W'_{i, 2m+1 - i} \sim_{\text{i.i.d.}} Q'_{\lambda}$$
for all $i \in [m]$
\item[c.] Generate a permutation $\sigma$ on $[2m]$ uniformly at random
\item[d.] Update $W \gets (W')^{\sigma, \sigma}$, $m \gets 2m$ and $\lambda \gets g_{\text{cl}}(\lambda)$
\end{enumerate}
\item Output the graph $H$ with $\{i, j \} \in E(H)$ if $W_{ij} > t$
\end{enumerate}
\vspace{1mm}
\end{algbox}
\caption{Distributional lifting procedure in Theorem \ref{lem:5cc}.}
\label{fig:distlift}
\end{figure}

\begin{theorem}[Distributional Lifting] \label{lem:5cc}
Suppose that $n$ and $\ell$ are such that $\ell = O(\log n)$ and are sufficiently large. Let $p', q' \in [0, 1]$ and define the parameters:
\begin{itemize}
\item target planted and noise distribution families $P_{\lambda}$ and $Q_{\lambda}$ parameterized by $\lambda$;
\item a rejection kernel $\textsc{rk}$ that can be computed in randomized $\text{poly}(n)$ time and parameter $\lambda_0$ such that $\textsc{rk}(\text{Bern}(p')) \sim \tilde{P}_{\lambda_0}$ and $\textsc{rk}(\text{Bern}(q')) \sim \tilde{Q}_{\lambda_0}$;
\item a cloning map $f_{\text{cl}}$ that can be computed in randomized $\text{poly}(n)$ time and parameter map $g_{\text{cl}}$ such that
$$f_{\text{cl}}(P_\lambda) \sim P_{g_{\text{cl}(\lambda)}}^{\otimes 4} \quad \text{and} \quad f_{\text{cl}}(Q_\lambda) \sim Q_{g_{\text{cl}(\lambda)}}^{\otimes 4}$$
for each parameter $\lambda$;
\item a randomized $\text{poly}(n)$ time algorithm for sampling from $Q'_{\lambda_i}$ for each $1 \le i \le \ell$ where the sequence of parameters $\lambda_i$ are such that $\lambda_{i+1} = g_{\text{cl}}(\lambda_i)$ for each $i$; and
\item a threshold $t \in \mathbb{R}$.
\end{itemize}
Then $ \phi = \textsc{Distributional-Lifting}$ with these parameters is a randomized polynomial time computable map $\phi : \mG_n \to \mG_{2^\ell n}$ such that under both $H_0$ and $H_1$, it holds that
\begin{align*}
\TV\left( \phi(\textsc{PDS}(n, k, p', q')), \textsc{PDS}\left(2^\ell n, 2^\ell k, p, q \right) \right) &\le \binom{n}{2} \cdot \max\left\{ \TV\left( \tilde{P}_{\lambda_0}, P_{\lambda_0} \right), \TV\left( \tilde{Q}_{\lambda_0}, Q_{\lambda_0} \right) \right\} \\
&\quad \quad + \sum_{i = 1}^{\ell} \left( 2^i n \cdot \TV\left(Q_{\lambda_i}, Q'_{\lambda_i}\right) + \sqrt{\frac{\chi^2(Q_{\lambda_i}, P_{\lambda_i})}{2}} \right)
\end{align*}
where $p = \bP_{X \sim P_{\lambda_\ell}}[ X > t]$ and $q = \bP_{X \sim Q_{\lambda_\ell}}[ X > t]$.
\end{theorem}

\begin{proof}
If $\ell = O(\log n)$, the algorithm $\textsc{Distributional-Lifting}$ runs in randomized polynomial time. Let $\phi_i(W)$ be the algorithm that outputs the value of $W$ after $i$ iterations of Step 4 given the value of $W$ in Step 2 as its input. Let $\phi'_i(G)$ be the algorithm that outputs the value of $W$ after $i$ iterations of Step 4 given the original graph $G$ as its input. Note that $\phi'_0$ outputs the value of $M$ after Step 1.

We first consider an iteration $\phi_1(M)$ of Step 3 applied to $M \sim M_m(S, P_{\lambda}, Q_{\lambda})$ where $|S| = k$. By the definition of $f_{\text{cl}}$, if $i, j$ are distinct and both in $S$, then $(x^1_{ij}, x^2_{ij}, x^3_{ij}, x^4_{ij}) \sim P_{g_{\text{cl}}(\lambda)}^{\otimes 4}$. Similarly, if at least one of $i$ or $j$ is not in $S$, then $(x^1_{ij}, x^2_{ij}, x^3_{ij}, x^4_{ij}) \sim Q_{g_{\text{cl}}(\lambda)}^{\otimes 4}$. Therefore the symmetric matrix $W'$ constructed in Step 3b has independent entries below its main diagonal and satisfies that:
\begin{itemize}
\item $W'_{ij} \sim P_{g_{\text{cl}}(\lambda)}$ for all distinct $i, j \in S' = S \cup \{ 2m + 1 - i : i \in S\}$ with $i + j \neq 2m + 1$;
\item $W'_{ij} \sim Q_{g_{\text{cl}}(\lambda)}$ for all distinct $(i, j) \not \in S' \times S'$;
\item $W'_{ij} \sim Q'_{g_{\text{cl}}(\lambda)}$ with $i + j = 2m + 1$; and
\item $W'_{ii} = 0$.
\end{itemize}
Let $W_r'$ be the matrix with each of its entries identically distributed to those of $W'$ except $(W_r')_{ij} \sim Q_{g_{\text{cl}}(\lambda)}$ if $i + j = 2m + 1$. Coupling entries individually yields that
$$\TV(\mL(W'), \mL(W_r')) \le m \cdot \TV\left( Q_{g_{\text{cl}}(\lambda)}, Q'_{g_{\text{cl}}(\lambda)} \right)$$
Now consider the matrix $W_r = (W_r')^{\sigma, \sigma}$ conditioned on the two sets $\sigma(S)$ and $\sigma(S' \backslash S)$ where $\sigma$ is a uniformly at random chosen permutation on $[2m]$. We will show that this matrix is close in total variation to $M_{2m}(\sigma(S'), P_{g_{\text{cl}}(\lambda)}, Q_{g_{\text{cl}}(\lambda)})$. Note that fully conditioned on $\sigma$, the entries of $W_r$ below the main diagonal are independent and identically distributed to $M_{2m}(\sigma(S'), P_{g_{\text{cl}}(\lambda)}, Q_{g_{\text{cl}}(\lambda)})$ other than the entries with indices $(\sigma(i), \sigma(2m + 1 - i))$ where $i \in S'$. These entries are distributed as $Q_{g_{\text{cl}}(\lambda)}$ in $W_r$ conditioned on $\sigma$ and as $P_{g_{\text{cl}}(\lambda)}$ in the target distribution $M_{2m}(\sigma(S'), P_{g_{\text{cl}}(\lambda)}, Q_{g_{\text{cl}}(\lambda)})$. Marginalizing to only condition on the sets $\sigma(S)$ and $\sigma(S', S)$, yields that all entries $(W_r)_{ij}$ with $(i, j) \not \in \sigma(S) \times \sigma(S', S) \cup \sigma(S', S) \times \sigma(S)$ are identically distributed in $W_r | \{ \sigma(S), \sigma(S' \backslash S)\}$ and the target distribution. Coupling these corresponding entries yields that the total variation between $W_r | \{ \sigma(S), \sigma(S' \backslash S)\}$ and the target distribution satisfies that
\begin{align*}
\TV\left( \mL(W_r|\sigma(S), \sigma(S' \backslash S)), \right. &\left. M_{2m}(\sigma(S'), P_{g_{\text{cl}}(\lambda)}, Q_{g_{\text{cl}}(\lambda)}) \right) \\
&= \TV\left( \mL(W_r[\sigma(S) \times \sigma(S' \backslash S)]), M_k(P_{g_{\text{cl}}(\lambda)}) \right)
\end{align*}
Now observe that $W_r[\sigma(S) \times \sigma(S' \backslash S)]$ is distributed as $\mL(A^{\text{id}, \tau})$ where $\tau$ is permutation of $[k]$ selected uniformly at random and $A$ is a $k \times k$ matrix with its diagonal entries i.i.d. $Q_{g_{\text{cl}}(\lambda)}$ and its other entries i.i.d. $P_{g_{\text{cl}}(\lambda)}$. By Lemma \ref{lem:4a}, we therefore have that
\begin{align*}
\TV\left( \mL(W_r[\sigma(S) \times \sigma(S' \backslash S)]), M_k(P_{g_{\text{cl}}(\lambda)}) \right) &= \TV\left( \mL(A^{\text{id}, \tau}), M_k(P_{g_{\text{cl}}(\lambda)}) \right) \\
&\le \sqrt{\frac{1}{2} \cdot \chi^2\left( Q_{g_{\text{cl}}(\lambda)}, P_{g_{\text{cl}}(\lambda)}\right)}
\end{align*}
Now consider the matrix $\phi_1(M) = W = (W')^{\sigma, \sigma}$. By the data processing inequality, we have that
\begin{align*}
\TV\left( \mL(W_r|\sigma(S), \sigma(S' \backslash S)), \mL(W|\sigma(S), \sigma(S' \backslash S)\right) &\le \TV(\mL(W'), \mL(W_r')) \\
&\le m \cdot \TV\left( Q_{g_{\text{cl}}(\lambda)}, Q'_{g_{\text{cl}}(\lambda)} \right)
\end{align*}
The triangle inequality now implies that
\begin{align*}
\TV\left( \mL(W|\sigma(S), \sigma(S' \backslash S)), \right. &\left. M_{2m}(\sigma(S'), P_{g_{\text{cl}}(\lambda)}, Q_{g_{\text{cl}}(\lambda)}) \right) \\
&\le m \cdot \TV\left( Q_{g_{\text{cl}}(\lambda)}, Q'_{g_{\text{cl}}(\lambda)} \right) + \sqrt{\frac{1}{2} \cdot \chi^2\left( Q_{g_{\text{cl}}(\lambda)}, P_{g_{\text{cl}}(\lambda)}\right)}
\end{align*}
Letting $S$ be chosen uniformly at random over all subsets of $[n]$ of size $k$ and the triangle inequality now imply that
\begin{align*}
&\TV\left( \phi_1(M_m(k, P_{\lambda}, Q_{\lambda})), M_{2m}(2k, P_{g_{\text{cl}}(\lambda)}, Q_{g_{\text{cl}}(\lambda)}) \right) \\
&\quad \quad \quad \le \bE_{S} \bE_{\sigma(S), \sigma(S'\backslash S)} \left[ \TV\left( \mL(W|\sigma(S), \sigma(S' \backslash S)), M_{2m}(\sigma(S'), P_{g_{\text{cl}}(\lambda)}, Q_{g_{\text{cl}}(\lambda)}) \right) \right] \\
&\quad \quad \quad \le m \cdot \TV\left( Q_{g_{\text{cl}}(\lambda)}, Q'_{g_{\text{cl}}(\lambda)} \right) + \sqrt{\frac{1}{2} \cdot \chi^2\left( Q_{g_{\text{cl}}(\lambda)}, P_{g_{\text{cl}}(\lambda)}\right)}
\end{align*}
For each $i \ge 0$, combining this inequality with the triangle inequality and data processing inequality yields that
\begin{align*}
&\TV\left( \phi_{i+1}(M_m(k, P_{\lambda_0}, Q_{\lambda_0})), M_{2^{i+1} m}\left(2^{i+1} k, P_{\lambda_{i+1}}, Q_{\lambda_{i+1}}\right) \right) \\
&\quad \quad \quad \le \TV\left( \phi_1 \circ \phi_{i}(M_m(k, P_{\lambda_0}, Q_{\lambda_0})), \phi_1\left( M_{2^{i} m}\left(2^{i} k, P_{\lambda_{i}}, Q_{\lambda_{i}}\right) \right) \right) \\
&\quad \quad \quad \quad + \TV\left( \phi_1\left( M_{2^{i} m}\left(2^{i} k, P_{\lambda_{i}}, Q_{\lambda_{i}}\right) \right), M_{2^{i+1} m}\left(2^{i+1} k, P_{\lambda_{i+1}}, Q_{\lambda_{i+1}}\right) \right) \\
&\quad \quad \quad \le \TV\left( \phi_{i}(M_m(k, P_{\lambda_0}, Q_{\lambda_0})), M_{2^{i} m}\left(2^{i} k, P_{\lambda_{i}}, Q_{\lambda_{i}}\right) \right) \\
&\quad \quad \quad \quad + 2^i m \cdot \TV\left( Q_{\lambda_{i+1}}, Q'_{\lambda_{i+1}} \right) + \sqrt{\frac{1}{2} \cdot \chi^2\left( Q_{\lambda_{i+1}}, P_{\lambda_{i+1}}\right)}
\end{align*}
Now note that the adjacency matrix $A_{ij}(G) = \mathbf{1}_{\{i, j \} \in E(G)}$ of $G \sim G(n, p', q', S)$ is distributed as $M_n(S, \text{Bern}(p'), \text{Bern}(q'))$. Note that $\phi_0'$ applies $\textsc{rk}$ element-wise to the entries below the main diagonal of $A_{ij}(G)$. Coupling each of the independent entries below the diagonals of $\phi_0'(G(n, p', q', S))$ and $M_n(S, P_{\lambda_0}, Q_{\lambda_0})$ separately, we have that if $|S| = k$ then
$$\TV\left( \phi_0'(G(n, p', q', S)), M_n(S, P_{\lambda_0}, Q_{\lambda_0}) \right) \le \binom{k}{2} \cdot \TV\left( \tilde{P}_{\lambda_0}, P_{\lambda_0} \right) + \left( \binom{n}{2} - \binom{k}{2} \right) \cdot \TV\left( \tilde{Q}_{\lambda_0}, Q_{\lambda_0} \right)$$
Taking $S$ to be uniformly distributed over all $k$ element subsets of $[n]$ yields by triangle inequality,
\begin{align*}
\TV\left( \phi_0'(G(n, k, p', q')), \right. &\left. M_n(k, P_{\lambda_0}, Q_{\lambda_0}) \right) \\
&\le \bE_S \left[ \TV\left( \phi_0'(G(n, p', q', S)), M_n(S, P_{\lambda_0}, Q_{\lambda_0}) \right) \right] \\
&\le \binom{n}{2} \cdot \max\left\{ \TV\left( \tilde{P}_{\lambda_0}, P_{\lambda_0} \right), \TV\left( \tilde{Q}_{\lambda_0}, Q_{\lambda_0} \right) \right\}
\end{align*}
Applying the bounds above iteratively, the triangle inequality and the data processing inequality now yields that
\begin{align*}
&\TV\left( \phi_\ell'(G(n, k, p', q')), M_{2^{\ell} n}(2^{\ell} k, P_{\lambda_\ell}, Q_{\lambda_\ell}) \right) \\
&\quad \le \TV\left( \phi_\ell \circ \phi_0'(G(n, k, p)), \phi_\ell \left( M_n(k, P_{\lambda_0}, Q_{\lambda_0}) \right) \right) + \TV\left( \phi_{\ell}(M_n(k, P_{\lambda_0}, Q_{\lambda_0})), M_{2^{\ell} n}(2^{\ell} k, P_{\lambda_\ell}, Q_{\lambda_\ell}) \right) \\
&\quad \le \TV\left( \phi_0'(G(n, k, p)), M_n(k, P_{\lambda_0}, Q_{\lambda_0}) \right) + \sum_{i = 1}^{\ell} \left( 2^{i-1}n \cdot \TV\left( Q_{\lambda_{i}}, Q'_{\lambda_{i}} \right) + \sqrt{\frac{\chi^2(Q_{\lambda_i}, P_{\lambda_i})}{2}} \right) \\
&\quad \le \binom{n}{2} \cdot \max\left\{ \TV\left( \tilde{P}_{\lambda_0}, P_{\lambda_0} \right), \TV\left( \tilde{Q}_{\lambda_0}, Q_{\lambda_0} \right) \right\} + \sum_{i = 1}^{\ell} \left( 2^{i-1} n \cdot \TV\left( Q_{\lambda_{i}}, Q'_{\lambda_{i}} \right) + \sqrt{\frac{\chi^2(Q_{\lambda_i}, P_{\lambda_i})}{2}} \right)
\end{align*}
We now deal with the case of $H_0$. By the same reasoning, we have that
$$\TV\left( \phi_0'(G(n, q')), M_n(Q_{\lambda_0}) \right) \le \binom{n}{2} \cdot \TV\left( \tilde{Q}_{\lambda_0}, Q_{\lambda_0} \right)$$
Now note that if $M \sim M_m(Q_{\lambda_0})$, every entry of $W'_{ij}$ in Step 3b below the main diagonal is i.i.d. sampled from $Q_{\lambda_1}$ other than those with $i + j = 2m + 1$, which are sampled from $Q'_{\lambda_1}$. Coupling entries individually implies that
$$\TV\left( \phi_1(M), M_{2m}(Q_{\lambda_1}) \right) \le m \cdot \TV\left( Q_{\lambda_1}, Q'_{\lambda_1} \right)$$
By induction, the data processing and triangle inequalities imply that
$$\TV\left( \phi_\ell(M), M_{2m}(Q_{\lambda_\ell}) \right) \le \sum_{i = 1}^\ell 2^{i-1} m \cdot \TV\left( Q_{\lambda_i}, Q'_{\lambda_i} \right)$$
Therefore it follows that
\begin{align*}
\TV\left( \phi_\ell'(G(n, q')), M_{2^\ell n}(Q_{\lambda_\ell}) \right) &\le \TV\left( \phi_\ell \circ \phi_0'(G(n, q')), \phi_\ell\left(M_n(Q_{\lambda_0})\right) \right) \\
&\quad \quad \quad \quad + \TV\left( \phi_\ell(M_n(Q_{\lambda_0})), M_{2^\ell n}(Q_{\lambda_\ell}) \right) \\
&\le \TV\left( \phi_0'(G(n, q')), M_n(Q_{\lambda_0}) \right) + \sum_{i = 1}^\ell 2^{i-1} m \cdot \TV\left( Q_{\lambda_i}, Q'_{\lambda_i} \right) \\
&\le \binom{n}{2} \cdot \TV\left( \tilde{Q}_{\lambda_0}, Q_{\lambda_0} \right) + \sum_{i = 1}^\ell 2^{i-1} m \cdot \TV\left( Q_{\lambda_i}, Q'_{\lambda_i} \right)
\end{align*}
If $W \sim M_{2^{\ell} n}(2^{\ell} k, P_{\lambda_\ell}, Q_{\lambda_\ell})$, then the graph with adjacency matrix $A_{ij} = \mathbf{1}_{\{W_{ij} > t\}}$ is distributed as $G(2^\ell n, 2^\ell k, p, q)$ where $p = \bP_{X \sim P_{\lambda_\ell}}[ X > t]$ and $q = \bP_{X \sim Q_{\lambda_\ell}}[ X > t]$. Similarly if $W \sim M_{2^\ell n}(Q_{\lambda_\ell})$ then the graph with adjacency matrix $A_{ij} = \mathbf{1}_{\{W_{ij} > t\}}$ is distributed as $G(2^\ell n, q)$. Now combining the total variation bounds above with the data processing inequality proves the theorem.
\end{proof}

\section{Planted Dense Subgraph and Biclustering}

\subsection{Poisson Lifting and Lower Bounds for Low-Density PDS}

In this section, we introduce Poisson lifting to give a reduction from planted clique to $\textsc{PDS}(n, k, p, q)$ in the regime where $p = cq$ for some fixed $c > 1$ and $q = \tilde{\Theta}(n^{-\alpha})$ for some fixed $\alpha > 0$. This is the same parameter regime as considered in \cite{hajek2015computational} and strengthens their lower bounds to hold for a fixed planted dense subgraph size $k$. Poisson lifting is a specific instance of $\textsc{Distributional-Lifting}$ with Poisson target distributions. The guarantees of Poisson lifting are captured in the following lemma.

\begin{figure}[t!]
\begin{algbox}
\textbf{Algorithm} \textsc{Poisson-Lifting}

\vspace{2mm}

\textit{Inputs}: Graph $G \in \mG_n$, iterations $\ell$, parameters $\gamma, \epsilon \in (0, 1)$ and $c > 1$ with $3\epsilon^{-1} \le \log_c \gamma^{-1}$

\vspace{2mm}

Return the output of $\textsc{Distributional-Lifting}$ applied to $G$ with $\ell$ iterations and parameters:
\begin{itemize}
\item initial edge densities $p' = 1$ and $q' = \gamma$
\item target families $P_{\lambda} = \text{Pois}(c\lambda)$ and $Q_{\lambda} = Q'_\lambda = \text{Pois}(\lambda)$
\item rejection kernel $\textsc{rk}_{\text{P1}} = \textsc{rk}\left( 1 \to \text{Pois}(c\lambda_0), \gamma \to \text{Pois}(\lambda_0), \lceil 6 \log_{\gamma^{-1}} n \rceil \right)$ and $\lambda_0 = n^{-\epsilon}$
\item cloning map $f_{\text{cl}}(x) = (x_1, x_2, x_3, x_4)$ computed as follows:
\begin{enumerate}
\item Generate $x$ numbers in $[4]$ uniformly at random
\item Let $x_i$ be the number of $i$'s generated for each $i \in [4]$
\end{enumerate}
\item parameter map $g_{\text{cl}}(\lambda) = \lambda/4$
\item threshold $t = 0$
\end{itemize}
\vspace{1mm}
\end{algbox}
\caption{Poisson lifting procedure in Lemma \ref{lem:6a}.}
\label{fig:pois}
\end{figure}

\begin{lemma}[Poisson Lifting] \label{lem:6a}
Suppose that $n$ and $\ell$ are such that $\ell = O(\log n)$ and are sufficiently large. Fix arbitrary constants $\epsilon \in (0, 1)$ and $c > 1$ and let $\lambda_0 = n^{-\epsilon}$. Suppose that $\gamma$ is a small enough constant satisfying that $\log_c \gamma^{-1} \ge 3\epsilon^{-1}$. Then $\phi = \textsc{Poisson-Lifting}$ is a randomized polynomial time computable map $\phi : \mG_n \to \mG_{2^\ell n}$ such that under both $H_0$ and $H_1$, it holds that
$$\TV\left( \phi(\textsc{PC}(n, k, \gamma)), \textsc{PDS}\left(2^\ell n, 2^\ell k, p, q \right) \right) = O\left( n^{-\epsilon/2}\right)$$
where $p = 1 - e^{-4^{-\ell} c\lambda_0}$ and $q = 1 - e^{-4^{-\ell} \lambda_0}$.
\end{lemma}

\begin{proof}
Let $\textsc{Poisson-Lifting}$ be the algorithm $\textsc{Distributional-Lifting}$ applied with the parameters in Figure \ref{fig:pois}. Let $\lambda_{i + 1} = g_{\text{cl}}(\lambda_i) = \lambda_i/4$ for each $0 \le i \le \ell - 1$. Note that the cloning map $f_{\text{cl}}(x)$ can be computed in $O(1)$ operations. Furthermore, if $x \sim \text{Pois}(\lambda)$ then Poisson splitting implies that if $f_{\text{cl}}(x) = (x_1, x_2, x_3, x_4)$ then the $x_i$ are independent and $x_i \sim \text{Pois}(\lambda/4)$. Therefore,
\begin{align*}
f_{\text{cl}}(P_\lambda) = f_{\text{cl}}(\text{Pois}(c\lambda)) &\sim \text{Pois}(c\lambda/4)^{\otimes 4} = P_{g_{\text{cl}}(\lambda)}^{\otimes 4} \quad \text{and} \\ 
f_{\text{cl}}(Q_\lambda) = f_{\text{cl}}(\text{Pois}(\lambda)) &\sim \text{Pois}(\lambda/4)^{\otimes 4} = Q_{g_{\text{cl}}(\lambda)}^{\otimes 4}
\end{align*}
Both $P_{\lambda} = \text{Pois}(c\lambda)$ and $Q_{\lambda} = Q'_\lambda = \text{Pois}(\lambda)$ can be sampled in $O(1)$ time, and the $\chi^2$ divergence between these distributions is
\begin{align*}
\chi^2\left( Q_{\lambda}, P_\lambda \right) &= -1 + \sum_{t = 0}^\infty \frac{\left( \frac{1}{t!} e^{-\lambda} \lambda^t \right)^2}{\frac{1}{t!} e^{-c\lambda} (c\lambda)^t} = -1 + \exp\left( c^{-1}(c-1)^2 \lambda \right) \cdot \sum_{t = 0}^\infty \frac{e^{-\lambda/c} (\lambda/c)^t}{t!} \\
&= \exp\left( c^{-1}(c-1)^2 \lambda \right) - 1 \le 2c^{-1}(c-1)^2 \lambda
\end{align*}
as long as $c^{-1}(c-1)^2 \lambda \le 1$ since $e^x \le 1 + 2x$ for $x \in [0, 1]$. By Lemma \ref{lem:5a}, the rejection kernel $\textsc{rk}_{\text{P1}}$ can be computed in $O(\log n)$ time and satisfies that
$$\TV\left(\textsc{rk}_{\text{P1}}(1), P_{\lambda_0} \right) = O(n^{-3}) \quad \text{and} \quad \TV\left(\textsc{rk}_{\text{P1}}(\text{Bern}(\gamma)), Q_{\lambda_0} \right) = O(n^{-3})$$
Now note that $P_{\lambda_\ell} = \text{Pois}(4^{-\ell}c\lambda_0)$ and $Q_{\lambda_\ell} = \text{Pois}(4^{-\ell}\lambda_0)$ which implies that $p = \bP_{X \sim P_{\lambda_\ell}}[X > 0] = 1 - e^{-4^{-\ell}c\lambda_0}$ and $q = \bP_{X \sim Q_{\lambda_\ell}}[X > 0] = 1 - e^{-4^{-\ell}\lambda_0}$. Since $\textsc{PDS}(n, k, 1, \gamma)$ is the same problem as $\textsc{PC}(n, k, \gamma)$, applying Theorem \ref{lem:5cc} yields that under both $H_0$ and $H_1$, we have
\begin{align*}
&\TV\left( \phi(\textsc{PC}(n, k, \gamma)), \textsc{PDS}\left( 2^\ell n, 2^\ell k, p, q \right) \right) \\
&\quad \quad \le \binom{n}{2} \cdot \max\left\{ \TV\left(\textsc{rk}_{\text{P1}}(1), P_{\lambda_0} \right), \TV\left(\textsc{rk}_{\text{P1}}(\text{Bern}(\gamma)), Q_{\lambda_0} \right) \right\} + \sum_{i = 1}^\ell \sqrt{\frac{\chi^2(Q_{\lambda_i}, P_{\lambda_i})}{2}} \\
&\quad \quad \le \binom{n}{2} \cdot O(n^{-3}) + c^{-1/2} (c - 1) \sum_{i = 1}^\ell \sqrt{\lambda_i} \\
&\quad \quad = O(n^{-1}) + c^{-1/2} (c - 1) n^{-\epsilon/2} \sum_{i = 1}^\ell 2^{-i} = O\left(n^{-1} + n^{-\epsilon/2}\right)
\end{align*}
which completes the proof of the lemma.
\end{proof}

We now use the reduction based on Poisson lifting analyzed above to prove hardness for the sparsest regime of PDS. The proof of this theorem is similar to that of Theorem \ref{lem:4c} and is deferred to Appendix \ref{app6}.

\begin{theorem} \label{thm:poissonhard}
Fix some $c > 1$. Let $\alpha \in [0, 2)$ and $\beta \in (0, 1)$ be such that $\beta < \frac{1}{2} + \frac{\alpha}{4}$. There is a sequence $\{ (N_n, K_n, p_n, q_n) \}_{n \in \mathbb{N}}$ of parameters such that:
\begin{enumerate}
\item The parameters are in the regime $q = \tilde{\Theta}(N^{-\alpha})$ and $K = \tilde{\Theta}(N^\beta)$ or equivalently,
$$\lim_{n \to \infty} \frac{\log q_n^{-1}}{\log N_n} = \alpha, \quad \quad \lim_{n \to \infty} \frac{\log K_n}{\log N_n} = \beta \quad \text{and} \quad \lim_{n \to \infty} \frac{p_n}{q_n} = c$$
\item For any sequence of randomized polynomial-time tests $\phi_n : \mG_{N_n} \to \{0, 1\}$, the asymptotic Type I$+$II error of $\phi_n$ on the problems $\textsc{PDS}_D(N_n, K_n, p_n, q_n)$ is at least $1$ assuming the PC conjecture holds for each fixed density $p \le 1/2$.
\end{enumerate}
Therefore the computational boundary for $\textsc{PDS}_D(n, k, p, q)$ in the parameter regime $q = \tilde{\Theta}(n^{-\alpha})$, $\frac{p}{q} \to c$ and $k = \tilde{\Theta}(n^\beta)$ is $\beta^* = \frac{1}{2} + \frac{\alpha}{4}$.
\end{theorem}

In this section, we gave a planted clique lower bound for $\textsc{PDS}_D(n, k, p, q)$ with $\frac{p}{q} \to c$ as opposed to $p = cq$ exactly. We now will describe a simple reduction from $\textsc{PDS}_D(n, k, p, q)$ with $\frac{p}{q} \to c_1$ where $c_1 > c$ to $\textsc{PDS}_D(n, k, p_1, q_1)$ where $p_1 = cq_1$ and $q_1 = \Theta(q)$. Given an instance of $\textsc{PDS}_D(n, k, p, q)$, add in every non-edge independently with probability $\rho = \frac{p - cq}{c - 1 + p - cq}$ which is in $(0, 1)$ since $c_1 > c$ implies that $p > cq$ for large enough $n$. This yields an instance of $\textsc{PDS}_D$ with $p_1 = 1 - (1 - \rho)(1 - p) = p + \rho - \rho p$ and $q_1 = 1 - (1 - \rho)(1 - q) = q + \rho - \rho q$. The choice of $\rho$ implies that $p_1 = cq_1$ exactly and $\rho = \Theta(q)$ since $\frac{p}{q} \to c_1 > c$. Applying this reduction after $\textsc{Poisson-Lifting}$ yields that $\textsc{PDS}_D(n, k, cq, q)$ has the same planted clique lower bound as in the previous theorem.

\subsection{Gaussian Lifting and Lower Bounds for High-Density PDS and BC}

In parallel to the previous section, here we introduce Gaussian lifting to give a reduction from planted clique to the dense regime of $\textsc{PDS}(n, k, p, q)$ where $q = \Theta(1)$ and $p - q = \tilde{\Theta}(n^{-\alpha})$ for some fixed $\alpha > 0$. As an intermediate step, we also give a reduction strengthening the lower bounds for $\textsc{BC}$ shown in \cite{ma2015computational} to hold for the canonical simple vs. simple hypothesis testing formulation of $\textsc{BC}$. 

\begin{figure}[t!]
\begin{algbox}
\textbf{Algorithm} \textsc{Gaussian-Lifting}

\vspace{2mm}

\textit{Inputs}: Graph $G \in \mG_n$, iterations $\ell$

\vspace{2mm}

Return the output of $\textsc{Distributional-Lifting}$ applied to $G$ with $\ell$ iterations and parameters:
\begin{itemize}
\item initial densities $p' = 1$ and $q' = 1/2$
\item target families $P_{\lambda} = N(\lambda, 1)$ and $Q_{\lambda} = Q'_\lambda = N(0, 1)$
\item rejection kernel $\textsc{rk}_{\text{G}} = \textsc{rk}\left( 1 \to N(\lambda_0, 1), 1/2 \to N(0, 1), N \right)$ where $N = \lceil 6 \log_2 n \rceil$ and $\lambda_0 = \frac{\log 2}{2 \sqrt{6 \log n + 2\log 2}}$
\item cloning map $f_{\text{cl}}(x) = (x_1, x_2, x_3, x_4)$ computed as follows:
\begin{enumerate}
\item Generate $G_1, G_2, G_3 \sim_{\text{i.i.d.}} N(0, 1)$
\item Compute $(x_1, x_2, x_3, x_4)$ as
\begin{align*}
x_1 &= \frac{1}{2} \left( x + G_1 + G_2 + G_3 \right) \\
x_2 &= \frac{1}{2} \left( x - G_1 + G_2 - G_4 \right) \\
x_3 &= \frac{1}{2} \left( x + G_1 - G_2 - G_3 \right) \\
x_4 &= \frac{1}{2} \left( x - G_1 - G_2 + G_3 \right)
\end{align*}
\end{enumerate}
\item parameter map $g_{\text{cl}}(\lambda) = \lambda/2$
\item threshold $t = 0$
\end{itemize}
\vspace{1mm}
\end{algbox}
\caption{Gaussian lifting procedure in Lemma \ref{lem:6b}.}
\label{fig:gaussian}
\end{figure}

The next lemma we prove is an analogue of Lemma \ref{lem:6a} for $\textsc{Gaussian-Lifting}$ and follows the same structure of verifying the preconditions for and applying Lemma \ref{lem:5cc}.

\begin{lemma}[Gaussian Lifting] \label{lem:6b}
Suppose that $n$ and $\ell$ are such that $\ell = O(\log n)$ and are sufficiently large and let
$$\mu = \frac{\log 2}{2 \sqrt{6 \log n + 2\log 2}}$$
Then $\phi = \textsc{Gaussian-Lifting}$ is a randomized polynomial time computable map $\phi : \mG_n \to \mG_{2^\ell n}$ such that under both $H_0$ and $H_1$, it holds that
$$\TV\left( \phi(\textsc{PC}(n, k, 1/2)), \textsc{PDS}\left(2^\ell n, 2^\ell k, \Phi\left(2^{-\ell}\mu\right), 1/2 \right) \right) = O\left( \frac{1}{\sqrt{\log n}} \right)$$
\end{lemma}

\begin{proof}
Let $\textsc{Gaussian-Lifting}$ be the algorithm $\textsc{Distributional-Lifting}$ applied with the parameters in Figure \ref{fig:gaussian}. Let $\lambda_{i + 1} = g_{\text{cl}}(\lambda_i) = \lambda_i/2$ for each $0 \le i \le \ell - 1$. Note that the cloning map $f_{\text{cl}}(x)$ can be computed in $O(1)$ operations. Now suppose that $x \sim N(\lambda, 1)$. If $f_{\text{cl}}(x) = (x_1, x_2, x_3, x_4)$, then it follows that
$$\left[ \begin{matrix} x_1 \\ x_2 \\ x_3 \\ x_4 \end{matrix} \right] = \frac{\lambda}{2} \left[ \begin{matrix} 1 \\ 1 \\ 1 \\ 1 \end{matrix} \right] + \frac{1}{2} \left[ \begin{array}{rrrr} 1 & 1 & 1 & 1 \\ 1 & -1 & 1 & -1 \\ 1 & 1 & -1 & -1 \\ 1 & -1 & -1 & 1 \end{array} \right] \cdot \left[ \begin{matrix} x - \lambda \\ G^1 \\ G^2 \\ G^3 \end{matrix} \right]$$
Since $x - \mu, G^1, G^2$ and $G^3$ are zero-mean and jointly Gaussian with covariance matrix $I_4$, it follows that the entries of $(x_1, x_2, x_3, x_4)$ are also jointly Gaussian. Furthermore, the coefficient matrix above is orthonormal, implying that the covariance matrix of $(x_1, x_2, x_3, x_4)$ remains $I_4$. Therefore it follows that $f_{\text{cl}}(N(\lambda, 1)) \sim N(\lambda/2, 1)^{\otimes 4}$. Applying this identity with $\lambda = 0$ yields that $f_{\text{cl}}(N(0, 1)) \sim N(0, 1)^{\otimes 4}$. Thus $f_{\text{cl}}$ is a valid cloning map for $P_\lambda$ and $Q_{\lambda}$ with parameter map $ g_{\text{cl}}(\lambda) = \lambda/2$.

Observe that $P_{\lambda} = N(\lambda, 1)$ and $Q_{\lambda} = Q'_\lambda = N(0, 1)$ can be sampled in $O(1)$ time in the given computational model. Note that the $\chi^2$ divergence between these distributions is
$$\chi^2\left( Q_\lambda, P_\lambda \right) = -1 + \int_{-\infty}^\infty \frac{\left( \frac{1}{\sqrt{2\pi}} e^{-x^2/2}\right)^2}{\frac{1}{\sqrt{2\pi}} e^{-(x-\lambda)^2/2}}dx = -1 + \frac{e^{2\lambda^2}}{\sqrt{2\pi}} \int_{-\infty}^\infty e^{-(x+\lambda)^2/2} dx = e^{2\lambda^2} - 1 \le 4 \lambda^2$$
as long as $4\lambda^2 \le 1$ since $e^x \le 1 + 2x$ for $x \in [0, 1]$, which is the case for all $\lambda = \lambda_i$. By Lemma \ref{lem:5c}, the rejection kernel $\textsc{rk}_{\text{G}}$ can be computed in $\text{poly}(n)$ time and satisfies that
$$\TV\left(\textsc{rk}_{\text{G}}(1), P_{\lambda_0} \right) = O(n^{-3}) \quad \text{and} \quad \TV\left(\textsc{rk}_{\text{G}}(\text{Bern}(1/2)), Q_{\lambda_0} \right) = O(n^{-3})$$
Now note that $P_{\lambda_\ell} = N(2^{-\ell}\mu, 1)$ and $Q_{\lambda_\ell} = N(0, 1)$ which implies that $p = \bP_{X \sim P_{\lambda_\ell}}[X > 0] = \Phi\left(2^{-\ell}\mu\right)$ and $q = \bP_{X \sim Q_{\lambda_\ell}}[X > 0] = 1/2$. Since $\textsc{PDS}(n, k, 1, 1/2)$ is the same problem as $\textsc{PC}(n, k, 1/2)$, applying Theorem \ref{lem:5cc} yields that under both $H_0$ and $H_1$, we have
\begin{align*}
&\TV\left( \phi(\textsc{PC}(n, k, 1/2)), \textsc{PDS}\left( 2^\ell n, 2^\ell k, p, q \right) \right) \\
&\quad \quad \le \binom{n}{2} \cdot \max\left\{ \TV\left(\textsc{rk}_{\text{G}}(1), P_{\lambda_0} \right), \TV\left(\textsc{rk}_{\text{G}}(\text{Bern}(1/2)), Q_{\lambda_0} \right) \right\} + \sum_{i = 1}^\ell \sqrt{\frac{\chi^2(Q_{\lambda_i}, P_{\lambda_i})}{2}} \\
&\quad \quad \le \binom{n}{2} \cdot O(n^{-3}) + \sqrt{2} \cdot \sum_{i = 1}^\ell \lambda_i \\
&\quad \quad = O(n^{-1}) + \mu \sqrt{2} \cdot \sum_{i = 1}^\ell 2^{-i} = O\left(n^{-1} + \frac{1}{\sqrt{\log n}} \right)
\end{align*}
which completes the proof of the lemma.
\end{proof}

We now use this $\textsc{Gaussian-Lifting}$ reduction to deduce hardness for the dense variant of planted dense subgraph, which has a slightly different computational boundary of than the sparsest variant, given by $p - q \approx n/k^2$. The proof of the next theorem is deferred to Appendix \ref{app6}

\begin{theorem} \label{thm:gauss-hard}
Let $\alpha \in [0, 2)$ and $\beta \in (0, 1)$ be such that $\beta < \frac{1}{2} + \frac{\alpha}{2}$. There is a sequence $\{ (N_n, K_n, p_n, q_n) \}_{n \in \mathbb{N}}$ of parameters such that:
\begin{enumerate}
\item The parameters are in the regime $q = \Theta(1)$, $p - q = \tilde{\Theta}(N^{-\alpha})$ and $K = \tilde{\Theta}(N^\beta)$ or equivalently,
$$\lim_{n \to \infty} \frac{\log (p_n - q_n)^{-1}}{\log N_n} = \alpha \quad \text{and} \quad \lim_{n \to \infty} \frac{\log K_n}{\log N_n} = \beta$$
\item For any sequence of randomized polynomial-time tests $\phi_n : \mG_{N_n} \to \{0, 1\}$, the asymptotic Type I$+$II error of $\phi_n$ on the problems $\textsc{PDS}_D(N_n, K_n, p_n, q_n)$ is at least $1$ assuming the PC conjecture holds with density $p = 1/2$.
\end{enumerate}
Therefore the computational boundary for $\textsc{PDS}_D(n, k, p, q)$ in the parameter regime $q = \Theta(1)$, $p - q = \tilde{\Theta}(n^{-\alpha})$ and $k = \tilde{\Theta}(n^\beta)$ is $\beta^* = \frac{1}{2} + \frac{\alpha}{2}$.
\end{theorem}

Note that to prove this theorem, it was only necessary to map to instances with ambient density $q = 1/2$. We remark that it is possible to map from $q = 1/2$ to any constant $q$ by removing edges with a constant probability $\rho < 1$ or removing non-edges with probability $\rho$. Note that this still preserves the asymptotic regime $p - q = \tilde{\Theta}(n^{-\alpha})$. We now use $\textsc{Gaussian-Lifting}$ to give a reduction from planted clique to biclustering. This next lemma uses similar ingredients to the proof of Theorem \ref{lem:5cc} and analysis of the cloning technique in $\textsc{Gaussian-Lifting}$. The proof is deferred to Appendix \ref{app6}.

\begin{figure}[t!]
\begin{algbox}
\textbf{Algorithm} \textsc{BC-Reduction}

\vspace{2mm}

\textit{Inputs}: Graph $G \in \mG_n$, iterations $\ell$
\begin{enumerate}
\item Set $W$ to be the output of $\textsc{Gaussian-Lifting}$ applied to $G$ with $\ell$ iterations without the thresholding in Step 4 of $\textsc{Distributional-Lifting}$
\item Replace the diagonal entries with $W_{ii} \sim_{\text{i.i.d.}} N(0, 2)$
\item Generate an antisymmetric $2^\ell n \times 2^\ell n$ matrix $A$ of with i.i.d. $N(0, 1)$ random variables below its main diagonal and set
$$W \gets \frac{1}{\sqrt{2}} \left( W + A \right)$$
\item Generate a permutation $\sigma$ of $[2^\ell n]$ uniformly at random and output $W^{\text{id}, \sigma}$
\end{enumerate}
\vspace{1mm}
\end{algbox}
\caption{Reduction to biclustering in Lemma \ref{lem:bc}.}
\label{fig:bc}
\end{figure}

\begin{lemma} \label{lem:bc}
Suppose that $n$ and $\ell$ are such that $\ell = O(\log n)$ and are sufficiently large and
$$\mu = \frac{\log 2}{2 \sqrt{6 \log n + 2\log 2}}$$
Then there is a randomized polynomial time computable map $\phi = \textsc{BC-Reduction}$ with $\phi : \mG_n \to \mathbb{R}^{2^\ell n \times 2^\ell n}$ such that under $H_0$ and $H_1$, it holds that
$$\TV\left( \phi(\textsc{PC}(n, k, 1/2)), \textsc{BC}\left(2^\ell n, 2^\ell k, 2^{-\ell - 1/2} \mu \right) \right) = O\left(\frac{1}{\sqrt{\log n}} \right)$$
\end{lemma}

Note that Lemma \ref{lem:bc} provides a randomized polynomial time map that exactly reduces from $\textsc{PC}_D(n, k, 1/2)$ to $\textsc{BC}_D(2^\ell n, 2^\ell k, 2^{-\ell-1/2} \mu)$. This reduction yields tight computational lower bounds for a simple vs. simple hypothesis testing variant of biclustering as stated in Theorem \ref{thm:bc}. This follows from setting $\ell_n, k_n, N_n$ and $K_n$ as in Theorem \ref{thm:gauss-hard} and $\mu_n = 2^{-\ell_n-1/2} \mu$, then applying an identical analysis as in Theorem \ref{thm:gauss-hard}. Note that when $\beta < \frac{1}{2}$, this choice sets $\ell_n = 0$ and deduces that $\textsc{BC}_D$ is hard when $\alpha > 0$.

\begin{theorem} \label{thm:bc}
Let $\alpha > 0$ and $\beta \in (0, 1)$ be such that $\beta < \frac{1}{2} + \frac{\alpha}{2}$. There is a sequence $\{ (N_n, K_n, \mu_n) \}_{n \in \mathbb{N}}$ of parameters such that:
\begin{enumerate}
\item The parameters are in the regime $\mu = \tilde{\Theta}(N^{-\alpha})$ and $K = \tilde{\Theta}(N^\beta)$ or equivalently,
$$\lim_{n \to \infty} \frac{\log \mu_n^{-1}}{\log N_n} = \alpha \quad \text{and} \quad \lim_{n \to \infty} \frac{\log K_n}{\log N_n} = \beta$$
\item For any sequence of randomized polynomial-time tests $\phi_n : \mathbb{R}^{N_n \times N_n} \to \{0, 1\}$, the asymptotic Type I$+$II error of $\phi_n$ on the problems $\textsc{BC}_D(N_n, K_n, \mu_n)$ is at least $1$ assuming the PC conjecture holds with density $p = 1/2$.
\end{enumerate}
Therefore the computational boundary for $\textsc{BC}_D(n, k, \mu)$ in the parameter regime $\mu = \tilde{\Theta}(n^{-\alpha})$ and $k = \tilde{\Theta}(n^\beta)$ is $\beta^* = \frac{1}{2} + \frac{\alpha}{2}$ and $\alpha^* = 0$ when $\beta < \frac{1}{2}$.
\end{theorem}

We now deduce the computational barrier for the biclustering recovery problem from the PDS recovery conjecture. The obtained boundary of $\beta^* = \frac{1}{2} + \alpha$ is stronger than the detection boundary $\beta^* = \frac{1}{2} + \frac{\alpha}{2}$ in the previous theorem. First we will need the following lemma, which gives the necessary total variation guarantees for our reduction. We omit details that are identical to the proof of Lemma \ref{lem:bc}.

\begin{figure}[t!]
\begin{algbox}
\textbf{Algorithm} \textsc{BC-Recovery}

\vspace{2mm}

\textit{Inputs}: Graph $G \in \mG_n$, density bias $\rho$
\begin{enumerate}
\item Let $\textsc{rk}_{G} = \textsc{rk}\left( \frac{1}{2} + \rho \to N(\mu, 1), \frac{1}{2} \to N(0, 1), N\right)$ where $\mu = \frac{\log (1 + 2\rho)}{2 \sqrt{6 \log n + 2\log 2}}$ and $N = \lceil 6 \log_{1 + 2 \rho} n \rceil$ and compute the symmetric matrix $W \in \mathbb{R}^{n \times n}$ with
$$W_{ij} = \textsc{rk}_G\left( \mathbf{1}_{\{i, j \} \in E(G)} \right)$$
for all $i \neq j$ and $W_{ii} \sim_{\text{i.i.d.}} N(0, 2)$
\item Generate an antisymmetric matrix $A \in \mathbb{R}^{n \times n}$ of with i.i.d. $N(0, 1)$ random variables below its main diagonal and set
$$W \gets \frac{1}{\sqrt{2}} \left( W + A \right)$$
\item Generate a permutation $\sigma$ of $[n]$ uniformly at random and output $W^{\text{id}, \sigma}$
\end{enumerate}
\vspace{1mm}
\end{algbox}
\caption{Reduction to biclustering recovery in Lemma \ref{lem:bcrec}.}
\label{fig:bcrec}
\end{figure}

\begin{lemma} \label{lem:bcrec}
Suppose that $n, \mu$ and $\rho \ge n^{-1}$ are such that
$$\mu = \frac{\log (1 + 2\rho)}{2 \sqrt{6 \log n + 2\log 2}}$$
Then there is a randomized polynomial time computable map $\phi = \textsc{BC-Recovery}$ with $\phi : \mG_n \to \mathbb{R}^{n \times n}$ such that for any subset $S \subseteq [n]$ with $|S| = k$, it holds that
$$\TV\left( \phi\left(G(n, 1/2 + \rho, 1/2, S)\right), \int \mL\left( \mu \cdot \mathbf{1}_S \mathbf{1}_T^\top + N(0, 1)^{\otimes n \times n} \right) d\pi(T) \right) = O\left(\frac{1}{\sqrt{\log n}} \right)$$
where $\pi$ is the uniform distribution on subsets of $[n]$ of size $k$.
\end{lemma}

\begin{proof}
Let $\phi = \textsc{BC-Recovery}$ be as in Figure \ref{fig:bcrec}. Applying Lemma \ref{lem:5c}, it holds that $\textsc{rk}_G$ can be computed in $\text{poly}(n)$ time and that
$$\TV\left( \textsc{rk}_G(\text{Bern}(1/2 + \rho)), N(\mu, 1) \right) = O(n^{-3}) \quad \text{and} \quad \TV\left( \textsc{rk}_G(\text{Bern}(1/2)), N(0, 1) \right) = O(n^{-3})$$
Let $W_1$ and $W_2$ be the values of $W$ after Steps 1 and 2, respectively, applied to an input graph $G \sim G(n, 1/2 + \rho, 1/2, S)$. Let $M$ be a sample from $M_n(S, N(\mu, 1), N(0, 1))$ with i.i.d. $N(0, 2)$ random variables on its diagonal. Coupling entries individually yields that
\begin{align*}
\TV\left( \mL(W_1), \mL(M) \right) &\le \binom{k}{2} \cdot \TV\left( \textsc{rk}_G(\text{Bern}(1/2 + \rho)), N(\mu, 1) \right) \\
&\quad \quad + \left( \binom{n}{2} - \binom{k}{2} \right) \cdot \TV\left( \textsc{rk}_G(\text{Bern}(1/2)), N(0, 1) \right) \\
&= \binom{n}{2} \cdot O(n^{-3}) = O(n^{-1})
\end{align*}
An identical argument as in Lemma \ref{lem:bc} now shows that $W_2$ is at total variation distance $O(n^{-1})$ from $\mu \cdot \mathbf{1}_S \mathbf{1}_S^\top + N(0, 1)^{\otimes n \times n}$ with all of its diagonal entries replaced with i.i.d. samples from $N(0, 1)$. The same permutation argument applying Lemma \ref{lem:4a} now yields that
$$\TV\left( \mL((W'_2)^{\text{id}, \sigma}), \int \mL\left( \mu \cdot \mathbf{1}_S \mathbf{1}_T^\top + N(0, 1)^{\otimes n \times n} \right) d\pi(T) \right) = O\left(\frac{1}{\sqrt{\log n}} \right)$$
Applying the triangle and data processing inequalities as in the conclusion of Lemma \ref{lem:bc} completes the proof of the lemma.
\end{proof}

With this lemma, we now deduce the recovery barrier for biclustering from the PDS and PC conjectures. Note that the recovery barrier of $\beta^* = \frac{1}{2} + \alpha$ and detection barrier of $\beta^* = \frac{1}{2} + \frac{\alpha}{2}$ indicates that recovery is conjectured to be strictly harder than detection for the formulations we consider in the regime $\beta > \frac{1}{2}$.

\begin{theorem} \label{thm:bcrec}
Let $\alpha > 0$ and $\beta \in (0, 1)$. There is a sequence $\{ (N_n, K_n, \mu_n) \}_{n \in \mathbb{N}}$ of parameters such that:
\begin{enumerate}
\item The parameters are in the regime $\mu = \tilde{\Theta}(N^{-\alpha})$ and $K = \tilde{\Theta}(N^\beta)$ or equivalently,
$$\lim_{n \to \infty} \frac{\log \mu_n^{-1}}{\log N_n} = \alpha \quad \text{and} \quad \lim_{n \to \infty} \frac{\log K_n}{\log N_n} = \beta$$
\item If $\beta \ge \frac{1}{2}$ and $\beta < \frac{1}{2} + \alpha$, then the following holds. Let $\epsilon > 0$ be fixed and let $M_n$ be an instance of $\textsc{BC}_R(N_n, K_n, \mu_n)$. There is no sequence of randomized polynomial-time computable functions $\phi_n : \mathbb{R}^{N_n \times N_n} \to \binom{[N_n]}{k}^2$ such that for all sufficiently large $n$ the probability that $\phi_n(M_n)$ is exactly the pair of latent row and column supports of $M_n$ is at least $\epsilon$, assuming the PDS recovery conjecture.
\item If $\beta < \frac{1}{2}$ and $\alpha > 0$, then the following holds. There is no sequence of randomized polynomial-time computable functions $\phi_n : \mathbb{R}^{N_n \times N_n} \to \binom{[N_n]}{k}^2$ such that for all sufficiently large $n$ the probability that $\phi_n(M_n)$ is exactly the pair of latent row and column supports of $M_n$ is at least $\epsilon$, assuming the PC conjecture.
\end{enumerate}
Therefore, given the PDS recovery conjecture, the computational boundary for $\textsc{BC}_R(n, k, \mu)$ in the parameter regime $\mu = \tilde{\Theta}(n^{-\alpha})$ and $k = \tilde{\Theta}(n^\beta)$ is $\beta^* = \frac{1}{2} + \alpha$ when $\beta \ge \frac{1}{2}$ and $\alpha^* = 0$ when $\beta < \frac{1}{2}$.
\end{theorem}

\begin{proof}
First we consider the case when $\beta \ge 1/2$ and $\beta < \frac{1}{2} + \alpha$. Now set
$$k_n = \lceil n^{\beta} \rceil, \quad \quad \rho_n = n^{-\alpha}, \quad \quad N_n = n \quad \quad K_n = k_n, \quad \quad \mu_n = \frac{\log (1 + 2\rho_n)}{2 \sqrt{6 \log n + 2\log 2}}$$
Assume for contradiction that there is a sequence of randomized polynomial-time computable functions $\phi_n$ as described above and let $\phi_n^r$ denote the restriction of $\phi_n$ to output latent row support only. Let $\varphi_n = \textsc{BC-Recovery}$ be the reduction in Lemma \ref{lem:bcrec}, let $G_n \sim G(n, S, 1/2 + \rho_n, 1/2)$ and let $M_n = \varphi_n(G_n)$ where $S$ is a $k_n$-subset of $[n]$. Let $\mL_{n, S, T} = \mL\left( \mu_n \mathbf{1}_S \mathbf{1}_T^\top + N(0, 1)^{\otimes n \times n} \right)$ and let $\mL_{n, S} = \int \mL_{n, S, T} d\pi(T)$ where $\pi$ is the uniform distribution over $k_n$-subsets of $[n]$ be the distribution of an instance of $\text{BC}_R(N_n, K_n, \mu_n)$ conditioned on the event that the row support of its planted submatrix is $S$. Now observe that
$$\left| \bP_{M \sim \mL(M_n)} \left[ \phi_n^r(M) = S \right] - \bP_{M \sim \mL_{n, S}}\left[ \phi_n^r(M) = S\right]\right| \le \TV\left( \mL(M_n), \mL_{n, S} \right) = O\left( \frac{1}{\sqrt{\log n}} \right)$$
because of Lemma \ref{lem:bcrec}. Now note that $\bP_{M \sim \mL_{n, S}}\left[\phi_n^r(M) = S \right] = \bE_{T \sim \pi} \bP_{M \sim \mL_{n, S, T}}\left[\phi_n^r(M) = S \right] \ge \epsilon$ for sufficiently large $n$ by assumption. Therefore it follows that
$$\bP[\phi_n^r \circ \varphi_n(G_n) = S] \ge \epsilon - O\left( \frac{1}{\sqrt{\log n}} \right)$$
which is at least $\epsilon/2$ for sufficiently large $n$. Now observe that
$$\lim_{n \to \infty} \frac{\log k_n}{\log n} = \beta \quad \text{and} \quad \lim_{n \to \infty} \log_n \left( \frac{k_n^2 \rho_n^2}{\frac{1}{4} - \rho_n^2} \right) = 2\beta - 2\alpha < 1$$
Since the sequence of functions $\phi_n^r \circ \varphi_n$ can be computed in randomized polynomial time, this contradicts the PDS recovery conjecture. Therefore no such sequence of functions $\phi_n$ exists for the parameter sequence $\{ (N_n, K_n, \mu_n) \}_{n \in \mathbb{N}}$ defined above. Now note that as $n \to \infty$,
$$\mu_n = \frac{\log (1 + 2\rho_n)}{2 \sqrt{6 \log n + 2\log 2}} \sim \frac{\rho_n}{\sqrt{6 \log n}} = \frac{n^{-\alpha}}{\sqrt{6 \log n}}$$
Therefore it follows that
$$\lim_{n \to \infty} \frac{\log \mu_n^{-1}}{\log N_n} = \lim_{n \to \infty} \frac{\alpha \log n + \frac{1}{2} \log (6\log n)}{\log n} = \alpha \quad \text{and} \quad \lim_{n \to \infty} \frac{\log K_n}{\log N_n} = \beta$$
This completes the proof in the case that $\beta \ge 1/2$. Now consider the case where $\beta < 1/2$. Set $\rho_n = 1/2$ and all other parameters as above. Let $G_n \sim G(n, k_n, S)$ and repeat the same argument as above to obtain that $\bP[\phi_n^r \circ \varphi_n(G_n) = S] \ge \epsilon - o(1)$. Now consider the algorithm $\phi'_n : \mG_n \to \{0, 1\}$ that computes $S' = \phi_n^r \circ \varphi_n(G_n)$ and checks if $S'$ is a clique, outputting a $1$ if it is and $0$ or $1$ uniformly at random otherwise. If $G_n \sim G(n, 1/2)$, then with probability $1 - o(1)$ the largest clique of $G_n$ is less than $(2 + \epsilon) \log_2 n$ for any fixed $\epsilon > 0$. It follows by the definition of $\phi_n$ that $|S'| = k_n = \Theta(n^{\beta}) = \omega(\log n)$ and thus with probability $1 - o(1)$, $\phi'_n$ outputs a random bit. Therefore $\bP_{G_n \sim G(n, 1/2)}[\phi'_n(G_n) = 1] = 1/2 + o(1)$. If $G_n \sim G(n, k, 1/2)$, then with probability at least $\epsilon - o(1)$, $S'$ is the support of the planted clique and $\phi_n'$ outputs a $1$. Otherwise, $\phi_n'$ outputs a random bit. Therefore $\bP_{G_n \sim G(n, k, 1/2)}[\phi'_n(G_n) = 0] = (1 - \epsilon)/2 + o(1)$. Therefore it follows that the Type I$+$II error of $\phi'_n$ is
$$\bP_{G_n \sim G(n, 1/2)}[\phi'_n(G_n) = 1] + \bP_{G_n \sim G(n, k, 1/2)}[\phi'_n(G_n) = 0] = 1 - \frac{\epsilon}{2} + o(1)$$
which contradicts the PC conjecture. This completes the proof of the theorem.
\end{proof}

Since the reduction $\textsc{BC-Recovery}$ exactly preserves the latent support $S$ of the instance of $\textsc{PDS}_R$ when mapping to $\textsc{BC}_R$, the same reduction shows hardness of partial recovery if the PDS conjecture is strengthened to hold for partial recovery. The same is true for weak recovery.

\subsection{Lower Bounds for General PDS}

In this section we give a reduction to the general regime of PDS where $q = \tilde{\Theta}(n^{-\alpha})$ and $p - q = \tilde{\Theta}(n^{-\beta})$ where $\beta > \alpha$. Note that in order to completely characterize PDS when $p - q = O(q)$, we also need the computational lower bound shown in Section 6.1 when $\alpha = \beta$. We now give this reduction, which applies $\textsc{Gaussian-Lifting}$ and $\textsc{Poisson-Lifting}$ in sequence. Its correctness follows from combining the guarantees of rejection kernels, $\textsc{Gaussian-Lifting}$ and $\textsc{Poisson-Lifting}$ with the data processing and triangle inequalities. The proof is deferred to Appendix \ref{app6}.

\begin{figure}[t!]
\begin{algbox}
\textbf{Algorithm} \textsc{General-PDS-Reduction}

\vspace{2mm}

\textit{Inputs}: Graph $G \in \mG_n$, iterations $\ell_1, \ell_2$
\begin{enumerate}
\item Let $H$ be the output of $\textsc{Gaussian-Lifting}$ applied to $G$ with $\ell_1$ iterations
\item Update $H$ to be the output of $\textsc{Poisson-Lifting}$ applied to $H$ with $\ell_2$ iterations where the rejection kernel is replaced with
$$\textsc{rk}_{\text{P2}} = \textsc{rk}\left(\frac{1}{2} + \rho \to \text{Pois}(c\lambda), \frac{1}{2} \to \text{Pois}(\lambda), N \right)$$
where the rejection kernel has natural parameter $2^{\ell_1} n$ and satisfies $\lambda = (2^{\ell_1} n)^{-\epsilon}$, $N = \lceil 6 \rho^{-1} \log (2^{\ell_1}n) \rceil$, $\rho = \Phi(2^{-\ell_1} \mu) - 1/2$ and $c = \left( 2\Phi\left(2^{-\ell_1} \mu \right)\right)^{\epsilon/4}$
\item Output $H$
\end{enumerate}
\vspace{1mm}
\end{algbox}
\caption{Reduction to the general regime of planted dense subgraph in Lemma \ref{lem:gpds}.}
\label{fig:bc}
\end{figure}

\begin{lemma} \label{lem:gpds}
Fix some arbitrary $\epsilon \in (0, 1)$. Suppose that $n$, $\ell_1$ and $\ell_2$ are such that $\ell_1, \ell_2 = O(\log n)$ and are sufficiently large. Let $\ell = \ell_1 + \ell_2$ and
$$\mu = \frac{\log 2}{2 \sqrt{6 \log n + 2\log 2}}$$
Then there is a randomized polynomial time computable map $\phi = \textsc{General-PDS-Reduction}$ with $\phi : \mG_n \to \mG_{2^\ell n}$ such that under both $H_0$ and $H_1$, it holds that
$$\TV\left( \phi(\textsc{PC}(n, k, 1/2)), \textsc{PDS}\left(2^\ell n, 2^\ell k, p_{\ell_1, \ell_2}, q_{\ell_1, \ell_2} \right) \right) = O\left( \frac{1}{\sqrt{\log n}} \right)$$
where $p_{\ell_1, \ell_2}$ and $q_{\ell_1, \ell_2}$ are defined to be
$$p_{\ell_1, \ell_2} = 1 - \exp\left( 4^{-\ell_2} \left(2^{\ell_1}n\right)^{-\epsilon} \cdot \left( 2 \Phi\left(2^{-\ell_1} \mu \right) \right)^{\epsilon/4} \right) \quad \text{and} \quad q_{\ell_1, \ell_2} = 1 - \exp\left( 4^{-\ell_2} \left(2^{\ell_1} n\right)^{-\epsilon} \right)$$
\end{lemma}

We now use this reduction combining $\textsc{Gaussian-Lifting}$ and $\textsc{Poisson-Lifting}$ to identify the computational barrier in the general PDS problem. As in previous theorems deducing the computational barrier implied by a reduction, the proof resolves to a calculation we defer to Appendix \ref{app6}.

\begin{theorem} \label{thm:pdsgenhard}
Let $\alpha, \gamma \in [0, 2)$ and $\beta \in (0, 1)$ be such that $\gamma \ge \alpha$ and $\beta < \frac{1}{2} + \frac{\gamma}{2} - \frac{\alpha}{4}$. There is a sequence $\{ (N_n, K_n, p_n, q_n) \}_{n \in \mathbb{N}}$ of parameters such that:
\begin{enumerate}
\item The parameters are in the regime $p - q = \tilde{\Theta}(N^{-\gamma})$, $q = \tilde{\Theta}(N^{-\alpha})$ and $K = \tilde{\Theta}(N^\beta)$ or equivalently,
$$\lim_{n \to \infty} \frac{\log q_n^{-1}}{\log N_n} = \alpha, \quad \lim_{n \to \infty} \frac{\log (p_n - q_n)^{-1}}{\log N_n} = \gamma \quad \text{and} \quad \lim_{n \to \infty} \frac{\log K_n}{\log N_n} = \beta$$
\item For any sequence of randomized polynomial-time tests $\phi_n : \mG_{N_n} \to \{0, 1\}$, the asymptotic Type I$+$II error of $\phi_n$ on the problems $\textsc{PDS}_D(N_n, K_n, p_n, q_n)$ is at least $1$ assuming the PC conjecture holds for $p = 1/2$.
\end{enumerate}
Therefore the computational boundary for $\textsc{PDS}_D(n, k, p, q)$ in the parameter regime $p - q = \tilde{\Theta}(n^{-\gamma})$, $q = \tilde{\Theta}(n^{-\alpha})$ and $k = \tilde{\Theta}(n^\beta)$ where $\gamma \ge \alpha$ is $\beta^* = \frac{1}{2} + \frac{\gamma}{2} - \frac{\alpha}{4}$.
\end{theorem}

\section{Reflection Cloning and Subgraph Stochastic Block Model}

\subsection{Reflecting Cloning and Rank-1 Submatrix}

Suppose that $n$ is even and fixed. Let $\mathcal{R}$ denote the linear operator on $\mathbb{R}^{n \times n}$ matrices that reflects a matrix horizontally about its vertical axis of symmetry. Let $\mathcal{F}$ denote the linear operator that multiplies each entry on the right half of a matrix by $-1$. The reflection cloning reduction is to iteratively update the matrix $W$ to
$$W \gets \frac{1}{\sqrt{2}} \left( \mathcal{R} W^\sigma + \mathcal{F} W^\sigma \right)$$
where $\sigma$ is chosen uniformly at random and then to update $W$ similarly with vertical analogues of $\mathcal{R}$ and $\mathcal{F}$. This achieves the same scaling of $\mu$ as $\textsc{Gaussian-Lifting}$ but does not increase $n$. This ends up tightly achieving the right parameter scaling to deduce the sharper hardness of $\textsc{ROS}_D$ and, indirectly, $\textsc{SPCA}_D$ over problems that admit sum-tests such as $\textsc{PIS}_D$, $\textsc{PDS}_D$, $\textsc{BC}_D$ and $\textsc{BSPCA}_D$. The parameter scaling in these problems is captured exactly by the cloning methods in the previous two sections. Note that $\textsc{Reflection-Cloning}$ causes $r'$ and $c'$ to have negative entries and hence cannot show hardness for $\textsc{BC}$, unlike $\textsc{Gaussian-Lifting}$. We remark that all previous cloning methods are in some sense lossy, introducing independent sources of randomness at each entry of the input matrix. In contrast, the only randomness introduced in $\textsc{Reflection-Cloning}$ are random permutations of rows and columns, and ends up achieving a much sharper scaling in $\mu$.

We remark that if $r, c \in \{-1, 0, 1\}$ then $r'$ and $c'$ have most of their entries in $\{-1, 0, 1\}$. Entries outside of $\{-1, 0, 1\}$ result from permutations $\sigma$ in Step 2a that yield collisions between elements in the support of $r$ and $c$ i.e. when $i$ and $n + 1 - i$ are both in the support of either $r$ or $c$ in an iteration of Step 2. Although few in number, these entries turn out to be information-theoretically detectable. We note that if it were possible to reduce from $r, c \in \{-1, 0, 1\}$ to $r', c' \in \{-1, 0, 1\}$ with this property, then this would strengthen our hardness results for $\textsc{SSBM}_D, \textsc{SSW}_D, \textsc{ROS}_D$ and $\textsc{SPCA}_D$ to hold for the canonical simple vs. simple hypothesis testing formulations of these problems. Given a vector $v \in \mathbb{R}^n$ and permutation $\sigma$ of $[n]$, let $r^\sigma$ denote the vector formed by permuting the indices of $r$ according to $\sigma$.

\begin{figure}[t!]
\begin{algbox}
\textbf{Algorithm} \textsc{Reflection-Cloning}

\vspace{2mm}

\textit{Inputs}: Matrix $M \in \mathbb{R}^{n \times n}$ where $n$ is even, number of iterations $\ell$
\begin{enumerate}
\item Initialize $W \gets M$
\item For $i = 0, 1, \dots, \ell - 1$ do:
\begin{enumerate}
\item[a.] Generate a permutation $\sigma$ of $[n]$ uniformly at random
\item[b.] Let $W' \in \mathbb{R}^{n \times n}$ have entries
\begin{align*}
W'_{ij} &= \frac{1}{2} \left( W_{ij}^{\sigma, \sigma} + W_{(n+1-i)j}^{\sigma, \sigma} + W_{(n+1-i)(n+1-j)}^{\sigma, \sigma} + W_{(n+1-i)(n+1-j)}^{\sigma, \sigma} \right) \\
W'_{(n+1-i)j} &= \frac{1}{2} \left( W_{ij}^{\sigma, \sigma} - W_{(n+1-i)j}^{\sigma, \sigma} + W_{(n+1-i)(n+1-j)}^{\sigma, \sigma} - W_{(n+1-i)(n+1-j)}^{\sigma, \sigma} \right) \\
W'_{i(n+1-j)} &= \frac{1}{2} \left( W_{ij}^{\sigma, \sigma} + W_{(n+1-i)j}^{\sigma, \sigma} - W_{(n+1-i)(n+1-j)}^{\sigma, \sigma} - W_{(n+1-i)(n+1-j)}^{\sigma, \sigma} \right) \\
W'_{(n+1-i)(n+1-j)} &= \frac{1}{2} \left( W_{ij}^{\sigma, \sigma} - W_{(n+1-i)j}^{\sigma, \sigma} - W_{(n+1-i)(n+1-j)}^{\sigma, \sigma} + W_{(n+1-i)(n+1-j)}^{\sigma, \sigma} \right)
\end{align*}
for each $1 \le i, j \le n/2$
\item[c.] Set $W \gets W'$
\end{enumerate}
\item Output $W$
\end{enumerate}
\vspace{1mm}
\end{algbox}
\caption{Reflection cloning procedure in Lemma \ref{lem:refc}.}
\end{figure}

\begin{lemma}[Reflection Cloning] \label{lem:refc}
Suppose $n$ is even and $\ell = O(\log n)$. There is a randomized polynomial-time computable map $\phi = \textsc{Reflection-Cloning}$ with $\phi : \mathbb{R}^{n \times n} \to \mathbb{R}^{n \times n}$ and
\begin{enumerate}
\item It holds that
$$\phi\left(N(0, 1)^{\otimes n \times n}\right) \sim N(0, 1)^{\otimes n \times n}$$
\item Consider any $\lambda > 0$ and any pair of vectors $r, c \in \mathbb{Z}^n$. Then there is a distribution $\pi$ over vectors $r', c' \in \mathbb{Z}^n$ with $\| r' \|_2^2 = 2^\ell \| r \|_2^2$ and $\| c' \|_2^2 = 2^\ell \| c \|_2^2$ such that
$$\phi\left(\lambda \cdot r c^\top + N(0, 1)^{\otimes n \times n}\right) \sim \int \mL\left( \frac{\lambda}{2^\ell} \cdot r' c'^\top + N(0, 1)^{\otimes n \times n} \right) d\pi(r', c')$$
where it holds with probability at least $1 - 4\| r \|_0^{-1} - 4\| c \|_0^{-1}$ over $\pi$ that
\begin{align*}
2^\ell \| r \|_0 \ge \| r' \|_0 \ge 2^\ell \| r \|_0 \left( 1 - \max\left(\frac{2C\ell \cdot \log (2^\ell \| r \|_0)}{\| r \|_0}, \frac{2^\ell \| r \|_0}{n}\right) \right) \\ 
2^\ell \| c \|_0 \ge \| c' \|_0 \ge 2^\ell \| c \|_0 \left( 1 - \max\left(\frac{2C\ell \cdot \log (2^\ell \| c \|_0)}{\| c \|_0}, \frac{2^\ell \| c \|_0}{n}\right) \right)
\end{align*}
for some constant $C$ if $\| r \|_0$ and $\| c \|_0$ are sufficiently large and at most $2^{-\ell-1} n$. Furthermore, if $r = c$ then $r' = c'$ holds almost surely.
\end{enumerate}
\end{lemma}

\begin{proof}
Let $\phi(M)$ be implemented by the procedure $\textsc{Reflection-Cloning}(M, \ell)$ as in Figure 6. If $\ell = O(\log n)$, this algorithm runs in randomized polynomial time. Let $\phi_1(W)$ denote the map that takes the value $W$ prior to an iteration of Step 2 as its input and outputs the value of $W$ after this iteration.

If $W \sim N(0, 1)^{\otimes n \times n}$, then it follows by a similar argument as in Lemma 7 that $\phi_1(W) \sim N(0, 1)^{\otimes n \times n}$. Specifically, for each $1 \le i, j \le n/2$ we have that
$$\left[ \begin{matrix} W_{ij}' \\ W_{(n+1-i)j}' \\ W_{i(n+1-j)}' \\ W_{(n+1-i)(n+1-j)}' \end{matrix} \right] = \frac{1}{2} \left[ \begin{array}{rrrr} 1 & 1 & 1 & 1 \\ 1 & -1 & 1 & -1 \\ 1 & 1 & -1 & -1 \\ 1 & -1 & -1 & 1 \end{array} \right] \cdot \left[ \begin{matrix} W_{ij}^{\sigma, \sigma} \\ W_{(n+1-i)j}^{\sigma, \sigma} \\ W_{i(n+1-j)}^{\sigma, \sigma} \\ W_{(n+1-i)(n+1-j)}^{\sigma, \sigma} \end{matrix} \right]$$
where $\sigma$ is the random permutation generated in Step 2a. It holds that $W^{\sigma, \sigma} \sim N(0, 1)^{\otimes n \times n}$ and therefore that the vector of entries on the right hand side above is distributed as $N(0, 1)^{\otimes 4}$. Since the coefficient matrix is orthogonal, it follows that the vector on the left hand side is also distributed as $N(0, 1)^{\otimes 4}$. Since the $\sigma$-algebras $\sigma\{W_{ij}^{\sigma, \sigma}, W_{(n+1-i)j}^{\sigma, \sigma}, W_{i(n+1-j)}^{\sigma, \sigma}, W_{(n+1-i)(n+1-j)}^{\sigma, \sigma} \}$ are independent as $(i, j)$ ranges over $[n/2]^2$, it follows that $W' = \phi_1(W) \sim N(0, 1)^{\otimes n \times n}$. Iterating, it now follows that $\phi(N(0, 1)^{\otimes n \times n}) \sim N(0, 1)^{\otimes n \times n}$, establishing Property 1.

Now consider the case when $W = \lambda \cdot r c^\top + U$ where $U \sim N(0, 1)^{\otimes n \times n}$. Note that $W'$ can be expressed in terms of $W^{\sigma, \sigma}$ as
$$W' = \frac{1}{2} \left( A + B \right)^\top W^{\sigma, \sigma} \left( A + B \right)^\top = \frac{\lambda}{2} \left( A r^\sigma + B r^\sigma \right) \left( Ac^\sigma + Bc^\sigma \right)^\top + \frac{1}{2} \left( A + B \right)^\top U^{\sigma, \sigma} \left( A + B \right)$$
where $B$ is the $n \times n$ matrix with ones on its anti-diagonal and zeros elsewhere, and $A$ is given by
$$A = \left[ \begin{matrix} I_{n/2} & 0 \\ 0 & -I_{n/2} \end{matrix} \right]$$
Note that $\frac{1}{2} \left( A + B \right)^\top U^{\sigma, \sigma} \left( A + B \right)$ is distributed as $\phi_1(U) \sim N(0, 1)^{\otimes n \times n}$. Since $A + B$ is symmetric and satisfies that $(A + B)^2 = 2 \cdot I_n$, we have that
$$\| Ar^\sigma + Br^\sigma \|_2^2 = 2 \| r^\sigma \|_2^2 = 2\| r \|_2^2 \quad \text{and} \quad \| Ac^\sigma + Bc^\sigma \|_2^2 = 2 \| c^\sigma \|_2^2 = 2\| c \|_2^2$$
Let $r_0 = r$ and $r_{i+1} = A r^{\sigma_i} + B r^{\sigma_i}$ where $\sigma_i$ is the permutation generated in the $i$th iteration of Step 2. It follows by induction that $r_\ell \in \mathbb{Z}^n$ and $\| r_\ell \|_2^2 = 2^\ell \| r \|_2^2$ hold almost surely. Analogously define $c_i$ for each $0 \le i \le \ell$ and note that
$$\phi\left(\lambda \cdot r c^\top + N(0, 1)^{\otimes n \times n}\right) \sim \mL\left( \frac{\lambda}{2^\ell} \cdot r_\ell c_\ell^\top + N(0, 1)^{\otimes n \times n}\right)$$
Furthermore note that if $r_0 = r = c = c_0$, then $r_i = c_i$ for all $i$ holds almost surely. Thus it suffices to show that the desired bounds on $\| r_\ell \|_0$ and $\| c_\ell \|_0$ hold with high probability in order to establish Property 2.

Note that since $A + B$ has two nonzero entries per row and column, it follows that $2\| r_i \|_0 \ge \| r_{i+1} \|_0$ for all $i$. Now consider the collision set
$$S_i = \left\{ \{j, n + 1 - j\} : j, n+1 - j \in \text{supp}\left(r_i^\sigma\right) \right\}$$
Note that if $j \in \text{supp}\left(r_i^\sigma\right)$, then $j$ is only not in the support of $r_{i+1}$ if it is in some unordered pair in $S_i$. Also note that $(r_{i+1})_j + (r_{i+1})_{n+1 - j} = 2r^\sigma_j \neq 0$ if $j$ is in some unordered pair in $S_i$. Therefore at most one element per unordered pair of $S_i$ can be absent from the support of $r_{i+1}$. This implies that $2\| r_i \|_0 - \| r_{i+1} \|_0 \le |S_i|$ for each $i$. For each $1 \le j \le n/2$, let $X_j$ be the indicator for the event that $\{j, n+1 - j\} \in S_i$. Let $t = \| r_i \|_0$ and note that $|S_i| = X_1 + X_2 + \cdots + X_{n/2}$. For any subset $T \subseteq [n/2]$, it follows that if $|T| \le n/2$ then
$$\bE\left[ \prod_{j \in T} X_j \right] = \frac{t(t-1)\cdots(t - 2|T| + 1)}{n(n-1)\cdots(n - 2|T| + 1)} \le \left( \frac{t}{n} \right)^{2|T|}$$
and if $|T| > n/2$, then this expectation is zero. Let $Y \sim \text{Bin}(n/2, t^2/n^2)$ and note that the above inequality implies that $\bE[|S_i|^k] \le \bE[Y^k]$ for all $j \ge 0$. This implies that if $\theta \ge 0$, then
$$\bE[\exp(\theta |S_i|)] \le \bE[\exp(\theta Y)] = \left( 1 + (e^\theta - 1) \cdot \frac{t^2}{n^2} \right)^{n/2} \le \exp\left( (e^\theta - 1) \cdot \frac{t^2}{2n} \right)$$
A Chernoff bound now yields that
$$\bP[|S_i| \ge k] \le \exp\left( (e^\theta - 1) \cdot \frac{t^2}{2n} - \theta k \right)$$
Setting $k = t^2/n$ and $\theta = \ln 2$ yields that
$$\bP\left[|S_i| \ge \frac{t^2}{n}\right] \le \left( \frac{e}{4} \right)^{\frac{t^2}{2n}} \le \frac{1}{t}$$
if $t^2/2n \ge \log_{4/e} t$. If $t^2/2n < \log_{4/e} t$, setting $\theta = \ln 2$ and $k = \frac{1}{\ln 2} \left( \log_{4/e} t + \ln t \right) = C\log t$ yields
$$\bP\left[|S_i| \ge C \log t \right] \le \exp\left( \frac{t^2}{2n} - (\ln 2) k \right) = \frac{1}{t}$$
Therefore with probability at least $1 - 1/\| r_i \|_0$, it follows that $|S_i| < \max(C\log \| r_i \|_0, \| r_i \|_0^2/n)$. Note that this inequality implies that
$$\| r_{i+1} \|_0 \ge 2\| r_i \|_0 - |S_i| \ge 2\| r_i \|_0 \left(1 - \max\left(\frac{C \log \| r_i \|_0}{\| r_i\|_0}, \frac{\| r_i \|_0}{n}\right) \right)$$
We now will show by induction on $0 \le j \le \ell$ that as long as $\|r \|_0$ is sufficiently large, we have
\begin{equation}
\| r_j \|_0 \ge 2^j \| r \|_0 \left( 1 - \max\left(\frac{2Cj \cdot \log (2^j \| r \|_0)}{\| r \|_0}, \frac{2^j \| r \|_0}{n}\right) \right)
\end{equation}
holds with probability at least
$$1 - \frac{4(1 - 2^{-j})}{\| r \|_0}$$
The claim is vacuously true when $j = 0$. Now assume that this holds for a particular $j$. Now note that since $j \le \ell$ and $\| r \|_0$ is sufficiently large but at most $2^{-\ell - 1} n$, we have that
$$\max\left(\frac{2Cj \cdot \log (2^j \| r \|_0)}{\| r \|_0}, \frac{2^j \| r \|_0}{n}\right) \le \frac{1}{2}$$
Therefore $\| r_j \|_0 \ge 2^{j-1} \| r \|_0$. We now condition on the value of $\| r_j \|_0$ and the event (7.1) holds. Under this conditioning, it follows by the induction hypothesis that with probability at least $1 - 1/\| r_j \|_0 \ge 1 - 2^{1-j} \| r\|_0^{-1}$, we have that
\begin{align*}
\| r_{j+1} \|_0 &\ge 2^{j+1} \| r \|_0 \left( 1 - \max\left(\frac{2Cj \cdot \log (2^j \| r \|_0)}{\| r \|_0}, \frac{2^j \| r \|_0}{n}\right) \right) \left(1 - \max\left(\frac{C \log \| r_j \|_0}{\| r_j\|_0}, \frac{\| r_j \|_0}{n}\right) \right) \\
&2^{j+1} \| r \|_0 \left( 1 - \max\left(\frac{2C(j+1) \cdot \log (2^{j+1} \| r \|_0)}{\| r \|_0}, \frac{2^{j+1} \| r \|_0}{n}\right) \right)
\end{align*}
since $2^{j-1} \| r \|_0 \le \| r_j \| \le 2^j \| r \|_0$. Marginalizing over $\| r_j \|_0$ and whether or not (7.1) holds yields that the above inequality holds with probability at least
$$\left( 1 - \frac{4(1 - 2^{-j})}{\| r \|_0} \right) \left( 1 - \frac{2^{1-j}}{\| r \|_0^{-1}} \right) \ge 1 - \frac{4(1 - 2^{-j})}{\| r \|_0} - \frac{2^{1-j}}{\| r \|_0} = 1 - \frac{4(1 - 2^{-j-1})}{\| r \|_0}$$
which completes the induction. Obtaining symmetric results for $c$ and taking $j = \ell$ completes the proof of the lemma.
\end{proof}

Combining this reflection cloning procedure with the reduction from planted clique to biclustering yields a reduction from planted clique to rank-1 submatrix. The proof of the next lemma and theorem are deferred to Appendix \ref{app7}.

\begin{figure}[t!]
\begin{algbox}
\textbf{Algorithm} \textsc{ROS-Reduction}

\vspace{2mm}

\textit{Inputs}: Graph $G \in \mG_n$, number of iterations $\ell$
\begin{enumerate}
\item Compute the output $W$ of $\textsc{BC-Reduction}$ applied to $G$ with zero iterations
\item Return the output of $\textsc{Reflection-Cloning}$ applied to $W$ with $\ell$ iterations
\end{enumerate}
\vspace{1mm}
\end{algbox}
\caption{Reduction to rank-1 submatrix in Lemma \ref{lem:ros}.}
\label{fig:ros}
\end{figure}

\begin{lemma} \label{lem:ros}
Suppose that $n$ is even and $2^{\ell} k < \frac{n}{\log k}$ where $n$ is sufficiently large. Let
$$\mu = \frac{\log 2}{2 \sqrt{6 \log n + 2\log 2}}$$
There is a randomized polynomial time computable map $\phi = \textsc{ROS-Reduction}$ with $\phi : \mG_n \to \mathbb{R}^{n \times n}$ such that if $G$ is an instance of $\text{PC}(n, k, 1/2)$ then under $H_0$, it holds that
$$\TV\left( \mL_{H_0}(\phi(G)), N(0, 1)^{\otimes n \times n} \right) = O\left( \frac{1}{\sqrt{\log n}} \right)$$
and, under $H_1$, there is a prior $\pi$ on pairs of unit vectors in $\mathcal{V}_{n, 2^\ell k}$ such that
$$\TV\left( \mL_{H_1}(\phi(G)), \int \mL\left( \frac{\mu k}{\sqrt{2}} \cdot  uv^\top + N(0, 1)^{\otimes n \times n} \right) d\pi(u, v) \right) = O\left( \frac{1}{\sqrt{\log n}} + k^{-1} \right)$$
\end{lemma}

This lemma provides a polynomial time map from an instance of $\textsc{PC}_D(n, k, 1/2)$ to $N(0, 1)^{\otimes n \times n}$ under $H_0$ and to a distribution in the composite hypothesis $H_1$ of $\textsc{ROS}_D(n, 2^\ell k, 2^{-1/2} \mu)$ under $H_1$. Now we deduce the hard regime of $\textsc{ROS}_D$ given the planted clique conjecture as in the next theorem. Here, we consider the asymptotic regime $\frac{\mu}{k} = \tilde{\Theta}(n^{-\alpha})$ to be consistent with Figure \ref{fig:diagrams}. The purpose of this parameterization is to focus on the factor $\frac{\mu}{k}$ required to normalize entries in the planted submatrix to have magnitude approximately $1$. This enables a valid comparison between the hardness of $\textsc{ROS}_D$ and $\textsc{BC}_D$.

\begin{theorem} \label{thm:ros}
Let $\alpha > 0$ and $\beta \in (0, 1)$ be such that $\beta < \frac{1}{2} + \alpha$. There is a sequence $\{ (N_n, K_n, \mu_n) \}_{n \in \mathbb{N}}$ of parameters such that:
\begin{enumerate}
\item The parameters are in the regime $\frac{\mu}{K} = \tilde{\Theta}(N^{-\alpha})$ and $K = \tilde{\Theta}(N^\beta)$ or equivalently,
$$\lim_{n \to \infty} \frac{\log (K_n \mu_n^{-1})}{\log N_n} = \alpha \quad \text{and} \quad \lim_{n \to \infty} \frac{\log K_n}{\log N_n} = \beta$$
\item For any sequence of randomized polynomial-time tests $\phi_n : \mG_{N_n} \to \{0, 1\}$, the asymptotic Type I$+$II error of $\phi_n$ on the problems $\textsc{ROS}_D(N_n, K_n, \mu_n)$ is at least $1$ assuming the PC conjecture holds with density $p = 1/2$.
\end{enumerate}
Therefore the computational boundary for $\textsc{ROS}_D(n, k, \mu)$ in the parameter regime $\frac{\mu}{k} = \tilde{\Theta}(n^{-\alpha})$ and $k = \tilde{\Theta}(n^\beta)$ is $\beta^* = \frac{1}{2} + \alpha$ and $\alpha^* = 0$ when $\beta < \frac{1}{2}$.
\end{theorem}

\subsection{Sparse Spiked Wigner Matrix}


We now show a computational lower bound for sparse spiked Wigner matrix detection given the planted clique conjecture. As observed in Section 2, it suffices to reduce from planted clique to $\textsc{SROS}_D$. This is because if $M$ is an instance of $\textsc{SROS}_D(n, k, \mu)$, then $\frac{1}{\sqrt{2}} (M + M^\top)$ is an instance of $\textsc{SSW}_D(n, k, \mu/\sqrt{2})$. This transformation implies that any computational lower bound that applies to $\textsc{SROS}_D(n, k, \mu)$ also applies to $\textsc{SSW}_D(n, k, \mu/\sqrt{2})$. The rest of this section is devoted to giving a reduction to $\textsc{SROS}_D(n, k, \mu)$.

Our reduction uses the symmetry preserving property of reflection cloning and yields the same computational barrier as for $\textsc{ROS}_D$. However, there are several subtle differences between this reduction and that in Lemma \ref{lem:ros}. In order to show hardness for $\textsc{SROS}_D$, it is important to preserve the symmetry of the planted sparse structure. This requires planting the hidden entries along the diagonal of the adjacency matrix of the input graph $G$, which we do by an averaging trick and generating additional randomness to introduce independence. However, this induces an arbitrarily small polynomial loss in the size of the spike, unlike in the reduction to $\textsc{ROS}_D$. Although this does not affect our main theorem statement for $\textsc{SROS}_D$, which only considers $\text{poly}(n)$ size factors, it yields a weaker lower bound than that in Theorem \ref{thm:ros} when examined up to sub-polynomial factors. This reduction to $\textsc{SROS}_D$ will also serve as the main sub-routine in our reduction to $\textsc{SSBM}_D$.

\begin{figure}[t!]
\begin{algbox}
\textbf{Algorithm} \textsc{SROS-Reduction}

\vspace{2mm}

\textit{Inputs}: Planted clique instance $G \in \mG_n$ with clique size $k$ where $n$ is even, iterations $\ell$
\begin{enumerate}
\item Let $\textsc{rk}_{\text{G}} = \textsc{rk}\left( 1 \to N(\mu, 1), 1/2 \to N(0, 1), N \right)$ where $N = \lceil 6 \log_2 n \rceil$ and $\mu = \frac{\log 2}{2 \sqrt{6 \log n + 2\log 2}}$ and form the symmetric matrix $W \in \mathbb{R}^{n \times n}$ with $W_{ii} = 0$ and for all $i < j$,
$$W_{ij} = \textsc{rk}_G\left( \mathbf{1}_{\{i, j\} \in E(G)}\right)$$
\item Sample an antisymmetric matrix $A \in \mathbb{R}^{n \times n}$ with i.i.d. $N(0, 1)$ entries below its main diagonal and sample two matrices $B, C \in \mathbb{R}^{n \times n}$ with i.i.d. $N(0, 1)$ off-diagonal and zero diagonal entries
\item Form the matrix $M \in \mathbb{R}^{n \times n}$ with
$$M_{ij} = \frac{k - 1}{2 \sqrt{n - 1}} \left( W_{ij} + A_{ij} + B_{ij} \cdot \sqrt{2} \right) + C_{ij} \cdot \sqrt{1 - \frac{(k-1)^2}{n - 1}}$$
for all $i \neq j$ and
$$M_{ii} = \frac{1}{2\sqrt{n - 1}} \sum_{j = 1}^n \left( W_{ij} + A_{ij} - B_{ij} \cdot \sqrt{2} \right)$$
\item Update $M$ to be the output of $\textsc{Reflection-Cloning}$ applied with $\ell$ iterations to $M$
\item Output $M^{\sigma, \sigma}$ where $\sigma$ is a permutation of $[n]$ chosen uniformly at random
\end{enumerate}
\vspace{1mm}
\end{algbox}
\caption{Reduction to sparse spiked Wigner matrix in Lemma \ref{lem:ssw}.}
\label{fig:ssw}
\end{figure}

\begin{lemma} \label{lem:ssw}
Suppose that $n$ is even and $2^{\ell} k < \frac{n}{\log k}$ where $n$ is sufficiently large. Let
$$\mu = \frac{\log 2}{2 \sqrt{6 \log n + 2\log 2}}$$
There is a randomized polynomial time computable map $\phi = \textsc{SROS-Reduction}$ with $\phi : \mG_n \to \mathbb{R}^{n \times n}$ such that if $G$ is an instance of $\text{PC}(n, k, 1/2)$ then under $H_0$, it holds that
$$\TV\left( \mL_{H_0}(\phi(G)), N(0, 1)^{\otimes n \times n} \right) = O( n^{-1} )$$
and, under $H_1$, there is a prior $\pi$ on unit vectors in $\mathcal{V}_{n, 2^\ell k}$ such that
$$\TV\left( \mL_{H_1}(\phi(G)), \int \mL\left( \frac{\mu k(k-1)}{2\sqrt{(n-1)}} \cdot  vv^\top + N(0, 1)^{\otimes n \times n} \right) d\pi(v) \right) = O\left( n^{-1} \right)$$
\end{lemma}

\begin{proof}
Let $\phi = \textsc{SROS-Reduction}$ be as in Figure \ref{fig:ssw}. Applying the total variation bounds in Lemma \ref{lem:5c} and entry-wise coupling as in Lemma \ref{lem:bcrec} yields that
$$\TV\left( \mL_{H_0}(W), M_n(N(0, 1)) \right) = O(n^{-1}) \quad \text{and} \quad \TV\left( \mL_{H_1}(W), M_n(k, N(\mu, 1), N(0, 1)) \right) = O(n^{-1})$$
Let $W' \in \mathbb{R}^{n \times n}$ be such that $W' \sim M_n(N(0, 1))$ under $H_0$ and $W' \sim M_n(k, N(\mu, 1), N(0, 1))$ under $H_1$. Let $M'$ denote the matrix computed in Step 3 using the matrix $W'$ in place of $W$. We will first argue that under $H_0$, the variables $W'_{ij} + A_{ij} + B_{ij} \cdot \sqrt{2}$ and $W'_{ij} + A_{ij} - B_{ij} \cdot \sqrt{2}$ for all $i \neq j$ are independent. First note that the $\sigma$-algebras $\sigma\{ W'_{ij}, A_{ij}, B_{ij} \}$ for all $i < j$ are independent. Therefore it suffices to verify that the four variables $W'_{ij} + A_{ij} \pm B_{ij} \cdot \sqrt{2}$ and $W'_{ji} + A_{ji} \pm B_{ji} \cdot \sqrt{2}$ are independent. Observe that these four variables are jointly Gaussian and satisfy that
$$\frac{1}{2} \cdot \left[ \begin{matrix} W'_{ij} + A_{ij} + B_{ij} \cdot \sqrt{2} \\ W'_{ij} + A_{ij} - B_{ij} \cdot \sqrt{2} \\ W'_{ji} + A_{ji} + B_{ji} \cdot \sqrt{2} \\ W'_{ji} + A_{ji} - B_{ji} \cdot \sqrt{2}\end{matrix} \right] = \frac{1}{2} \cdot \left[ \begin{matrix} 1 & 1 & \sqrt{2} & 0 \\ 1 & 1 & - \sqrt{2} & 0 \\ 1 & - 1 & 0 & \sqrt{2} \\ 1 & - 1 & 0 & -\sqrt{2} \end{matrix} \right] \cdot \left[ \begin{matrix} W'_{ij} \\ A_{ij} \\ B_{ij} \\ B_{ji} \end{matrix} \right]$$
since $W'$ is symmetric and $A$ is antisymmetric. Observe that $W'_{ij}, A_{ij}, B_{ij}, B_{ji}$ are independent Gaussians and since the coefficient matrix above is orthogonal, it follows that the vector on the left hand side above is distributed as $N(0, 1)^{\otimes 4}$ and thus has independent entries. Since $C$ has i.i.d. $N(0, 1)$ entries off of its diagonal and $W', A, B$ and $C$ all have zero diagonals, it follows that $\text{Var}(M'_{ij}) = 1$ for all $i, j \in [n]$. Since the entries of $M'$ are independent and each entry is Gaussian with variance $1$, it follows that $M' \sim N(0, 1)^{\otimes n \times n}$.

Now suppose that $H_1$ holds and let $S \subseteq [n]$ be indices of the rows and columns containing the planted $N(\mu, 1)$ entries of $W'$. It now follows that $W'_{ij} = \mu + W''_{ij}$ if $i \neq j$ and $i, j \in S$ and $W'_{ij} = W''_{ij}$ where $W'' \sim M_n(N(0, 1))$. Now note that if $i \neq j$, we have that conditioned on $S$,
\begin{align*}
M'_{ij} &= \frac{(k-1)\mu}{2 \sqrt{n - 1}} \cdot \mathbf{1}_{\{i, j \in S\}}+ \frac{k-1}{2 \sqrt{n - 1}} \left( W''_{ij} + A_{ij} + B_{ij} \cdot \sqrt{2} \right) + C_{ij} \cdot \sqrt{1 - \frac{(k-1)^2}{n - 1}} \\
&\sim \frac{(k-1)\mu}{2 \sqrt{n - 1}} \cdot \mathbf{1}_{\{i, j \in S\}} + N(0, 1)
\end{align*}
by applying the previous argument to $W''$. Furthermore $M'_{ii}$ has diagonal entries
\begin{align*}
M'_{ij} &= \frac{(k-1)\mu}{2 \sqrt{n - 1}} \cdot \mathbf{1}_{\{i, j \in S\}}+ \frac{1}{2\sqrt{n - 1}} \sum_{j = 1}^n \left( W''_{ij} + A_{ij} - B_{ij} \cdot \sqrt{2} \right) \\
&\sim \frac{(k-1)\mu}{2 \sqrt{n - 1}} \cdot \mathbf{1}_{\{i \in S\}} + N(0, 1)
\end{align*}
conditioned on $S$. Therefore $M' | S\sim \frac{(k-1)\mu}{2\sqrt{n - 1}} \cdot \mathbf{1}_S \mathbf{1}_S^\top + N(0, 1)^{\otimes n \times n}$. Now by the data processing inequality, we now have that under $H_0$,
$$\TV\left( \mL_{H_0}(M), N(0, 1)^{\otimes n \times n} \right) \le \TV\left( \mL_{H_0}(W), M_n(N(0, 1)) \right) = O(n^{-1})$$
By the data processing and triangle inequalities, we have that
\begin{align*}
&\TV\left( \mL_{H_1}(M), \int \mL\left( \frac{(k-1)\mu}{2\sqrt{(n-1)}} \cdot  \mathbf{1}_S \mathbf{1}_S^\top + N(0, 1)^{\otimes n \times n} \right) d\pi'(S) \right) \\
&\quad \quad \le \bE_S \TV\left( \mL_{H_1}(M), \frac{(k-1)\mu}{2\sqrt{(n-1)}} \cdot  \mathbf{1}_S \mathbf{1}_S^\top + N(0, 1)^{\otimes n \times n} \right) = O(n^{-1})
\end{align*}
where $\pi'$ is the uniform distribution on $k$-subsets of $[n]$. Now applying the same argument as in the proof of Lemma \ref{lem:ros} and the fact that $\textsc{Reflection-Cloning}$ preserves the fact that the rank-1 mean submatrix is symmetric as shown in Lemma \ref{lem:refc}, proves the lemma.
\end{proof}

We now use this lemma to deduce the computational barriers for $\textsc{SROS}_D$ and $\textsc{SSW}_D$. Although the barrier matches Theorem \ref{thm:ros}, the parameter settings needed to achieve it are slightly different due to the polynomial factor loss in the reduction in Lemma \ref{lem:ssw}. The proof is deferred to Appendix \ref{app7}.

\begin{theorem} \label{thm:ssw}
Let $\alpha > 0$ and $\beta \in (0, 1)$ be such that $\beta < \frac{1}{2} + \alpha$. There is a sequence $\{ (N_n, K_n, \mu_n) \}_{n \in \mathbb{N}}$ of parameters such that:
\begin{enumerate}
\item The parameters are in the regime $\frac{\mu}{K} = \tilde{\Theta}(N^{-\alpha})$ and $K = \tilde{\Theta}(N^\beta)$ or equivalently,
$$\lim_{n \to \infty} \frac{\log (K_n \mu_n^{-1})}{\log N_n} = \alpha \quad \text{and} \quad \lim_{n \to \infty} \frac{\log K_n}{\log N_n} = \beta$$
\item For any sequence of randomized polynomial-time tests $\phi_n : \mG_{N_n} \to \{0, 1\}$, the asymptotic Type I$+$II error of $\phi_n$ on the problems $\textsc{SROS}_D(N_n, K_n, \mu_n)$ and $\textsc{SSW}_D(N_n, K_n, \mu_n/\sqrt{2})$ is at least $1$ assuming the PC conjecture holds with density $p = 1/2$.
\end{enumerate}
Therefore the computational boundaries for $\textsc{SROS}_D(n, k, \mu)$ and $\textsc{SSW}_D(n, k, \mu)$ in the parameter regime $\frac{\mu}{k} = \tilde{\Theta}(n^{-\alpha})$ and $k = \tilde{\Theta}(n^\beta)$ is $\beta^* = \frac{1}{2} + \alpha$ and $\alpha^* = 0$ when $\beta < \frac{1}{2}$.
\end{theorem}

\subsection{Subgraph Stochastic Block Model}

\begin{figure}[t!]
\begin{algbox}
\textbf{Algorithm} \textsc{SSBM-Reduction}

\vspace{2mm}

\textit{Inputs}: Planted clique instance $G \in \mG_n$ with clique size $k$ where $n$ is even, iterations $\ell$
\begin{enumerate}
\item Compute the output $M$ of $\textsc{SROS-Reduction}$ applied to $G$ with $\ell$ iterations
\item Sample $n$ i.i.d. Rademacher random variables $x_1, x_2, \dots, x_n$ and update each entry of $M$ to be $M_{ij} \gets x_i x_j M_{ij}$
\item Output the graph $H$ where $\{i, j\} \in E(H)$ if and only if $M_{ij} > 0$ for each $i < j$
\end{enumerate}
\vspace{1mm}
\end{algbox}
\caption{Reduction to the subgraph stochastic block model in Lemma \ref{lem:ssbm}.}
\end{figure}

In this section, we deduce tight hardness for the subgraph stochastic block model from the reflection cloning reduction from planted clique to $\textsc{SROS}_D$. This captures the sharper hardness of detection in the subgraph stochastic block model over planted dense subgraph and the lack of a sum test. Note that here total variation distance under $H_1$ refers to the total variation distance to some prior over the distributions in $H_1$.

\begin{lemma} \label{lem:ssbm}
Suppose that $n$ is even and $\delta \in (0, 1/2)$ is such that $(2^{\ell}k)^{1 + 2\delta} = O(n)$ and $k = \Omega\left(n^\delta\right)$. Let $\mu > 0$ be such that
$$\mu = \frac{\log 2}{2 \sqrt{6 \log n + 2\log 2}}$$
There is a polynomial time map $\phi = \textsc{SSBM-Reduction}$ with $\phi : \mG_n \to \mG_n$ such that for some mixture $\mL_{\textsc{SSBM}}$ of distributions in $G_B\left(n, 2^\ell k, 1/2, \rho \right)$, it holds that
$$\TV\left( \phi(G(n, 1/2)), G(n, 1/2) \right) = O\left(n^{-1} \right) \quad \text{and} \quad \TV\left( \phi(G(n, k, 1/2)), \mL_{\text{SSBM}}\right) = O\left( k^{-1} \right)$$
where $\rho$ is given by
$$\rho = \Phi\left( \frac{\mu(k - 1)}{2^{\ell + 1}\sqrt{n- 1}} \right) - \frac{1}{2}$$
\end{lemma}

\begin{proof}
Given a vector $v \in \mathbb{Z}^n$ and $\theta > 0$, let $M(\theta, v) \in \mathbb{R}^{n \times n}$ be distributed as $\theta \cdot vv^\top + N(0, 1)^{\otimes n \times n}$. Let $S \subseteq [n]$ and $P \in [0, 1]^{|S| \times |S|}$ be a symmetric matrix with zero diagonal entries and rows and columns indexed by vertices in $S$. Let $G\left(n, S, P, q\right)$ be the distribution on graphs $G$ generated as follows:
\begin{enumerate}
\item if $i \neq j$ and $i, j \in S$ then the edge $\{i, j\} \in E(G)$ independently with probability $P_{ij}$; and
\item all other edges of $G$ are included independently with probability $q$.
\end{enumerate}
Now let $\tau : \mathbb{R}^{n \times n} \to \mG_n$ be the map sending a matrix $M$ to the graph $G$ such that
$$E(G) = \left\{ \{i, j \} : M_{ij} > 0 \text{ and } i < j \right\}$$
In other words, $G$ is the graph formed by thresholding the entries of $M$ below its main diagonal at zero. Suppose that $v \in \mathbb{Z}^n$ is $k$-sparse and $S = \text{supp}(v)$. Since the entries of $M(\theta, v)$ are independent, it follows that
$$\tau\left(M(\theta, v)\right) \sim G(n, S, P, q) \quad \text{where} \quad P_{ij} = \Phi\left( \theta \cdot v_i v_j \right) \text{ for each } i, j \in S$$
Now let $S = A \cup B$ where $A = A(v) = \{ i \in S : v_i > 0 \}$ and $B = B(v) = \{ i \in S : v_i < 0 \}$. Note that since $v \in \mathbb{Z}^n$, it follows that if $i, j \in A$ or $i, j \in B$ then $v_i v_j \ge 1$ and thus $P_{ij} = \Phi\left( \theta \cdot v_i v_j \right) \ge \Phi(\theta)$. Furthermore if $(i, j) \in A \times B$ or $(i, j) \in B \times A$ then $v_i v_j \le - 1$ and $P_{ij} = \Phi\left( \theta \cdot v_i v_j \right) \le 1 - \Phi(\theta)$. Now note that if $\sigma$ is a permutation of $[n]$ chosen uniformly at random then
$$\mL\left( \tau\left(M(\theta, v)\right)^{\sigma} \right) = \mL\left( \tau\left(M(\theta, v)^{\sigma, \sigma}\right) \right) \in G_B\left(n, k, 1/2, \Phi(\theta) - 1/2 \right)$$
if it also holds that $\frac{k}{2} - k^{1 - \delta} \le |A|, |B| \le \frac{k}{2} + k^{1 - \delta}$, by the definition of $G_B$ in Section 2.2.

Let $\phi_1 = \textsc{SROS-Reduction}$ and let $W \sim N(0, 1)^{\otimes n \times n}$. As shown in Lemma \ref{lem:ssw},
$$\TV\left( \phi_1(G(n,1/2)), \mL(W) \right) = O( n^{-1} )$$
Now observe that since $N(0, 1)$ is symmetric and since the entries of $W$ are independent, the distribution of $W$ is invariant to flipping the signs of any subset of the entries of $W$. Therefore $xx^\top \circ W \sim N(0, 1)$ where $\circ$ denotes the entry-wise or Schur product on $\mathbb{R}^{n \times n}$ and $x \in \{-1, 1\}^n$ is chosen uniformly at random. Therefore the data processing inequality implies that
\begin{align*}
\TV\left( \mL\left( xx^\top \circ \phi_1(G(n,1/2)) \right), \mL(W) \right) &= \TV\left( \mL\left( xx^\top \circ \phi_1(G(n,1/2)) \right), \mL\left( xx^\top \circ W\right) \right) \\
&\le \TV\left( \phi_1(G(n,1/2)), \mL(W) \right) = O( n^{-1} )
\end{align*}
Recalling that $\phi$ is the output of $\textsc{SSBM-Reduction}$, note that $\phi(G(n, 1/2))$ is distributed as $\tau\left( xx^\top \circ \phi_1(G(n,1/2)) \right)$ and that $\tau\left(N(0, 1)^{\otimes n \times n}\right)$ is distributed as $G(n, 1/2)$. It follows by the data processing inequality that
$$\TV\left( \phi(G(n, 1/2)), G(n, 1/2) \right) = \TV\left( \mL\left( \tau\left( xx^\top \circ \phi_1(G(n,1/2)) \right)\right), \mL(\tau(W)) \right) = O(n^{-1})$$
Repeating the analysis in Lemmas \ref{lem:ros} and \ref{lem:ssw} without converting the prior in Lemma \ref{lem:refc} to be over unit vectors yields that there is a prior $\pi$ on vectors in $u \in \mathbb{Z}^n$ such that
$$\TV\left( \phi_1(G(n, k, 1/2)), \int \mL\left( \frac{\mu (k-1)}{2^{\ell + 1}\sqrt{n-1}} \cdot  uu^\top + N(0, 1)^{\otimes n \times n} \right) d\pi(u) \right) = O\left( n^{-1} \right)$$
and such that with probability at least $1 - 8k^{-1}$ it holds that
$$2^\ell k \ge \| u \|_0 \ge 2^{\ell} k \left( 1 - \max \left( \frac{2C \ell \cdot \log(2^\ell k)}{k}, \frac{2^\ell k}{n} \right) \right) = 2^{\ell} k - O\left( 2^\ell (\log n)^2 + (2^\ell k)^{1-2\delta} \right)$$
where $C > 0$ is the constant in Lemma \ref{lem:refc}. Now let $\pi'$ be the prior $\pi$ conditioned on the inequality in the last displayed equation. It follows by Lemma \ref{lem:5tv} that $\TV( \pi, \pi') \le 8k^{-1}$. Now let $\theta = \frac{\mu (k-1)}{2^{\ell + 1}\sqrt{n-1}}$ and let the matrix $M'$ be distributed as
$$M' \sim \int \mL\left( \theta \cdot  uu^\top + N(0, 1)^{\otimes n \times n} \right) d\pi(u)$$
As above, let $x \in \{-1, 1\}^n$ be chosen uniformly at random. The same argument above shows that
$$\mL\left( xx^\top \circ M' \Big| x \circ u = v \right) = \mL\left( \theta \cdot vv^\top + N(0, 1)^{\otimes n \times n}\right) = \mL(M(\theta, v))$$
As shown above, this implies that 
$$\mL\left( \tau\left(xx^\top \circ M'\right) \Big| x \circ u = v \right) \in G_B\left(n, 2^\ell k, 1/2, \Phi(\theta) - 1/2 \right)$$
as long as $2^{\ell - 1}k - (2^\ell k)^{1 - \delta} \le |A(v)|, |B(v)| \le 2^{\ell - 1}k + (2^\ell k)^{1 - \delta}$. Let $\pi''(x, u)$ be the product distribution of $\mL(x)$ and $\pi'$ conditioned on the event that these inequalities hold for $v = x \circ u$. Now note that conditioning on $u$ yields that $|A(x \circ u)| + |B(x \circ u)| = \| u \|_0$ and $|A(x \circ u)|$ is distributed as $\text{Bin}(\| u\|_0, 1/2)$. By Hoeffding's inequality, we have that conditioned on $u$,
$$\bP\left[ \left| |A(x \circ u)| - \frac{1}{2} \| u \|_0 \right| > \sqrt{\| u \|_0 \log k} \right] \le \frac{2}{k^2}$$
Note that if this inequality holds for $ |A(x \circ u)|$, then it also holds for $|B(x \circ u)|$ since these sum to $\| u \|_0$. Therefore with probability at least $1 - 2k^{-2}$, it holds that
\begin{align*}
\left| |A(x \circ u)| - 2^{\ell - 1} k \right| &\le \left| |A(x \circ u)| - \frac{1}{2} \| u \|_0 \right| + O\left( 2^\ell (\log n)^2 + (2^\ell k)^{1-2\delta} \right) \\
&= O\left( \sqrt{2^\ell k \log k} + 2^\ell (\log n)^2 + (2^\ell k)^{1-2\delta} \right) \le (2^{\ell} k)^{1 - \delta}
\end{align*}
for sufficiently large $k$ and the same inequalities hold for $|B(x \circ u)|$. This verifies that the desired inequalities on $|A(x \circ u)|$ and $|B(x \circ u)|$ hold with probability at least $1 - 2k^{-2}$ over $\mL(x) \otimes \pi(u)$. Applying Lemma \ref{lem:5tv} therefore yields that $\TV\left( \mL(x) \times \pi'(u), \pi''(x, u) \right) \le 2k^{-2}$. The data processing and triangle inequalities then imply that
\begin{align*}
\TV\left( \mL(x) \otimes \pi(u), \pi''(x, u) \right) &\le \TV\left( \mL(x) \otimes \pi(u), \mL(x) \otimes \pi'(u) \right) + \TV\left( \mL(x) \otimes \pi'(u), \pi''(x, u) \right) \\
&\le \TV( \pi, \pi') + 2k^{-2} = O(k^{-1})
\end{align*}
Now observe that
\begin{align*}
&\TV\left( \phi(G(n, k, 1/2), \int \mL\left( \tau\left(xx^\top \circ M'\right) \Big| x, u \right) d\pi''(x, u) \right) \\
&\quad \quad \quad \quad \le \TV\left( \mL\left( \tau\left( xx^\top \circ \phi_1(G(n, k, 1/2) \right) \right), \int \mL\left( \tau\left(xx^\top \circ M'\right) \Big| x, u \right) d\mL(x) d\pi(u) \right) \\
&\quad \quad \quad \quad \quad \quad + \TV\left( \int \mL\left( \tau\left(xx^\top \circ M'\right) \Big| x, u \right) d\mL(x) d\pi(u), \int \mL\left( \tau\left(xx^\top \circ M'\right) \Big| x, u \right) d\pi''(x, u) \right) \\
&\quad \quad \quad \quad \le \TV\left( \phi_1(G(n, k, 1/2) , \mL(M') \right) + \TV\left( \mL(x) \otimes \pi(u), \pi''(x, u) \right) \\
&\quad \quad \quad \quad = O(n^{-1}) + O(k^{-1}) = O(k^{-1})
\end{align*}
By the definition of $\pi''$, the distribution
$$\mL_{\textsc{SSBM}} = \int \mL\left( \tau\left(xx^\top \circ M'\right) \Big| x, u \right) d\pi''(x, u)$$
is a mixture of distributions in $G_B\left(n, 2^\ell k, 1/2, \Phi(\theta) - 1/2 \right)$, completing the proof of the lemma.
\end{proof}

Applying this reduction and setting parameters similarly to Theorem \ref{thm:ssw} yields the following computational lower bound for the subgraph stochastic block model. The proof is a calculation deferred to Appendix \ref{app7}.

\begin{theorem} \label{thm:SSBMguar}
Let $\alpha \in [0, 2)$ and $\beta \in (\delta, 1 - 3\delta)$ be such that $\beta < \frac{1}{2} + \alpha$. There is a sequence $\{ (N_n, K_n, q_n, \rho_n) \}_{n \in \mathbb{N}}$ of parameters such that:
\begin{enumerate}
\item The parameters are in the regime $q = \Theta(1)$, $\rho = \tilde{\Theta}(N^{-\alpha})$ and $K = \tilde{\Theta}(N^\beta)$ or equivalently,
$$\lim_{n \to \infty} \frac{\log \rho_n^{-1}}{\log N_n} = \alpha, \quad \quad \lim_{n \to \infty} \frac{\log K_n}{\log N_n} = \beta \quad \text{and} \quad \lim_{n \to \infty} q_n = q$$
\item For any sequence of randomized polynomial-time tests $\phi_n : \mG_{N_n} \to \{0, 1\}$, the asymptotic Type I$+$II error of $\phi_n$ on the problems $\textsc{SSBM}_D(N_n, K_n, q_n, \rho_n)$ is at least $1$ assuming the PC conjecture holds with density $p = 1/2$.
\end{enumerate}
Therefore the computational boundary for $\textsc{SSBM}_D(n, k, q, \rho)$ in the parameter regime $q = \Theta(1)$, $\rho = \tilde{\Theta}(n^{-\alpha})$ and $k = \tilde{\Theta}(n^\beta)$ is $\beta^* = \frac{1}{2} + \alpha$.
\end{theorem}

\section{Random Rotations and Sparse PCA}

In this section, we deduce tight lower bounds for detection in sparse PCA when $k \gg \sqrt{n}$. We also show a suboptimal lower bound when $k \ll \sqrt{n}$ matching the results of \cite{berthet2013complexity} and \cite{gao2017sparse}. In both cases, we reduce from biclustering and rank-1 submatrix to biased and unbiased variants of sparse PCA. The reductions in this section are relatively simple, with most of the work behind the lower bounds we show residing in the reductions from PC to biclustering and rank-1 submatrix. This illustrates the usefulness of using natural problems as intermediates in constructing average-case reductions. The lower bounds we prove make use of the following theorem of Diaconis and Freedman showing that the first $m$ coordinates of a unit vector in $n$ dimensions where $m \ll n$ are close to independent Gaussians in total variation \cite{diaconis1987dozen}.

\begin{theorem}[Diaconis and Freedman \cite{diaconis1987dozen}]
Suppose that $(v_1, v_2, \dots, v_n)$ is uniformly distributed according to the Haar measure on $\mathbb{S}^{n-1}$. Then for each $1 \le m \le n - 4$,
$$\TV\left( \mL\left( v_1, v_2, \dots, v_m \right), N(0, n^{-1})^{\otimes m} \right) \le \frac{2(m+3)}{n - m - 3}$$
\end{theorem}

The next lemma is the crucial ingredient in the reductions of this section. Note that the procedure $\textsc{Random-Rotation}$ in Figure \ref{fig:randrot} requires sampling the Haar measure on $\mathcal{O}_{\tau n}$. This can be achieved efficiently by iteratively sampling a Gaussian, projecting it onto the orthogonal complement of the vectors chosen so far, normalizing it and adding it to the current set. Repeating this for $n$ iterations yields an implementation for the sampling part of Step 2 in Figure \ref{fig:randrot}. In the next lemma, we show that $\textsc{Random-Rotation}$ takes an instance $\lambda \cdot uv^\top + N(0, 1)^{\otimes m \times n}$ of rank-1 submatrix to an $m \times n$ matrix with its $n$ columns sampled i.i.d. from $N(0, I_m + \theta vv^\top)$.

\begin{figure}[t!]
\begin{algbox}
\textbf{Algorithm} \textsc{Random-Rotation}

\vspace{2mm}

\textit{Inputs}: Matrix $M \in \mathbb{R}^{m \times n}$, parameter $\tau \in \mathbb{N}$
\begin{enumerate}
\item Construct the $m \times \tau n$ matrix $M'$ such that the leftmost $m \times n$ submatrix of $M'$ is $M$ and the remaining entries of $M'$ are sampled i.i.d. from $N(0, 1)$
\item Let $R$ be the leftmost $\tau n \times n$ submatrix of a random orthogonal matrix sampled from the normalized Haar measure on the orthogonal group $\mathcal{O}_{\tau n}$
\item Output $M'R$
\end{enumerate}
\vspace{1mm}
\end{algbox}
\caption{Random rotation procedure in Lemma \ref{lem:randrot}.}
\label{fig:randrot}
\end{figure}

\begin{lemma}[Random Rotation] \label{lem:randrot}
Let $\tau : \mathbb{N} \to \mathbb{N}$ be an arbitrary function with $\tau(n) \to \infty$ as $n \to \infty$. Consider the map $\phi : \mathbb{R}^{m \times n} \to \mathbb{R}^{m \times n}$ that sends $M$ to $\textsc{Random-Rotation}$ with inputs $M$ and $\tau$. It follows that $\phi(N(0, 1)^{\otimes m \times n}) \sim N(0, 1)^{\otimes m \times n}$ and for any unit vectors $u \in \mathbb{R}^m, v \in \mathbb{R}^n$ we have that
$$\TV\left( \phi\left( \lambda \cdot uv^\top + N(0, 1)^{\otimes m \times n} \right), N\left(0, I_m + \frac{\lambda^2}{\tau n} \cdot uu^\top\right)^{\otimes n} \right) \le \frac{2(n + 3)}{\tau n - n - 3}$$
\end{lemma}

\begin{proof}
Let $R' \in \mathcal{O}_{\tau n}$ be the original $\tau n \times \tau n$ sampled in Step 2 of $\textsc{Random-Rotation}$ and let $R$ be its upper $\tau n \times n$ submatrix. Let $M$ and $M'$ be the matrices input to $\textsc{Random-Rotation}$ and computed in Step 1, respectively, as shown in Figure \ref{fig:randrot}. If $M \sim N(0, 1)^{\otimes m \times n}$, then it follows that $M' \sim N(0, 1)^{\otimes m \times \tau n}$. Since the rows of $M'$ are independent and distributed according to the isotropic distribution $N(0, 1)^{\otimes \tau n}$, multiplication on the right by any orthogonal matrix leaves the distribution of $M'$ invariant. Therefore $M' R' \sim N(0, 1)^{\otimes m \times \tau n}$ and since $M' R$ consists of the first $n$ columns of $M'R'$, it follows that $\phi(M) = M'R \sim N(0, 1)^{\otimes m \times n}$.

Now suppose that $M$ is distributed as $\lambda \cdot uv^\top + N(0, 1)^{\otimes m \times n}$ for some unit vectors $u, v$. Let $v'$ be the unit vector in $\mathbb{R}^{\tau n}$ formed by appending $\tau n - n$ zeros to the end of $v$. It follows that $M'$ is distributed as $\lambda \cdot uv'^\top + N(0, 1)^{\otimes m \times \tau n}$. Let $M' = \lambda \cdot uv'^\top + W$ where $W \sim N(0, 1)^{\otimes m \times \tau n}$. Now let $W' = WS^{-1}$ where $S \in \mathcal{O}_{\tau n}$ is sampled according to the Haar measure and independently of $R'$. Observe that
$$M'R' = \lambda \cdot u \cdot \left( R'^\top v' \right)^\top + W' S R'$$
Now note that conditioned on $R'$, the product $SR'$ is distributed according to the Haar measure on $\mathcal{O}_{\tau n}$. This implies that $SR'$ is independent of $R'$. Therefore $W' S R'$ is independent of $R'$ and distributed according to $N(0, 1)^{\otimes m \times \tau n}$. This also implies that $R'^\top v$ is independent of $W' S R'$ and distributed uniformly over $\mathbb{S}^{\tau n - 1}$. Let $r \in \mathbb{R}^n$ denote the vector consisting of the first $n$ coordinates of $R'^\top v'$ and let $W''$ denote the $m \times n$ matrix consisting of the first $n$ columns of $W' S R'$. It follows that $\phi(M) = M' R = \lambda \cdot u r^\top + W''$ where $r$ and $W''$ are independent. Now let $g \in \mathbb{R}^n$ be a Gaussian vector with entries i.i.d. sampled from $N(0, n^{-1})$. Also let $Z = \lambda \cdot ug^\top + W''$ and note that by Diaconis-Freedman's theorem and coupling the noise terms $W''$, the data processing inequality implies
$$\TV\left( \mL\left( \lambda \cdot ug^\top + W'' \right), \mL\left( \lambda \cdot ur^\top + W'' \right) \right) \le \TV\left( \mL(r), \mL(g) \right) \le \frac{2(n + 3)}{\tau n - n - 3}$$
Now note that since the entries of $g$ are independent, the matrix $\lambda \cdot u g^\top + W''$ has independent columns. Its $i$th row has jointly Gaussian entries with covariance matrix
\begin{align*}
\bE\left[ (\lambda \cdot u g_i + W_i) (\lambda \cdot u g_i + W_i)^\top \right] &= \bE\left[ \lambda^2 \cdot uu^\top g_i^2 + \lambda \cdot g_i  \cdot u W_i^\top + \lambda \cdot g_i \cdot W_i u^\top + W_i W_i^\top \right] \\
&= \frac{\lambda^2}{\tau n} \cdot uu^\top + I_m
\end{align*}
Therefore $\lambda \cdot ug^\top + W'' \sim N\left(0, I_m + \frac{\lambda^2}{\tau n} \cdot uu^\top\right)^{\otimes n}$. Combining these results yields that
$$\TV\left( \phi(M), N\left(0, I_m + \frac{\lambda^2}{\tau n} \cdot uu^\top\right)^{\otimes n} \right) \le \frac{2(n + 3)}{\tau n - n - 3}$$
which completes the proof of the lemma.
\end{proof}

Applying reflection cloning to produce an instance of rank-1 submatrix and then randomly rotating to obtain an instance of sparse PCA yields a reduction from $\textsc{PC}$ to $\textsc{SPCA}$ as given in $\textsc{SPCA-High-Sparsity}$. This establishes tight lower bounds in the regime $k \gg \sqrt{n}$. This reduction is stated in the next lemma, which takes an instance of $\textsc{PC}(n, k, 1/2)$ to an instance of sparse PCA with sparsity $2^\ell k$ and $\theta = \frac{\mu^2 k^2}{2\tau n}$ where $\tau$ and $\mu$ can be taken to be polylogarithmically small in $n$. The proof involves a simple application of the data processing and triangle inequalities and is deferred to Appendix \ref{app8}.

\begin{figure}[t!]
\begin{algbox}
\textbf{Algorithm} \textsc{SPCA-High-Sparsity}

\vspace{2mm}

\textit{Inputs}: Graph $G \in \mG_n$, number of iterations $\ell$, function $\tau : \mathbb{N} \to \mathbb{N}$ with $\tau(n) \to \infty$
\begin{enumerate}
\item Compute the output $M$ of $\textsc{ROS-Reduction}$ applied to $G$ with $\ell$ iterations
\item Output the matrix returned by $\textsc{Random-Rotation}$ applied with inputs $M$ and $\tau$
\end{enumerate}
\vspace{2mm}
\textbf{Algorithm} \textsc{SPCA-Low-Sparsity}

\vspace{2mm}

\textit{Inputs}: Graph $G \in \mG_n$, number of iterations $\ell$, function $\tau : \mathbb{N} \to \mathbb{N}$ with $\tau(n) \to \infty$
\begin{enumerate}
\item Compute the output $M$ of $\textsc{BC-Reduction}$ applied to $G$ with $\ell$ iterations
\item Output the matrix returned by $\textsc{Random-Rotation}$ applied with inputs $M$ and $\tau$
\end{enumerate}
\vspace{2mm}
\textbf{Algorithm} \textsc{SPCA-Recovery}

\vspace{2mm}

\textit{Inputs}: Graph $G \in \mG_n$, density bias $\rho$, function $\tau : \mathbb{N} \to \mathbb{N}$ with $\tau(n) \to \infty$
\begin{enumerate}
\item Let $M$ be the output of $\textsc{BC-Recovery}$ applied to $G$ with density bias $\rho$
\item Output the matrix returned by $\textsc{Random-Rotation}$ applied with inputs $M$ and $\tau$
\end{enumerate}
\vspace{1mm}
\end{algbox}
\caption{Reductions to $\textsc{SPCA}_D$ when $k \gtrsim \sqrt{n}$ and $k \lesssim \sqrt{n}$ in Lemmas \ref{lem:hsspca} and \ref{lem:lsspca} and reduction to $\textsc{SPCA}_{R}$ in Theorem \ref{thm:spcarec}.}
\label{fig:randrot}
\end{figure}

\begin{lemma} \label{lem:hsspca}
Suppose that $n$ and $\ell$ are such that $\ell = O(\log n)$ and are sufficiently large,
$$\mu = \frac{\log 2}{2 \sqrt{6 \log n + 2\log 2}}$$
and $\tau : \mathbb{N} \to \mathbb{N}$ is an arbitrary function with $\tau(n) \to \infty$ as $n \to \infty$. Then $\phi = \textsc{SPCA-High-Sparsity}$ is a randomized polynomial time computable map $\phi : \mG_n \to \mathbb{R}^{n \times n}$ such that if $G$ is an instance of $\text{PC}(n, k, 1/2)$ then under $H_0$, it holds that $\phi(G) \sim N(0, 1)^{\otimes n \times n}$ and under $H_1$, there is a prior $\pi$ such that
$$\TV\left( \mL_{H_1}(\phi(G)), \int N\left(0, I_n + \frac{\mu^2 k^2}{2 \tau n} \cdot uu^\top\right)^{\otimes n} d\pi(u) \right) \le \frac{2(n + 3)}{\tau n - n - 3} + O\left( \frac{1}{\sqrt{\log n}} + k^{-1} \right)$$
where $\pi$ is supported on unit vectors in $\mathcal{V}_{n, 2^\ell k}$.
\end{lemma}

The next lemma gives the guarantees of $\textsc{SPCA-Low-Sparsity}$, which maps from planted clique to an instance of biclustering and then to sparse PCA. This reduction shows hardness for the canonical simple vs. simple hypothesis testing formulation of sparse PCA. In particular, the output in Lemma \ref{lem:randrot} is close in total variation to the simple vs. simple model $\textsc{UBSPCA}$. After multiplying the rows of the matrix output in Lemma \ref{lem:randrot} by $\pm 1$, each with probability $1/2$, this also yields a reduction to $\textsc{USPCA}$. The lemma can be proven with the same applications of the triangle and data processing inequalities as in Lemma \ref{lem:hsspca} using the total variation bound in Lemma \ref{lem:bc} instead of Lemma \ref{lem:ros}. 

Before stating the lemma, we determine the parameters of the sparse PCA instance that $\textsc{SPCA-Low-Sparsity}$ produces. Under $H_1$, $\textsc{BC-Reduction}$ takes an instance of $G(n, k, 1/2)$ approximately in total variation to $2^{-\ell-1/2} \mu \cdot \mathbf{1}_S \mathbf{1}_T^\top + N(0, 1)^{\otimes 2^\ell n \times 2^\ell n}$ where $S, T \subseteq [2^\ell n]$ have size $2^\ell k$ and $\mu$ is subpolynomial in $n$. This matrix can be rewritten as $2^{-1/2} \mu k \cdot uv^\top + N(0, 1)^{\otimes 2^\ell n \times 2^\ell n}$ where $u, v$ are $2^\ell k$-sparse unit vectors. Now $\textsc{Random-Rotation}$ takes this matrix to an instance of $\textsc{UBSPCA}_D$ with the resulting parameters $d = n' = 2^\ell n$, $k' = 2^\ell k$ and $\theta = \frac{\mu^2 k^2}{2^{\ell + 1} \tau n}$ where $\tau, \mu$ are subpolynomial in $n$.

\begin{lemma} \label{lem:lsspca}
Suppose that $n$ and $\ell$ are such that $\ell = O(\log n)$ and are sufficiently large,
$$\mu = \frac{\log 2}{2 \sqrt{6 \log n + 2\log 2}}$$
and $\tau : \mathbb{N} \to \mathbb{N}$ is an arbitrary function with $\tau(n) \to \infty$ as $n \to \infty$. Then $\phi = \textsc{SPCA-Low-Sparsity}$ is a randomized polynomial time computable map $\phi : \mG_n \to \mathbb{R}^{2^\ell n \times 2^\ell n}$ such that if $G$ is an instance of $\text{PC}(n, k, 1/2)$ then under $H_0$, it holds that $\phi(G) \sim N(0, 1)^{\otimes 2^\ell n \times 2^\ell n}$ and
$$\TV\left( \mL_{H_1}(\phi(G)), \int N\left(0, I_n + \frac{\mu^2 k^2}{2^{\ell + 1} \tau n} \cdot uu^\top\right)^{\otimes n} d\pi(u) \right) \le \frac{2(2^\ell n + 3)}{\tau \cdot 2^\ell n - 2^\ell n - 3} + O\left( \frac{1}{\sqrt{\log n}} \right)$$
where $\pi$ is the uniform distribution over all $2^\ell k$-sparse unit vectors in $\mathbb{R}^{2^\ell n}$ with nonzero entries equal to $1/\sqrt{2^\ell k}$.
\end{lemma}

We now apply these reductions to deduce planted clique hardness for sparse PCA and its variants. The proofs of the next two theorems are deferred to Appendix \ref{app8}. Note that when $k \ll \sqrt{n}$, the lower bounds for sparse PCA are not tight. For biased sparse PCA, the lower bounds are only tight at the single point when $\theta = \tilde{\Theta}(1)$. The next theorem deduces tight hardness for $\textsc{SPCA}_D$ when $k \gtrsim \sqrt{n}$ with Lemma \ref{lem:hsspca}.

\begin{theorem} \label{thm:spca}
Let $\alpha > 0$ and $\beta \in (0, 1)$ be such that $\alpha > \max(1 - 2\beta, 0)$. There is a sequence $\{ (N_n, K_n, d_n, \theta_n) \}_{n \in \mathbb{N}}$ of parameters such that:
\begin{enumerate}
\item The parameters are in the regime $d = \Theta(N)$, $\theta = \tilde{\Theta}(N^{-\alpha})$ and $K = \tilde{\Theta}(N^\beta)$ or equivalently,
$$\lim_{n \to \infty} \frac{\log \theta_n^{-1}}{\log N_n} = \alpha \quad \text{and} \quad \lim_{n \to \infty} \frac{\log K_n}{\log N_n} = \beta$$
\item For any sequence of randomized polynomial-time tests $\phi_n : \mG_{N_n} \to \{0, 1\}$, the asymptotic Type I$+$II error of $\phi_n$ on the problems $\textsc{SPCA}_D(N_n, K_n, d_n, \theta_n)$ is at least $1$ assuming the PC conjecture holds for $p = 1/2$.
\end{enumerate}
\end{theorem}

Similarly, varying the parameters $\ell$ and $k$ in Lemma \ref{lem:lsspca} yields the following hardness for the simple vs. simple hypothesis testing formulations of biased and ordinary sparse PCA. Since $\textsc{UBSPCA}_D$ and $\textsc{USPCA}_D$ are instances of $\textsc{BSPCA}_D$ and $\textsc{SPCA}_D$, respectively, the next theorem also implies the lower bounds when $\alpha > 1 - 2\beta$ in Theorem \ref{thm:spca} when $k \lesssim \sqrt{n}$.

\begin{theorem} \label{thm:uspca}
Let $\alpha > 0$ and $\beta \in (0, 1)$ be such that $\frac{1 - \alpha}{2} < \beta < \frac{1 + \alpha}{2}$. There is a sequence $\{ (N_n, K_n, d_n, \theta_n) \}_{n \in \mathbb{N}}$ of parameters such that:
\begin{enumerate}
\item The parameters are in the regime $d = \Theta(N)$, $\theta = \tilde{\Theta}(N^{-\alpha})$ and $K = \tilde{\Theta}(N^\beta)$ or equivalently,
$$\lim_{n \to \infty} \frac{\log \theta_n^{-1}}{\log N_n} = \alpha \quad \text{and} \quad \lim_{n \to \infty} \frac{\log K_n}{\log N_n} = \beta$$
\item For any sequence of randomized polynomial-time tests $\phi_n : \mG_{N_n} \to \{0, 1\}$, the asymptotic Type I$+$II error of $\phi_n$ on the problems $\textsc{USPCA}_D(N_n, K_n, d_n, \theta_n)$ and $\textsc{UBSPCA}_D(N_n, K_n, d_n, \theta_n)$ is at least $1$ assuming the PC conjecture holds for $p = 1/2$.
\end{enumerate}
\end{theorem}

To conclude this section, we observe that the reduction $\textsc{SPCA-Recovery}$ shows that recovery in $\textsc{UBSPCA}_R$ is hard if $\theta \ll 1$ given the PDS recovery conjecture. The proof of the following theorem follows the same structure as Lemma \ref{lem:hsspca} and Theorem \ref{thm:bcrec}. Note that $\textsc{BC-Recovery}$ approximately maps from $\textsc{PDS}_R(n, k, 1/2 + \rho, 1/2)$ to $\textsc{BC}_R(n, k, \mu)$ where $\mu = \frac{\log (1 + 2\rho)}{2 \sqrt{6 \log n + 2\log 2}} = \tilde{\Theta}(\rho)$. This map preserves the support of the planted dense subgraph in the row support of the planted matrix in $\textsc{BC}_R$. Then $\textsc{Random-Rotation}$ approximately maps from this $\textsc{BC}_R$ instance to a $\textsc{UBSPCA}_R(n, k, n, \theta)$ instance with $\theta = \frac{k^2 \mu^2}{\tau n} = \tilde{\Theta}\left( \frac{k^2 \rho^2}{n} \right)$. This map ensures that the planted vector $u$ is supported on the same indices as the original $\textsc{PDS}_R$ instance. Furthermore, the PDS conjecture is that the original $\textsc{PDS}_R$ instance is hard if $\rho^2 \ll \frac{n}{k^2}$ which corresponds to the barrier $\theta \ll 1$ under this reduction.

\begin{theorem} \label{thm:spcarec}
Let $\alpha \in \mathbb{R}$ and $\beta \in (0, 1)$. There is a sequence $\{ (N_n, K_n, D_n, \theta_n) \}_{n \in \mathbb{N}}$ of parameters such that:
\begin{enumerate}
\item The parameters are in the regime $d = \Theta(N)$, $\theta = \tilde{\Theta}(N^{-\alpha})$ and $K = \tilde{\Theta}(N^\beta)$ or equivalently,
$$\lim_{n \to \infty} \frac{\log \theta_n^{-1}}{\log N_n} = \alpha \quad \text{and} \quad \lim_{n \to \infty} \frac{\log K_n}{\log N_n} = \beta$$
\item If $\alpha > 0$ and $\beta > \frac{1}{2}$, then the following holds. Let $\epsilon > 0$ be fixed and let $X_n$ be an instance of $\textsc{UBSPCA}_R(N_n, K_n, D_n, \theta_n)$. There is no sequence of randomized polynomial-time computable functions $\phi_n : \mathbb{R}^{D_n \times N_n} \to \binom{[N_n]}{k}^2$ such that for all sufficiently large $n$ the probability that $\phi_n(X_n)$ is exactly the pair of latent row and column supports of $X_n$ is at least $\epsilon$, assuming the PDS recovery conjecture.
\end{enumerate}
Therefore, given the PDS recovery conjecture, the computational boundary for $\textsc{UBSPCA}_R(n, k, d, \theta)$ in the parameter regime $\theta = \tilde{\Theta}(n^{-\alpha})$ and $k = \tilde{\Theta}(n^\beta)$ is $\alpha^* = 0$ when $\beta > \frac{1}{2}$.
\end{theorem}

\begin{proof}
Suppose that $\alpha > 0$ and $\beta \ge \frac{1}{2}$. Let $\gamma =  \beta - \frac{1 - \alpha}{2} > 0$ and define 
$$K_n = k_n = \lceil n^{\beta} \rceil, \quad \quad \rho_n = n^{-\gamma}, \quad \quad N_n = D_n = n, \quad \quad \mu_n = \frac{\log (1 + 2\rho_n)}{2 \sqrt{6 \log n + 2\log 2}}, \quad \quad \theta_n = \frac{k_n^2 \mu_n^2}{\tau n}$$
where $\tau$ is an arbitrarily slowly growing function of $n$. Let $\varphi_n = \textsc{SPCA-Recovery}$ be the reduction in Figure \ref{fig:randrot}. Let $G_n \sim G(n, S, 1/2 + \rho_n, 1/2)$ and $X_n = \varphi_n(G_n)$ where $S$ is a $k_n$-subset of $[n]$. Let $u_S$ denote the unit vector supported on indices in $S$ with nonzero entries equal to $1/\sqrt{k_n}$. Lemma \ref{lem:bcrec} and Lemma \ref{lem:randrot} together imply that
$$\TV\left( \mL(X_n), N\left( 0, I_n + \theta_n u_S u_S^\top \right) \right) \le O\left( \frac{1}{\sqrt{\log n}} \right) + \frac{2(n+3)}{\tau n - n - 3} \to 0 \text{ as } n \to \infty$$
Let $\mL_{n, S} = N\left( 0, I_n + \theta_n u_S u_S^\top \right)$. Assume for contradiction that there is a sequence of randomized polynomial-time computable functions $\phi_n$ as described above. Now observe that
$$\left| \bP_{X \sim \mL(X_n)} \left[ \phi_n(X) = S \right] - \bP_{X \sim \mL_{n, S}}\left[ \phi_n(X) = S\right]\right| \le \TV\left( \mL(X_n), \mL_{n, S} \right) \to 0 \text{ as } n \to \infty$$
Since $\bP_{X \sim \mL_{n, S}}\left[\phi_n(X) = S \right] \ge \epsilon$ for sufficiently large $n$, it follows that $\bP_{X \sim \mL(X_n)} \left[ \phi_n \circ \varphi_n(G_n) = S \right] = \bP_{X \sim \mL(X_n)} \left[ \phi_n(X) = S \right] \ge \epsilon/2$ for sufficiently large $n$. Furthermore observe
$$\lim_{n \to \infty} \frac{\log k_n}{\log n} = \beta \quad \text{and} \quad \lim_{n \to \infty} \log_n \left( \frac{k_n^2 \rho_n^2}{\frac{1}{4} - \rho_n^2} \right) = 2\beta - 2\gamma = 1 - \alpha < 1$$
Since the sequence of functions $\phi_n \circ \varphi_n$ can be computed in randomized polynomial time, this contradicts the PDS recovery conjecture. Therefore no such sequence of functions $\phi_n$ exists for the parameter sequence $\{ (N_n, K_n, D_n, \theta) \}_{n \in \mathbb{N}}$ defined above. As in Theorem \ref{thm:bcrec}, $\mu_n \sim \frac{\rho_n}{\sqrt{6 \log n}}$ as $n \to \infty$.
Therefore it follows that
$$\lim_{n \to \infty} \frac{\log \theta_n^{-1}}{\log N_n} = \lim_{n \to \infty} \frac{2\gamma \log n + \log (6\log n) + \log n - 2\beta \log n}{\log n} = \alpha \quad \text{and} \quad \lim_{n \to \infty} \frac{\log K_n}{\log N_n} = \beta$$
which completes the proof of the theorem.
\end{proof}

As is the case for $\textsc{BC-Recovery}$, the reduction $\textsc{SPCA-Recovery}$ also shows hardness for partial and weak recovery if the PDS recovery conjecture is strengthened to assume hardness of partial and weak recovery, respectively, for $\textsc{PDS}_R$.

\section{Algorithms and Information-Theoretic Thresholds}

In this section, we give the algorithms and information-theoretic lower bounds necessary to prove Theorem \ref{lem:2a}. Specifically, for each problem, we give an information-theoretic lower bound, an inefficient algorithm that achieves the information-theoretic lower bound and a polynomial-time algorithm. As the computational lower bounds and reductions previously presented are the main novel contribution of the paper, the details in this section are succinctly presented only as needed for Theorem \ref{lem:2a}. 

Many of the problems we consider have pre-existing algorithms and information-theoretic lower bounds. In these cases, we cite the relevant literature and state the results needed for Theorem \ref{lem:2a}. Note that we only require algorithms and lower bounds optimal up to sub-polynomial factors for Theorem \ref{lem:2a}. For some problems, we only give an information-theoretic lower bound for detection and show that this implies the recovery lower bound in the next section.


\subsection{Biclustering, Planted Dense Subgraph and Independent Set}

\paragraph{Information-Theoretic Lower Bounds.} The information-theoretic lower bound for $\textsc{BC}_D$ was shown in \cite{butucea2013detection}. More precisely, they showed the following theorem re-written in our notation. Note that they showed a lower bound for a composite hypothesis testing version of $\textsc{BC}_D$, but took a uniform prior over the support of the hidden submatrix, matching our formulation.

\begin{theorem}[Theorem 2.2 in \cite{butucea2013detection}]
Suppose that $k, \mu$ are such that as $n \to \infty$, it holds that $k/n \to 0$ and one of the following holds
$$\frac{\mu k^2}{n} \to 0 \quad \text{and} \quad \limsup_{n \to \infty} \frac{\mu}{2\sqrt{k^{-1} \log(n/k)}} < 1$$
Then if $M_n$ denotes an instance of $\textsc{BC}_D(n, k, \mu)$,
$$\TV\left( \mL_{H_0}(M_n), \mL_{H_1}(M_n) \right) \to 0 \quad \text{as} \quad n \to \infty$$
\end{theorem}

This corresponds exactly to the information-theoretic barrier of $\mu \ll \frac{1}{\sqrt{k}}$ and $\mu \ll \frac{n}{k^2}$. We remark that the information-theoretic lower bounds for biclustering can also be deduced from the information-theoretic lower bounds for planted dense subgraph with $q = 1/2$ using the reduction from Lemma \ref{lem:bcrec} and the data-processing inequality. Tight information-theoretic lower bounds for $\textsc{PDS}_D$ and $\textsc{PIS}_D$ can be deduced from a mild adaptation of \cite{hajek2015computational}. The argument from the proof of Proposition 3 in \cite{hajek2015computational} yields the following lemma.

\begin{lemma}[Proposition 3 in \cite{hajek2015computational}]
If $G$ is an instance of $\textsc{PDS}_D(n, k, p, q)$, then
$$\TV\left( \mL_{H_0}(G), \mL_{H_1}(G) \right) \le \frac{1}{2} \sqrt{\bE \left[ \exp\left( \frac{(p - q)^2}{q(1 - q)} \cdot H^2 \right) - 1 \right]}$$
where $H \sim \text{Hypergeometric}(n, k, k)$.
\end{lemma}

This lemma can be derived by combining the $\chi^2$ computation in the proof of Proposition 3 with Cauchy-Schwarz. In \cite{hajek2015computational}, Proposition 3 is specifically for the case when $p = cq$ where $c > 1$ and also for a mixture over $\textsc{PDS}_D(n, K, p, q)$ where $K \sim \text{Bin}(n, k/n)$. However, the first step in the proof of Proposition 3 is to condition on $K$ and prove this bound for each fixed $K$. When combined with Lemma 14 from \cite{hajek2015computational}, we obtain the desired information-theoretic lower bounds.

\begin{lemma}[Lemma 14 in \cite{hajek2015computational}] \label{lem:hgm}
There is an increasing function $\tau : \mathbb{R}^+ \to \mathbb{R}^+$ with $\lim_{x \to 0^+} \tau(x) = 1$ and
$$\bE[\exp(\lambda H^2)] \le \tau(b)$$
where $H \sim \text{Hypergeometric}(n, k, k)$, $\lambda = b \cdot \max \left\{ \frac{1}{k} \log \left( \frac{en}{k} \right), \frac{n^2}{k^4} \right\}$ and $0 < b < (16e)^{-1}$.
\end{lemma}

Combining these two lemmas and setting $b$ as
$$b = \frac{(p - q)^2}{q(1 - q)} \cdot \left( \max \left\{ \frac{1}{k} \log \left( \frac{en}{k} \right), \frac{n^2}{k^4} \right\} \right)^{-1}$$
yields the following theorem on the information-theoretic lower bound for the general regime of $\textsc{PDS}_D$.

\begin{theorem}
Suppose $p, q, k$ are such that as $n \to \infty$, it holds that $\frac{(p - q)^2}{q(1 - q)} \ll \frac{1}{k}$ and $\frac{(p - q)^2}{q(1 - q)} \ll \frac{n^2}{k^4}$. Then if $G_n$ is an instance of $\textsc{PDS}_D(n, k, p, q)$, it follows that
$$\TV\left( \mL_{H_0}(G_n), \mL_{H_1}(G_n) \right) \to 0 \quad \text{as} \quad n \to \infty$$
\end{theorem}

Note that when $p = cq$ for some constant $c > 1$ or when $p = 0$, this barrier is $q \ll \frac{1}{k}$ and $q \ll \frac{n^2}{k^4}$. This recovers the information-theoretic lower bounds for $\textsc{PIS}_D$ and $\textsc{PDS}_D$ when $p = cq$. The information-theoretic lower bounds for the weak and strong recovery variants $\textsc{BC}_R$ and $\textsc{PDS}_R$ are derived in \cite{hajek2016information}. The following theorems of \cite{hajek2016information} characterize these lower bounds.

\begin{theorem}[Corollary 2 in \cite{hajek2016information}] \label{thm:infbcrec}
Suppose $k$ and $\mu$ are such that as $n \to \infty$,
\begin{equation} \label{eqn:gaussweak}
k\mu^2 \to \infty \quad \text{and} \quad \liminf_{n \to \infty} \frac{(k-1)\mu^2}{\log \frac{n}{k}} > 4
\end{equation}
then weak recovery in $\textsc{BC}_R(n, k, \mu)$ is possible. If weak recovery is possible, then (\ref{eqn:gaussweak}) holds as a non-strict inequality.
\end{theorem}

\begin{theorem}[Corollary 4 in \cite{hajek2016information}] \label{thm:bicexact}
Suppose $k$ and $\mu$ are such that as $n \to \infty$, condition (\ref{eqn:gaussweak}) holds and
\begin{equation} \label{eqn:gaussexact}
\liminf_{n \to \infty} \frac{k\mu^2}{\left( \sqrt{2 \log n} + \sqrt{2 \log k} \right)^2} > 1
\end{equation}
then exact recovery in $\textsc{BC}_R(n, k, \mu)$ is possible. If exact recovery is possible, then (\ref{eqn:gaussweak}) and (\ref{eqn:gaussexact}) hold as a non-strict inequalities.
\end{theorem}

\begin{theorem}[Corollary 1 in \cite{hajek2016information}]
Suppose $p$ and $q$ are such that the ratios $\log \frac{p}{q}$ and $\log \frac{1 - p}{1 - q}$ are bounded as $n \to \infty$. If $k$ satisfies
\begin{equation} \label{eqn:bernweak}
k \cdot \KL(p, q) \to \infty \quad \text{and} \quad \liminf_{n \to \infty} \frac{k \cdot \KL(p, q)}{\log \frac{n}{k}} > 2
\end{equation}
then there is an algorithm achieving weak recovery for $\textsc{PDS}_R(n, k, p, q)$ is possible. If weak recovery is possible, then (\ref{eqn:bernweak}) holds as a non-strict inequality.
\end{theorem}

\begin{theorem}[Corollary 3 in \cite{hajek2016information}] \label{thm:pdsexact}
Suppose $p$ and $q$ are such that the ratios $\log \frac{p}{q}$ and $\log \frac{1 - p}{1 - q}$ are bounded as $n \to \infty$. If $k$ satisfies
$$\tau = \frac{\log \frac{1 - q}{1 - p} + \frac{1}{k} \log \frac{n}{k}}{\log \frac{p(1 - q)}{q(1 - p)}}$$
If (\ref{eqn:bernweak}) holds and
\begin{equation} \label{eqn:bernexact}
\liminf_{n \to \infty} \frac{k \cdot \KL(\tau, q)}{\log n} > 1
\end{equation}
then exact recovery in $\textsc{PDS}_R(n, k, p, q)$ is possible. If exact recovery is possible, then (\ref{eqn:bernweak}) and (\ref{eqn:bernexact}) hold as non-strict inequalities.
\end{theorem}

These theorems show that both weak and strong recovery for $\textsc{BC}_R(n, k, \mu)$ are information-theoretically impossible when $\mu \lesssim \frac{1}{\sqrt{k}}$ by the first condition in (\ref{eqn:gaussweak}). Now note that if $p - q = O(q)$ and $q \to 0$ as $n \to \infty$, then
\begin{align*}
\KL(p, q) &= p \log \left( \frac{p}{q} \right) + (1 - p) \cdot \log \left( \frac{1 - p}{1 - q} \right) \\
&= p \cdot \left( \frac{p - q}{q} \right) - O\left( p \cdot \left( \frac{p - q}{q} \right)^2 \right) - (1 - p) \cdot \left( \frac{p - q}{1 - q} \right) - O\left( (1 - p) \cdot \left( \frac{p - q}{1 - q} \right)^2 \right) \\
&= \frac{(p - q)^2}{q(1 - q)} + O\left( p \cdot \left( \frac{p - q}{q} \right)^2 + (p - q)^2 \right) = O\left( \frac{(p - q)^2}{q(1 - q)} \right)
\end{align*}
as $n \to \infty$. Therefore it follows that if $p - q = O(q)$, $q \to 0$ as $n \to \infty$ and $\log \frac{p}{q}$ and $\log \frac{1 - p}{1 - q}$ are bounded as $n \to \infty$ then both weak and strong recovery in $\textsc{PDS}_R(n, k, p, q)$ are information-theoretically impossible if $\frac{(p - q)^2}{q(1 - q)} \lesssim \frac{1}{k}$ by the first condition in (\ref{eqn:bernweak}). Note that these theorems of \cite{hajek2016information} do not imply the necessary information-theoretic lower bound for $\textsc{PIS}_R$ since $p = 0$ violates the condition that $\log \frac{p}{q}$ is bounded. However, the genie argument in the necessary part of Theorem 1 in \cite{hajek2016information} can be mildly adapted to obtain the following theorem, the proof of which is deferred to Appendix \ref{app9}.

\begin{theorem} \label{thm:pisrecit}
If $k \ge 2$, $q \ll \frac{1}{k}$ and $n - k = \Omega(n)$ as $n \to \infty$, then weak recovery in $\textsc{PIS}_R(n, k, q)$ is impossible.
\end{theorem}

\paragraph{Information-Theoretically Optimal Algorithms.} A corresponding algorithm achieving the information-theoretic lower bound for $\textsc{BC}_D$ was also shown in \cite{butucea2013detection}. Their algorithm outputs the hypothesis $H_1$ if either the maximum sum over all $k \times k$ submatrices of the input exceeds a threshold or if the total sum of the input exceeds another threshold. The guarantees of this algorithm are summarized in the following theorem.

\begin{theorem}[Theorem 2.1 in \cite{butucea2013detection}]
Suppose that $k, \mu$ are such that as $n \to \infty$, it holds that $k/n \to 0$ and one of the following holds
$$\frac{\mu k^2}{n} \to \infty \quad \text{or} \quad \limsup_{n \to \infty} \frac{\mu}{2\sqrt{k^{-1} \log(n/k)}} > 1$$
Then there is an algorithm solving $\textsc{BC}_D(n, k, \mu)$ with Type I$+$II error tending to zero as $n \to \infty$.
\end{theorem}

A very similar algorithm is optimal for $\textsc{PDS}_D$. Generalizing the concentration bounds in the proof of Proposition 4 in \cite{hajek2015computational} to any $p, q$ with $p - q = O(q)$ and $q \to 0$ yields the following theorem, the proof of which is deferred to Appendix \ref{app9}.

\begin{theorem} \label{thm:pdsdet}
Suppose that $p$, $q$ and $k$ are such that $|p - q| = O(q)$, $q \to 0$ and
$$\frac{(p - q)^2}{q(1 - q)} = \omega\left( \frac{n^2}{k^4} \right) \quad \text{or} \quad \frac{(p - q)^2}{q(1 - q)} = \omega\left(\frac{\log(n/k)}{k}\right)$$
as $n \to \infty$. Then there is an algorithm solving $\textsc{PDS}_D(n, k, p, q)$ with Type I$+$II error tending to zero as $n \to \infty$.
\end{theorem}

This theorem gives the necessary algorithm matching the information-theoretic lower bound for $\textsc{PDS}_D$ in the general regime $p - q = O(q)$, including $p = cq$ for some constant $c > 1$. The algorithm needed for $\textsc{PIS}_D$ can be obtained by setting $p = 0$ in this theorem. An algorithm matching the information-theoretic lower bound for $\textsc{BC}_R$ follows from Theorem \ref{thm:bicexact}, which asserts that exact recovery is possible as long as
$$\mu > (1 + \epsilon) \cdot \frac{\sqrt{2\log n} + \sqrt{2 \log k}}{\sqrt{k}}$$
for some fixed $\epsilon > 0$. Specializing Corollary 2.4 in \cite{chen2016statistical} to the case of $r = 1$ clusters yields an analogous algorithm for $\textsc{PDS}_R$.

\begin{theorem}[Corollary 2.4 in \cite{chen2016statistical}]
Suppose that $p, q$ and $k$ are such that $p > q$ and
$$\frac{(p - q)^2}{q(1 - q)} \ge \frac{C\log n}{k}, \quad q \ge \frac{C\log k}{k} \quad \text{and} \quad kq \log \frac{p}{q} \ge C \log n$$
for some sufficiently large constant $C > 0$. Then the maximum likelihood estimator for the planted dense subgraph in $\textsc{PDS}_R(n, k, p, q)$ solves strong recovery with error probability tending to zero.
\end{theorem}

This implies that if $p > q$, $p - q = O(q)$, $q \to 0$ and $\frac{(p - q)^2}{q(1 - q)} \gtrsim \frac{1}{k}$ as $n \to \infty$, then exact recovery is possible. Specializing the result to $p = 1$ and applying this algorithm to the complement graph of a $\textsc{PIS}_R$ instance yields that there is an algorithm for $\textsc{PIS}_R$ if $q \gtrsim \frac{1}{k}$. We remark that the necessary algorithm for $\textsc{PDS}_R$ can also be deduced from Theorem \ref{thm:pdsexact}. However, the constraints that $\log \frac{p}{q}$ and $\log \frac{1 - p}{1 - q}$ must be bounded does not yield the desired algorithm for $\textsc{PIS}_R$.

\paragraph{Polynomial-Time Algorithms.} The polynomial time algorithm matching our planted clique lower bound for $\textsc{BC}_D$ is another simple algorithm thresholding the maximum and sum of the input matrix. Given an instance $M$ of $\textsc{BC}_D(n, k, \mu)$, let $\max(M) = \max_{i, j \in [n]} M_{ij}$ and $\text{sum}(M) =  \sum_{i, j = 1}^n M_{ij}$. Specializing Lemma 1 of \cite{ma2015computational} to our setup yields the following lemma.

\begin{lemma}[Lemma 1 in \cite{ma2015computational}]
If $M$ is an instance of $\textsc{BC}_D(n, k, \mu)$ then
$$\bP_{H_0} \left[ \textnormal{sum}(M) > \frac{\mu k^2}{2} \right] + \bP_{H_1} \left[ \textnormal{sum}(M) \le \frac{\mu k^2}{2} \right] \le \exp \left( - \frac{\mu^2 k^4}{8n^2} \right)$$
If $c > 0$ is any absolute constant and $\tau = \sqrt{(4 + c) \log n}$, then
$$\bP_{H_0} \left[ \max(M) > \tau \right] + \bP_{H_1} \left[ \max(M) \le \tau \right] \le n^{-c/2} + \exp \left( - \frac{1}{2} \left| \mu - \tau \right|_+ \right)$$
\end{lemma}

It follows that the algorithm that outputs $H_1$ if $\max(M) > \sqrt{5 \log n}$ or $\text{sum}(M) > \frac{\mu k^2}{2}$ solves $\textsc{BC}_D$ with Type I$+$II error tending to zero as $n \to \infty$ if either $\mu \ge \sqrt{6 \log n}$ or $\mu = \omega\left(\frac{n}{k^2}\right)$. By Theorem \ref{thm:pdsdet}, if $\frac{(p - q)^2}{q(1 - q)} = \omega\left( \frac{n^2}{k^4} \right)$ it follows that thresholding the number of edges of an instance $G$ of $\textsc{PDS}_D(n, k, p, q)$ has Type I$+$II error tending to zero as $n \to \infty$. Setting $p = 0$ recovers the computational barrier of $\textsc{PIS}_D$. Polynomial-time algorithms for the recovery variants of these problems were given in \cite{chen2016statistical}. The following are three theorems of \cite{chen2016statistical} written in our notation.

\begin{theorem}[Theorem 2.5 in \cite{chen2016statistical}]
A polynomial-time convex relaxation of the MLE solves exact recovery in $\textsc{PDS}_R(n, k, p, q)$ with error probability at most $n^{-10}$ if $p > q$ and
$$k^2(p - q) \ge C\left[ p(1 - q) k \log n + q(1 - q) n \right]$$
where $C > 0$ is a fixed constant.
\end{theorem}

\begin{theorem}[Theorem 3.3 in \cite{chen2016statistical}]
A polynomial-time convex relaxation of the MLE solves exact recovery in $\textsc{BC}_R(n, k, \mu)$ with error probability at most $n^{-10}$ if
$$\mu^2 \ge C\left[ \frac{\log n}{k} + \frac{n}{k^2} \right]$$
where $C > 0$ is a fixed constant.
\end{theorem}

\begin{theorem}[Theorem 3.3 in \cite{chen2016statistical}]
A polynomial-time element-wise thresholding algorithm solves exact recovery in $\textsc{BC}_R(n, k, \mu)$ with error probability at most $n^{-3}$ if $\mu^2 \ge C \log n$ where $C > 0$ is a fixed constant.
\end{theorem}

Since $k \log n = \tilde{O}(n)$ and $p = O(q)$ if $p - q = O(q)$, the first theorem above implies that exact recovery is possible in polynomial time for the general regime of $\textsc{PDS}_R$ if $\frac{(p - q)^2}{q(1 - q)} \gg \frac{n}{k^2}$. Taking the complement graph of the input and setting $p = 1$ and $q = 1 - \tilde{\Theta}(n^{-\alpha})$ in the first theorem yields that $\text{PIS}_R(n, k, 1 - q)$ can be solved in polynomial time if $1 - q \gg \frac{n}{k^2}$. The second and third theorems above imply that exact recovery for $\textsc{BC}_R$ is possible in polynomial time if $\mu \gg \frac{1}{\sqrt{k}}$ or $\mu \gg 1$. These polynomial-time algorithms for detection and recovery match the computational lower bounds shown in previous sections.

\subsection{Rank-1 Submatrix, Sparse Spiked Wigner and Subgraph SBM}

\paragraph{Information-Theoretic Lower Bounds.} Applying a similar $\chi^2$ computation as in information-theoretic lower bounds for sparse PCA and planted dense subgraph, we can reduce showing an information-theoretic lower bound for $\textsc{SROS}_D$ to bounding an MGF. In the case of $\textsc{SROS}_D$, this MGF turns out to be that of the square of a symmetric random walk on $\mathbb{Z}$ terminated after a hypergeometric number of steps. An asymptotically tight upper bound on this MGF was obtained in \cite{cai2015optimal} through the following lemma.

\begin{lemma} [Lemma 1 in \cite{cai2015optimal}] \label{lem:hgmw}
Suppose that $d \in \mathbb{N}$ and $k \in [p]$. Let $B_1, B_2, \dots, B_k$ be independent Rademacher random variables. Let the symmetric random walk on $\mathbb{Z}$ stopped at the $m$th step be
$$G_m = \sum_{i = 1}^m B_i$$
If $H \sim \text{Hypergeometric}(d, k, k)$ then there is an increasing function $g : (0, 1/36) \to (1, \infty)$ such that $\lim_{x \to 0^+} g(x) = 1$ and for any $a \in (0, 1/36)$, it holds that
$$\bE\left[ \exp\left( G_H^2 \cdot \frac{a}{k} \log \frac{ed}{k} \right) \right] \le g(a)$$
\end{lemma}

With this bound, we obtain the following information-theoretic lower bound for $\textsc{SROS}_D$, which matches Theorem \ref{lem:2a}.

\begin{theorem} \label{thm:sswinf}
Suppose that $M$ is an instance of $\textsc{SROS}_D(n, k, \mu)$ where under $H_1$, the planted vector $v$ is chosen uniformly at random from all $k$-sparse unit vectors in $\mathbb{R}^n$ with nonzero coordinates equal to $\pm \frac{1}{\sqrt{k}}$. Suppose it holds that $\mu \le \sqrt{\beta_0 k \log \frac{en}{k}}$ for some $0 < \beta_0 < (16e)^{-1}$. Then there is a function $w : (0, 1) \to (0, 1)$ satisfying that $\lim_{\beta_0 \to 0^+} w(\beta_0) = 0$ and
$$\TV\left( \mL_{H_0}(M), \mL_{H_1}(M) \right) \le w(\beta_0)$$
\end{theorem}

\begin{proof}
Let $\bP_0$ denote $\mL_{H_0}(M) = N(0, 1)^{\otimes n \times n}$ and $\bP_u$ denote $\mL\left( \mu \cdot uu^\top + N(0, 1)^{\otimes n \times n} \right)$ where $u$ is in the set $S$ of $k$-sparse unit vectors $u$ with nonzero entries equal to $\pm 1/\sqrt{k}$. Now let $\mP_1$ denote $\mL_{H_1}(M)$ which can also be written as
$$\bP_1 = \frac{1}{|S|} \sum_{u \in S} \bP_u$$
Given two matrices $A, B \in \mathbb{R}^{n \times n}$, let $\langle A, B \rangle = \sum_{i, j = 1}^n A_{ij} B_{ij}$ denote their inner product. Now note that for any $X \in \mathbb{R}^{n \times n}$,
\begin{align*}
\frac{d\bP_u}{d\bP_0}(X) &= \exp\left( - \frac{1}{2} \sum_{i, j = 1}^n (X_{ij} - \mu \cdot u_i u_j)^2 + \frac{1}{2} \sum_{i, j = 1}^n X_{ij}^2 \right) \\
&= \exp\left( \mu \cdot \langle X, uu^\top \rangle - \frac{\mu^2}{2} \| u \|_2^4 \right) = \exp\left( \mu \cdot \langle X, uu^\top \rangle - \frac{\mu^2}{2} \right)
\end{align*}
since $\| u \|_2 = 1$. Now observe that
\begin{align*}
\chi^2(\bP_1, \bP_0) &= \bE_{X \sim \bP_0} \left[ \left( \frac{d\bP_1}{d\bP_0}(X) - 1 \right)^2 \right] = -1 + \frac{1}{|S|^2} \sum_{u, v \in S} \bE_{X \sim \bP_0} \left[ \frac{d\bP_u}{d\bP_0}(X) \cdot \frac{d\bP_v}{d\bP_0}(X) \right] \\
&= -1 + \frac{1}{|S|^2} \sum_{u, v \in S} \bE_{X \sim \bP_0} \left[ \exp\left( \mu \cdot \langle X, uu^\top + vv^\top \rangle - \mu^2 \right) \right] \\
&= -1 + \frac{1}{|S|^2} \sum_{u, v \in S} \exp\left( \frac{\mu^2}{2} \left\| uu^\top + vv^\top \right\|_F^2 - \mu^2 \right) \\
&= -1 + \frac{1}{|S|^2} \sum_{u, v \in S} \exp\left( \frac{\mu^2}{2} \langle uu^\top, uu^\top \rangle + \mu^2 \langle uu^\top, vv^\top \rangle + \frac{\mu^2}{2} \langle vv^\top, vv^\top \rangle - \mu^2 \right) \\
&= -1 + \frac{1}{|S|^2} \sum_{u, v \in S} \exp\left( \mu^2 \langle u, v \rangle^2 \right) = -1 + \bE_{u, v \sim \text{Unif}[S]}\left[ \exp\left( \mu^2 \langle u, v \rangle^2 \right) \right]
\end{align*}
where the third inequality follows since $\bE[\exp\left(\langle t, X\rangle \right)] = \exp\left(\frac{1}{2} \| t \|_2^2\right)$ and the last inequality follows since $\langle uu^\top, uu^\top \rangle = \| u\|_2^4 = \langle vv^\top, vv^\top \rangle = \| v \|_2^4 = 1$ and $\langle uu^\top, vv^\top \rangle = \langle u, v \rangle^2$. Let $G_m$ denote a symmetric random walk on $\mathbb{Z}$ stopped at the $m$th step and $H \sim \text{Hypergeometric}(n, k, k)$ as in Lemma \ref{lem:hgmw}. Now note that if $u, v \sim \text{Unif}[S]$ are independent, then $\langle u, v\rangle$ is distributed as $G_H/k$. Now let $a = \mu^2 \left( k \log \frac{en}{k} \right)^{-1} \le \beta_0$ and note that Lemma \ref{lem:hgmw} along with Cauchy-Schwarz implies that
$$\TV(\bP_0, \bP_1) \le \frac{1}{2} \sqrt{\chi^2(\bP_1, \bP_0)} \le \frac{1}{2} \sqrt{g(\beta_0) - 1}$$
where $g$ is the function from Lemma \ref{lem:hgmw}. Setting $w(\beta_0) = \frac{1}{2} \sqrt{g(\beta_0) - 1}$ proves the theorem.
\end{proof}

Note that any instance of $\textsc{SROS}_D$ is also an instance of $\textsc{ROS}_D$ and thus the information theoretic lower bound in Theorem \ref{thm:sswinf} also holds for $\textsc{ROS}_D$. Symmetrizing $\textsc{SROS}_D$ as in Section 7 yields that the same information-theoretic lower bound holds for $\textsc{SSW}_D$. Now consider the function $\tau : \mathbb{R}^{n \times n} \to \mG_n$ that such that if $\tau(M) = G$ then $\{i, j\} \in E(G)$ if and only if $M_{ij} > 0$ for all $i < j$. In other words, $\tau$ thresholds the above-diagonal entries of $M$ as in Step 3 of $\textsc{SSBM-Reduction}$ from Lemma \ref{lem:ssbm}. Note that $\tau$ maps $N(0, 1)^{\otimes n \times n}$ to $G(n, 1/2)$ and takes $\mL_{H_1}(M)$ from Theorem \ref{thm:sswinf} to a distribution in $\mL_{\text{SSBM}} \in G_B(n, k, 1/2, \rho)$ where
$$\rho = \Phi\left(\frac{\mu}{k} \right) - \frac{1}{2} = \frac{1}{\sqrt{2\pi}} \cdot \frac{\mu}{k}$$
As in the proof of Theorem \ref{thm:SSBMguar} there is a method $e : \mG_n \to \mG_n$ that either adds or removes edges with a fixed probability mapping $G(n, 1/2)$ to $G(n, q)$ and $\mL_{\text{SSBM}}$ to $e(\mL_{\text{SSBM}}) \in G_B(n, k, 1/2, \rho')$ where $\rho' = \Theta(\rho)$ as long as $q = \Theta(1)$. By the data processing inequality, we now have that
$$\TV\left( G(n, 1/2), \mL_{\text{SSBM}}'\right) \le \TV\left( \mL_{H_0}(M), \mL_{H_1}(M) \right) \to 0 \quad \text{as } n \to \infty$$
if $\mu \ll \sqrt{k}$ which corresponds to $\rho \ll 1/\sqrt{k}$, establishing the information theoretic lower bound for $\textsc{SSBM}_D$ in the regime $q = \Theta(1)$ matching Theorem \ref{lem:2a}.

Corresponding recovery lower bounds for these problems follow from information-theoretic lower bounds for biclustering. Note that an instance of $\textsc{BC}_{WR}(n, k, \mu)$ is an instance of $\textsc{ROS}_{WR}(n, k, \mu/k)$. By Theorem \ref{thm:infbcrec}, $\textsc{ROS}_{WR}(n, k, \mu)$ is therefore information-theoretically impossible if $\mu \le 2\sqrt{k \log \frac{n}{k}}$. An analogous information-theoretic lower bound is given for a symmetric variant of $\textsc{BC}_{WR}$ in \cite{hajek2016information}, which implies the corresponding lower bounds for $\textsc{SROS}_{WR}$ and $\textsc{SSW}_{WR}$.

\paragraph{Information-Theoretically Optimal Algorithms.} Unlike existing maximum likelihood estimators for recovery such as those for $\textsc{BC}_R$ in \cite{chen2016statistical} and \cite{cai2015computational}, the definition of $\mathcal{V}_{n, k}$ requires that algorithms solving $\textsc{ROS}_R$ are adaptive to the sparsity level $k$. We introduce a modified exhaustive search algorithm $\textsc{ROS-Search}$ that searches over all possible sparsity levels and checks whether each resulting output is reasonable using an independent copy $B$ of the data matrix.

We first establish the notation that will be used in this section. Given some $v \in \mathbb{R}^n$, let $\text{supp}_+(v)$ denote the set of $i$ with $v_i > 0$ and $\text{supp}_-(v)$ denote the set of $i$ with $v_i < 0$. If $A, B \in \mathbb{R}^{n \times n}$, let $\langle A, B \rangle = \text{Tr}(A^\top B)$. Let $S_t$ be the set of $v \in \mathbb{R}^n$ with exactly $t$ nonzero entries each in $\{-1, 1\}$. In order to show that $\textsc{ROS-Search}$ succeeds at solving $\textsc{ROS}_R$ asymptotically down to its information-theoretic limit, we begin by showing the following lemma.

\begin{figure}[t!]
\begin{algbox}
\textbf{Algorithm} $\textsc{ROS-Search}$
\vspace{2mm}

\textit{Inputs}: Matrix $M \in \mathbb{R}^{n \times n}$, sparsity upper bound $k$, threshold $\rho > 0$, constant $c_1 \in (0, 1)$
\begin{enumerate}
\item Sample $G \sim N(0, 1)^{\otimes n \times n}$ and form $A = \frac{1}{\sqrt{2}} (M + G)$ and $B = \frac{1}{\sqrt{2}} (M - G)$
\item For each pair $k_1, k_2 \in [c_1 k, k]$ do:
\begin{enumerate}
\item[a.] Let $S_t$ be the set of $v \in \mathbb{R}^n$ with exactly $t$ nonzero entries each in $\{-1, 1\}$ and compute 
$$(u, v) = \text{argmax}_{(u, v) \in S_{k_1} \times S_{k_2}} \left\{ u^\top A v \right\}$$
\item[b.] Mark the pair $(u, v)$ if it satisfies that
\begin{itemize}
\item The set of $i$ with $\sum_{j = 1}^n u_i v_j B_{ij} \ge \frac{1}{2} k_2 \rho$ is exactly $\text{supp}(u)$
\item The set of $j$ with $\sum_{i = 1}^n u_i v_j B_{ij} \ge \frac{1}{2} k_1 \rho$ is exactly $\text{supp}(v)$
\end{itemize}
\end{enumerate}
\item Output $\text{supp}(u)$, $\text{supp}(v)$ where $(u, v)$ is the marked pair maximizing $|\text{supp}(u)| + |\text{supp}(v)|$
\end{enumerate}
\vspace{1mm}
\end{algbox}
\caption{Exhaustive search algorithm for sparse rank-1 submatrix recovery in Theorem \ref{lem:rossearch}.}
\end{figure}

\begin{lemma} \label{lem:rossearch}
Let $R$ and $C$ be subsets of $[n]$ such that $|R| = k_1$ and $|C| = k_2$ where $k_1, k_2 \in [c_1 k, k]$ for some constant $c_1 \in (0, 1)$. Let $\rho > 0$ and $M \in \mathbb{R}^{n \times n}$ be a random matrix and with independent sub-Gaussian entries with sub-Gaussian norm at most $1$ such that:
\begin{itemize}
\item $\bE[M_{ij}] \ge \rho$ if $(i, j) \in R \times C$; and
\item $\bE[M_{ij}] = 0$ if $(i, j) \not \in R \times C$.
\end{itemize}
There is an absolute constant $c_2 > 0$ such that if $k\rho^2 \ge c_2 \log n$, then
$$\textnormal{argmax}_{(u, v) \in S_{k_1} \times S_{k_2}} \left\{ u^\top M v \right\}$$
is either $(\mathbf{1}_R, \mathbf{1}_C)$ and $(-\mathbf{1}_R, -\mathbf{1}_C)$ with probability at least $1 - n^{-1}$ for sufficiently large $n$.
\end{lemma}

\begin{proof}
For each pair $(u, v) \in S_{k_1} \times S_{k_2}$, let $A_1(u, v)$ be the set of pairs $(i, j) \in R \times C$ with $u_i v_j = -1$, let $A_2(u, v)$ be the set of pairs $(i, j) \in \text{supp}(u) \times \text{supp}(v)$ that are not in $R \times C$ and let $A_3(u, v)$ be the set of $(i, j) \in R \times C$ that are not in $\text{supp}(u) \times \text{supp}(v)$. Now observe that
\begin{align*}
\mathbf{1}_R^\top M \mathbf{1}_C - u^\top M v &= \langle M, \mathbf{1}_R \mathbf{1}_C^\top - uv^\top \rangle = \sum_{(i, j) \in A_1(u, v)}  2M_{ij} - \sum_{(i, j) \in A_2(u, v)} u_i v_j M_{ij} + \sum_{(i, j) \in A_3(u, v)} M_{ij} \\
&\ge \rho \left( 2 |A_1(u, v)| + |A_3(u, v)| \right) + \sum_{(i, j) \in A_1(u, v)}  2\left( M_{ij} - \bE[M_{ij}] \right) \\
&\quad \quad - \sum_{(i, j) \in A_2(u, v)} u_i v_j M_{ij} + \sum_{(i, j) \in A_3(u, v)} \left( M_{ij} - \bE[M_{ij}] \right)
\end{align*}
Since $R \times C$ and $\text{supp}(u) \times \text{supp}(v)$ both have size $k_1 k_2$, it follows that $|A_2(u, v)| = |A_3(u, v)|$. Note that the random variables in the sum above are independent, zero mean and sub-Gaussian with norm at most $1$. By Hoeffding's inequality for sub-Gaussian random variables as in Proposition 5.10 in \cite{vershynin2010introduction}, it follows that
\begin{align*}
\bP\left[ \langle M, \mathbf{1}_R \mathbf{1}_C^\top - uv^\top \rangle \le 0 \right] &\le e \cdot \exp\left( -\frac{c \rho^2 \left( 2 |A_1(u, v)| + |A_2(u, v)| \right)^2}{4 |A_1(u, v)| + |A_2(u, v)| + |A_3(u, v)|} \right) \\
&= e \cdot \exp\left( - \frac{1}{2} c \rho^2 \left( 2|A_1(u, v)| + |A_2(u, v)| \right) \right) \\
&\le e \cdot n^{-c_1^{-2} \cdot \frac{16}{k} \left( 2|A_1(u, v)| + |A_2(u, v)| \right)}
\end{align*}
for some absolute constant $c > 0$ as long as $c\rho^2 \ge 16kc_1^{-2} \log n$. Let $S(a_1, a_2, b_1, b_2)$ be the set of all pairs $(u, v)$ such that $a_1 = |\text{supp}(u) \backslash R|$, $a_2 = |\text{supp}_-(u) \cap R|$, $b_1 = |\text{supp}(v) \backslash C|$ and $b_2 = |\text{supp}_-(v) \cap C|$. Suppose that $a_2 \le \frac{1}{2}(k_1 - a_1)$. Note that for any $(u, v) \in S(a_1, a_2, b_1, b_2)$, we have
$$|A_1(u, v)| = a_2(k_2 - b_1 - b_2) + b_2(k_1 - a_1 - a_2) \quad \text{and} \quad |A_2(u, v)| = a_1 k_2 + b_1 k_1 - a_1 b_1$$
Note that $k_1, k_2 \ge c_1 k$, $a_1 + a_2 \le k_1$ and $b_1 + b_2 \le k_2$. Therefore we have that $\frac{1}{k} |A_2(u, v)| \ge a_1 \cdot \frac{k_2}{k} \ge c_1 a_1$ and $\frac{1}{k} |A_2(u, v)| \ge c_1 b_1$. This implies that $\frac{1}{k} |A_2(u, v)| \ge \frac{1}{2} c_1 (a_1 + b_1)$. Now note that if $b_2 \ge c_1 a_2$, then it holds that
$$\frac{1}{k} |A_1(u, v)| \ge \frac{b_2}{k}(k_1 - a_1 - a_2) \ge \frac{1}{2k} b_2 (k_1 - a_1) \ge \frac{c_1}{2} \cdot b_2 - \frac{a_1}{2} \ge \frac{c_1^2}{4} (a_2 + b_2) - \frac{a_1}{2}$$
Otherwise if $b_2 < c_1 a_2$ then it follows that $b_2 < c_1 a_2 \le c_1 \cdot \frac{1}{2}(k_1 - a_1) \le \frac{k_2}{2}$ since $k_2 \ge c_1 k \ge c_1 k_1$. Now we have that
$$\frac{1}{k} |A_1(u, v)| \ge \frac{a_2}{k}(k_2 - b_1 - b_2) \ge \frac{a_2}{k}(k_2 - b_2) - b_1 \ge \frac{c_1}{2} a_2 - b_1 \ge \frac{c_1}{4} (a_2 + b_2) - b_1$$
Therefore in either case it follows that
$$\frac{1}{k} |A_1(u, v)| \ge \frac{c_1^2}{4}(a_2 + b_2) - a_1 - b_1$$
Combining these inequalities and the fact that $c_1 \in (0, 1)$ yields that
\begin{align*}
\frac{2}{k} |A_1(u, v)| + \frac{1}{k} |A_2(u, v)| &\ge \frac{c_1}{4k} \cdot |A_1(u, v)| + \frac{1}{k} |A_2(u, v)| \\
&\ge \frac{c_1^2}{4}(a_2 + b_2) + \frac{c_1}{4}(a_1 + b_1) \ge \frac{c_1^2}{4}(a_1 + a_2 + b_1 + b_2)
\end{align*}
as long as $a_2 \le \frac{1}{2}(k_1 - a_1)$. Furthermore, we have that
$$|S(a_1, a_2, b_1, b_2)| = \binom{k_1}{a_1} \binom{n - k_1}{a_1} \binom{k_1 - a_1}{a_2} \binom{k_2}{b_1} \binom{n - k_2}{b_1} \binom{k_2 - b_1}{b_2} \le n^{2a_1 + 2b_1 + a_2 + b_2}$$
since $\binom{n}{k} \le n^k$ and $k_1, k_2 \le n$. Let $T$ be the set of $(a_1, a_2, b_1, b_2) \neq (0, 0, 0, 0)$ such that $a_1, a_2, b_1, b_2 \ge 0$, $a_1 + a_2 \le k_1$, $b_1 + b_2 \le k_2$ and $a_2 \le \frac{1}{2}(k_1 - a_1)$. Now note for all $(u, v) \in S^2$, it holds that at least one of the pairs $(u, v)$ or $(-u, -v)$ satisfies that $a_2 \le \frac{1}{2}(k_1 - a_1)$. Since $(u, v)$ and $(-u, -v)$ yield the same value of $u^\top M v$, we can restrict to $T$ in the following union bound. Now note that
\begin{align*}
&\bP\left[ \text{there is } (u, v) \in S^2 \text{ with } (u, v) \neq \pm(\mathbf{1}_R, \mathbf{1}_C) \text{ and } \langle M, \mathbf{1}_R \mathbf{1}_C^\top - uv^\top \rangle \le 0 \right] \\
&\quad \quad \quad \quad \le \sum_{(a_1, a_2, b_1, b_2) \in T} \left( \sum_{(u, v) \in S(a_1, a_2, b_1, b_2)} e \cdot n^{-c_1^{-2} \cdot \frac{16}{k} \left( 2|A_1(u, v)| + |A_2(u, v)| \right)} \right) \\
&\quad \quad \quad \quad \le \sum_{(a_1, a_2, b_1, b_2) \in T} |S(a_1, a_2, b_1, b_2)| \cdot n^{-4(a_1 + a_2 + b_1 + b_2)} \\
&\quad \quad \quad \quad \le \sum_{(a_1, a_2, b_1, b_2) \in T} n^{-2a_1 - 3a_2 - 2b_1 -3b_2} \\
&\quad \quad \quad \quad \le -1 + \sum_{a_1, a_2, b_1, b_2 = 0}^\infty n^{-2a_1 - 3a_2 - 2b_1 -3b_2} \\
&\quad \quad \quad \quad = -1 + \left( \sum_{i = 0}^\infty n^{-2i} \right)^2 \left( \sum_{j = 0}^\infty n^{-3j} \right)^2 \\
&\quad \quad \quad \quad = -1 + (1 - n^{-2})^{-2}(1 - n^{-3})^{-2} = O(n^{-2})
\end{align*}
which as at most $n^{-1}$ for sufficiently large $n$, completing the proof of the lemma.
\end{proof}

We now use this lemma to prove the following theorem, which shows that $\textsc{ROS-Search}$ solves $\textsc{ROS}_R$ and $\textsc{SSW}_R$ as long as $\mu \gtrsim \sqrt{k}$, asymptotically matching their information theoretic limits.

\begin{theorem}
Suppose that $M \sim \mu \cdot rc^\top + N(0, 1)^{\otimes n \times n}$ where $r, c \in \mathcal{V}_{n, k}$. There is an an absolute constant $c > 0$ such that if $\mu \ge c \sqrt{k \log n}$, then $\textsc{ROS-Search}$ applied with $c_1 = 1/2$ and $\rho = \mu/k$ outputs $\text{supp}(r)$ and $\text{supp}(c)$ with probability at least $1 - 4n^{-1}$ for sufficiently large $n$.
\end{theorem}

\begin{proof}
Suppose that $(u, v) \in S_{k_1} \times S_{k_2}$ where $k_1, k_2 \in [c_1 k, k]$ are random vectors that are independent of $B$ and either $\text{supp}(u) \not \subseteq \text{supp}(r)$ or $\text{supp}(v) \not \subseteq \text{supp}(c)$. Note that the definition of $\mathcal{V}_{n, k}$ is such that any fixed $c_1 \in (0, 1)$ suffices for sufficiently large $k$. We first observe that if $\rho = \mu/k$ and $\mu \ge c \sqrt{k \log n}$ then $(u, v)$ is not marked in Step 2b of $\textsc{ROS-Search}$ with probability at least $1 - n^{-3}$. If $\text{supp}(u) \not \subseteq \text{supp}(r)$, then let $i \in \text{supp}(u) \backslash \subseteq \text{supp}(r)$. It follows that $\sum_{j = 1}^n u_i v_j B_{ij} \sim N(0, k_2)$ since $\| v \|_0 = k_2$ and by Gaussian tail bounds that
$$\bP\left[ \sum_{j = 1}^n u_i v_j B_{ij} \ge \frac{1}{2} k_2 \rho \right] \le \frac{1}{\sqrt{2\pi}} \cdot \frac{2}{\rho \sqrt{k_2}} \cdot \exp\left(- \frac{\rho^2 k_2}{8} \right) \le n^{-3}$$
if $\mu^2 \ge \sqrt{3c_1 k \log n}$. This implies that if $\text{supp}(u) \not \subseteq \text{supp}(r)$, then $(u, v)$ is marked in Step 2b of $\textsc{ROS-Search}$ with probability at most $n^{-3}$. A symmetric argument shows that the same is true if $\text{supp}(v) \not \subseteq \text{supp}(c)$. Now for each pair $k_1, k_2 \in [c_1 k, k]$, let $u_{k_1}$ and $v_{k_2}$ be such that
$$(u_{k_1}, v_{k_2}) = \text{argmax}_{(u, v) \in S_{k_1} \times S_{k_2}} \left\{ u^\top A v \right\}$$
if the pair is marked and let $(u_{k_1}, v_{k_2}) = (0, 0)$ otherwise. By Lemma \ref{lem:gausscloning} in the next section, $A$ and $B$ from Step 1 of $\textsc{ROS-Search}$ are i.i.d. and distributed as $\frac{1}{\sqrt{2}} \cdot rc^\top + N(0, 1)^{\otimes n \times n}$. In particular, the pairs $(u_{k_1}, v_{k_2})$ are in the $\sigma$-algebra generated by $A$ and hence are independent of $B$. By a union bound, we have
\begin{align*}
&\bP\left[ \text{supp}(u_{k_1}) \subseteq \text{supp}(r) \text{ and } \text{supp}(v_{k_2}) \subseteq \text{supp}(c) \text{ for all } k_1, k_2 \in [c_1 k, k] \right] \\
&\quad \quad \quad \quad \ge 1 - \sum_{k_1, k_2 \in [c_1 k, k]} \bP\left[ \text{supp}(u_{k_1}) \not \subseteq \text{supp}(r) \text{ or } \text{supp}(v_{k_2}) \not \subseteq \text{supp}(c) \right] \\
&\quad \quad \quad \quad \ge 1 - k^2 \cdot n^{-3} \ge 1 - n^{-1}
\end{align*}
Now let $k_1' = \| r \|_0$ and $k_2' = \| c \|_0$ and let $r'$ be the vector such that $r'_i = 0$ if $v_i = 0$, $r'_i = 1$ if $v_i > 0$ and $r'_i = -1$ if $v_i < 0$. Define the map $\tau_r : \mathbb{R}^n \to \mathbb{R}^n$ such that $\tau_r$ maps $v$ to the vector with $i$th entry $r'_i v_i$. Define $c'$ and $\tau_c$ analogously. Now let $M'$ be the matrix with $(i, j)$th entry $r'_i c'_j A_{ij}$. Since the entries of $M$ are independent Gaussians, it follows that $M' \sim \mu \cdot \tau_r(r) \tau_c(c)^\top + N(0, 1)^{\otimes n \times n}$. Now observe that since $S_{k_1'}$ and $S_{k_2'}$ are preserved by $\tau_r$ and $\tau_c$, respectively, we have that $(u_{k_1'}, v_{k_2'}) = (\tau_r(u), \tau_c(v))$ where
$$(u, v) = \text{argmax}_{(u, v) \in S_{k_1'} \times S_{k_2'}} \left\{ u^\top M' v \right\}$$
Now note that the mean matrix $\mu \cdot \tau_r(r) \tau_c(c)^\top$ of $M'$ has all nonnegative entries. Furthermore, the entries in its support are at least $\mu \cdot \min_{(i, j) \in \text{supp}(r) \times \text{supp}(c)} |r_i c_j| \ge \frac{\mu}{k} = \rho$ by the definition of $\mathcal{V}_{n, k}$. Applying Lemma \ref{lem:rossearch} now yields that with probability at least $1 - n^{-1}$, it follows that $(u, v) = (\mathbf{1}_{\text{supp}(r)}, \mathbf{1}_{\text{supp}(c)})$ or $(u, v) = (-\mathbf{1}_{\text{supp}(r)}, -\mathbf{1}_{\text{supp}(c)})$. This implies that $u_{k_1'} = r'$ and $v_{k_2'} = c'$ and thus are supported on all of $\text{supp}(r)$ and $\text{supp}(c)$, respectively.

We now will show that $r'$ and $c'$ are marked by the test in Step 2b with high probability. If $i \in \text{supp}(r)$, then $\sum_{j = 1}^n r'_i c'_j B_{ij} \sim N\left( \sum_{j = 1}^n \mu \cdot |r_i| \cdot |c_j|, k_2 \right)$. Since $\sum_{j = 1}^n \mu \cdot |r_i| \cdot |c_j| \ge \mu \cdot \frac{k_2}{k} = k_2 \rho$ by the definition of $\mathcal{V}_{n, k}$. Therefore it follows by the same Gaussian tail bound as above that
$$\bP\left[ \sum_{j = 1}^n r'_i c'_j B_{ij} < \frac{1}{2} k_2 \rho \right] \le \bP\left[ N(0, k_2) < -\frac{1}{2} k_2 \rho \right] \le n^{-3}$$
Furthermore, if $i \not \in \text{supp}(r)$ it follows that $\sum_{j = 1}^n r'_i c'_j B_{ij} \sim N(0, k_2)$ and thus
$$\bP\left[ \sum_{j = 1}^n r'_i c'_j B_{ij} \ge \frac{1}{2} k_2 \rho \right] = \bP\left[ N(0, k_2) \ge \frac{1}{2} k_2 \rho \right] \le n^{-3}$$
Now a union bound yields that
\begin{align*}
\bP\left[ \text{supp}(r) = \left\{ i \in [n] : \sum_{j = 1}^n r'_i c'_j B_{ij} \ge \frac{1}{2} k_2 \rho \right\} \right] &\ge 1 - \sum_{i \in \text{supp}(r)} \bP\left[ \sum_{j = 1}^n r'_i c'_j B_{ij} < \frac{1}{2} k_2 \rho \right] \\
&\quad \quad - \sum_{i \not\in \text{supp}(r)} \bP\left[ \sum_{j = 1}^n r'_i c'_j B_{ij} \ge \frac{1}{2} k_2 \rho \right] \\
&\ge 1 - k_1 \cdot n^{-3} - (n - k_1) n^{-3} = 1 - n^{-2}
\end{align*}
An identical argument yields that
$$\bP\left[ \text{supp}(c) = \left\{ j \in [n] : \sum_{i = 1}^n r'_i c'_j B_{ij} \ge \frac{1}{2} k_2 \rho \right\} \right] \ge 1 - n^{-2}$$
A union bound now yields that $(r', c')$ is marked by the test in Step 2b with probability at least $1 - 2n^{-2}$. A further union bound now yields that with probability at least $1 - 2n^{-1} - 2n^{-2}$, the following events all hold:
\begin{itemize}
\item $\text{supp}(u_{k_1}) \subseteq \text{supp}(r)$ and $\text{supp}(v_{k_2}) \subseteq \text{supp}(c)$ for all $k_1, k_2 \in [c_1 k, k]$;
\item $(u_{k_1'}, v_{k_2'}) = (r', c')$; and
\item $(r', c')$ is marked when input to the test in Step 2b.
\end{itemize}
These three events imply that that the vector $(r', c')$ is marked in $\textsc{ROS-Search}$ and hence the maximum of $|\text{supp}(u_{k_1})| + |\text{supp}(v_{k_2})|$ over marked pairs $(u_{k_1}, v_{k_2})$ is $k_1' + k_2'$. Furthermore, first event implies that any marked pair $(u_{k_1}, v_{k_2})$ with $|\text{supp}(u_{k_1})| + |\text{supp}(v_{k_2})| = k_1' + k_2'$ must satisfy that $\text{supp}(u_{k_1}) = \text{supp}(r)$ and $\text{supp}(v_{k_2}) = \text{supp}(v)$. Thus the algorithm correctly recovers the supports of $r$ and $c$ with probability at least $1 - 2n^{-1} - 2n^{-2}$, proving the theorem.
\end{proof}

The last theorem of this section gives a simple test solving the detection problems $\textsc{SSBM}_D$, $\textsc{ROS}_D$ and $\textsc{SSW}_D$ asymptotically down to their information-theoretic limits. More precisely, this test solves $\textsc{SSBM}_D$ if $\rho \gtrsim \frac{1}{\sqrt{k}}$ and by setting $\rho = \mu/k$, solves $\textsc{ROS}_D$ and $\textsc{SSW}_D$ as long as $\mu \gtrsim \sqrt{k}$.

\begin{theorem}
Suppose that $c_1 \in (0, 1)$ is a fixed constant and let $R_+$ and $R_-$ be disjoint subsets of $[n]$ with $c_1 k \le |R_+| + |R_-| \le k$. Let $C_+$ and $C_-$ be defined similarly. Let $M \in \mathbb{R}^{n \times n}$ be a random matrix with independent sub-Gaussian entries with sub-Gaussian norm at most $1$ such that:
\begin{itemize}
\item $\bE[M_{ij}] \ge \rho$ if $(i, j) \in R_+ \times C+$ or $(i, j) \in R_- \times C_-$;
\item $\bE[M_{ij}] \le -\rho$ if $(i, j) \in R_+ \times C_-$ or $(i, j) \in R_- \times C_+$; and
\item $\bE[M_{ij}] = 0$ if $(i, j) \not \in (R_+ \cup R_-) \times (C_+ \cup C_-)$.
\end{itemize}
There is a constant $c_2 > 0$ such that if $k \rho \ge c_2 \sqrt{\log n}$, then $\max_{(u, v) \in S_k^2} u^\top M v \ge \frac{1}{2} c_1^2 k^2 \rho$ with probability at least $1 - en^{-1}$. If $\bE[M_{ij}] = 0$ for all $(i, j) \in [n]^2$, then there is some constant $c_3 > 0$ such that if $k \rho^2 \ge c_2 \log n$ then $\max_{(u, v) \in S_k^2} u^\top M v < \frac{1}{2} c_1^2 k^2 \rho$ with probability at least $1 - en^{-1}$.
\end{theorem}

\begin{proof}
Let $u \in S_k$ satisfy that $u_i = 1$ for each $i \in R_+$ and $u_i = - 1$ for each $i \in R_-$. Similarly let $v \in S_k$ satisfy that $v_i = 1$ for each $i \in C_+$ and $v_i = -1$ for each $i \in C_-$. It follows that
$$\sum_{i, j = 1}^n u_i v_j \cdot \bE[M_{ij}] \ge (|R_+| + |R_-|)(|C_+| + |C_-|) \rho \ge c_1^2 k^2 \rho$$
Now note that if $c \cdot c_1^4 k^2\rho^2 \ge 4 \log n$, then
\begin{align*}
\bP\left[ u^\top M v < \frac{1}{2} c_1^2 k^2 \rho \right] &= \bP\left[ \sum_{i, j = 1}^n u_i v_j \left(M_{ij} - \bE[M_{ij}] \right) < \frac{1}{2} c_1^2 k^2 \rho -\sum_{i, j = 1}^n u_i v_j \cdot \bE[M_{ij}] \right] \\
&\le \bP\left[ \sum_{i, j = 1}^n u_i v_j \left(M_{ij} - \bE[M_{ij}] \right) < -\frac{1}{2} c_1^2 k^2 \rho \right] \\
&\le e \cdot \exp\left( - \frac{c \cdot \left( c_1^2 k^2 \rho \right)^2}{4k^2} \right) \le en^{-1}
\end{align*}
for some constant $c > 0$ by Hoeffding's inequality for sub-Gaussian random variables as in Proposition 5.10 in \cite{vershynin2010introduction}. This proves the first claim of the theorem.

Now suppose that $\bE[M_{ij}] = 0$ for all $(i, j) \in [n]^2$. For a fixed pair $(u, v) \in S_k^2$, we have by the same application of Hoeffding's inequality that
$$\bP\left[ u^\top M v \ge \frac{1}{2} c_1^2 k^2 \rho \right] \le e \cdot \exp\left( - \frac{c \cdot c_1^4 k^2 \rho^2}{4} \right)$$
Now note that $|S_k| = 2^k \binom{n}{k} \le (2n)^k$. Thus a union bound yields that
$$\bP\left[ \max_{(u, v) \in S_k^2} u^\top M v \ge \frac{1}{2} c_1^2 k^2 \rho \right] \le |S_k| \cdot e \cdot \exp\left( - \frac{c \cdot c_1^4 k^2 \rho^2}{4} \right) \le e \cdot \exp\left( k \log (2n) - \frac{c \cdot c_1^4 k^2 \rho^2}{4} \right) \le en^{-1}$$
if $c \cdot c_1^4 k^2 \rho^2 \ge 4k \log(2n) + 4 \log n$. This completes the proof of the theorem.
\end{proof}

\paragraph{Polynomial-Time Algorithms.} The proofs in this section are given in Appendix \ref{app9}. The polynomial-time algorithms achieving the tight boundary for the problems in this section are very different in the regimes $k \lesssim \sqrt{n}$ and $k \gtrsim \sqrt{n}$. When $k \lesssim \sqrt{n}$, the simple linear-time algorithm thresholding the absolute values of the entries of the data matrix matches the planted clique lower bounds in Theorem \ref{lem:2a} up to polylogarithmic factors. This is captured in the following simple theorem, which shows that if $\mu \gtrsim k$ then this algorithm solves $\textsc{ROS}_R$ and $\textsc{SSW}_R$. Setting $u = v = 0$ in the Theorem yields that the test outputting $H_1$ if $\max_{i, j \in [n]^2} |M_{ij}| > \sqrt{6 \log n}$ solves the detection variants $\textsc{ROS}_D$ and $\textsc{SSW}_D$ if $\mu \gtrsim k$.

\begin{theorem} \label{thm:maxros}
Let $M \sim \mu \cdot uv^\top + N(0, 1)^{\otimes n \times n}$ where $u, v \in \mathcal{V}_{n, k}$ and suppose that $\mu \ge 2k\sqrt{6\log n}$, then the set of $(i, j)$ with $|M_{ij}| > \sqrt{6\log n}$ is exactly $\text{supp}(u) \times \text{supp}(v)$ with probability at least $1 - O(n^{-1})$. 
\end{theorem}

In the regime $k \gtrsim \sqrt{n}$, the spectral projection algorithm in Figure \ref{fig:rosspectral} from \cite{cai2015computational} achieves exact recovery in $\textsc{ROS}_R$ down to the planted clique lower bounds in Theorem \ref{lem:2a}. This method is described in Algorithm 1 and its guarantees established in Lemma 1 of \cite{cai2015computational}. Although it is stated as a recovery algorithm for a sub-Gaussian variant of $\textsc{BC}_R$, as indicated in Remark 2.1 in \cite{cai2015computational}, the guarantees of the algorithm extend more generally to rank one perturbations of a sub-Gaussian noise matrix. Our model of $\textsc{ROS}_R$ does not exactly fit into the extended model in Remark 2.1, but the argument in Lemma 1 can be applied to show $\textsc{ROS-Spectral-Projection}$ solves $\textsc{ROS}_R$. The details of this argument are show below. For brevity, we omit parts of the proof that are identical to \cite{cai2015computational}.

\begin{figure}[t!]
\begin{algbox}
\textbf{Algorithm} $\textsc{ROS-Spectral-Projection}$
\vspace{2mm}

\textit{Inputs}: Matrix $M \in \mathbb{R}^{n \times n}$
\begin{enumerate}
\item Let $G \sim N(0, 1)^{\otimes n \times n}$ and let $A = \frac{1}{\sqrt{2}}(M + G)$ and $B = \frac{1}{\sqrt{2}}(M - G)$
\item Compute the top left and right singular vectors $U$ and $V$ of $A$
\item Sort the $n$ entries of $U^\top B$ in decreasing order and separate the entries into two clusters $R$ and $[n] \backslash R$ at the largest gap between consecutive values
\item Sort the $n$ entries of $B V$ in decreasing order and separate the entries into two clusters $C$ and $[n] \backslash C$ at the largest gap between consecutive values
\item Output $R$ and $C$
\end{enumerate}
\vspace{1mm}
\end{algbox}
\caption{Algorithm for sparse rank-1 submatrix recovery from \cite{cai2015computational} and in Theorem \ref{thm:rosrec}.}
\label{fig:rosspectral}
\end{figure}

\begin{theorem} \label{thm:rosrec}
Suppose that $M = \mu \cdot rc^\top + N(0, 1)^{\otimes n \times n}$ where $r, c \in \mathcal{V}_{n, k}$. There is a constant $C_1 > 0$ such that if $\mu \ge C_1(\sqrt{n} + \sqrt{k \log n})$ then the algorithm $\textsc{ROS-Spectral-Projection}$ correctly outputs $\text{supp}(r)$ and $\text{supp}(c)$ with probability at least $1 - 2n^{-C_2} - 2\exp(-2C_2n)$ for some constant $C_2 > 0$.
\end{theorem}

A simple singular value thresholding algorithm solves $\textsc{ROS}_D$ if $\mu \gtrsim \sqrt{n}$ and is comparatively simpler to analyze. As previously mentioned, this algorithm also solves $\textsc{SSW}_D$ since any instance of $\textsc{SSW}_D$ is an instance of $\textsc{ROS}_D$. Let $\sigma_1(M)$ denote the largest singular value of the matrix $M$. 

\begin{theorem} \label{thm:rossvd}
Suppose that $M$ is an instance of $\textsc{ROS}_D(n, k, \mu)$. There is a constant $C_1 > 0$ such that if $\mu > 4 \sqrt{n} + 2 \sqrt{2 \log n}$ then the algorithm that outputs $H_1$ if $\sigma_1(M) \ge \frac{1}{2} \mu$ and $H_0$ otherwise has Type I$+$II error tending to zero as $n \to \infty$.
\end{theorem}

To complete this section, we give a simple spectral thresholding algorithm for $\textsc{SSBM}_D$ if $\rho \gtrsim \frac{\sqrt{n}}{k}$. Let $\lambda_1(X)$ denote the largest eigenvalue of $X$ where $X$ is a square symmetric matrix. 

\begin{theorem} \label{thm:ssbmalg}
Let $G$ be an instance of $\textsc{SSBM}_D(n, k, q, \rho)$ where $q = \Theta(1)$, $k = \Omega(\sqrt{n})$ and $\rho \ge \frac{6\sqrt{n}}{k}$. Let $A \in \mathbb{R}^{n \times n}$ be the adjacency matrix of $G$ and $J \in \mathbb{R}^{n \times n}$ be the matrix with zeros on its diagonal and ones elsewhere. Then the algorithm that outputs $H_1$ if $\lambda_1(A - qJ) \ge 2\sqrt{n}$ and $H_0$ otherwise has Type I$+$II error tending to zero as $n \to \infty$.
\end{theorem}

\subsection{Sparse PCA and Biased Sparse PCA}

\paragraph{Information-Theoretic Lower Bounds.} Other than for Theorem \ref{thm:reclowerbound}, the proofs in these sections on sparse PCA and biased sparse PCA are given in Appendix \ref{app9}. In \cite{berthet2013optimal}, information-theoretic lower bounds for $\textsc{SPCA}_D(n, k, d, \theta)$ were considered and it was shown that if
$$\theta \le \min\left\{\frac{1}{\sqrt{2}}, \sqrt{\frac{k \log(1 + o(d/k^2))}{n}} \right\}$$
then the optimal Type I$+$II error of any algorithm for $\textsc{SPCA}_D$ tends to $1$ as $n \to \infty$. The proof of this information-theoretic lower bound follows a similar $\chi^2$ and MGF argument as in the previous section. If $d \ll k^2$, then this bound degrades to $\theta = o\left( \frac{d}{kn} \right)$. When $d = \Theta(n)$, there is a gap between this information-theoretic lower bound and the best known algorithm based on thresholding the $k$-sparse eigenvalue of the empirical covariance matrix, which only requires $\theta \gtrsim \sqrt{k/n}$. This information-theoretic lower bound was improved in \cite{cai2015optimal} to match the algorithmic upper bound. The following is their theorem in our notation.

\begin{theorem}[Proposition 2 in \cite{cai2015optimal}] \label{thm:cmw15}
Let $\beta_0 \in (0, 1/36)$ be a constant. Suppose that $X = (X_1, X_2, \dots, X_n)$ is an instance of $\textsc{SPCA}_D(n, k, d, \theta)$ where under $H_1$, the planted vector $v$ is chosen uniformly at random from all $k$-sparse unit vectors with nonzero coordinates equal to $\pm \frac{1}{\sqrt{k}}$. If it holds that
$$\theta \le \min \left\{ 1, \sqrt{\frac{\beta_0 k}{n} \log \left( \frac{ed}{k} \right)} \right\}$$
then there is a function $w : (0, 1/36) \to (0, 1)$ satisfying that $\lim_{\beta_0 \to 0^+} w(\beta_0) = 0$ and
$$\TV\left( \mL_{H_0}(X), \mL_{H_1}(X) \right) \le w(\beta_0)$$
\end{theorem}

In particular if $\theta \ll \sqrt{k/n}$, then this inequality eventually applies for every $\beta_0 > 0$ and $\TV\left( \mL_{H_0}(X), \mL_{H_1}(X) \right) \to 0$, establishing the desired information-theoretic lower bound for $\textsc{SPCA}_D$. We now show that $\textsc{BSPCA}_D$ satisfies a weaker lower bound. The prior on the planted vector $v$ used in Theorem 5.1 of \cite{berthet2013optimal} to derive the suboptimal lower bound for $\textsc{SPCA}_D$ shown above placed entries equal to $1/\sqrt{k}$ on a random $k$-subset of the $d$ coordinates of $v$. Therefore their bound also applies to $\textsc{BSPCA}_D$. In order strengthen this to the optimal information-theoretic lower bound for $\textsc{BSPCA}_D$, we need apply Lemma \ref{lem:hgm} in place of the weaker hypergeometric squared MGF bounds used in \cite{berthet2013optimal}. To simplify the proof, we will use the following lemma from \cite{cai2015optimal}.

\begin{lemma}[Lemma 7 in \cite{cai2015optimal}] \label{lem:chispca}
Let $v$ be a distribution on $d \times d$ symmetric random matrices $M$ such that $\| M \| \le 1$ almost surely. If $\bE_v[N(0, I_d + M)^{\otimes n}] = \int N(0, I_d + M)^{\otimes n} dv(M)$, then
$$\chi^2\left( \bE_v[N(0, I_d + M)^{\otimes n}], N(0, I_d)^{\otimes n} \right) + 1 = \bE \left[ \det(I_d - M_1 M_2)^{-n/2} \right]$$
where $M_1$ and $M_2$ are independently drawn from $v$.
\end{lemma}

Applying this lemma with the hypergeometric squared MGF bounds in Lemma \ref{lem:hgm} yields the following information-theoretic lower bound for $\textsc{BSPCA}_D$.

\begin{theorem} \label{thm:bscpait}
Suppose that $X = (X_1, X_2, \dots, X_n)$ is an instance of $\textsc{BSPCA}_D(n, k, d, \theta)$ where under $H_1$, the planted vector $v$ is chosen uniformly at random from all $k$-sparse unit vectors in $\mathbb{R}^d$ with nonzero coordinates equal to $\frac{1}{\sqrt{k}}$. Suppose it holds that $\theta \le 1/\sqrt{2}$ and
$$\theta \le \min\left\{ \sqrt{\frac{\beta_0 k}{n} \log \left( \frac{ed}{k} \right)}, \sqrt{\frac{\beta_0 d^2}{nk^2}}\right\}$$
for some $0 < \beta_0 < (16e)^{-1}$. Then there is a function $w : (0, 1) \to (0, 1)$ satisfying that $\lim_{\beta_0 \to 0^+} w(\beta_0) = 0$ and
$$\TV\left( \mL_{H_0}(X), \mL_{H_1}(X) \right) \le w(\beta_0)$$
\end{theorem}

We remark that this same proof technique applied to the ensemble of $k$-sparse unit vectors $v$ chosen uniformly at random from those with nonzero coordinates equal to $\pm 1/\sqrt{k}$ with Lemma \ref{lem:hgmw} proves Theorem \ref{thm:cmw15}. This difference in the lower bounds resulting from these two choices of ensembles illustrates the information-theoretic difference between $\textsc{SPCA}_D$ and $\textsc{BSPCA}_D$. We now will show information-theoretic lower bounds for the weak recovery problems $\textsc{SPCA}_{WR}$ and $\textsc{BSCPA}_{WR}$. The argument presented here is similar to the proof of Theorem 3 in \cite{wang2016statistical}. For this argument, we will need a variant of the Gilbert-Varshamov lemma and generalized Fano's lemma as in \cite{wang2016statistical}. Given two $u, v\in \mathbb{R}^d$, let $d_H(u, v)$ denote the Hamming distance between $u$ and $v$.

\begin{lemma}[Gilbert-Varshamov, Lemma 4.10 in \cite{massart2007concentration}]
Suppose that $\alpha, \beta \in (0, 1)$ and $k \le \alpha \beta d$. Let
$$\rho = \frac{\alpha}{\log(\alpha \beta)^{-1}}\left( \beta - \log \beta - 1 \right)$$
Then there is a subset $S$ of $\{ v \in \{0, 1\}^d : \| v \|_0 = k \}$ of size at least $\left( \frac{d}{k} \right)^{\rho k}$ such that for any two $u, v \in S$ with $u \neq v$, it holds that $d_H(u, v) \ge 2(1 - \alpha)k$. 
\end{lemma}

\begin{lemma}[Generalized Fano's Lemma, Lemma 3 in \cite{yu1997assouad}]
Let $P_1, P_2, \dots, P_M$ be probability distributions on a measurable space $(\mathcal{X}, \mathcal{B})$ and assume that $\KL(P_i, P_j) \le \beta$ for all $i \neq j$. Any measurable function $\phi : \mathcal{X} \to \{1, 2, \dots, M\}$ satisfies that
$$\max_{1 \le i \le M} P_i(\phi \neq i) \ge 1 - \frac{\beta + \log 2}{\log M}$$
\end{lemma}

With these two lemmas, we now will show that the weak recovery problems $\textsc{SPCA}_{WR}$ and $\textsc{BSCPA}_{WR}$ cannot be solved if $\theta \lesssim \sqrt{k/n}$, matching Theorem \ref{lem:2a}. Note that in the theorem below, $\phi(X)$ and $\text{supp}(v)$ have size $k$ for all $X \in \mathbb{R}^{d \times n}$ and $v \in S$. Therefore it holds that $|\phi(X) \Delta \textnormal{supp}(v)| = 2k - 2|\phi(X) \cap \textnormal{supp}(v)|$ for all such $X$ and $v$.

\begin{theorem} \label{thm:reclowerbound}
Fix positive integers $n, k, d$ and real numbers $\theta > 0$ and a constant $\epsilon \in (0, 1)$ such that $k \le \epsilon d/4$. Let $P_v$ denote the distribution $N(0, I_d + \theta vv^\top)$ and let $S$ be the set of all $k$-sparse unit vectors with nonzero entries equal to $1/\sqrt{k}$. If
$$\frac{n\theta^2}{2(1 + \theta)} + \log 2 \le \frac{\epsilon^2}{2\log 4\epsilon^{-1}} \cdot k \log \left( \frac{d}{k} \right)$$
then for any function $\phi : \mathbb{R}^{d \times n} \to \binom{[n]}{k}$, it holds that
$$\min_{v \in S} \bE_{X \sim P_v^{\otimes n}} \left[ |\phi(X) \cap \textnormal{supp}(v)| \right] \le \left( \frac{1}{2} + \epsilon \right) k$$
\end{theorem}

\begin{proof}
Note that $\KL(N(0, \Sigma_0), N(0, \Sigma_1)) = \frac{1}{2} \left[ \text{Tr}\left( \Sigma_1^{-1} \Sigma_0 \right) - d + \ln \frac{\det(\Sigma_1)}{\det(\Sigma_0)} \right]$ for any positive semidefinite $\Sigma_0, \Sigma_1 \in \mathbb{R}^{d\times d}$. Therefore for any $\| u \|_2 = \| v \|_2 = 1$,
\begin{align*}
\KL(P_u^{\otimes n}, P_v^{\otimes n}) &= n \cdot \KL(P_u, P_v) = \frac{n}{2} \cdot \text{Tr}\left( \left(I_d + \theta uu^\top \right)^{-1} \left(I_d + \theta vv^\top \right) - I_d \right) \\
&= \frac{n\theta}{2} \cdot \text{Tr}\left( \left( I_d + \theta uu^\top \right)^{-1}\left( vv^\top - uu^\top \right) \right) \\
&= \frac{n\theta}{2} \cdot \text{Tr}\left( \left( I_d - \frac{\theta}{1 + \theta} \cdot uu^\top \right) \left( vv^\top - uu^\top \right) \right) \\
&= \frac{n\theta}{2} \cdot \text{Tr}\left(vv^\top - uu^\top - \frac{\theta}{1 + \theta} \cdot \langle u, v \rangle \cdot uv^\top + \frac{\theta}{1 + \theta} \cdot uu^\top \right) \\
&= \frac{n\theta^2}{4(1 + \theta)} \cdot \| u - v \|_2^2 \le \frac{n\theta^2}{2(1 + \theta)}
\end{align*}
since $\det(I_d + \theta uu^\top) = \det(I_d + \theta vv^\top) = 1 + \theta$. Let $S_0$ be the subset from Gilbert-Varshamov's lemma applied with $\alpha = \epsilon$ and $\beta = \frac{1}{4}$. It follows that
$$\log |S_0| \ge \frac{\epsilon}{2\log 4\epsilon^{-1}} \cdot k \log \left( \frac{d}{k} \right)$$
For each $u \in S$, let $\hat{u}$ be an element of $S_0$ such that $|\text{supp}(u) \cap \text{supp}(\hat{u})|$ is maximal. Let $\hat{\phi}$ denote the function that maps $X \in \mathbb{R}^{d \times n}$ to $\text{supp}(\hat{u})$ where $u = \mathbf{1}_{\phi(X)}$. Suppose that $v \in S_0$ is such that $v \neq \hat{u}$. Observe by the triangle inequality that
$$4k - 2|\phi(X) \cap \text{supp}(v)| - 2|\hat{\phi}(X) \cap \phi(X)| = d_H(\hat{u}, u) + d_H(u, v) \ge d_H(\hat{u}, v) \ge 2(1 - \epsilon) k$$
Rearranging and using the fact that $|\phi(X) \cap \text{supp}(v)| \le |\hat{\phi}(X) \cap \phi(X)|$ yields that
$$k - |\phi(X) \cap \textnormal{supp}(v)| \ge \frac{1}{2} \left( 1 - \epsilon \right) k$$
Note that $k - |\phi(X) \cap \textnormal{supp}(v)| \ge 0$ is true for all $v$. Therefore if $v \in S_0$, we have
$$\bE_{X \sim P_v^{\otimes n}} \left[ k - |\phi(X) \cap \textnormal{supp}(v)| \right] \ge \bP_{X \sim P_{v}^{\otimes n}} \left[ \hat{\phi}(X) \neq v \right] \cdot \frac{1}{2} \left( 1 - \epsilon \right)k$$
Now observe by generalized Fano's Lemma,
\begin{align*}
\max_{v \in S} \bE_{X \sim P_v^{\otimes n}} \left[ k - |\phi(X) \cap \textnormal{supp}(v)| \right] &\ge \max_{v \in S_0} \bE_{X \sim P_v^{\otimes n}} \left[ k - |\phi(X) \cap \textnormal{supp}(v)| \right] \\
&\ge \frac{1}{2} \left( 1 - \epsilon \right)k \cdot \max_{v \in S_0} \bP_{X \sim P_{v}^{\otimes n}} \left[ \hat{\phi}(X) \neq v \right] \\
&\ge \frac{1}{2} \left( 1 - \epsilon \right)k \cdot \left[ 1 - \frac{\frac{n\theta^2}{2(1 + \theta)} + \log 2}{\log |S_0|} \right] \\
&\ge \frac{1}{2} \left( 1 - \epsilon \right)k \cdot \left(1 - \epsilon \right) \ge \left( \frac{1}{2} - \epsilon \right) k
\end{align*}
Rearranging completes the proof of the lemma.
\end{proof}

\paragraph{Information-Theoretically Optimal Algorithms.} Given a positive semidefinite matrix $M \in \mathbb{R}^{d\times d}$, let the $k$-sparse maximum eigenvalue and eigenvector be
$$\lambda_{\max}^k(M) = \max_{\|u\|_2 = 1, \| u \|_0 = k} u^\top M u \quad \text{and} \quad v_{\max}^k(M) = \arg\max_{\| u \|_2  = 1, \| u \|_0 = k} u^\top M u$$
Note that $\lambda_{\max}^k(M)$ and $v_{\max}^k(M)$ can be computed by searching over all principal $k \times k$ minors for the maximum eigenvalue and corresponding eigenvector. In \cite{berthet2013complexity} and \cite{wang2016statistical}, the $k$-sparse maximum eigenvalue and eigenvector of the empirical covariance matrix are shown to solve sparse PCA detection and estimation in the $\ell_2$ norm under general distributions satisfying a restricted covariance concentration condition. Specializing these results to Gaussian formulations of sparse PCA yields the following theorems. Let $L(u, v) = \sqrt{1 - \langle u, v \rangle}$ for $u, v \in \mathbb{R}^d$ with $\| u \|_2 = \| v \|_2 = 1$. The statement of Theorem 2 in \cite{berthet2013complexity} is for $k \ll \sqrt{d}$, but the same argument also shows the result for $k \ll d$. Note that the following theorem applies to $\textsc{BSPCA}_D$ since any instance of $\textsc{BSPCA}_D$ is also an instance of  $\textsc{SPCA}_D$.

\begin{theorem}[Theorem 2 in \cite{berthet2013complexity}]
Suppose that $X = (X_1, X_2, \dots, X_n)$ is an instance of $\textsc{SPCA}_D(n, k, d, \theta)$ and let $\hat{\Sigma}$ be the empirical covariance matrix of $X$. If $\theta, \delta \in (0, 1)$ are such that
$$\theta > 15 \sqrt{\frac{k \log \left( \frac{3ed}{k\delta} \right)}{n}}$$
then the algorithm that outputs $H_1$ if $\lambda_{\max}^k(\hat{\Sigma}) > 1 + 8 \sqrt{\frac{k \log \left( \frac{3ed}{k\delta} \right)}{n}}$ and $H_0$ otherwise has Type I$+$II error at most $\delta$.
\end{theorem}

\begin{theorem}[Theorem 2 in \cite{wang2016statistical}]
Suppose that $k, d$ and $n$ are such that $2k \log d \le n$. Let $P_v$ denote the distribution $N(0, I_d + \theta vv^\top)$ and given some $X = (X_1, X_2, \dots, X_n) \sim P_v^{\otimes n}$, let $\hat{\Sigma}(X)$ be the empirical covariance matrix of $X$. It follows that
$$\sup_{v \in \mathcal{V}_{d, k}} \bE_{X \sim P_v} L\left( v^k_{\max}\left(\hat{\Sigma}(X)\right), v\right) \le 7 \sqrt{\frac{k \log d}{n\theta^2}}$$
\end{theorem}

The latter result on estimation in the $\ell_2$ norm yields a weak recovery algorithm for $\textsc{SPCA}_{WR}$ and $\textsc{BSPCA}_{WR}$ by thresholding the entries of $v^k_{\max}(\hat{\Sigma})$, as in the following theorem. If $\frac{k \log d}{n\theta^2} \to 0$ then this algorithm achieves weak recovery. 

\begin{theorem} \label{thm:weakrecspca}
Suppose that $k, d$ and $n$ are such that $2k \log d \le n$. Let $S(X) \subseteq [d]$ be the set of coordinates of $v^k_{\max}\left(\hat{\Sigma}(X) \right)$ with magnitude at least $\frac{1}{2\sqrt{k}}$. It follows that
$$\sup_{v \in \mathcal{V}_{d, k}} \bE_{X \sim P_v} \left[\frac{1}{k} \left|S(X) \Delta \textnormal{supp}(v)\right| \right] \le 56 \sqrt{\frac{2k \log d}{n\theta^2}}$$
\end{theorem}

We next analyze an algorithm thresholding the sum of the entries of the empirical covariance matrix. For instances of $\textsc{BSPCA}_D$, this sum test solves the detection problem when $\theta \gtrsim \frac{\sqrt{n}}{k}$ in the regime $d = \Theta(n)$ and becomes optimal when $k \gtrsim n^{2/3}$. Recall that $\delta = \delta_{\text{BSPCA}} > 0$ is the constant in the definition of $\mathcal{BV}_{d, k}$.

\begin{theorem} \label{thm:bspcadet}
Suppose that $X = (X_1, X_2, \dots, X_n)$ is an instance of $\textsc{BSPCA}_D(n, k, d, \theta)$ and let $\hat{\Sigma}(X)$ be the empirical covariance matrix of $X$. Suppose that $2 \delta^2 k \theta \le d$ and $\frac{n k^2 \theta^2}{d^2} \to \infty$ as $n \to \infty$. Then the test that outputs $H_1$ if $\mathbf{1}^\top \hat{\Sigma}(X) \mathbf{1} > d + 2\delta^2 k \theta$ and $H_0$ otherwise has Type I$+$II error tending to zero as $n \to \infty$.
\end{theorem}

\paragraph{Polynomial-Time Algorithms.} As shown in \cite{berthet2013complexity}, $\textsc{SPCA}_D$ and $\textsc{BSPCA}_D$ can be solved with a semidefinite program in the regime $k \lesssim \sqrt{n}$. Their algorithm computes a semidefinite relaxation of the maximum $k$-sparse eigenvalue, first forming the empirical covariance matrix $\hat{\Sigma}(X)$ and solving the convex program
\begin{align*}
\text{SDP}(X) = \max_Z \quad &\text{Tr}\left(\hat{\Sigma}(X) Z\right) \\
\text{s.t.} \quad &\text{Tr}(Z) = 1, |Z|_1 \le k, Z \succeq 0
\end{align*}
Thresholding $\text{SDP}(X)$ yields a detection algorithm for $\textsc{SPCA}_D$ with guarantees captured in the following theorem. In \cite{berthet2013complexity}, a more general model is considered with sub-Gaussian noise. We specialize the more general theorem in \cite{berthet2013complexity} to our setup.

\begin{theorem}[Theorem 5 in \cite{berthet2013complexity}]
Suppose that $\delta \in (0, 1)$. Let $X = (X_1, X_2, \dots, X_n)$ be an instance of $\textsc{SPCA}_D(n, k, d, \theta)$ and suppose that $\theta \in [0, 1]$ satisfies that
$$\theta \ge 23 \sqrt{\frac{k^2 \log(d^2/\delta)}{n}}$$
Then the algorithm that outputs $H_1$ if $\textnormal{SDP}(X) \ge 16 \sqrt{\frac{k^2 \log(d^2/\delta)}{n}} + \frac{1}{\sqrt{n}}$ and $H_0$ otherwise has Type I$+$II error at most $\delta$.
\end{theorem}

In \cite{wang2016statistical}, a semidefinite programming approach is shown to solve the sparse PCA estimation task under the $\ell_2$ norm. As in the proof of Theorem \ref{thm:weakrecspca}, this yields a weak recovery algorithm for $\textsc{SPCA}_{WR}$ and $\textsc{BSPCA}_{WR}$, achieving the tight barrier when $k \lesssim \sqrt{n}$. The SDP algorithm in \cite{wang2016statistical} is shown in Figure \ref{fig:spcasdp} and its guarantees specialized to the case of Gaussian data are in the following theorem.

\begin{figure}[t!]
\begin{algbox}
\textbf{Algorithm} $\textsc{SPCA-SDP}$
\vspace{2mm}

\textit{Inputs}: $X = (X_1, X_2, \dots, X_n) \in \mathbb{R}^{d \times n}$
\begin{enumerate}
\item Compute the empirical covariance matrix $\hat{\Sigma} = \hat{\Sigma}(X)$ and set
$$\epsilon = \frac{\log d}{4n} \quad \text{and} \quad \lambda = 4 \sqrt{\frac{\log d}{n}}$$
\item Let $f : \mathbb{R}^{d \times d} \to \mathbb{R}$ be
$$f(M) = \text{Tr}\left(\hat{\Sigma} M\right) - \lambda \| M \|_1$$
and compute an $\epsilon$-maximizer $\hat{M}^{\epsilon}$ of $f(M)$ subject to the constraints that $M$ is symmetric, $M \succeq 0$ and $\text{Tr}(M) = 1$
\item Output $v_{\text{SDP}} = \arg \max_{\| u \|_2 = 1} u^\top \hat{M}^\epsilon u$
\end{enumerate}
\vspace{1mm}
\end{algbox}
\caption{SDP algorithm for weak $\ell_2$ estimation in sparse PCA from \cite{wang2016statistical} and in Theorem \ref{thm:spcasdp}.}
\label{fig:spcasdp}
\end{figure}

\begin{theorem}[Theorem 5 in \cite{wang2016statistical}] \label{thm:spcasdp}
Suppose that $k, d$ and $n$ are such that $4 \log d \le n \le k^2 d^2 \theta^{-2} \log d$ and $0 < \theta \le k$. Let $P_v$ denote the distribution $N(0, I_d + \theta vv^\top)$ and given some $X = (X_1, X_2, \dots, X_n) \sim P_v^{\otimes n}$, let $v_{\textnormal{SDP}}(X)$ be the output of $\textsc{SPCA-SDP}$ applied to $X$. Then
$$\sup_{v \in \mathcal{V}_{d, k}} \bE_{X \sim P_v} L\left( v_{\textnormal{SDP}}(X), v\right) \le \min\left( (16\sqrt{2} + 2) \sqrt{\frac{k^2 \log d}{n \theta^2}}, 1 \right)$$
\end{theorem}

Using an identical argument to Theorem \ref{thm:weakrecspca}, we obtain the following theorem.

\begin{theorem}
Suppose that $k, d$ and $n$ are such that $4 \log d \le n \le k^2 d^2 \theta^{-2} \log d$ and $0 < \theta \le k$. Let $S(X) \subseteq [d]$ be the set of coordinates of $v_{\textnormal{SDP}}(X)$ with magnitude at least $\frac{1}{2\sqrt{k}}$. It follows that
$$\sup_{v \in \mathcal{V}_{d, k}} \bE_{X \sim P_v} \left[\frac{1}{k} \left|S(X) \Delta \textnormal{supp}(v)\right| \right] \le 8\sqrt{2} \cdot \min\left( (16\sqrt{2} + 2) \sqrt{\frac{k^2 \log d}{n \theta^2}}, 1 \right)$$
\end{theorem}

To establish the polynomial-time upper bounds for $\textsc{SPCA}$ and $\textsc{BSPCA}$ in Theorem \ref{lem:2a}, it now suffices to consider $k \gtrsim \sqrt{n}$. First consider the detection problems $\text{SPCA}_D$ and $\textsc{BSPCA}_D$. The following theorem establishes that a spectral algorithm directly applied to the empirical covariance matrix solves $\textsc{SPCA}_D$ when $\theta \gtrsim 1$.

\begin{theorem} \label{thm:spectralspca}
Suppose that $X = (X_1, X_2, \dots, X_n)$ is an instance of $\textsc{SPCA}_D(n, k, d, \theta)$ and let $\hat{\Sigma}(X)$ be the empirical covariance matrix of $X$. Suppose that $d \le cn$ for some constant $c > 0$, $d \to \infty$ and $n(1 + \theta)^{-2} \to \infty$ as $n \to \infty$ and it holds that $\theta > 4\sqrt{c}$. Then the test that outputs $H_1$ if $\lambda_1(\hat{\Sigma}(X)) > 1 + 2\sqrt{c}$ and $H_0$ otherwise has Type I$+$II error tending to zero as $n \to \infty$.
\end{theorem}

The algorithm summing the entries of the empirical covariance matrix in Theorem \ref{thm:bspcadet} runs in polynomial time and shows that $\text{BSPCA}_D$ can be solved in polynomial time as long as $\theta \gtrsim \frac{\sqrt{n}}{k}$. Note that this algorithm gives an upper bound matching Theorem \ref{lem:2a} and can detect smaller signal levels $\theta$ in $\text{BSPCA}_D$ than the spectral algorithm.

For the recovery problem when $k \gtrsim \sqrt{n}$, the spectral algorithm considered in Theorem 1.1 of \cite{krauthgamer2015semidefinite} achieves the upper bound in Theorem \ref{lem:2a} for the exact recovery problems $\text{SPCA}_R$ and $\text{BSPCA}_R$. As given in \cite{krauthgamer2015semidefinite}, this spectral algorithm is not adaptive to the support size of the planted vector and assumes that the planted sparse vector has nonzero entries of the form $\pm 1/\sqrt{k}$. We mildly adapt this algorithm to only require that the planted vector is in $\mathcal{V}_{d, k}$. The proof of its correctness follows a similar argument as Theorem 1.1 in \cite{krauthgamer2015semidefinite}. We omit details that are identical for brevity.

\begin{theorem} \label{thm:spcaspectralalg}
Suppose that $k, d$ and $n$ are such that $k, d \to \infty$, $\frac{d}{n} \to c$ and $\frac{k \log d}{n} \to 0$ as $n \to \infty$ for some constant $c > 0$. Let $X = (X_1, X_2, \dots, X_n) \sim P_v^{\otimes n}$ and let $\hat{v}$ be the leading eigenvector of $\hat{\Sigma}(X)$. Let $S \subseteq [d]$ be the set of $i$ such that $|\hat{v}_i|^4 \ge \frac{\log d}{kd}$. If $\theta > \sqrt{c}$ is fixed then $S = \textnormal{supp}(v)$ with probability tending to one as $n \to \infty$.
\end{theorem}

Note that the theorem statement assumes that $d/n \to c$ where $c > 1$. The algorithm can more generally accommodate inputs with $d = \Theta(n)$ by padding the input $X$ with i.i.d. $N(0, 1)$ entries so that $d/n \to c$ where $c > 1$.

\section{Detection-Recovery Reductions}

In this section, we show that our computational lower bounds for detection problems imply corresponding lower bounds for the recovery. The idea is to produce two instances of each problem that are coupled to have the same planted sparse structure but are conditionally independent given this structure. If there is a polynomial-time recovery algorithm that outputs a set $S$ containing a constant fraction of the planted sparse structure, then we restrict to the indices in $S$ yields an instance with a planted structure of linear size and apply a detection algorithm from Section 9. For each of the problems that we consider, this solves detection within a sub-polynomial factor of when detection first becomes information-theoretically possible. The reduction also always runs in polynomial time. Therefore our detection lower bounds also imply recovery lower bounds.

To carry out this argument, we first require methods of creating two such instances. We first give such a cloning method for $\textsc{PDS}$ instances. This method is similar to the cloning map in $\textsc{PC-Lifting}$ from Section 4 but produces two copies rather than four. The proof is given in Appendix \ref{app10}.

\begin{figure}[t!]
\begin{algbox}
\textbf{Algorithm} $\textsc{PDS-Cloning}$
\vspace{2mm}

\textit{Inputs}: Graph $G \in \mG_n$, parameters $p, q \in (0, 1]$ with $p > q$
\begin{enumerate}
\item Set $P = 1 - \sqrt{1 - p}$ and $Q = 1 - \sqrt{1 - q}$
\item For each pair $i, j \in [n]$ with $i < j$, independently generate $x^{ij} \in \{0, 1\}^2$ such that
\begin{itemize}
\item If $\{i, j\} \in E(G)$, then generate $x^{ij}$ from
$$\bP[x^{ij} = v] = \frac{1 - q}{p - q} \cdot P^{|v|_1}(1 - P)^{2 - |v|_1} - \frac{1 - p}{p - q} \cdot Q^{|v|_1}(1 - Q)^{2 - |v|_1}$$
\item If $\{i, j\} \not \in E(G)$, then generate $x^{ij}$ from
$$\bP[x^{ij} = v] = \frac{p}{p - q} \cdot Q^{|v|_1}(1 - Q)^{2 - |v|_1} - \frac{q}{p - q} \cdot P^{|v|_1}(1 - P)^{2 - |v|_1}$$
\end{itemize}
\item Construct the graphs $G^1$ and $G^2$ such that $\{i, j\} \in E(G^k)$ if and only if $x^{ij}_k = 1$
\item Output $(G^1, G^2)$
\end{enumerate}
\vspace{4mm}
\textbf{Algorithm} $\textsc{Gaussian-Cloning}$
\vspace{2mm}

\textit{Inputs}: Matrix $M \in \mathbb{R}^{n \times n}$
\begin{enumerate}
\item Generate a matrix $G \sim N(0, 1)^{\otimes n \times n}$ with independent Gaussian entries
\item Compute the two matrices
$$M^1 = \frac{1}{\sqrt{2}} (M + G) \quad \text{and} \quad M^2 = \frac{1}{\sqrt{2}} (M - G)$$
\item Output $(M^1, M^2)$
\end{enumerate}
\vspace{1mm}
\end{algbox}
\caption{Cloning procedures in Lemmas \ref{lem:pdscloning} and \ref{lem:gausscloning}.}
\end{figure}

\begin{lemma} \label{lem:pdscloning}
Suppose that $S \subseteq [n]$ and $p, q \in (0, 1]$ are such that $p > q$. Also suppose that $P, Q \in [0,1]$ are such that $Q \neq 0, 1$ and the quotients $\frac{P}{Q}$ and $\frac{1 - P}{1 - Q}$ are both between $\sqrt{\frac{1 - p}{1 - q}}$ and $\sqrt{\frac{p}{q}}$. If $G \sim G(n, S, p, q)$ and $(G^1, G^2)$ is the output of $\textsc{PDS-Cloning}$ applied to $G$ with parameters $p, q, P$ and $Q$, then $(G^1, G^2) \sim G(n, S, P, Q)^{\otimes 2}$. Furthermore, if $G \sim G(n, q)$ then $(G^1, G^2) \sim G(n, Q)^{\otimes 2}$.
\end{lemma}

A similar argument as in Lemma \ref{lem:6b} on $\textsc{Gaussian-Lifting}$ yields the following lemma. The proof is also deferred to Appendix \ref{app10}.

\begin{lemma} \label{lem:gausscloning}
If $M \sim \mL\left( A + N(0, 1)^{\otimes n \times n} \right)$ for any fixed matrix $A \in \mathbb{R}^{n \times n}$ and $(M^1, M^2)$ is the output of $\textsc{Gaussian-Cloning}$ applied to $M$, then $(M^1, M^2) \sim \mL\left( \frac{1}{\sqrt{2}} A + N(0, 1)^{\otimes n \times n} \right)^{\otimes 2}$.
\end{lemma}

With these two lemmas, we now overview how detection lower bounds imply partial recovery lower bounds for each problem that we consider. In the cases of sparse PCA and biased sparse PCA, rather than give a direct detection-recovery reduction, we outline modifications to our detection reductions that yield recovery reductions from planted clique. 

\paragraph{Biclustering.} Suppose that there is a randomized algorithm $\phi$ that solves $\textsc{BC}_{PR}$. More precisely, if $M$ is an instance of $\textsc{BC}_{PR}(n, k, \mu)$ with latent row and column supports $S$ and $T$, then $\bE[|\phi_1(M) \cap S|] = \Omega(k)$ and $\bE[|\phi_2(M) \cap T|] = \Omega(k)$. Consider the detection algorithm for $\textsc{BC}_{D}(n, k, \mu \sqrt{2})$ that applies $\textsc{Gaussian-Cloning}$ to the input to produce $(M^1, M^2)$ and then takes the sum of the entries of $M^2$ restricted to the indices in $\phi_1(M^1) \times \phi_2(M^1)$. The algorithm then outputs $H_1$ if this sum is at least $k \cdot \tau(k)$ where $\tau(k) \to \infty$ arbitrarily slowly as $k \to \infty$. Since $M^1$ and $M^2$ are independent, it follows that $M^2$ is independent of $\phi(M^1)$. Under $H_0$, the sum is distributed as $N(0, k^2)$ which is less than $k \cdot \tau(k)$ with probability tending to $1$ as $k \to \infty$. Under $H_1$, let $S$ and $T$ denote the latent row and column supports of the planted submatrix. If $k_1 = |\phi_1(M^1) \cap S|$ and $k_2 = |\phi_2(M^1) \cap T|$, then the sum is distributed as $N(\mu \cdot k_1 k_2, k^2)$. If $\mu \cdot k_1 k_2 \ge 2k \cdot \tau(k)$, then the algorithm outputs $H_1$ with probability tending to $1$ as $k \to \infty$. If $\mu \ge \frac{2\tau(k)^3}{k}$ then $\mu \cdot k_1 k_2 < 2k \cdot \tau(k)$ is only possible if either $k_1 \cdot \tau(k) < k$ or $k_2 \cdot \tau(k) < k$. By Markov's inequality, each of these events occurs with probability tending to zero as $k \to \infty$ over the randomness of $\phi$. Therefore this algorithm has Type I$+$II error tending to zero as $k \to \infty$ if $\mu \ge \frac{2\tau(k)^3}{k}$, which is true for some $\tau$ as long as $\textsc{BC}_R$ is information-theoretically possible.

In summary, if $\phi$ solves $\textsc{BC}_{PR}(n, k, \mu)$ then there is a polynomial-time algorithm using $\phi$ as a blackbox that solves $\textsc{BC}_{D}(n, k, \mu \sqrt{2})$. Hence our computational and information-theoretic lower bounds for $\textsc{BC}_D$ imply partial recovery lower bounds in the same parameter regimes. We now give similar detection-recovery reductions using an initial cloning step for other problems.

\paragraph{Sparse Spiked Wigner and Rank-1 Submatrix.} Similarly to biclustering, suppose that $\phi$ solves $\textsc{ROS}_{PR}(n, k, \mu)$. Consider the detection algorithm for $\textsc{ROS}_D(n, k, \mu \sqrt{2})$ that applies $\textsc{Gaussian-Cloning}$ to the input to produce $(M^1, M^2)$, forms the $k \times k$ matrix $W$ given by $M^2$ restricted to indices in $\phi_1(M^1) \times \phi_2(M^1)$ and outputs $H_1$ if and only if $\sigma_1(W) > 2\sqrt{k} + \sqrt{2 \log k}$ where $\sigma_1(W)$ is the largest singular value of $W$. By Corollary 5.35 in \cite{vershynin2010introduction}, the algorithm outputs $H_0$ under $H_0$ with probability at least $1 - 2k^{-1}$. Under $H_1$, let $k_1$ and $k_2$ be the sizes of the intersections of $\phi_1(M^1)$ and $\phi_2(M^2)$ with the row and column supports, respectively, of the rank-1 spike. By the definition of $\mathcal{V}_{n, k}$, it follows that the rank-1 spike has largest singular value at least $\mu k^{-1} \sqrt{k_1 k_2}$. By Weyl's interlacing inequality, it follows that $\sigma_1(W) \ge \mu k^{-1} \sqrt{k_1 k_2} - 2\sqrt{k} - \sqrt{2 \log k}$. Suppose that $\mu \ge \tau(k) \sqrt{k}$ where $\tau(k) \to \infty$ as $k \to \infty$ arbitrarily slowly. It follows that $\mu k^{-1} \sqrt{k_1 k_2} < 4 \sqrt{k} + 2 \sqrt{2 \log k} < 8 \sqrt{k}$ implies that either $\tau(k) \cdot k_1 < 64k$ or $\tau(k) \cdot k_2 < 64k$. Both of these events occur with probability tending to zero as $k \to \infty$. It follows that this algorithm has Type I$+$II error tending to zero as $k \to \infty$ if $\mu \ge \tau(k) \sqrt{k}$, which aligns with the information theoretic lower bound on $\textsc{ROS}_R$ of $\mu \gtrsim \sqrt{k}$. Similar reductions yield an analogous result for $\textsc{SROS}$ and therefore also $\textsc{SSW}$.

\paragraph{Planted Independent Set and Planted Dense Subgraph.} We first give a detection-recovery algorithm for $\textsc{PIS}$. Suppose that $\phi$ solves $\textsc{PIS}_{PR}(n, k, q)$. Consider the detection algorithm that takes the complement graph of an input $\textsc{PIS}_{D}(n, k, q)$, applies $\textsc{PDS-Cloning}$ with $P = p' = 1$, $q' = 1 - q$ and $Q = 1 - q/2$ to produce $(G^1, G^2)$ and then outputs $H_1$ if and only if $\overline{G^2}$ restricted to the vertices in $\phi(\overline{G^1})$ contains at most $\binom{k}{2} q - k\sqrt{q(1 - q) \log k}$ edges. Here $\overline{G}$ denotes the complement of the graph $G$. First note that these inputs to $\textsc{PDS-Cloning}$ are valid since $1 - P = 1 - p' = 0$ and
$$\frac{P}{Q} = \frac{1}{1-q/2} = \frac{2}{1 + q'} \le \sqrt{\frac{1}{q'}} = \sqrt{\frac{1}{1-q}}$$
Under $H_0$, it follows that $\overline{G^1}$ and $\overline{G^2}$ are independent and distributed as $G(n, q/2)$. By the same applications of Bernstein's inequality as in Theorem \ref{thm:pdsdet}, the algorithm outputs $H_1$ with probability tending to zero as $k \to \infty$. Under $H_1$, it follows that $\overline{G^1}$ and $\overline{G^2}$ are independent and distributed as $G_I(n, k, q/2)$. Let $k_1$ be the size of the intersection between $\phi(\overline{G^1})$ and the latent support of the planted independent set. The number of edges in $\overline{G^2}$ restricted to the vertices in $\phi(\overline{G^1})$ is distributed as $\text{Bin}( \binom{k}{2} - \binom{k_1}{2}, q )$. If $\binom{k_1}{2} q \ge 2k\sqrt{q(1 - q) \log k}$, then with high probability the algorithm outputs $H_1$. If $q \ge \frac{\tau(k) \log k}{k^2}$ where $\tau(k) \to \infty$ then it would have to hold that $\tau(k) \cdot k_1 \le 4k$ for this inequality not to be true. However, this event occurs with probability tending to zero by Markov's inequality. This gives a reduction from $\textsc{PIS}_{PR}(n, k, q)$ to $\textsc{PIS}_D(n, k, 2q)$.

Similar detection-recovery reductions apply for planted dense subgraph. Fix some constant $w \in (0, 1/2)$. For an instance of $\textsc{PDS}_D(n, k, p, q)$ with $p > q$, consider $\textsc{PDS-Cloning}$ with $P = \frac{wp+(1-w)q}{2}$ and $Q = q/2$. Note that these are valid inputs to $\textsc{PDS-Cloning}$ when
$$\sqrt{\frac{1 - p}{1 - q}} \le 1 - \frac{p - q}{2(1 - q)} \le 1 - \frac{w(p - q)}{2 - q} = \frac{1 - P}{1 - Q} \quad \text{and} \quad \frac{P}{Q} = \frac{wp+(1-w)q}{q} = 1 + \frac{w(p - q)}{q} \le \sqrt{\frac{p}{q}}$$
where the second inequality holds as long as $\frac{p - q}{q} \le \frac{1 - 2w}{w^2}$. Taking $w$ to be sufficient small covers the entire parameter space $p - q = O(q)$. Producing two independent copies of planted dense subgraph and thresholding the edge count now yields a detection-recovery reduction by the same argument as for planted independent set.

\paragraph{Sparse PCA and Biased Sparse PCA.} Rather than give a generic reduction between detection and recovery for variants of sparse PCA, we modify our existing reductions to produce two copies. Note that the reductions $\textsc{SPCA-High-Sparsity}$ and $\textsc{SPCA-Low-Sparsity}$ approximately produce instances of $\textsc{ROS}_D$ and $\textsc{BC}_D$, respectively, as intermediates. Consider the reduction that applies $\textsc{Gaussian-Cloning}$ to these intermediates and then the second steps of the reductions to both copies. Under $H_1$, this yields two independent copies of $N\left(0, I_n + \theta uu^\top \right)$ with a common latent spike $u$. Furthermore the resulting parameter $\theta$ is only affected up to a constant factor.

Now given two independent samples from $\textsc{SPCA}_D(n, k, n, \theta)$ constrained to have the same hidden vector under $H_1$ and a randomized algorithm $\phi$ that solves the weak recovery problem, we will show that a spectral algorithm solves detection. Let $(X^1, X^2)$ denote the $n \times n$ data matrices for the two independent samples. Consider the algorithm that computes the $k \times k$ empirical covariance matrix $\hat{\Sigma}$ using the columns in $X^2$ restricted to the indices in $\phi(X^1)$ and then outputs $H_1$ if and only if $\lambda_1(\hat{\Sigma}) \ge 1 + 2 \sqrt{k/n}$. Under $H_0$, the same argument in Theorem \ref{thm:spectralspca} implies that $\lambda_1(\hat{\Sigma}) < 1 + 2 \sqrt{k/n}$ with probability at least $1 - 2e^{-k/2}$. Under $H_1$, let $k_1$ denote the size of the intersection between $\phi(X^1)$ and the latent support of the hidden vector. It follows that each column of $X^2$ restricted to the indices in $S = \phi(X^1)$ is distributed as $N\left(0, I_k + \theta u_S u_S^\top \right)$ where $u$ is the hidden vector. Now note that $\| u_S \|_2 \ge \sqrt{k_1/k}$ and by the argument in Theorem \ref{thm:spectralspca}, it follows that $u_S^\top \hat{\Sigma} u_S \ge 1 + \frac{\theta}{2k} \cdot \sqrt{k_1 k_2}$ with probability tending to 1 as $k \to \infty$. Therefore if $\theta \ge \tau(k) \sqrt{k/n}$ for some $\tau(k) \to \infty$, then this algorithm has Type I$+$II error tending to zero as $k \to \infty$, which aligns with the information-theoretic lower bound on $\textsc{SPCA}_D$ and $\textsc{SPCA}_R$.

%
%

\section{Future Directions}

A general direction for future work is to add more problems to the web of reductions established here. This work also has left number of specific problems open, including the following.
\begin{enumerate}
\item Collisions between support elements in $\textsc{Reflection-Cloning}$ causes our formulations of $\textsc{SSW}_D, \textsc{ROS}_D, \textsc{SSBM}_D$ and $\textsc{SPCA}_D$ to be as composite hypothesis testing problems rather than the canonical simple hypothesis testing formulations. Is there an alternative reduction from planted clique that can strengthen these lower bounds to hold for the simple hypothesis testing analogues?
\item All previous planted clique lower bounds for $\textsc{SPCA}_D$ are not tight over the parameter regime $k \ll \sqrt{n}$. Is there a reduction from planted clique yielding tight computational lower bounds for $\textsc{SPCA}_D$ in this highly sparse regime?
\item Is there a polynomial time algorithm for the recovery variant of $\textsc{SSBM}$ matching the computational barrier for the detection variant?
\item Can the PDS recovery conjecture be shown to follow from the planted clique conjecture?
\end{enumerate}

\section*{Acknowledgements}

We thank Philippe Rigollet for inspiring discussions on related topics. We are grateful for helpful comments on this work from Philippe Rigollet, Yury Polyanskiy, Piotr Indyk, Jonathan Weed, Frederic Koehler, Vishesh Jain and the anonymous reviewers. This work was supported in part by the grants ONR N00014-17-1-2147 and NSF CCF-1565516.

\bibliography{GB_BIB.bib}
\bibliographystyle{alpha}

\pagebreak

\appendix

\section{Deferred Proofs from Section 3}
\label{app3}

\begin{proof}[Proof of Lemma \ref{lem:5tv}]
Let $B$ be any event in $\sigma\{Y \}$. It follows that
$$\bP\left[ Y \in B | Y \in A \right] - \bP\left[ Y \in B \right] = \bP\left[ Y \in B | Y \in A \right] \cdot \left( 1 - \bP[Y \in A] \right) - \bP\left[ Y \in B \cap A^c \right]$$
Since $\bP\left[ Y \in B \cap A^c \right] \le \bP[Y \in A^c] = 1 - \bP[Y \in A]$ and $\bP\left[ Y \in B | Y \in A \right] \in [0, 1]$, we have that the quantity above is between $-\bP[Y \in A^c]$ and $\bP[Y \in A^c]$. Therefore, by the definition of total variation distance
$$\TV\left( \mL(Y | A), \mL(Y) \right) = \sup_{B \in \sigma\{Y\}} \left| \bP\left[ Y \in B | Y \in A \right] - \bP\left[ Y \in B \right] \right| = \bP[Y \in A^c]$$
where the equality case is achieved by setting $B = A$. This proves the lemma.
\end{proof}

\section{Deferred Proofs from Section 5}
\label{app5}

\begin{proof}[Proof of Lemma \ref{lem:5a}]
Let $f_X(m)$ and $g_X(m)$ be the PMFs of $\text{Pois}(c \lambda)$ and $\text{Pois}(\lambda)$, respectively. Note that
$$f_X(m) = \frac{e^{-c \lambda} (c\lambda)^m}{m!} \quad \text{and} \quad g_X(m) = \frac{e^{-\lambda} (\lambda)^m}{m!}$$
can be computed and sampled in $O(1)$ operations. Therefore Lemma \ref{lem:5zz} implies that $\textsc{rk}_{\text{P1}}$ can be computed in $O(N) = O(\log n)$ time. The rest of the proof entails bounding $\Delta$ in Lemma \ref{lem:5zz} with $p = 1$. Let the set $S$ be as defined in Lemma \ref{lem:5zz}, let $M =  \log_c q^{-1} \ge 3\epsilon^{-1}$ and define the set $S' = \left\{ m \in \mathbb{Z}_{\ge 0} : m \le M \right\}$. Now note that if $m \in S'$ then
$$\frac{f_X(m)}{g_X(m)} = \frac{\frac{e^{-c \lambda} (c\lambda)^m}{m!}}{\frac{e^{-\lambda} (\lambda)^m}{m!}} = e^{-(c - 1)\lambda} c^m \le c^M \le q^{-1}$$
and therefore it follows that $S' \subseteq S$. For sufficiently large $n$, we have that $M = \log_c q^{-1} \ge cn^{-\epsilon} \ge c\lambda > \lambda$ and therefore a standard Poisson tail bound yields that
$$\bP_{X \sim g_X}[X \not \in S] \le \bP_{X \sim g_X}[X > M] \le e^{-\lambda} \left(\frac{e\lambda}{M} \right)^M \le \left(\frac{e}{M} \right)^M \lambda^{3\epsilon^{-1}} \le \left(\frac{e}{M} \right)^M n^{-3}$$
since $\lambda \le n^{-\epsilon}$. Similarly, we have that $\bP_{X \sim f_X}[X \not \in S] \le \left(\frac{ce}{M} \right)^M n^{-3}$. Now note that for sufficiently large $n$, we have that
\begin{align*}
\frac{\bP_{X \sim f_X} [X \not \in S]}{1 - q} + \left( \bP_{X \sim g_X}[X \not \in S] + q \right)^{N} &\le \frac{\left(\frac{ce}{M} \right)^M n^{-3}}{1 - q} + \left( \left(\frac{e}{M} \right)^M n^{-3} + q \right)^N \\
&\le \frac{\left(\frac{ce}{M} \right)^M n^{-3}}{1 - q} + \left( q^{1/2} \right)^N \\
&\le \left( (1 - q)^{-1} \left(\frac{ce}{M} \right)^M + 1 \right) n^{-3}
\end{align*}
By similar reasoning, we have that for sufficiently large $n$
$$\frac{\bP_{X \sim g_X} [X \not \in S]}{1 - q} + \left( \bP_{X \sim f_X}[X \not \in S] \right)^{N} \le \left( (1 - q)^{-1} \left(\frac{e}{M} \right)^M + 1 \right) n^{-3}$$
Therefore $\Delta \le \left( (1 - q)^{-1} \left(\frac{ce}{M} \right)^M + 1 \right) n^{-3}$ for sufficiently large $n$ and applying Lemma \ref{lem:5zz} proves the lemma.
\end{proof}

\begin{proof}[Proof of Lemma \ref{lem:5b}]
As in Lemma \ref{lem:5a}, let $f_X(m)$ and $g_X(m)$ be the PMFs of $\text{Pois}(c \lambda)$ and $\text{Pois}(\lambda)$ and note that they can be computed and sampled in $O(1)$ operations. Lemma \ref{lem:5zz} implies that $\textsc{rk}_{\text{P2}}$ can be computed in $O(N) = O(n^K \log n) = \text{poly}(n)$ time. Let the set $S$ be as defined in Lemma \ref{lem:5zz}, let $M =  \log_c (1 + 2 \rho) \ge (K + 3)\epsilon^{-1}$ and define the set $S' = \left\{ m \in \mathbb{Z}_{\ge 0} : m \le M \right\}$. Now note that if $n > 1$, then $\lambda \le n^{-\epsilon} < 1$ and it follows that
$$e^{-(c - 1)\lambda} \ge 1 - (c - 1)\lambda > 2 - c \ge 2 - (1 + 2\rho)^{\epsilon/(K + 3)} > 1 - 2\rho$$
since $\epsilon \in (0, 1)$. Therefore if $m \in S'$, then
$$1 - 2\rho < e^{-(c - 1)\lambda} \le \frac{f_X(m)}{g_X(m)} = e^{-(c - 1)\lambda} c^m \le c^M \le 1 + 2\rho$$
and it follows that $S' \subseteq S$. By the same Poisson tail bounds as in Lemma \ref{lem:5a}, we have that for sufficiently large $n$
$$\bP_{X \sim g_X}[X \not \in S] \le \left(\frac{e}{M} \right)^M n^{-K-3} \quad \text{and} \quad \bP_{X \sim f_X}[X \not \in S] \le \left(\frac{ce}{M} \right)^M n^{-K-3}$$
Now note that for sufficiently large $n$, we have that $\rho^{-1} \le n^K$ and $\left(\frac{e}{M} \right)^M n^{-K - 3} \le \frac{1}{2} n^{-K} \le \frac{1}{2}\rho$ since $M = O_n(1)$. Therefore
\begin{align*}
\rho^{-1} \cdot \bP_{X \sim f_X} [X \not \in S] + \left( \bP_{X \sim g_X}[X \not \in S] + \frac{1}{1 + 2\rho} \right)^{N} &\le \rho^{-1} \left(\frac{ce}{M} \right)^M n^{-K-3} \\
&\quad \quad + \left( \left(\frac{e}{M} \right)^M n^{-K - 3} + 1 - \rho \right)^N \\
&\le \left(\frac{ce}{M} \right)^M n^{-3} + \left( 1 - \frac{\rho}{2} \right)^N \\
&\le \left(\frac{ce}{M} \right)^M n^{-3} + e^{-\rho N/2} \\
&\le \left( \left(\frac{ce}{M} \right)^M + 1 \right) n^{-3}
\end{align*}
By similar reasoning, we have that for sufficiently large $n$,
$$\rho^{-1} \cdot \bP_{X \sim g_X} [X \not \in S] + \left( \bP_{X \sim f_X}[X \not \in S] + 1 - 2\rho \right)^{N} \le \left( \left(\frac{e}{M} \right)^M + 1 \right) n^{-3}$$
Therefore we have that $\Delta \le \left( \left(\frac{ce}{M} \right)^M + 1 \right) n^{-3}$ for sufficiently large $n$ and applying Lemma \ref{lem:5zz} proves the lemma.
\end{proof}

\begin{proof}[Proof of Lemma \ref{lem:5c}]
Let $f_X(x)$ and $g_X(x)$ be the PDFs of $N(\mu, 1)$ and $N(0, 1)$, respectively, given by
$$f_X(x) = \frac{1}{\sqrt{2\pi}} e^{-(x - \mu)^2/2} \quad \text{and} \quad g_X(x) = \frac{1}{\sqrt{2\pi}} e^{-x^2/2}$$
which can be computed and sampled in $O(1)$ operations in the given computational model. To bound $N$, note that since $\log(1 + x) \ge x/2$ for $x \in (0, 1)$, we have that
$$\log \left( \frac{p}{q} \right) \ge \frac{p - q}{2q} \ge \frac{1}{2}(p - q) \ge \frac{1}{2} n^{-O_n(1)}$$
and similarly that $\log\left( \frac{1 - q}{1 - p} \right) \ge \frac{p - q}{2(1 - p)} \ge \frac{1}{2}(p - q) \ge \frac{1}{2} n^{-O_n(1)}$. Therefore $N = \text{poly}(n)$ and Lemma \ref{lem:5zz} implies that $\textsc{rk}_{\text{G}}$ can be computed in $\text{poly}(n)$ time. Let the set $S$ be as defined in Lemma \ref{lem:5zz}, let $M = \sqrt{6 \log n + 2\log(p - q)^{-1}}$ and define the set $S' = \left\{ x \in \mathbb{R} : |x| \le M \right\}$. Since $M\mu \le \delta/2$, we have for any $x \in S'$ that
$$\frac{1 - p}{1 - q} \le \exp\left( - 2M\mu \right) \le \exp\left( -M \mu - \frac{\mu^2}{2} \right) \le \frac{f_X(x)}{g_X(x)} = \exp\left( x\mu - \frac{\mu^2}{2} \right) \le \exp\left( M \mu \right) \le \frac{p}{q}$$
for sufficiently large $n$ since $M \to \infty$ as $n \to \infty$ and $\mu \in (0, 1)$. This implies that $S' \subseteq S$. Using the bound $1 - \Phi(t) \le \frac{1}{\sqrt{2\pi}} \cdot t^{-1} e^{-t^2/2}$ for $t \ge 1$, we have that
$$\bP_{X \sim g_X}[X \not \in S] \le \bP_{X \sim g_X}[X \not \in S'] = 2\left(1 - \Phi(M) \right) \le \frac{2}{\sqrt{2\pi}} \cdot M^{-1} e^{-M^2/2}$$
Similarly since $M/(M - \mu) \le 2$, we have for sufficiently large $n$ that
\begin{align*}
\bP_{X \sim f_X}[X \not \in S] \le \bP_{X \sim f_X}[X \not \in S'] &= \left(1 - \Phi(M - \mu) \right) + \left(1 - \Phi(M + \mu) \right) \\
&\le \frac{1}{\sqrt{2\pi}} \cdot (M - \mu)^{-1} e^{-(M - \mu)^2/2} + \frac{1}{\sqrt{2\pi}} \cdot (M + \mu)^{-1} e^{-(M + \mu)^2/2} \\
&\le \frac{1}{\sqrt{2\pi}} \cdot M^{-1} e^{-(\mu^2 + M^2)/2} \left[ \frac{M}{M - \mu} \cdot e^{M\mu} + 1 \right] \\
&\le \frac{1}{\sqrt{2\pi}} \left( 1 + 2\sqrt{\frac{p}{q}} \right) \cdot M^{-1} e^{-M^2/2}
\end{align*}
Now note that $M^{-1} e^{-M^2/2} \le e^{-M^2/2} \le (p - q)n^{-3}$ for sufficiently large $n$. If $n$ is large enough then $\frac{2}{\sqrt{2\pi}} \cdot n^{-3} \le \frac{q^{1/2}}{p(p^{1/2} + q^{1/2})} = \Omega_n(1)$ since $q = \Omega_n(1)$. Rearranging yields that $\frac{2}{\sqrt{2\pi}} (p - q) n^{-3} \le \sqrt{\frac{q}{p}} - \frac{q}{p}$. Together with the previous two equations and the fact that $N \ge 6\delta^{-1} \log n$, this implies
\begin{align*}
\frac{\bP_{X \sim f_X} [X \not \in S]}{p - q} + \left( \bP_{X \sim g_X}[X \not \in S] + \frac{q}{p} \right)^{N} &\le \frac{1}{\sqrt{2\pi}} \left( 1 + 2\sqrt{\frac{p}{q}} \right)n^{-3} + \left( \frac{2}{\sqrt{2\pi}} (p - q) n^{-3} + \frac{q}{p} \right)^N \\
&\le \frac{1}{\sqrt{2\pi}} \left( 1 + 2\sqrt{\frac{p}{q}} \right)n^{-3} + \left( \frac{q}{p} \right)^{N/2} \\
&\le \left( 1 + \frac{1}{\sqrt{2\pi}} + \frac{4}{\sqrt{2\pi}} \sqrt{\frac{p}{q}} \right)n^{-3} = O_n(n^{-3})
\end{align*}
Now note that if $\frac{1}{\sqrt{2\pi}} \left( 1 + 4\sqrt{\frac{p}{q}} \right) n^{-3} > \frac{(1 - p)^{1/2}}{(1 - q)((1 - q)^{1/2} + (1 - p)^{1/2})}$ then it follows that
$$\frac{1 - p}{1 - q} \le \left( \frac{(1 - q)}{\sqrt{2\pi}} \left( 1 + 4\sqrt{\frac{p}{q}} \right) \cdot n^3 - 1 \right)^{-1} \le Cn^{-3}$$
for some constant $C > 0$ if $n$ is sufficiently large, since $1 - q = \Omega_n(1)$. Otherwise, the same manipulation as above implies that $\frac{1}{\sqrt{2\pi}} \left( 1 + 4\sqrt{\frac{p}{q}} \right) n^{-3} \le \sqrt{\frac{1 - p}{1 - q}} - \frac{1 - p}{1 - q}$. Therefore we have in either case that
$$\frac{1}{\sqrt{2\pi}} \left( 1 + 4\sqrt{\frac{p}{q}} \right) n^{-3} + \frac{1 - p}{1 - q} \le \max \left\{ \sqrt{\frac{1 - p}{1 - q}}, \frac{1}{\sqrt{2\pi}} \left( 1 + 4\sqrt{\frac{p}{q}} \right) n^{-3} +Cn^{-3} \right\}$$
For sufficiently large $n$, the second term in the maximum above is at most $n^{-2}$. Therefore
\begin{align*}
\frac{\bP_{X \sim g_X} [X \not \in S]}{p - q} + \left( \bP_{X \sim f_X}[X \not \in S] + \frac{1 - p}{1 - q}\right)^{N} &\le \left( \frac{2}{\sqrt{2\pi}} \right) n^{-3} + \left( \frac{1}{\sqrt{2\pi}} \left( 1 + 2\sqrt{\frac{p}{q}} \right) n^{-3} + \frac{1 - p}{1 - q} \right)^N \\
&\le \left( \frac{2}{\sqrt{2\pi}} \right) n^{-3} + \max \left\{ \left( \frac{1 - p}{1 - q} \right)^{N/2}, n^{-2N} \right\} \\
&= O_n(n^{-3})
\end{align*}
Therefore $\Delta = O_n(n^{-3})$ for sufficiently large $n$ since $q = \Omega_n(1)$. Now applying Lemma \ref{lem:5zz} proves the lemma.
\end{proof}

\section{Deferred Proofs from Section 6}
\label{app6}

\begin{proof}[Proof of Theorem \ref{thm:poissonhard}]
If $\beta < \alpha$ then PDS in this regime is information-theoretically impossible. Thus we may assume that $\beta \ge \alpha$. Take $\epsilon > 0$ to be a small enough constant so that
$$\beta + \epsilon(1 - \beta) < \frac{1}{2} + \frac{\alpha}{4}$$
Now let $\gamma = \frac{2\beta - \alpha + \epsilon(1 - \beta)}{2 - \alpha}$. Rearranging the inequality above yields that $\gamma \in (0, 1/2)$. Now set
$$\ell_n = \left\lceil \frac{(\alpha - \epsilon) \log_2 n}{2 - \alpha} \right\rceil, \quad \quad k_n = \lceil n^{\gamma} \rceil, \quad \quad N_n = 2^{\ell_n} n \quad \quad K_n = 2^{\ell_n} k_n,$$
$$p_n = 1 - e^{-4^{-\ell_n} cn^{-\epsilon}}, \quad \quad q_n = 1 - e^{-4^{-\ell_n} n^{-\epsilon}}$$
Take $p$ to be a small enough constant so that $p < \frac{1}{2} c^{-3\epsilon^{-1}}$. By Lemma \ref{lem:6a}, there is a randomized polynomial time algorithm mapping $\text{PC}_D(n, k_n, p)$ to $\text{PDS}_D(N_n, K_n, p_n, q_n)$ with total variation converging to zero as $n \to \infty$. This map with Lemma \ref{lem:3a} now implies that property 2 above holds. We now verify property 1. Note that
$$\lim_{n \to \infty} \frac{\log K_n}{\log N_n} = \lim_{n \to \infty} \frac{\left\lceil \frac{(\alpha - \epsilon) \log_2 n}{2 - \alpha} \right\rceil \cdot \log 2 + \left( \frac{2\beta - \alpha + \epsilon(1 - \beta)}{2 - \alpha} \right) \log n}{\left\lceil \frac{(\alpha - \epsilon) \log_2 n}{2 - \alpha} \right\rceil\cdot \log 2 + \log n} = \frac{\frac{\alpha - \epsilon}{2 - \alpha} + \frac{2\beta - \alpha + \epsilon(1 - \beta)}{2 - \alpha}}{\frac{\alpha - \epsilon}{2 - \alpha} + 1} = \beta$$
Note that as $n \to \infty$, it follows that since $4^{-\ell_n n^{-\epsilon}} \to 0$,
$$q_n = 1 - e^{-4^{-\ell_n} n^{-\epsilon}} \sim 4^{-\ell_n} n^{-\epsilon}$$
Similarly $p_n \sim 4^{-\ell_n} c n^{-\epsilon}$ and thus $\frac{p_n}{q_n} \to c$. Note that
$$\lim_{n \to \infty} \frac{\log q_n^{-1}}{\log N_n} = \lim_{n \to \infty} \frac{2\left\lceil \frac{(\alpha - \epsilon) \log_2 n}{2 - \alpha} \right\rceil \log 2 + \epsilon \log n}{\left\lceil \frac{(\alpha - \epsilon) \log_2 n}{2 - \alpha} \right\rceil\cdot \log 2 + \log n} = \frac{\frac{2(\alpha - \epsilon)}{2 - \alpha} + \epsilon}{\frac{\alpha - \epsilon}{2 - \alpha} + 1} = \alpha$$
which completes the proof.
\end{proof}

\begin{proof}[Proof of Theorem \ref{thm:gauss-hard}]
If $\beta < 2\alpha$ then PDS is information-theoretically impossible. Thus we may assume that $\beta \ge 2\alpha$. Let $\gamma = \frac{\beta - \alpha}{1 - \alpha}$ and note that $\gamma \in (0, 1/2)$. Now set
$$\ell_n = \left\lceil \frac{\alpha \log_2 n}{1 - \alpha} \right\rceil, \quad \quad k_n = \lceil n^{\gamma} \rceil, \quad \quad N_n = 2^{\ell_n} n \quad \quad K_n = 2^{\ell_n} k_n,$$
$$p_n = \Phi\left( 2^{-\ell_n} \cdot \frac{\log 2}{2 \sqrt{6 \log n + 2\log 2}} \right)$$
By Lemma \ref{lem:6a}, there is a randomized polynomial time algorithm mapping $\text{PC}_D(n, k_n, 1/2)$ to the detection problem $\text{PDS}_D(N_n, K_n, p_n, 1/2)$ with total variation converging to zero as $n \to \infty$. This map with Lemma \ref{lem:3a} now implies that property 2 above holds. We now verify property 1. Note that
$$\lim_{n \to \infty} \frac{\log K_n}{\log N_n} = \lim_{n \to \infty} \frac{\left\lceil \frac{\alpha \log_2 n}{1 - \alpha} \right\rceil \cdot \log 2 + \left( \frac{\beta - \alpha}{1 - \alpha} \right) \log n}{\left\lceil \frac{\alpha \log_2 n}{1 - \alpha} \right\rceil\cdot \log 2 + \log n} = \frac{\frac{\alpha}{1 - \alpha} + \frac{\beta - \alpha}{1 - \alpha}}{\frac{\alpha}{1 - \alpha} + 1} = \beta$$
Let $\mu = \frac{\log 2}{2 \sqrt{6 \log n + 2\log 2}}$ and note that as $n \to \infty$, we have that
$$\lim_{n \to \infty} \frac{\Phi\left( 2^{-\ell_n} \mu \right) - \frac{1}{2}}{2^{-\ell_n} \mu} = \lim_{\tau \to 0} \left( \frac{1}{\tau \sqrt{2\pi}} \int_0^\tau e^{-x^2/2} dx \right) = \frac{1}{\sqrt{2\pi}}$$
Therefore $p_n - q_n \sim \frac{2^{-\ell_n} \mu}{\sqrt{2\pi}}$ as $n \to \infty$. This implies
$$\lim_{n \to \infty} \frac{\log (p_n - q_n)^{-1}}{\log N_n} = \lim_{n \to \infty} \frac{2\left\lceil \frac{\alpha \log_2 n}{1 - \alpha} \right\rceil \cdot \log 2 - \log \mu}{\left\lceil \frac{\alpha \log_2 n}{1 - \alpha} \right\rceil\cdot \log 2 + \log n} = \frac{\frac{2\alpha}{1 - \alpha}}{\frac{\alpha}{1 - \alpha} + 1} = \alpha$$
which completes the proof.
\end{proof}

\begin{proof}[Proof of Lemma \ref{lem:bc}]
Let $\phi = \textsc{BC-Reduction}$ be as in Figure \ref{fig:bc}. Let $\phi_\ell'$ denote $\textsc{Gaussian-Lifting}$ applied to $G$ that outputs $W$ after $\ell$ iterations without applying the thresholding in Step 4 of $\textsc{Distributional-Lifting}$. Lemmas \ref{lem:5cc} and \ref{lem:6b} imply $\textsc{Gaussian-Lifting}$ ensures that
$$\TV\left( \phi_\ell'(G(n, 1/2)), M_{2^{\ell} n}(N(0, 1)) \right) = O\left( \frac{1}{\sqrt{\log n}} \right)$$
Now suppose that $W \sim M_{2^{\ell} n}(N(0, 1))$ and let $W' = \phi_{\text{2-3}}(W)$ denote the value of $W$ after applying Steps 2 and 3 in Figure \ref{fig:bc} to $W$. Note that the diagonal entries of $W$ are i.i.d. $N(0, 1)$ since the diagonal entries of $A$ are zero. If $i < j$, then it follows that
$$W'_{ij} = \frac{1}{\sqrt{2}} \left( W_{ij} + G_{ij} \right) \quad \text{and} \quad W'_{ji} = \frac{1}{\sqrt{2}} \left( W_{ij} - G_{ij} \right)$$
Since $W_{ij}$ and $G_{ij}$ are independent and distributed as $N(0, 1)$, it follows that $W'_{ij}$ and $W'_{ji}$ are jointly Gaussian and uncorrelated, which implies that they are independent. Furthermore, $W'_{ij}$ and $W'_{ji}$ are both in the $\sigma$-algebra $\sigma\{W_{ij}, G_{ij}\}$ and collection of $\sigma$-algebras $\sigma\{W_{ij}, G_{ij}\}$ with $i < j$ is independent. Thus it follows that both $W'$ and $(W')^{\text{id}, \sigma}$ are distributed as $N(0, 1)^{\otimes 2^\ell n \times 2^\ell n}$. It follows by the data processing inequality that
$$\TV\left( \phi(G(n, 1/2)), N(0, 1)^{\otimes 2^\ell n \times 2^\ell n} \right) \le \TV\left( \phi'_{\ell}(G(n, 1/2)), \mL(W) \right) = O\left( \frac{1}{\sqrt{\log n}} \right)$$
Now consider $G \sim G(n, k, 1/2)$ and note that $\textsc{Gaussian-Lifting}$ ensures that
$$\TV\left( \phi_\ell'(G), M_{2^{\ell} n}(2^{\ell} k, N(2^{-\ell}\mu, 1), N(0, 1)) \right) = O\left( \frac{1}{\sqrt{\log n}} \right)$$
Now let $W' \sim M_{2^{\ell} n}(S, N(2^{-\ell}\mu, 1), N(0, 1))$ where $S$ is a subset of $[2^\ell n]$ of size $2^\ell k$ and let $W$ be the matrix formed by applying Steps 1 and 2 above to $W'$ in place of $\phi_\ell'(G)$. By the same jointly Gaussian independence argument above, it follows that the entries of $W$ are independent and distributed as:
\begin{itemize}
\item $W_{ij} \sim N(2^{-\ell-1/2} \mu, 1)$ if $(i, j) \in S \times S$ and $i \neq j$; and
\item $W_{ij} \sim N(0, 1)$ if $(i, j) \not \in S \times S$ or $i = j$.
\end{itemize}
Now consider the matrix $(W')^{\text{id}, \sigma}$ conditioned on the permutation $\sigma$. Its entries are independent and identically distributed to the corresponding entries of $2^{-\ell-1/2} \mu \cdot \mathbf{1}_S \mathbf{1}_T^\top + N(0, 1)^{\otimes 2^\ell n \times 2^\ell n}$ where $T = \sigma(S)$ other than at the indices $(i, \sigma(i))$ for $i \in S$. Marginalizing to condition only on $\sigma(S) = T$ and coupling all entries with indices outside of $S \times T$ yields that
\begin{align*}
\TV\left( \mL((W')^{\text{id}, \sigma} | \sigma(S) = T), \right. &\left. \mL\left( 2^{-\ell-1/2} \mu \cdot \mathbf{1}_S \mathbf{1}_T^\top + N(0, 1)^{\otimes 2^\ell n \times 2^\ell n} \right) \right) \\
&= \TV\left( \mL\left( (W')^{\text{id}, \sigma}[S \times T] | \sigma(S) = T \right), N(2^{-\ell-1/2} \mu, 1)^{\otimes 2^\ell k \times 2^\ell k} \right) \\
&\le \sqrt{\frac{1}{2} \cdot \chi^2\left( N(0, 1), N(2^{-\ell-1/2}\mu, 1) \right)} \\
&\le 2^{-\ell} \mu \le  \frac{\log 2}{2^{\ell + 1} \sqrt{6 \log n + 2\log 2}} = O\left(\frac{1}{2^\ell\sqrt{\log n}} \right)
\end{align*}
by applying Lemma \ref{lem:4a} and the $\chi^2$ upper bound shown in Lemma \ref{lem:6b}. Note that Lemma \ref{lem:4a} applies because $(W')^{\text{id}, \sigma}[S \times T]$ conditioned only on $\sigma(S) = T$ is distributed as a $2^\ell k \times 2^\ell k$ matrix with i.i.d. entries $N(2^{-\ell-1/2} \mu, 1)$, other than its diagonal entries which are i.i.d. $N(0, 1)$, and with its columns randomly permuted. Letting $\sigma(S) = T$ be chosen uniformly at random over all ordered pairs of size $2^\ell k$ subsets of $[2^\ell n]$ yields by the triangle inequality that
$$\TV\left( \mL((W')^{\text{id}, \sigma}), \int \mL\left( 2^{-\ell-1/2} \mu \cdot \mathbf{1}_S \mathbf{1}_T^\top + N(0, 1)^{\otimes 2^\ell n \times 2^\ell n} \right) d\pi'(T) \right) = O\left(\frac{1}{2^\ell\sqrt{\log n}} \right)$$
where $\pi'$ is the uniform distribution on all size $2^\ell k$ subsets of $[2^\ell n]$. Let $\pi(S, T)$ be the uniform distribution on all pairs of subsets of size $2^\ell k$ of $[2^\ell n]$. Taking $S$ to be also chosen uniformly at random yields by the data processing and triangle inequalities that
\begin{align*}
&\TV\left( \mL\left( \phi(G) \right), \int \mL\left(2^{-\ell-1/2} \mu \cdot \mathbf{1}_S \mathbf{1}_T^\top + N(0, 1)^{\otimes 2^\ell n \times 2^\ell n} \right) d\pi(S, T) \right) \\
&\quad \quad \quad \le \TV\left( \phi_\ell'(G), M_{2^{\ell} n}(2^{\ell} k, N(2^{-\ell}\mu, 1), N(0, 1)) \right) \\
&\quad \quad \quad \quad \quad + \bE_S\left[\TV\left( \mL((W')^{\text{id}, \sigma}), \int \mL\left( 2^{-\ell-1/2} \mu \cdot \mathbf{1}_S \mathbf{1}_T^\top + N(0, 1)^{\otimes 2^\ell n \times 2^\ell n} \right) d\pi'(T) \right) \right] \\
&\quad \quad \quad = O\left(\frac{1}{\sqrt{\log n}} \right)
\end{align*}
which completes the proof of the lemma.
\end{proof}

\begin{proof}[Proof of Lemma \ref{lem:gpds}]
Let $\phi = \textsc{General-PDS-Reduction}$ be as in Figure \ref{fig:bc}. Let $\phi_1(G)$ denote the map in Step 1 applying $\textsc{Gaussian-Lifting}$ for $\ell_1$ iterations and let $\phi_2(G)$ denote the map in Step 2 applying the modified version of $\textsc{Poisson-Lifting}$ for $\ell_2$ iterations. If $G \sim G(n, k, 1/2)$ then Lemma \ref{lem:6b} implies that
$$\TV\left( \phi_1(G), G\left(2^{\ell_1} n, 2^{\ell_1} k, \Phi\left( 2^{-\ell_1} \mu \right), 1/2 \right) \right) = O\left(\frac{1}{\sqrt{\log n}} \right)$$
For the values of $\lambda, \rho$ and $c$ in Step 2, we have that $c < 2^{1/4}$ and $\log_c (1 + 2\rho) = 4\epsilon^{-1} = O(1)$. Also observe that as $n \to \infty$,
$$\rho = \Phi\left(2^{-\ell_1} \mu\right) - 1/2 \sim \frac{1}{\sqrt{2\pi}} \cdot 2^{-\ell_1} \mu = \Omega \left( \frac{1}{2^{\ell_1} \sqrt{\log n}} \right) = \omega \left( \frac{1}{2^{\ell_1} n} \right)$$
Therefore the values $K = 1, \epsilon, \lambda, \rho, c$ and natural parameter $2^{\ell_1} n$ satisfy the preconditions to apply Lemma \ref{lem:5c}. It follows that
\begin{align*}
\TV\left(\textsc{rk}_{\text{P2}}(\text{Bern}(1/2 + \rho)), \text{Pois}(c\lambda) \right) &\le O\left(2^{-3\ell_1}n^{-3}\right), \quad \text{and} \\
\TV\left(\textsc{rk}_{\text{P2}}(\text{Bern}(1/2)), \text{Pois}(\lambda) \right) &\le O\left(2^{-3\ell_1}n^{-3}\right)
\end{align*}
Now let $H'$ be a sample from $G\left(2^{\ell_1} n, 2^{\ell_1} k, \Phi\left( 2^{-\ell_1} \mu \right), 1/2 \right)$. The argument in Lemma \ref{lem:6a} applied with these total variation bounds yields that
$$\TV\left( \phi_2(H'), G\left(2^\ell n, 2^\ell k, p_{\ell_1, \ell_2}, q_{\ell_1, \ell_2} \right) \right) = O\left(2^{-\ell_1 \epsilon/2} n^{-\epsilon/2} \right)$$
Applying the triangle inequality and data processing inequality yields that
\begin{align*}
&\TV\left( \phi(G(n, k, 1/2)), G\left(2^\ell n, 2^\ell k, p_{\ell_1, \ell_2}, q_{\ell_1, \ell_2} \right) \right) \\
&\quad \quad \quad \quad \quad \quad \le \TV\left( \phi_2 \circ \phi_1(G), \phi_2\left( H' \right) \right) + \TV\left( \phi_2(H'), G\left(2^\ell n, 2^\ell k, p_{\ell_1, \ell_2}, q_{\ell_1, \ell_2} \right) \right) \\
&\quad \quad \quad \quad \quad \quad = O\left(\frac{1}{\sqrt{\log n}} + 2^{-\ell_1 \epsilon/2} n^{-\epsilon/2} \right) = O\left( \frac{1}{\sqrt{\log n}} \right)
\end{align*}
By the same argument, if $G \sim G(n, 1/2)$ then
$$\TV\left( \phi(G(n, 1/2)), G\left(2^\ell n, q_{\ell_1, \ell_2} \right) \right) = O\left( \frac{1}{\sqrt{\log n}} \right)$$
which completes the proof of the lemma.
\end{proof}

\begin{proof}[Proof of Theorem \ref{thm:pdsgenhard}]
If $\beta < 2\gamma - \alpha$ then PDS in this regime is information-theoretically impossible. Thus we may assume that $1 > \beta \ge 2\gamma - \alpha$. Now let
$$\eta =1 - ( 1- \beta) \cdot \frac{2 - \epsilon}{2 - \alpha - (2 - \epsilon) ( \gamma - \alpha)}$$
Note that the given condition on $\alpha, \beta$ and $\gamma$ rearranges to $\frac{1 - \beta}{2 - \alpha - 2(\gamma - \alpha)} > \frac{1}{4}$. Therefore taking $\epsilon > 0$ to be small enough ensures that $2 - \alpha - (2 - \epsilon)(\gamma - \alpha) > 0$, $\eta \in (0, 1/2)$ and $\alpha > \epsilon$. Now set
$$\ell_n^1 = \left\lceil \frac{(\gamma - \alpha)(2 - \epsilon) \log_2 n}{2 - \alpha - (2 - \epsilon) ( \gamma - \alpha)} \right\rceil, \quad \quad \ell_n^2 = \left\lceil \frac{(\alpha - \epsilon)\log_2 n}{2 - \alpha - (2 - \epsilon) ( \gamma - \alpha)} \right\rceil,$$
$$k_n = \lceil n^{\eta} \rceil, \quad \quad N_n = 2^{\ell^1_n + \ell_n^2} n, \quad \quad K_n = 2^{\ell_n^1 + \ell_n^2} k_n,$$
$$p_n = 1 - \exp\left( 4^{-\ell^2_n} \left(2^{\ell^1_n}n\right)^{-\epsilon} \cdot \left( 2 \Phi\left(2^{-\ell^1_n} \mu \right) \right)^{\epsilon/4} \right), \quad \quad q_n = 1 - \exp\left( 4^{-\ell^2_n} \left(2^{\ell^1_n} n\right)^{-\epsilon} \right)$$
where $\mu = \frac{\log 2}{2 \sqrt{6 \log n + 2\log 2}}$. By Lemma \ref{lem:gpds}, there is a randomized polynomial time algorithm mapping $\text{PC}_D(n, k_n, 1/2)$ to $\text{PDS}_D(N_n, K_n, p_n, q_n)$ with total variation converging to zero as $n \to \infty$. This map with Lemma \ref{lem:3a} now implies that property 2 above holds. We now verify property 1. Note that
\begin{align*}
\lim_{n \to \infty} \frac{\log K_n}{\log N_n} &= \frac{\frac{(\gamma - \alpha)(2 - \epsilon)}{2 - \alpha - (2 - \epsilon) ( \gamma - \alpha)} + \frac{\alpha - \epsilon}{2 - \alpha - (2 - \epsilon) ( \gamma - \alpha)} + 1 - \frac{(1 - \beta)(2 - \epsilon)}{2 - \alpha - (2 - \epsilon) ( \gamma - \alpha)}}{\frac{(\gamma - \alpha)(2 - \epsilon)}{2 - \alpha - (2 - \epsilon) ( \gamma - \alpha)} + \frac{\alpha - \epsilon}{2 - \alpha - (2 - \epsilon) ( \gamma - \alpha)} + 1} \\
&= \frac{\frac{2 - \epsilon}{2 - \alpha - (2 - \epsilon) ( \gamma - \alpha)} - \frac{(1 - \beta)(2 - \epsilon)}{2 - \alpha - (2 - \epsilon) ( \gamma - \alpha)}}{\frac{2 - \epsilon}{2 - \alpha - (2 - \epsilon) ( \gamma - \alpha)}} = \beta
\end{align*}
Using the approximations in the proof of Theorem \ref{thm:gauss-hard}, we obtain as $n \to \infty$
\begin{align*}
q_n &\sim 4^{-\ell_n^2} \left( 2^{\ell^1_n} n \right)^{-\epsilon} \\
p_n - q_n &\sim 4^{-\ell_n^2} \left( 2^{\ell^1_n} n \right)^{-\epsilon} \left[ \left( 2\Phi\left( 2^{-\ell_n^1} \mu \right) \right)^{\epsilon/4} - 1\right] \sim 4^{-\ell_n^2} \left( 2^{\ell^1_n} n \right)^{-\epsilon} \cdot \frac{\epsilon}{2 \sqrt{2\pi}} \cdot 2^{-\ell^1_n} \mu
\end{align*}
Now it follows that
\begin{align*}
\lim_{n \to \infty} \frac{\log q_n^{-1}}{\log N_n} &= \frac{2 \cdot \frac{\alpha - \epsilon}{2 - \alpha - (2 - \epsilon) ( \gamma - \alpha)} + \epsilon \cdot \frac{(\gamma - \alpha)(2 - \epsilon)}{2 - \alpha - (2 - \epsilon) ( \gamma - \alpha)} + \epsilon}{\frac{2 - \epsilon}{2 - \alpha - (2 - \epsilon) ( \gamma - \alpha)}} = \alpha \\
\lim_{n \to \infty} \frac{\log (p_n - q_n)^{-1}}{\log N_n} &= \frac{2 \cdot \frac{\alpha - \epsilon}{2 - \alpha - (2 - \epsilon) ( \gamma - \alpha)} + (1 + \epsilon) \cdot \frac{(\gamma - \alpha)(2 - \epsilon)}{2 - \alpha - (2 - \epsilon) ( \gamma - \alpha)} + \epsilon}{\frac{2 - \epsilon}{2 - \alpha - (2 - \epsilon) ( \gamma - \alpha)}} = \gamma
\end{align*}
which completes the proof.
\end{proof}

\section{Deferred Proofs from Section 7}
\label{app7}

\begin{proof}[Proof of Lemma \ref{lem:ros}]
Let $\phi = \textsc{ROS-Reduction}$ be as in Figure \ref{fig:ros}. Let $\phi_1 : \mG_n \to \mathbb{R}^{n \times n}$ and $\phi_2 : \mathbb{R}^{n \times n} \to \mathbb{R}^{n \times n}$ denote the maps in Steps 1 and 2, respectively. By Lemma \ref{lem:bc}, it holds that
\begin{align*}
&\TV\left( \phi_1(G(n, 1/2)), N(0, 1)^{\otimes n \times n} \right) = O\left(\frac{1}{\sqrt{\log n}} \right) \\
&\TV\left( \phi_1(G(n, k, 1/2)), \int \mL\left(\frac{\mu}{\sqrt{2}} \cdot \mathbf{1}_S \mathbf{1}_T^\top + N(0, 1)^{\otimes n \times n} \right) d\pi'(S, T) \right) = O\left(\frac{1}{\sqrt{\log n}} \right) 
\end{align*}
where $\pi'$ is the uniform distribution over pairs of $k$-subsets $S, T \subseteq [n]$. By Lemma \ref{lem:refc}, it holds that $\phi_2\left( N(0, 1)^{\otimes n \times n} \right) \sim N(0, 1)^{\otimes n \times n}$. By the data processing inequality, we have that
\begin{align*}
\TV\left( \mL_{H_0}(\phi(G)), N(0, 1)^{\otimes n \times n} \right) &= \TV\left( \phi_2 \circ \phi_1(G(n, 1/2)), \phi_2 \left( N(0, 1)^{\otimes n \times n} \right) \right) \\
&\le \TV\left( \phi_1(G(n, 1/2)), N(0, 1)^{\otimes n \times n} \right) = O\left(\frac{1}{\sqrt{\log n}} \right)
\end{align*}
Let $M$ be a the matrix distributed as $\frac{\mu}{\sqrt{2}} \cdot \mathbf{1}_S \mathbf{1}_T^\top + N(0, 1)^{\otimes n \times n}$ where $S$ and $T$ are $k$-element subsets of $[n]$ chosen uniformly at random. It follows that the distribution of $\phi_2(M)$ conditioned on the sets $S$ and $T$ is given by
$$\mL\left( \phi_2\left( M \right) | S, T \right) \sim \int \mL\left( \frac{\mu}{2^\ell \sqrt{2}} \cdot rc^\top + N(0, 1)^{\otimes n \times n} \right) d\pi(r, c)$$
where $\pi$ is the prior in Lemma \ref{lem:refc}. As shown in Lemma \ref{lem:refc}, it follows that each pair $(r, c)$ in the support of $\pi$ satisfies that $\| r \|_2^2 = 2^\ell \| \mathbf{1}_S \|_2^2 = 2^\ell k$ and $\| c \|_2^2 = 2^\ell \| \mathbf{1}_T \|_2^2 = 2^\ell k$ and that $\| r \|_0, \| c \|_0 \le 2^\ell k$. Now consider the prior $\pi_{S, T}$ which is $\pi$ conditioned on the event that the following inequalities hold
$$\| r \|_0, \| c \|_0 \ge 2^\ell k \left( 1 - \max\left(\frac{2C\ell \cdot \log (2^\ell k)}{k}, \frac{2^\ell k}{n}\right) \right)$$
As shown in Lemma \ref{lem:refc}, this event occurs with probability at least $1 - 8/k$. Since $2^{\ell} k < \frac{n}{\log k}$, it follows that if $u = \frac{1}{\sqrt{2^\ell k}} \cdot r$ and $v = \frac{1}{\sqrt{2^\ell k}} \cdot c$ then $\pi_{S, T}$ induces a prior over pairs $(u, v)$ in $\mathcal{V}_{n, 2^\ell k}$. This follows from the fact that $r, c \in \mathbb{Z}^n$ implies that $u$ and $v$ have nonzero entries with magnitudes at least $1/\sqrt{2^\ell k}$. Applying Lemma \ref{lem:5tv} yields that since $2^{-\ell} \mu \cdot r c^\top = \mu k \cdot uv^\top$,
$$\TV\left( \mL\left( \phi_2\left( M \right) | S, T \right), \int \mL\left( \frac{\mu k}{\sqrt{2}} \cdot uv^\top + N(0, 1)^{\otimes n \times n} \right) d\pi_{S, T}(u, v)\right) \le \frac{8}{k}$$
Let $\pi(u, v) = \bE_{S, T}[\pi_{S,T}(u, v)]$ be the prior formed by marginalizing over $S$ and $T$. Note that $\pi$ is also supported on pairs of unit vectors in $\mathcal{V}_{n, 2^\ell k}$. By the triangle inequality, it follows that
\begin{align*}
&\TV\left( \mL\left( \phi_2\left( M \right) \right), \int \mL\left( \frac{\mu k}{\sqrt{2}} \cdot uv^\top + N(0, 1)^{\otimes n \times n} \right) d\pi(u, v) \right) \\
&\quad \le \bE_{S, T} \left[ \TV\left( \mL\left( \phi_2\left( M \right) | S, T \right), \int \mL\left( \frac{\mu k}{\sqrt{2}} \cdot uv^\top + N(0, 1)^{\otimes n \times n} \right) d\pi_{S, T}(u, v)\right) \right] \le \frac{8}{k}
\end{align*}
By the triangle inequality and data processing inequality, we now have that
\begin{align*}
&\TV\left( \mL_{H_1}(\phi(G)), \int \mL\left( \frac{\mu k}{\sqrt{2}} \cdot  uv^\top + N(0, 1)^{\otimes n \times n} \right) d\pi(u, v) \right) \\
&\quad \le \TV\left( \mL_{H_1}(\phi_2 \circ \phi_1(G)), \mL(\phi_2(M)) \right) + \TV\left( \mL\left( \phi_2\left( M \right) \right), \int \mL\left( \frac{\mu k}{\sqrt{2}} \cdot uv^\top + N(0, 1)^{\otimes n \times n} \right) d\pi(u, v) \right) \\
&\quad = O\left(\frac{1}{\sqrt{\log n}} \right) + \frac{8}{k} = O\left(\frac{1}{\sqrt{\log n}} + k^{-1} \right)
\end{align*}
since $M$ is a sample from the mixture $\int \mL\left(\frac{\mu}{\sqrt{2}} \cdot \mathbf{1}_S \mathbf{1}_T^\top + N(0, 1)^{\otimes n \times n} \right) d\pi'(S, T)$. This completes the proof of the lemma.
\end{proof}

\begin{proof}[Proof of Theorem \ref{thm:ros}]
If $\beta < 2\alpha$ then $\textsc{ROS}_D$ is information-theoretically impossible. Thus we may assume that $\beta \ge 2\alpha$. Let $\gamma = \beta - \alpha$ and note that $\gamma \in (0, 1/2)$. Now set
$$\ell_n = \lceil \alpha \log_2 n \rceil, \quad \quad k_n = \lceil n^{\gamma} \rceil, \quad \quad N_n = 2n, \quad \quad K_n = 2^{\ell_n} k_n, \quad \quad \mu_n =\frac{\mu k_n}{\sqrt{2}}$$
where $\mu = \frac{\log 2}{2 \sqrt{6 \log n + 2\log 2}}$. By Lemma \ref{lem:ros}, there is a randomized polynomial time algorithm mapping $\text{PC}_D(2n, k_n, 1/2)$ to the detection problem $\text{ROS}_D(N_n, K_n, \mu_n)$ under $H_0$ and to a prior over $H_1$ with total variation converging to zero as $n \to \infty$. This map with Lemma \ref{lem:3a} now implies that property 2 above holds. We now verify property 1. Note that
\begin{align*}
\lim_{n \to \infty} \frac{\log (K_n \mu_n^{-1})}{\log N_n} &= \lim_{n \to \infty} \frac{\lceil \alpha \log_2 n \rceil \cdot \log 2 - \log (\mu/\sqrt{2})}{\log n + \log 2} = \alpha \\
\lim_{n \to \infty} \frac{\log K_n}{\log N_n} &= \lim_{n \to \infty} \frac{\lceil \alpha \log_2 n \rceil \cdot \log 2 + \log k_n}{\log n + \log 2} = \alpha + (\beta - \alpha) = \beta
\end{align*}
which completes the proof.
\end{proof}

\begin{proof}[Proof of Theorem \ref{thm:ssw}]
When $2 \alpha > \beta$, it holds that $\textsc{SROS}_D$ is information-theoretically impossible. Therefore we may assume that $2\alpha \le \beta$ and in particular that $\alpha < \frac{1}{2}$. Now suppose that $\beta < \frac{1}{2} + \alpha$. First consider the case when $\beta \ge \frac{1}{2}$ and let $\epsilon = \frac{1}{2} \left( \alpha + \frac{1}{2} - \beta \right) \in \left( 0, \frac{1}{2} \right)$. Now set
$$\ell_n = \left\lceil \left( \beta - \frac{1}{2} + \epsilon \right) \log_2 n \right\rceil, \quad \quad k_n = \left\lceil n^{\frac{1}{2} - \epsilon} \right\rceil, \quad \quad N_n = 2n, \quad \quad K_n = 2^{\ell_n} k_n, \quad \quad \mu_n =\frac{\mu k_n(k_n-1)}{2\sqrt{n - 1}}$$
where $\mu = \frac{\log 2}{2 \sqrt{6 \log n + 2\log 2}}$. By Lemma \ref{lem:ssw}, there is a randomized polynomial time algorithm mapping $\text{PC}_D(2n, k_n, 1/2)$ to the detection problem $\text{SROS}_D(N_n, K_n, \mu_n)$ under $H_0$ and to a prior over $H_1$ with total variation converging to zero as $n \to \infty$. This map with Lemma \ref{lem:3a} now implies that property 2 above holds. We now verify property 1. Note that
\begin{align*}
\lim_{n \to \infty} \frac{\log (K_n \mu_n^{-1})}{\log N_n} &= \lim_{n \to \infty} \frac{\left\lceil \left( \beta - \frac{1}{2} + \epsilon \right) \log_2 n \right\rceil \cdot \log 2 - \log(k_n - 1) - \log (\mu/2) + \frac{1}{2} \log(n - 1)}{\log n + \log 2} \\
&=  (\beta - \frac{1}{2} + \epsilon) - \left( \frac{1}{2} - \epsilon \right) + \frac{1}{2} = \alpha \\
\lim_{n \to \infty} \frac{\log K_n}{\log N_n} &= \lim_{n \to \infty} \frac{\left\lceil \left( \beta - \frac{1}{2} + \epsilon \right) \log_2 n \right\rceil \cdot \log 2 + \log k_n}{\log n + \log 2} = \left( \beta - \frac{1}{2} + \epsilon \right) + \frac{1}{2} - \epsilon = \beta
\end{align*}
Now consider the case when $\beta < \frac{1}{2}$ and $\alpha > 0$. In this case, set $\ell_n = 0$, $k_n = \left\lceil n^{\frac{1}{2} - \epsilon} \right\rceil$, $K_n = k_n$ and
$$N_n = 2\left\lceil n^{\frac{1}{\beta} \left( \frac{1}{2} - \epsilon \right)}\right\rceil \quad \text{and} \quad \mu_n =\frac{\mu' k_n(k_n-1)}{2\sqrt{n - 1}}$$
where $\epsilon = \min \left( \frac{\alpha}{2(\alpha + \beta)}, \frac{1}{2} - \beta \right)$ and
$$\mu' = \frac{\log 2}{2 \sqrt{6 \log n + 2\log 2}} \cdot n^{\epsilon - \frac{\alpha}{\beta} \left( \frac{1}{2} - \epsilon \right)}$$
Note that $\epsilon \le \frac{\alpha}{2(\alpha + \beta)}$ implies that $\epsilon - \frac{\alpha}{\beta} \left( \frac{1}{2} - \epsilon \right) \le 0$. By Lemma \ref{lem:ssw}, there is a randomized polynomial time algorithm mapping $\text{PC}_D(2n, k_n, 1/2)$ to the detection problem $\text{SROS}_D(2n, K_n, \mu_n)$ under $H_0$ and to a prior over $H_1$ with total variation converging to zero as $n \to \infty$. Now consider the map that pads the resulting instance with i.i.d. $N(0, 1)$ random variables until it is $N_n \times N_n$. Note that since $\epsilon \le \frac{1}{2} - \beta$, we have that $N_n \ge 2n$. By the data processing and triangle inequalities, it follows that this map takes $\text{PC}_D(2n, k_n, 1/2)$ to $\text{SSW}_D(N_n, K_n, \mu_n)$ with total variation converging to zero as $n \to \infty$. This implies that property 2 above holds and we now verify property 1. Note
\begin{align*}
\lim_{n \to \infty} \frac{\log (K_n \mu_n^{-1})}{\log N_n} &= \lim_{n \to \infty} \frac{\log 2 + \frac{1}{2} \log(n - 1) - \log(k_n - 1) - \log \mu'}{\log 2 + \frac{1}{\beta} \left( \frac{1}{2} - \epsilon \right) \log n} \\
&= \frac{\frac{1}{2} - \left( \frac{1}{2} - \epsilon \right) - \left( \epsilon - \frac{\alpha}{\beta} \left( \frac{1}{2} - \epsilon \right) \right)}{\frac{1}{\beta} \left( \frac{1}{2} - \epsilon \right)} = \alpha \\
\lim_{n \to \infty} \frac{\log K_n}{\log N_n} &= \lim_{n \to \infty} \frac{\left(\frac{1}{2} - \epsilon \right) \log n}{\log 2 + \frac{1}{\beta} \left( \frac{1}{2} - \epsilon \right) \log n} = \beta
\end{align*}
which completes the proof of the theorem.
\end{proof}

\begin{proof}[Proof of Theorem \ref{thm:SSBMguar}]
When $2 \alpha > \beta$, it holds that $\textsc{SSBM}_D$ is information-theoretically impossible. Therefore we may assume that $2\alpha \le \beta$ and in particular that $\alpha < \frac{1}{2}$. Now suppose that $\beta < \frac{1}{2} + \alpha$. We consider the cases $\beta \ge \frac{1}{2}$ and $\beta < \frac{1}{2}$ and $q \le \frac{1}{2}$ and $q > \frac{1}{2}$, separately. First consider the case when $\beta \ge \frac{1}{2}$ and $q \le \frac{1}{2}$. Let $\epsilon = \frac{1}{2} \left( \alpha + \frac{1}{2} - \beta \right) \in \left( 0, \frac{1}{2} \right)$ and set parameters similarly to Theorem \ref{thm:ssw} with
$$\ell_n = \left\lceil \left( \beta - \frac{1}{2} + \epsilon \right) \log_2 n \right\rceil, \quad \quad k_n = \left\lceil n^{\frac{1}{2} - \epsilon} \right\rceil, \quad \quad N_n = 2n, \quad \quad K_n = 2^{\ell_n} k_n, \quad \quad q_n = q$$
$$\rho_n = 2q \cdot \Phi\left( \frac{\mu(k_n - 1)}{2^{\ell_n + 1}\sqrt{n- 1}} \right) - q$$
where $\mu = \frac{\log 2}{2 \sqrt{6 \log n + 2\log 2}}$. By Lemma \ref{lem:ssbm}, there is a randomized polynomial time algorithm mapping $\text{PC}_D(2n, k_n, 1/2)$ to the detection problem $\text{SSBM}_D(N_n, K_n, 1/2, (2q)^{-1} \rho_n)$ under $H_0$ and to a prior over $H_1$ with total variation converging to zero as $n \to \infty$. Now consider the algorithm that post-processes the resulting graph by keeping each edge independently with probability $2q$. This maps any instance of $\text{SSBM}_D(N_n, K_n, 1/2, (2q)^{-1} \rho_n)$ exactly to an instance of $\text{SSBM}_D(N_n, K_n, q, \rho_n)$. Therefore the data processing inequality implies that these two steps together yield a reduction that combined with Lemma \ref{lem:3a} implies property 2 holds. Now observe that as $n \to \infty$,
$$\rho_n = 2q \cdot \Phi\left( \frac{\mu(k_n - 1)}{2^{\ell_n + 1}\sqrt{n- 1}} \right) - q \sim 2q \cdot \frac{1}{\sqrt{2\pi}} \cdot \frac{\mu(k_n - 1)}{2^{\ell_n + 1}\sqrt{n- 1}}$$
The same limit computations as in Theorem \ref{thm:ssw} show that property 1 above holds. If $q > 1/2$, then instead set
$$\rho_n = 2(1 - q) \cdot \Phi\left( \frac{\mu(k_n - 1)}{2^{\ell_n + 1}\sqrt{n- 1}} \right) - (1 - q)$$
and post-process the graph resulting from the reduction in Lemma \ref{lem:ssbm} by adding each absent edge with probability $2q - 1$. By a similar argument, the resulting reduction shows properties 1 and 2.

Now consider the case when $\beta < \frac{1}{2}$, $\alpha > 0$ and $q \le \frac{1}{2}$. Set $\ell_n = 0$, $k_n = \left\lceil n^{\frac{1}{2} - \epsilon} \right\rceil$, $K_n = k_n$ and
$$N_n = 2\left\lceil n^{\frac{1}{\beta} \left( \frac{1}{2} - \epsilon \right)}\right\rceil \quad \text{and} \quad \rho_n = 2q \cdot \Phi\left( \frac{\mu'(k - 1)}{2\sqrt{n- 1}} \right) - q$$
where $\epsilon = \min \left( \frac{\alpha}{2(\alpha + \beta)}, \frac{1}{2} - \beta \right)$ and $\mu' = \frac{\log 2}{2 \sqrt{6 \log n + 2\log 2}} \cdot n^{\epsilon - \frac{\alpha}{\beta} \left( \frac{1}{2} - \epsilon \right)}$ as in the proof of Theorem \ref{thm:ssw}. Note that since $\epsilon \le \frac{1}{2} - \beta$, it follows that $N_n \ge 2n$. By Lemma \ref{lem:ssbm}, there is a randomized polynomial time algorithm mapping $\text{PC}_D(2n, k_n, 1/2)$ to the detection problem $\text{SSBM}_D(2n, K_n, 1/2, (2q)^{-1} \rho_n)$ under $H_0$ and to a prior over $H_1$ with total variation converging to zero as $n \to \infty$. Now consider the map that post-processes the graph by:
\begin{enumerate}
\item keeping each edge independently with probability $2q$;
\item adding $N_n - 2n$ vertices to the resulting graph and includes each edge incident to these new vertices independently with probability $q$; and
\item randomly permuting the vertices of the resulting $N_n$-vertex graph.
\end{enumerate}
Note that this maps any instance of $\text{SSBM}_D(2n, K_n, 1/2, (2q)^{-1} \rho_n)$ exactly to an instance of $\text{SSBM}_D(N_n, K_n, q, \rho_n)$. Thus the data processing inequality implies that this post-processing together with the reduction of Lemma \ref{lem:ssbm} yields a reduction showing property 2. The same limit computations as in Theorem \ref{thm:ssw} show that property 1 above holds. The same adaptation as in the case $\beta \ge \frac{1}{2}$ also handles $q < \frac{1}{2}$. This completes the proof of the theorem.
\end{proof}

\section{Deferred Proofs from Section 8}
\label{app8}

\begin{proof}[Proof of Lemma \ref{lem:hsspca}]
Let $M'$ be the matrix output in Step 2. Under $H_0$, Lemma \ref{lem:ros} implies that $M \sim N(0, 1)^{\otimes n \times n}$ and Lemma \ref{lem:randrot} implies that we also have $M' \sim N(0, 1)^{\otimes n \times n}$. It suffices to consider the case of $H_1$. By Lemma \ref{lem:ros}, 
$$\TV\left( \mL_{H_1}(M), \int \mL\left( \frac{\mu k}{\sqrt{2}} \cdot  uv^\top + N(0, 1)^{\otimes n \times n} \right) d\pi'(u, v) \right) = O\left( \frac{1}{\sqrt{\log n}} + k^{-1} \right)$$
where $\pi'(u, v)$ is a prior supported on pairs of unit vectors in $\mathcal{V}_{n, 2^\ell k}$. Let $W \sim \frac{\mu k}{\sqrt{2}} \cdot  uv^\top + N(0, 1)^{\otimes n \times n}$ where $(u, v)$ is distributed according to $\pi'$ and let $\varphi$ denote the map in Step 2 taking $A$ to $\textsc{Random-Rotation}(A, \tau)$. Conditioning on $(u, v)$ yields by Lemma \ref{lem:randrot} that
$$\TV\left( \mL\left( \varphi(W) | u, v \right), N\left(0, I_n + \frac{\mu^2 k^2}{2 \tau n} \cdot uu^\top\right)^{\otimes n} \right) \le \frac{2(n + 3)}{\tau n - n - 3}$$
Now consider the measure $\pi(u) = \bE_{v} \pi'(u, v)$ given by marginalizing over $v$ in $\pi'$. The triangle inequality implies that
\begin{align*}
&\TV\left( \mL(\varphi(W)), \int N\left(0, I_n + \frac{\mu^2 k^2}{2 \tau n} \cdot uu^\top\right)^{\otimes n} d\pi(u) \right) \\
&\quad \quad \le \bE_{u, v} \left[ \TV\left( \mL\left( \varphi(W) | u, v \right), N\left(0, I_n + \frac{\mu^2 k^2}{2 \tau n} \cdot uu^\top\right)^{\otimes n} \right) \right] \le \frac{2(n + 3)}{\tau n - n - 3}
\end{align*}
By the data processing inequality and triangle inequality, we now have that
\begin{align*}
&\TV\left( \mL_{H_1}(\phi(G)), \int N\left(0, I_n + \frac{\mu^2 k^2}{2 \tau n} \cdot uu^\top\right)^{\otimes n} d\pi(u) \right) \\
&\quad \quad \le \TV\left( \mL_{H_1}(M), \mL(W) \right) + \TV\left( \mL(\varphi(W)), \int N\left(0, I_n + \frac{\mu^2 k^2}{2 \tau n} \cdot uu^\top\right)^{\otimes n} d\pi(u) \right) \\
&\quad \quad \le \frac{2(n + 3)}{\tau n - n - 3} + O\left( \frac{1}{\sqrt{\log n}} + k^{-1} \right)
\end{align*}
since $\mL_{H_1}(\phi(G)) \sim \mL_{H_1}(\varphi(M))$. This completes the proof of the lemma.
\end{proof}

\begin{proof}[Proof of Theorem \ref{thm:spca}]
Note that if $\alpha \ge 1$, then sparse PCA is information theoretically impossible. Thus we may assume that $\alpha \in (0, 1)$ and $\beta > \frac{1 - \alpha}{2}$. Let $\gamma = \frac{1 - \alpha}{2} \in (0, 1/2)$. Now set $N_n = d_n = n$
$$\ell_n = \left\lceil \left( \beta - \frac{1 - \alpha}{2} \right) \log_2 n \right\rceil, \quad \quad k_n = \lceil n^{\gamma} \rceil, \quad \quad K_n = 2^{\ell_n} k_n, \quad \quad \theta_n = \frac{\mu^2 k_n^2}{2 \tau n}$$
where $\mu =  \frac{\log 2}{2 \sqrt{6 \log n + 2\log 2}}$ and $\tau$ is a sub-polynomially growing function of $n$. 
By Lemma \ref{lem:hsspca}, there is a randomized polynomial time algorithm mapping $\text{PC}_D(n, k_n, 1/2)$ to the detection problem $\text{SPCA}_D(N_n, K_n, d_n, \theta_n)$, exactly under $H_0$ and to a prior over $H_1$, with total variation converging to zero as $n \to \infty$. This map with Lemma \ref{lem:3a} now implies that property 2 above holds. We now verify property 1. Note that
$$\lim_{n \to \infty} \frac{\log K_n}{\log N_n} = \lim_{n \to \infty} \frac{\left\lceil \left( \beta - \frac{1 - \alpha}{2} \right) \log_2 n \right\rceil \cdot \log 2 + \left( \frac{1 - \alpha}{2} \right) \log n}{\log n} = \beta$$
$$\lim_{n \to \infty} \frac{\log \theta_n^{-1}}{\log N_n} = \lim_{n \to \infty} \frac{(1 - 2\gamma) \log n - 2 \log \mu + \log(2\tau)}{\log n} = \alpha$$
which completes the proof.
\end{proof}

\begin{proof}[Proof of Theorem \ref{thm:uspca}]
Note that if $\alpha \ge 1$, then $\text{USPCA}_D$ and $\textsc{UBSPCA}_D$ are information theoretically impossible. Thus we may assume that $\alpha \in (0, 1)$. Now observe that $\beta \in (0, 1)$ and $\frac{1 - \alpha}{2} < \beta < \frac{1 + \alpha}{2}$ imply that $\gamma = \frac{1 - \alpha}{3 - \alpha - 2\beta} \in (0, 1/2)$, that $\alpha + 2\beta - 1 > 0$ and $3 - \alpha - 2\beta > 0$. Now set
$$N_n = d_n = 2^{\ell_n} n, \quad \quad \ell_n = \left\lceil \left( \frac{\alpha + 2\beta - 1}{3 - \alpha - 2\beta} \right) \log_2 n \right\rceil, \quad \quad k_n = \lceil n^{\gamma} \rceil,$$
$$K_n = 2^{\ell_n} k_n, \quad \quad \theta_n = \frac{\mu^2 k_n^2}{2^{\ell_n + 1} \tau n}$$
where $\mu =  \frac{\log 2}{2 \sqrt{6 \log n + 2\log 2}}$ and $\tau$ is a sub-polynomially growing function of $n$. 
By Lemma \ref{lem:lsspca}, there is a randomized polynomial time algorithm mapping $\text{PC}_D(n, k_n, 1/2)$ to the detection problem $\textsc{UBSPCA}_D(N_n, K_n, d_n, \theta_n)$, exactly under $H_0$ and to a prior over $H_1$, with total variation converging to zero as $n \to \infty$. This map with Lemma \ref{lem:3a} now implies that property 2 above holds. We now verify property 1. Note that
\begin{align*}
\lim_{n \to \infty} \frac{\log K_n}{\log N_n} &= \lim_{n \to \infty} \frac{\left\lceil \left( \frac{\alpha + 2\beta - 1}{3 - \alpha - 2\beta} \right) \log_2 n \right\rceil \cdot \log 2 + \left( \frac{1 - \alpha}{3 - \alpha - 2\beta} \right) \log n}{\left\lceil \left( \frac{\alpha + 2\beta - 1}{3 - \alpha - 2\beta} \right) \log_2 n \right\rceil \cdot \log 2 + \log n} = \frac{\frac{\alpha + 2\beta - 1}{3 - \alpha - 2\beta} + \frac{1 - \alpha}{3 - \alpha - 2\beta}}{\frac{\alpha + 2\beta - 1}{3 - \alpha - 2\beta} + 1} = \beta \\
\lim_{n \to \infty} \frac{\log \theta_n^{-1}}{\log N_n} &= \lim_{n \to \infty} \frac{\left\lceil \left( \frac{\alpha + 2\beta - 1}{3 - \alpha - 2\beta} \right) \log_2 n \right\rceil \cdot \log 2 + (1 - 2\gamma) \log n - 2 \log \mu + \log(2\tau)}{\left\lceil \left( \frac{\alpha + 2\beta - 1}{3 - \alpha - 2\beta} \right) \log_2 n \right\rceil \cdot \log 2 + \log n} \\
&= \frac{\frac{\alpha + 2\beta - 1}{3 - \alpha - 2\beta} -2 \cdot \frac{1 - \alpha}{3 - \alpha - 2\beta} + 1}{\frac{\alpha + 2\beta - 1}{3 - \alpha - 2\beta} + 1} = \frac{\alpha + 2\beta - 1 - 2(1 - \alpha) + 3 - \alpha - 2\beta}{2}= \alpha
\end{align*}
which completes the proof. As described previously, the corresponding lower bound for $\textsc{USPCA}_D$ follows by randomly signing the rows of the data matrix of the resulting $\textsc{UBSPCA}_D$ instance.
\end{proof}

\section{Deferred Proofs from Section 9}
\label{app9}

\begin{proof}[Proof of Theorem \ref{thm:pisrecit}]
Let $G \sim G_I(n, k, q)$ where $S \subseteq [n]$ denotes the indices of the planted independent set in $G$ and satisfies $|S| = k$. Fix a vertex $i \in [n]$ and consider a random $J \in [n]$ where $J$ is chosen uniformly at random from $S$ if $i \not \in S$ and $J$ is chosen uniformly at random from $S^C$ if $i \in S$. Now consider the binary hypothesis testing problem with observations $(G, J, S \backslash \{i, J\})$ and the task of distinguishing between $H_0 : i \not\in S$ and $H_1 : i \in S$.

Since $S$ is chosen uniformly at random, it follows that the identity of the vertex $J$ is uniformly at random chosen from $[n] \backslash \{ i\}$ and is independent of the event $\{ i \in S\}$. It also holds that conditioned on $J$ and the outcome of the event $\{ i \in S\}$, the set $S \backslash \{i, J\}$ is uniformly distributed on $(k-1)$-subsets of $[n] \backslash \{i, J\}$. Therefore for any $J \in [n] \backslash \{i\}$ and $(k-1)$ subset $S\backslash \{i, J\}$ of $[n]\backslash \{i, J\}$, it holds that
\begin{align*}
\frac{\bP[G, J, S \backslash \{i, J\} | i \in S]}{\bP[G, J, S \backslash \{i, J\} | i \not\in S]} &= \frac{\bP[G, S \backslash \{i, J\} | i \in S, J]}{\bP[G, S \backslash \{i, J\} | i \not\in S, J]} \\
&= \frac{\bP[G | i \in S, J, S \backslash \{i, J\}]}{\bP[G | i \not\in S, J, S \backslash \{i, J\}]} = \prod_{k \in S \backslash \{i, J\}} \frac{(1 - A_{ik}) \cdot q^{A_{Jk}}(1 - q)^{1 - A_{Jk}}}{(1 - A_{Jk}) \cdot q^{A_{ik}}(1 - q)^{1 - A_{ik}}}
\end{align*}
where $A = A(G)$ is the adjacency matrix of $G$. From this factorization, it follows that the vector $v$ of values $A_{ik}$ and $A_{Jk}$ for all $k \in S\backslash \{i, J\}$ is therefore a sufficient statistic for this binary hypothesis testing problem. Furthermore, if $i \in S$ then $v$ has its first $k - 1$ coordinates equal to zero and its last $k - 1$ coordinates distributed as $\text{Bern}(q)^{\otimes (k - 1)}$. If $i \not \in S$, then $v$ has its first $k - 1$ coordinates distributed as $\text{Bern}(q)^{\otimes (k - 1)}$ and its last $k - 1$ coordinates equal to zero.  Thus the given hypothesis testing problem is equivalent to testing between these two distributions. Note that
$$\bP_{H_0}[v = 0] = \bP_{H_1}[v = 0] = (1 - q)^{k - 1} \ge 1 - (k - 1)q \to 1 \quad \text{as } n \to \infty$$
by Bernoulli's inequality if $k \ge 2$. Taking any coupling of $\mL_{H_0}(v)$ and $\mL_{H_1}(v)$ such that the events $\{ v = 0 \}$ under $H_0$ and $H_1$ coincide yields that
$$\TV\left( \mL_{H_0}(v), \mL_{H_1}(v) \right) \le 1 - (1 - q)^{k-1} \to 0 \quad \text{as } n \to \infty$$
Now let $p_{\text{E1}}$ and $p_{\text{E2}}$ be the optimal Type I and Type II error probabilities. Note that the prior on the hypotheses $H_0 : i \in S$ and $H_1 : i \not \in S$ is $\bP[i \in S] = k/n$. Let $\mathcal{E}$ be the optimal average probability of testing error under this prior. Also note that $p_{\text{E1}} + p_{\text{E2}} = 1 - \TV\left( \mL_{H_0}(v), \mL_{H_1}(v) \right) \to 1$ as $n \to \infty$.

Now assume for contradiction that there is some algorithm $A : \mG_n \to \binom{[n]}{k}$ that achieves weak recovery with $\bE[|S \cap A(G)|] = k - o(k)$ as $n \to \infty$. It follows that
\begin{align*}
k - \bE[|S \cap A(G)|] &= \sum_{i = 1}^n \bP\left[ \mathbf{1}_{\{i \in A(G)\}} \neq \mathbf{1}_{\{i \in S\}} \right] \ge \sum_{i = 1}^n \min_{\phi_i(G)} \bP\left[ \phi_i(G) \neq \mathbf{1}_{\{i \in S\}} \right] \\
&\ge \sum_{i = 1}^n \min_{\phi_i(G, J, S\backslash \{i, J\})} \bP\left[ \phi_i(G, J, S\backslash \{i, J\}) \neq \mathbf{1}_{\{i \in S\}} \right] = n \mathcal{E}
\end{align*}
The first minimum is over all functions $\phi_i : \mG_n \to \{0, 1\}$ that only observe the graph $G$, while the second minimum is over all functions $\phi_i$ that also observe $J$ and $S \backslash \{i, J\}$. From these inequalities, it must follow that $\mathcal{E} = o(k/n)$ which implies the test achieving $\mathcal{E}$ must have Type I and Type II errors that are $o(1)$ as $n \to \infty$ since $1 - \frac{k}{n} = \Omega(1)$. This implies that $p_{\text{E1}} + p_{\text{E2}} = o(1)$, which is a contradiction.
\end{proof}

\begin{proof}[Proof of Theorem \ref{thm:pdsdet}]
First suppose that $p > q$. Let $G$ be an instance of $\textsc{PDS}_D(n, k, p, q)$. Note that under $H_0$, the edge count $|E(G)| \sim \text{Bin}( \binom{n}{2}, q )$ and under $H_1$, $|E(G)|$ is the independent sum of $\text{Bin}( \binom{n}{2} - \binom{k}{2}, q )$ and $\text{Bin}( \binom{k}{2}, p )$. By Bernstein's inequality, we have that
\begin{align*}
\bP_{H_0}\left[ |E(G)| > \binom{n}{2} q + \binom{k}{2} \cdot \frac{p - q}{2} \right] &\le \exp\left( - \frac{\binom{k}{2}^2 ( p - q)^2/4}{2\binom{n}{2} q + \binom{k}{2} \cdot (p - q)/3} \right) \\
&= \exp\left( - \Omega\left( \frac{k^4}{n^2} \cdot \frac{(p-q)^2}{q(1 - q)} \right) \right)
\end{align*}
By the multiplicative Chernoff bound, it follows that
\begin{align*}
\bP_{H_1}\left[ |E(G)| \le \binom{n}{2} q + \binom{k}{2} \cdot \frac{p - q}{2} \right] &\le \exp\left( - \frac{\binom{k}{2}^2 (p - q)^2/4}{2\binom{n}{2} q + 2 \binom{k}{2} (p - q)}\right) \\
&= \exp\left( - \Omega\left( \frac{k^4}{n^2} \cdot \frac{(p-q)^2}{q(1 - q)} \right) \right)
\end{align*}
Now let $X$ be the maximum number of edges over all subgraphs of $G$ on $k$ vertices. By a union bound and Bernstein's inequality
\begin{align*}
\bP_{H_0}\left[ X \ge \binom{k}{2} \cdot \frac{p + q}{2} \right] &\le \sum_{R \in \binom{[n]}{k}} \bP_{H_0} \left[ \left|E\left(G[R]\right)\right| \ge \binom{k}{2} \cdot \frac{p + q}{2} \right] \\
&\le \left( \frac{en}{k} \right)^k \exp \left( - \frac{\binom{k}{2}^2 ( p - q)^2/4}{2\binom{k}{2} q + \binom{k}{2} \cdot (p - q)/3} \right) \\
&= \exp\left( k \log (en/k) - \Omega\left( k^2 \cdot \frac{(p-q)^2}{q(1 - q)} \right) \right)
\end{align*}
where the second inequality uses the fact that for any fixed $R$, $\left|E\left(G[R]\right)\right| \sim \text{Bin}(\binom{k}{2}, q)$. Under $H_1$, it holds that if $S$ is the vertex set of the latent planted dense subgraph then $|E(G[S])| \sim \text{Bin}(\binom{k}{2}, p)$. By the multiplicative Chernoff bound, it follows that
\begin{align*}
\bP_{H_1}\left[X < \binom{k}{2} \cdot \frac{p + q}{2} \right] &\le \bP_{H_1}\left[|E(G[S])| < \binom{k}{2} \cdot \frac{p + q}{2} \right] \\
&\le \exp\left( - \frac{\binom{k}{2}^2 (p - q)^2/4}{2\binom{k}{2} p} \right) \\
&= \exp\left( - \Omega\left( k^2 \cdot \frac{(p-q)^2}{q(1 - q)} \right) \right)
\end{align*}
Therefore the test that outputs $H_1$ if $|E(G)| > \binom{n}{2} q + \binom{k}{2} \cdot \frac{p - q}{2}$ or $X \ge \binom{k}{2} \cdot \frac{p + q}{2}$ and $H_0$ otherwise has Type I$+$II error tending to zero as $n \to \infty$ if one of the two given conditions holds. In the case when $p < q$, this test with inequalities reversed can be shown to have Type I$+$II error tending to zero as $n \to \infty$ by analogous concentration bounds.
\end{proof}

\begin{proof}[Proof of Theorem \ref{thm:maxros}]
This theorem follows from the Gaussian tail bound $1 - \Phi(t) \le \frac{1}{\sqrt{2\pi}} \cdot t^{-1} e^{-t^2/2}$ for all $t \ge 1$ and a union bound. Now observe that if $(i, j) \not \in \text{supp}(u) \times \text{supp}(v)$, then $M_{ij} \sim N(0, 1)$ and thus
$$\bP\left[ |M_{ij}| > \sqrt{6\log n} \right] = 2\left(1 - \Phi(\sqrt{6\log n})\right) \le \frac{2}{\sqrt{2\pi}} \cdot e^{-3\log n} = O(n^{-3})$$
If $(i, j) \not \in \text{supp}(u) \times \text{supp}(v)$, then $M_{ij} \sim N(\mu \cdot u_i v_j, 1)$ where $|\mu \cdot u_i v_j| \ge \sqrt{6 \log n}$ since $|u_i|, |v_j| \ge 1/\sqrt{k}$ by the definition of $\mathcal{V}_{n, k}$. This implies that
$$\bP\left[ |M_{ij}| \le \sqrt{6\log n} \right] \le \left(1 - \Phi(\sqrt{6\log n})\right) + \left(1 - \Phi(3\sqrt{6\log n})\right) \le \frac{2}{\sqrt{2\pi}} \cdot e^{-3\log n} = O(n^{-3})$$
Now the probability that the set of $(i, j)$ with $|M_{ij}| > \sqrt{6\log n}$ is not exactly $\text{supp}(u) \times \text{supp}(v)$ is, by a union bound, at most
$$\sum_{(i, j) \in \text{supp}(u) \times \text{supp}(v)} \bP\left[ |M_{ij}| \le \sqrt{6\log n} \right] + \sum_{(i, j) \not\in \text{supp}(u) \times \text{supp}(v)} \bP\left[ |M_{ij}| > \sqrt{6\log n} \right] = O(n^{-1})$$
which completes the proof of the theorem.
\end{proof}

\begin{proof}[Proof of Theorem \ref{thm:rosrec}]
Let $\mathcal{P}_u$ denote the projection operator onto the vector $u$ and let $\| r \|_0 = k_1$ and $\| c \|_0 = k_2$. By the definition of $\mathcal{V}_{n, k}$, it follows that
$$k\left( 1 - \frac{1}{\log k} \right) \le k_1, k_2 \le k$$
By the argument in Lemma 1 of \cite{cai2015computational}, there are constants $C_2, C_3 > 0$ such that
\begin{align*}
\| \mathcal{P}_{U} B_{\cdot j} - \mu c_j r \|_2 &\le C_3 \sqrt{\log n} + C_3\sqrt{\frac{n}{k_1}} \\
\| \mathcal{P}_{V} B_{i \cdot}^\top - \mu r_i c \|_2 &\le C_3 \sqrt{\log n} + C_3 \sqrt{\frac{n}{k_2}}
\end{align*}
hold for all $1 \le i, j \le n$ with probability at least $1 - 2n^{-C_2} - 2\exp(-2C_2n)$. Now note that if $j \in \text{supp}(c)$ and $j' \not \in \text{supp}(c)$, then it follows that
$$\| \mu c_j r - \mu c_{j'} r \|_2 = \mu \cdot |c_j| \ge \frac{\mu}{\sqrt{k}}$$
by the definition of $\mathcal{V}_{n, k}$. Similarly, if $i \in \text{supp}(r)$ and $i' \not \in \text{supp}(r)$ then $\| \mu r_i c - \mu r_{i'} c \|_2 \ge \mu/\sqrt{k}$. Therefore if for both $i = 1, 2$
$$\frac{\mu}{\sqrt{k}} \ge 6 C_3 \left( \sqrt{\log n} + \sqrt{\frac{n}{k_i}} \right)$$
then it holds that
$$2 \max_{j, j' \in \text{supp}(c)} \left\| \mathcal{P}_U B_{\cdot j} - \mathcal{P}_U B_{\cdot j'} \right\|_2 \le \max_{j\in \text{supp}(c), j' \not \in \text{supp}(c)} \left\| \mathcal{P}_U B_{\cdot j} - \mathcal{P}_U B_{\cdot j'} \right\|_2$$
$$2 \max_{i, i' \in \text{supp}(c)} \left\| \mathcal{P}_V B_{i \cdot}^\top - \mathcal{P}_U B_{i' \cdot}^\top \right\|_2 \le \max_{i\in \text{supp}(c), i' \not \in \text{supp}(c)} \left\| \mathcal{P}_V B_{i \cdot}^\top - \mathcal{P}_U B_{i' \cdot}^\top \right\|_2$$
and Steps 3 and 4 succeed in recovering $\text{supp}(r)$ and $\text{supp}(c)$.
\end{proof}

\begin{proof}[Proof of Theorem \ref{thm:rossvd}]
Under $H_0$, it holds that $M \sim N(0, 1)^{\otimes n \times n}$. By Corollary 5.35 in \cite{vershynin2010introduction}, we have
$$\sigma_1(M) \le 2 \sqrt{n} + \sqrt{2\log n}$$
with probability at least $1 - 2n^{-1}$. Now consider the case of $H_1$ and suppose that $M = \mu \cdot rc^\top + Z$ where $r, c \in \mathcal{V}_{n, k}$ and $Z \sim N(0, 1)^{\otimes n \times n}$. By Weyl's interlacing inequality, it follows that
$$|\mu - \sigma_1(M)| = |\sigma_1(\mu \cdot rc^\top) - \sigma_1(M)| \le \sigma_1(Z) \le 2\sqrt{n} + \sqrt{2\log n}$$
with probability at least $1 - 2n^{-1}$. If $\mu > 4 \sqrt{n} + 2 \sqrt{2 \log n}$ then the Type I$+$II error of the algorithm is at most $4n^{-1}$, proving the theorem.
\end{proof}

\begin{proof}[Proof of Theorem \ref{thm:ssbmalg}]
Suppose that $G$ is drawn from some distribution in $H_1$ and that the two hidden communities have index sets $S_1, S_2 \subseteq [n]$ where $k_1 = |S_1|$ and $k_2 = |S_2|$. For the remainder of the analysis of $H_1$, consider $A$ and $G$ conditioned on $S_1$ and $S_2$. Now let $v$ be the vector
$$v_i = \left\{ \begin{matrix} \frac{1}{\sqrt{k_1 + k_2}} & \text{if } i \in S_1 \\ -\frac{1}{\sqrt{k_1 + k_2}} & \text{if } i \in S_2 \\ 0 & \text{otherwise} \end{matrix} \right.$$
for each $i \in [n]$. Now observe that
$$v^\top (A - qJ) v = \frac{2}{k_1 + k_2} \left( \sum_{(i, j) \in S_1^2 \cup S_2^2 : i < j} \left( \mathbf{1}_{\{i, j\} \in E(G)} - q \right) + \sum_{(i, j) \in S_1 \times S_2} \left( q - \mathbf{1}_{\{i, j\} \in E(G)} \right) \right)$$
By the definition of $H_1$ in $\textsc{SSBM}_D$, the expression above is the sum of $\binom{k_1 + k_2}{2}$ independent shifted Bernoulli random variables each with expectation at least $\rho$. Therefore it follows that
$$\bE\left[ v^\top (A - qJ) v \right] \ge \frac{2}{k_1 + k_2} \cdot \binom{k_1 + k_2}{2} \cdot \rho = (k_1 + k_2 - 1)\rho \ge 3\sqrt{n}$$
since $k_1 + k_2 - 1 \ge k - 2k^{1 - \delta_{\textsc{SSBM}}} - 1 \ge \frac{k}{2}$ for sufficiently large $k$, as defined in Section 2.2. Now note that each of the centered random variables $\mathbf{1}_{\{i, j\} \in E(G)} - \bE[\mathbf{1}_{\{i, j\} \in E(G)}]$ and $\bE[\mathbf{1}_{\{i, j\} \in E(G)}] - \mathbf{1}_{\{i, j\} \in E(G)}$ are bounded in $[-1, 1]$ and therefore Bernstein's inequality implies that for all $t > 0$,
$$\bP\left[ v^\top (A - qJ) v < \bE[v^\top (A - qJ) v] - \frac{2t}{k_1 + k_2} \right] \le \exp\left( - \frac{\frac{1}{2} t^2}{n + \frac{1}{3} t} \right)$$
Note that $v$ is a unit vector and thus $\lambda_1(A - qJ) \ge v^\top (A - qJ) v$. Setting $t = \frac{1}{2} (k_1 + k_2) \sqrt{n}$ now yields that
$$\bP\left[ \lambda_1(A - qJ) < 2\sqrt{n} \right] \le \exp\left( - \frac{\frac{1}{8} (k_1 + k_2)^2 n}{n + \frac{1}{6} (k_1 + k_2) \sqrt{n}} \right) = \exp\left( - \Omega(n) \right)$$
since $k_1 + k_2 \ge \frac{k}{2} = \Omega(\sqrt{n})$ for sufficiently large $k$. Now suppose that $G$ is drawn from $G(n, q)$ as in $H_0$. By Theorem 1.5 in \cite{vu2005spectral}, it follows that with probability $1 - o_n(1)$, we have that
$$\lambda_1(A - qJ) \le 2\sqrt{q(1 - q)n} + C(q - q^2)^{1/4} n^{1/4} \log n$$
for some constant $C > 0$. Therefore $\lambda_1(A - qJ)$ is less than $2 \sqrt{n}$ for sufficiently large values of $n$ since $q(1 - q) \le 1/4$. Therefore the Type I$+$II error of this algorithm on $\textsc{SSBM}_D$ is $o_n(1) + \exp\left( - \Omega(n) \right) = o_n(1)$ as $n \to \infty$.
\end{proof}

\begin{proof}[Proof of Theorem \ref{thm:bscpait}]
Let $u_S$ denote the $d$-dimensional unit vector with entries in $S$ equal to $1/\sqrt{k}$ and all other entries equal to zero where $S$ is some $k$-subset of $d$. Let $v$ be the distribution on matrices $\theta u_Su_S^\top$ where $\theta \le 1$ and $S$ is chosen uniformly at random from the set of $k$-subsets of $[d]$. Note that $\mL_{H_0}(X) = N(0, I_d)^{\otimes n}$ and $\mL_{H_1}(X) = \bE_v[N(0, I_d + \theta u_S u_S^\top)^{\otimes n}]$. By Lemma \ref{lem:chispca}, it follows that
\begin{align*}
\chi^2\left( \mL_{H_1}(X), \mL_{H_0}(X) \right) &=  \bE \left[ \det\left(I_d - \theta^2 u_S u_S^\top u_T u_T^\top\right)^{-n/2} \right] - 1 \\
&= \bE\left[ \left( 1 - \frac{\theta^2}{k^2} \cdot |S \cap T|^2 \right)^{-n/2} \right] - 1 \\
&\le \bE\left[ \exp\left( \frac{n\theta^2}{k^2} \cdot |S \cap T|^2 \right) \right] - 1
\end{align*}
where $S$ and $T$ are independent random $k$-subsets of $[d]$. The last inequality above follows from the fact that $(1 - t)^{-1/2} \le e^t$ if $t \le 1/2$, $\theta^2 \le 1/2$ and $|S \cap T| \le k$. Now note that $|S \cap T| \sim \text{Hypergeometric}(d, k, k)$ and let
$$b = \frac{n\theta^2}{k^2} \cdot \left( \max\left\{ \frac{1}{k} \log \left( \frac{ed}{k} \right), \frac{d^2}{k^4} \right\} \right)^{-1}$$
The given condition on $\theta$ implies that $b \le \beta_0$. It follows by Lemma \ref{lem:hgm} that $\chi^2\left( \mL_{H_1}(X), \mL_{H_0}(X) \right) \le \tau(\beta_0) - 1$ and by Cauchy-Schwarz that if $w(\beta_0) = \frac{1}{2} \sqrt{\tau(\beta_0) - 1}$ then
$$\TV\left( \mL_{H_0}(X), \mL_{H_1}(X) \right) \le \frac{1}{2}\sqrt{\chi^2\left( \mL_{H_1}(X), \mL_{H_0}(X) \right)} \le w(\beta_0)$$
where $w(\beta_0) \to 0$ as $\beta_0 \to 0^+$, proving the theorem.
\end{proof}

\begin{proof}[Proof of Theorem \ref{thm:weakrecspca}]
Let $u = v^k_{\max}\left(\hat{\Sigma}(X) \right) - v$ and note that
$$\| u \|_2^2 = \| v \|_2^2 + \left\| v^k_{\max}\left(\hat{\Sigma}(X) \right) \right\|_2^2 - 2 \left\langle v, v^k_{\max}\left(\hat{\Sigma}(X) \right) \right\rangle = 2 \cdot L\left( v^k_{\max}\left(\hat{\Sigma}(X)\right), v\right)^2$$
If $i \in \text{supp}(v)$ where $v \in \mathcal{V}_{d, k}$, then $|v|_i \ge \frac{1}{\sqrt{k}}$. Therefore each $i \in S(X) \Delta \textnormal{supp}(v)$ satisfies that $|u|_i \ge \frac{1}{2\sqrt{k}}$, which implies that
$$\frac{1}{k} \left|S(X) \Delta \textnormal{supp}(v) \right| \le 4 \cdot \sum_{i \in S(X) \Delta \textnormal{supp}(v)} |u|_i^2 \le 4 \| u \|_2^2 \le 8\sqrt{2} \cdot L\left( v^k_{\max}\left(\hat{\Sigma}(X)\right), v\right)$$
using the fact that $L(u, v) \le \sqrt{2}$ if $\| u \|_2 = \| v \|_2 = 1$. This inequality along with the previous theorem completes the proof.
\end{proof}

\begin{proof}[Proof of Theorem \ref{thm:bspcadet}]
First assume that $H_0$ holds and $X \sim N(0, I_d)^{\otimes n}$. Observe that
$$\frac{n}{d} \cdot \mathbf{1}^\top \hat{\Sigma}(X) \mathbf{1} = \frac{1}{d}\sum_{i = 1}^n \langle \mathbf{1}, X_i \rangle^2$$
where the values $\frac{1}{\sqrt{d}} \langle \mathbf{1}, X_i \rangle$ are independent and distributed as $N(0, 1)$. Therefore $\frac{n}{d} \cdot \mathbf{1}^\top \hat{\Sigma}(X) \mathbf{1}$ is distributed as a $\chi^2$ distribution with $n$ degrees of freedom. Since $\frac{1}{d} \langle \mathbf{1}, X_i \rangle^2 - 1$ is zero-mean and sub-exponential with norm $1$, Bernstein's inequality implies that for all $t \ge 0$
$$\bP\left[ \frac{n}{d} \cdot \mathbf{1}^\top \hat{\Sigma}(X) \mathbf{1} \ge n + t \right] \le 2 \exp\left( - c \cdot \min \left( \frac{t^2}{n}, t \right) \right)$$
for some constant $c > 0$. Substituting $t = \frac{2n \delta^2 k \theta}{d} \le n$ yields that
$$\bP\left[ \mathbf{1}^\top \hat{\Sigma}(X) \mathbf{1} \ge d + 2\delta^2 k \theta \right] \le 2 \exp\left( - c \cdot \min \left( \frac{4n \delta^4 k^2 \theta^2}{d^2}, \frac{2n \delta^2 k \theta}{d} \right) \right) = 2 \exp\left( - \frac{4cn \delta^4 k^2 \theta^2}{d^2} \right)$$
which tends to zero as $n \to \infty$. Now assume that $H_1$ holds and $X \sim N(0, I_d + \theta vv^\top)^{\otimes n}$ for some $v \in \mathcal{BV}_{d, k}$. Note that each $X_i$ can be written as $X_i = \sqrt{\theta} \cdot g_i v + Z_i$ where $g_1, g_2, \dots, g_n \sim_{\text{i.i.d.}} N(0, 1)$ and $Z_1, Z_2, \dots, Z_n \sim_{\text{i.i.d.}} N(0, I_d)$. If $s(v) = \sum_{j = 1}^d v_j$ is the sum of the entries of $v$, then
$$\langle \mathbf{1}, X_i \rangle = \sqrt{\theta} \cdot g_i s(v) + \sum_{j = 1}^d Z_{ij} \sim N\left(0, d + \theta s(v)^2\right)$$
Furthermore, these inner products are independent for $i = 1, 2, \dots, n$. Therefore $\frac{n}{d + \theta s(v)^2} \cdot \mathbf{1}^\top \hat{\Sigma}(X) \mathbf{1}$ is also distributed as a $\chi^2$ distribution with $n$ degrees of freedom. Since $v \in \mathcal{BV}_{d, k}$, it either follows that $|\text{supp}_+(v)| \ge \left( \frac{1}{2} + \delta \right)k$ or $|\text{supp}_-(v)| \ge \left( \frac{1}{2} + \delta \right)k$. If $|\text{supp}_+(v)| \ge \left( \frac{1}{2} + \delta \right)k$, then by Cauchy-Schwarz we have that
\begin{align*}
s(v) &= \sum_{i \in \text{supp}_+(v)} v_i - \sum_{i \in \text{supp}_-(v)} |v_i| \ge \left( \frac{1}{2} + \delta \right)\sqrt{k} - \left( \sum_{i \in \text{supp}_-(v)} |v_i|^2 \right)^{1/2} \cdot \sqrt{|\text{supp}_-(v)|} \\
&= \left( \frac{1}{2} + \delta \right)\sqrt{k} - \left( 1 - \sum_{i \in \text{supp}_+(v)} |v_i|^2 \right)^{1/2} \cdot \sqrt{|\text{supp}_-(v)|} \\
&\ge \left( \frac{1}{2} + \delta \right)\sqrt{k} - \left( 1 - \frac{|\text{supp}_+(v)|}{k} \right)^{1/2} \cdot \sqrt{\left( \frac{1}{2} - \delta \right)k} \ge 2 \delta \sqrt{k}
\end{align*}
Bernstein's inequality with $t = \frac{2n \delta^2 k \theta}{d + \theta s(v)^2} \le n$ now implies that
$$\bP\left[ \mathbf{1}^\top \hat{\Sigma}(X) \mathbf{1} \le d + 2\delta^2 k \theta \right] \le 2 \exp\left( - \frac{4cn \delta^4 k^2 \theta^2}{(d + \theta s(v)^2)^2} \right)$$
which tends to zero as $n \to \infty$ since $\theta s(v)^2 \le \theta k \le \frac{d}{2\delta^2}$ by Cauchy-Schwarz. This completes the proof of the theorem.
\end{proof}

\begin{proof}[Proof of Theorem \ref{thm:spectralspca}]
First observe that $\hat{\Sigma}(X) = \frac{1}{n} XX^\top$ and thus $\lambda_1(\hat{\Sigma}(X)) = \frac{1}{n} \sigma_1(X)^2$. Under $H_0$, it follows that $X \sim N(0, 1)^{\otimes d \times n}$. By Corollary 5.35 in \cite{vershynin2010introduction}, it follows that
$$\bP\left[ \lambda_1(\hat{\Sigma}(X)) > 1 + 2\sqrt{c} \right] \le \bP\left[ \sigma_1(X) > \sqrt{n} + 2\sqrt{d} \right] \le 2e^{-d/2}$$
Under $H_1$, suppose that $X \sim N(0, I_d + \theta vv^\top)^{\otimes n}$ where $v \in \mathcal{V}_{d, k}$. As in the proof of Theorem \ref{thm:bspcadet}, write $X_i = \sqrt{\theta} \cdot g_i v + Z_i$ where $g_1, g_2, \dots, g_n \sim_{\text{i.i.d.}} N(0, 1)$ and $Z_1, Z_2, \dots, Z_n \sim_{\text{i.i.d.}} N(0, I_d)$. Now observe that
$$v^\top \hat{\Sigma}(X) v = \frac{1}{n} \sum_{i = 1}^n \langle v, X_i \rangle^2 = \frac{1}{n} \sum_{i = 1}^n \left( \sqrt{\theta} \cdot g_i + \langle v, Z_i \rangle \right)^2$$
Note that since $\| v \|_2 = 1$, it holds that $\sqrt{\theta} \cdot g_i + \langle v, Z_i \rangle \sim N(0, 1 + \theta)$ and are independent for $i = 1, 2, \dots, n$. Now note that $(1 + \theta)^{-1} \langle v, X_i \rangle^2 - 1$ is zero-mean and sub-exponential with norm $1$. Therefore Bernstein's inequality implies that
\begin{align*}
\bP\left[ \lambda_1(\hat{\Sigma}(X)) \le 1 + 2\sqrt{c} \right] &\le \bP\left[ v^\top \hat{\Sigma}(X) v \le 1 + 2\sqrt{c} \right] \\
&\le \bP\left[ \sum_{i = 1}^n \left[ (1 + \theta)^{-1} \langle v, X_i \rangle^2 - 1 \right] \le - \frac{2n\sqrt{c}}{1 + \theta} \right] \\
&\le 2 \exp \left( - c_1 \cdot \frac{4cn}{(1 + \theta)^2} \right) \to 0 \text{ as } n \to \infty
\end{align*}
for some constant $c_1 > 0$ and since $\frac{2n\sqrt{c}}{1 + \theta} < n$. This completes the proof of the theorem.
\end{proof}

\begin{proof}[Proof of Theorem \ref{thm:spcaspectralalg}]
Let $\hat{v} = g \cdot v + \sqrt{1 - g^2} \cdot u$ where $u$ is the projection of $\hat{v}$ onto the space orthogonal to $v$. Also assume that $g \in [0, 1]$, negating $v$ if necessary. By Theorem 4 in \cite{paul2007asymptotics}, it follows that
$$g \to \sqrt{\frac{\theta^2 - c}{\theta^2 + \theta c}} \quad \text{with probability } 1 - o(1) \text{ as } n \to \infty$$
By Theorem 6 in \cite{paul2007asymptotics}, $u$ is distributed uniformly on the $(d - 1)$-dimensional unit sphere of vectors in $\mathbb{R}^d$ orthogonal to $v$. By Lemma 4.1 in \cite{krauthgamer2015semidefinite}, it holds that $|u_i| \le h \sqrt{\frac{\log d}{d}}$ for all $i \in [d]$ with probability tending to one. Condition on this event. Since $\frac{k\log d}{d} \to 0$, it follows that $\frac{\log d}{d} = o\left( \sqrt{\frac{\log d}{kd}} \right)$. Therefore with probability tending to one, each $i \not \in \text{supp}(v)$ satisfies that $|\hat{v}_i|^4 < \frac{\log d}{kd}$ for sufficiently large $n, k$ and $d$. Now note that each $i \in \text{supp}(v)$ satisfies that
$$|v_i| \ge \frac{g}{\sqrt{k}} - \sqrt{1 - g^2} \cdot |u_i| \ge \frac{g}{\sqrt{k}} - h \sqrt{\frac{\log d}{d}} \ge \sqrt[4]{\frac{\log d}{kd}}$$
for sufficiently large $n, k$ and $d$ with probability tending to one. This is because $g = \Omega(1)$ with probability $1 - o(1)$ and $\frac{1}{\sqrt{k}} = \omega\left( \sqrt{\frac{\log d}{kd}} \right)$. Therefore it follows that $S = \text{supp}(v)$ with probability tending to one as $n \to \infty$.
\end{proof}

\section{Deferred Proofs from Section 10}
\label{app10}

\begin{proof}[Proof of Lemma \ref{lem:pdscloning}]
We first show that the distributions in Step 2 of $\textsc{PDS-Cloning}$ are well-defined. First note that both are normalized and thus it suffices to verify that they are nonnegative. First suppose that $p > q$, then the distributions are well-defined if
$$\frac{1 - p}{1 - q} \le \left( \frac{P}{Q} \right)^{|v|_1} \left( \frac{1 - P}{1 - Q} \right)^{2 - |v|_1} \le \frac{p}{q}$$
for all $v \in \{0, 1\}^2$, which follows from the assumption on $P$ and $Q$. Now observe that if $\mathbf{1}_{\{i, j \} \in E(G)} \sim \text{Bern}(q)$ then $x^{ij}$ has distribution
$$\bP[x^{ij} = v] = q \cdot \bP[x^{ij} = v | \{i, j\} \in E(G)] + (1 - q) \cdot \bP[x^{ij} = v | \{i, j\} \not \in E(G)] = Q^{|v|_1}(1 - Q)^{2 - |v|_1}$$
and hence $x^{ij} \sim \text{Bern}(Q)^{\otimes 2}$. Similarly, if $\mathbf{1}_{\{i, j \} \in E(G)} \sim \text{Bern}(p)$ then $x^{ij} \sim \text{Bern}(P)^{\otimes 2}$. It follows that if $G \sim G(n, q)$ then $(G^1, G^2) \sim G(n, Q)^{\otimes 2}$ and if $G \sim G(n, S, p, q)$ then $(G^1, G^2) \sim G(n, S, P, Q)^{\otimes 2}$, proving the lemma.
\end{proof}

\begin{proof}[Proof of Lemma \ref{lem:gausscloning}]
Since the entries of $M$ and $G$ are independent, the $\sigma$-algebras $\sigma \{ M^1_{ij}, M^2_{ij} \}$ for $i, j \in [n]$ are independent. Now note that $M^1_{ij}$ and $M^2_{ij}$ are jointly Gaussian and $\bE[M^1_{ij} M^2_{ij}] = \frac{1}{2} \cdot \bE[M^2 - G^2] = 0$, which implies that they are independent. It follows that $M^k_{ij}$ are independent for $k = 1, 2$ and $i, j \in [n]$. The lemma follows from the fact that each of $M^1$ and $M^2$ is identically distributed to $\frac{1}{\sqrt{2}} A + N(0, 1)^{\otimes n \times n}$. 
\end{proof}

\end{document}